\documentclass[twoside,runningheads]{llncs}
\makeatletter
\RequirePackage[bookmarks,unicode,colorlinks=true]{hyperref}%
   \def\@citecolor{blue}%
   \def\@urlcolor{blue}%
   \def\@linkcolor{blue}%

\def\orcidID#1{\smash{\href{http://orcid.org/#1}{\protect\raisebox{-1.25pt}{\protect\includegraphics{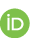}}}}}
\makeatother
\RequirePackage[T1]{fontenc}
\usepackage{xcolor,colortbl}

\usepackage{amsthm}
\usepackage{amssymb}
\usepackage{amssymb}
\usepackage{booktabs}
\usepackage{tabularx}
\usepackage{verbatim}
\usepackage{framed}
\usepackage{tikz}
\usetikzlibrary{shapes,snakes}
\usepackage{pgfplots}
\usepackage{alltt}
\usetikzlibrary{arrows}
\usetikzlibrary{calc}
\usetikzlibrary{fit}
\usetikzlibrary{positioning,shadows}
\usetikzlibrary{automata}
\usetikzlibrary{shapes,arrows}
\usetikzlibrary{decorations.text}
\usetikzlibrary{decorations.markings}
\usetikzlibrary{trees,snakes}
\usepackage{graphicx}
\usepackage{paralist}
\usepackage{proof}
\usepackage{lscape}
\usepackage{textcomp}
\usepackage{hyperref}
\usepackage[bbsets,Dfprime]{math}
\usepackage[prefixflatinterpret,bracketmodalinterpret,modernsign,substopindex,shortmquant,mquantifiertype,mconnectiveformal,bracketinterpret,fixformat,setfixinterpret,modifopindex,seqarrow,seqoptional,sidenotecalculus,abbrseqcontext,shortterms,nosigmaterms,novarterms]{logic}
\usepackage[pretest,nocommandblocks]{progreg}
\usepackage[bracketinterpret,prefixflatinterpret,bracketmodalinterpret,fixformat,differentialdL]{dL}

\let\citet\cite
\newif\iflongversion
\longversiontrue

\usepackage{booktabs} 

\irlabel{brandomI|{$[{:}{*}]$}}
\irlabel{bchoiceI|{$[\cup]$}}
\irlabel{btestI|{$[?]$}}
\irlabel{broll|{$[*]$r}}
\irlabel{bunroll|{$[*]${E}}}
\irlabel{bloopI|{$[*]$I}}
\irlabel{bsolve|{$[']$}}

\irlabel{di|DI}
\irlabel{dc|DC}
\irlabel{dw|DW}
\irlabel{dg|DG}
\irlabel{dv|DV}
\irlabel{dsolve|{$\langle'\rangle$}}
\irlabel{dloopI|{$\langle*\rangle$}I}
\irlabel{dstop|{$\langle*\rangle$}S}
\irlabel{dgo|{$\langle*\rangle$}G}
\irlabel{dchoiceIL|{$\langle\cup\rangle$}IL}
\irlabel{dchoiceIR|{$\langle\cup\rangle$}IR}
\irlabel{dtestI|{$\langle?\rangle$}I}
\irlabel{drandomI|{$\langle{:}*\rangle$}I}
\irlabel{seqI|{$\lstrike{;}\rstrike$}I}
\irlabel{seqE|{$\lstrike{;}\rstrike$}E}
\irlabel{asgnI|{$\lstrike{{:}{=}}\rstrike$}I}
\irlabel{asgnE|{$\lstrike{{:}{=}}\rstrike$}E}
\irlabel{drcase|{$\langle*\rangle${C}}}
\irlabel{hyp|hyp}
\irlabel{mon|M}

\newcommand{\stepsto}{\allowbreak\mapsto\allowbreak}
\newcommand{\bebecomes}{\mathrel{::=}}
\newcommand{\alternative}{~|~}

\newcommand{\sren}[3]{\urename[{#1}]{#2}{#3}}
\newcommand{\ssub}[3]{\subst[{#1}]{#2}{#3}}
\newcommand{\eren}[3]{\urename[{#1}]{#2}{#3}}
\newcommand{\earen}[2]{\urename[{#1}]{\boundvars{#2}}{\vec{y}}}

\usepackage{prettyref}
\newcommand{\rref}[2][]{\prettyref{#2}}
\newrefformat{ch}{Chapter\,\ref{#1}}
\newrefformat{sec}{Section\,\ref{#1}}
\newrefformat{app}{Appendix\,\ref{#1}}
\newrefformat{ass}{Assumption\,\ref{#1}}
\newrefformat{def}{Definition\,\ref{#1}}
\newrefformat{thm}{Theorem\,\ref{#1}}
\newrefformat{prop}{Proposition\,\ref{#1}}
\newrefformat{lem}{Lemma\,\ref{#1}}
\newrefformat{rem}{Remark\,\ref{#1}}
\newrefformat{cor}{Corollary\,\ref{#1}}
\newrefformat{ex}{Example\,\ref{#1}}
\newrefformat{tab}{Table\,\ref{#1}}
\newrefformat{fig}{Figure\,\ref{#1}}
\newrefformat{case}{case\,\ref{#1}}
\newrefformat{fn}{\ref{#1}}

\newcommand{\allrat}{\mathbb{Q}}
\newcommand{\allbool}{\mathbb{B}}

\newcommand{\allstate}{\sty}

\renewcommand*{\freevars}[1]{\mathop{\text{FV}}(#1)}
\renewcommand*{\boundvars}[1]{\mathop{\text{BV}}(#1)}
\renewcommand*{\mustboundvars}[1]{\mathop{\text{MBV}}(#1)}

\newcommand{\logname}{\text{\upshape\textsf{dH{\kern-0.05em}L}}\xspace}

\newcommand{\lequiv}{\leftrightarrow}

\newcommand{\om}{\omega}
\newcommand{\tom}{\tilde{\omega}}

\newcommand{\DL}{\textsf{DL}\xspace}
\newcommand{\GL}{\textsf{GL}\xspace}
\newcommand{\CGL}{\textsf{CGL}\xspace}

\newcommand{\tint}[2]{\lenvelope{#1}\renvelope{#2}}
\newcommand{\fint}[1]{\lenvelope{#1}\renvelope}
\newcommand{\fintR}[1]{\fint{#1}}

\newcommand{\rzNil}{\text{\rm\tt()}}
\newcommand{\rzCons}[2]{(#1,#2)}
\newcommand{\rzBLam}[2]{\lambda #1:\allstate.~#2}
\newcommand{\rzHOLam}[3]{\Lambda #1:\rzfor{#2}.~#3}
\newcommand{\rzFOLam}[3]{\Lambda #1:#2.~#3}
\newcommand{\rzApp}[2]{#1\,#2}
\newcommand{\rzIf}[3]{\text{\tt{if}}\:(#1)\:#2\:\text{\tt{else}}\:#3}

\newcommand{\myaa}{a}
\newcommand{\ab}{b}
\newcommand{\ac}{c}
\newcommand{\ad}{d}
\newcommand{\taa}{\tilde{\aa}}
\newcommand{\tab}{\tilde{\ab}}

\newcommand{\rzInd}[4]{\textsf{ind}(#1,~{#2}.{#3},~{#2}.{#4}{})}
\newcommand{\rzCoind}[4]{\textsf{gen}(#1,~{#2}.{#3},~{#2}.{#4})}

\newcommand{\allRz}{\mathcal{R}\mathbf{z}}
\newcommand{\rzfor}[1]{#1\,\allRz}

\newcommand*{\strategyforR}[2][]{{#2}\Langle{#1}\Rangle}
\newcommand*{\istrat}[3][]{\strategyforR[#1]{#2}^{#3}}

\newcommand*{\dstrategyforR}[2][]{{#2}\lenvelopewide{#1}\renvelopewide}
\newcommand*{\idstrat}[3][]{\dstrategyforR[#1]{#2}^{#3}}
\newcommand{\cint}[1]{\fint{\bigwedge #1}}
\newcommand{\cintR}[1]{\fintR{\bigwedge #1}}
\newcommand{\seq}[2]{#1 \vdash #2}

\newcommand{\proves}[3]{#1\allowbreak\vdash #2 \allowbreak \mathop{:}\, #3}

\newcommand{\edOpencons}{\langle}
\newcommand{\edClosecons}{\rangle}
\newcommand{\edSepcons}{,}
\newcommand{\edcons}[2]{\edOpencons{#1}\edSepcons{#2}\edClosecons}
\newcommand{\ebOpencons}{[}
\newcommand{\ebSepcons}{,}
\newcommand{\ebClosecons}{]}
\newcommand{\ebcons}[2]{\ebOpencons#1\ebSepcons#2\ebClosecons}
\newcommand{\eCons}[2]{\lstrike{#1,#2}\rstrike}

\newcommand{\pmodality}[2]{\llensehat{#1}\rlensehat{#2}}
\newcommand{\econs}[2]{\eCons{#1}{#2}}
\newcommand{\einjL}[1]{\ell \cdot #1}
\newcommand{\einjR}[1]{r \cdot #1}
\newcommand{\kwcase}{\textsf{case}}
\newcommand{\ecaseHead}[1]{\langle\kwcase\textrm{ }#1{\text{\textrm{ of }}}}
\newcommand{\ecaseLeft}[2]{#1\Rightarrow~#2}
\newcommand{\ecaseRight}[2]{~|~#1\Rightarrow~#2}
\newcommand{\ecaseEnd}{\rangle}
\newcommand{\ecasegen}[5]{\ecaseHead{#1}\allowbreak\ecaseLeft{#2}{#3}\allowbreak\ecaseRight{#4}{#5}\ecaseEnd}
\newcommand{\ecase}[3]{\ecasegen{#1}{\ell}{#2}{r}{#3}}
\newcommand{\edcase}[3]{\ecase{#1}{#2}{#3}}
\newcommand{\ercase}[3]{\langle\textsf{case}_*\ #1\text{ of }\pvs\Rightarrow~#2~|~\pvg\Rightarrow#3\rangle}
\newcommand{\eCase}[3]{\lstrike\kwcase\textsf{ }#1\textsf{ of }\allowbreak\ell\Rightarrow~#2~\allowbreak|~r\Rightarrow~#3\rstrike} 
\newcommand{\edinjL}[1]{\langle\ell \cdot #1\rangle}
\newcommand{\edinjR}[1]{\langle r \cdot #1\rangle}
\newcommand{\eInjL}[1]{\lstrike\ell \cdot #1\rstrike}
\newcommand{\eInjR}[1]{\lstrike{r \cdot #1}\rstrike}
\newcommand{\kwrep}{\textsf{rep}}
\newcommand{\erep}[4]{{#1}\text{ }\kwrep\text{ }#3.~{#2}\textsf{ in }{#4}}
\newcommand{\eapp}[2]{#1\ #2}

\newcommand{\elam}[3]{\lambda #1:#2. \:\, #3}
\newcommand{\eplam}[2]{\elam{\pvx}{#1}{#2}}
\newcommand{\etlam}[2]{\elam{x}{#1}{#2}}
\newcommand{\ebseq}[1]{[\iota~#1]}
\newcommand{\edseq}[1]{\langle\iota~#1\rangle}
\newcommand{\ebOpenseq}{[\iota~}
\newcommand{\ebCloseseq}{]}
\newcommand{\edOpenseq}{\langle\iota~}
\newcommand{\edCloseseq}{\rangle}

\newcommand{\eSeq}[1]{\lstrike\iota~#1\rstrike}

\newcommand{\ebOpenswap}{{[}\textsf{yield }}
\newcommand{\ebCloseswap}{{]}}
\newcommand{\edOpenswap}{\langle\textsf{yield }}
\newcommand{\edCloseswap}{\rangle}
\newcommand{\ebswap}[1]{\ebOpenswap#1\ebCloseswap}
\newcommand{\edswap}[1]{\edOpenswap#1\edCloseswap}
\newcommand{\eSwap}[1]{\lstrike\textsf{yield }#1\rstrike}

\newcommand{\emonInfix}[1]{{\circ_{#1}}}
\newcommand{\emon}[3]{{#1} \emonInfix{#3} {#2}}
\newcommand{\eQE}[2]{\textsf{FO}[#1](#2)}
\newcommand{\esplit}[3]{\textsf{split }{[#1,#2]}~#3}
\newcommand{\kwstop}{\textsf{stop}}
\newcommand{\kwgo}{\textsf{go}}
\newcommand{\estop}[1]{\langle\kwstop\ #1\rangle}
\newcommand{\ego}[1]{\langle\kwgo\ #1\rangle}
\newcommand{\kwroll}{\textsf{roll}}
\newcommand{\kwunroll}{\textsf{unroll}}
\newcommand{\ebroll}[1]{[\kwroll\ {#1}]}
\newcommand{\ebunroll}[1]{[\kwunroll\ #1]}

\newcommand{\eghost}[4]{\textsf{Ghost}[#1=#2](#3.~#4)}
\newcommand{\eAsgn}[4]{\lstrike\humod{#2}{\eren{f}{#2}{#1}}\text{ in }#3.~#4\rstrike}
\newcommand{\eAsgneq}[4]{\eAsgn{#1}{#2}{#3}{#4}}
\newcommand{\etconsgen}[5]{\langle{\eren{#4}{#1}{#2}}~{{:}{*}}~#3.~#5\rangle}
\newcommand{\etcons}[2]{\etconsgen{x}{y}{\pvx}{#1}{#2}}
\newcommand{\ietcons}[3]{\langle{\eren{#2}{#1}{~}}~{{:}{*}}~#3\rangle}

\newcommand{\ebasgneq}[4]{[\humod{#2}{\eren{f}{#2}{#1}}\text{ in }#3.~#4]}
\newcommand{\edasgneq}[4]{\langle\humod{#2}{\eren{f}{#2}{#1}}\text{ in }#3.~#4\rangle}
\newcommand{\edasgn}[4]{\edasgneq{#1}{#2}{#3}{#4}}
\newcommand{\ebasgn}[4]{\ebasgneq{#1}{#2}{#3}{#4}}
\newcommand{\iebasgn}[2]{[\humod{#1}{\eren{f}{#1}{~}}\text{ in }~#2]}
\newcommand{\iedasgn}[2]{\langle\humod{#1}{\eren{f}{#1}{~}}\text{ in }~#2\rangle}

\newcommand{\eunpack}[2]{\textsf{unpack}(#1,\pvx y.~#2)}

\newcommand{\efpgen}[5]{\textit{FP}(#1, #2.~#3, #4.~#5)}
\newcommand{\efp}[3]{\efpgen{#1}{\pvs}{#2}{\pvg}{#3}}

\newcommand{\met}{\ensuremath{\mathcal{M}}}

\newcommand{\metz}{{\boldsymbol{0}}}

\newcommand{\metgr}{\boldsymbol{\succ}}

\newcommand{\conv}{\varphi\xspace}
\newcommand{\Gemp}{\cdot}
\newcommand{\issimp}[1]{{#1}\ \textsf{\upshape simp}}
\newcommand{\eforHead}[4]{\textsf{for}(#1:\conv(\met)={#2};#3;{#4})} 
\newcommand{\eforBody}[1]{\{#1\}}
\newcommand{\eforgen}[5]{\eforHead{#1}{#2}{#3}{#4}\eforBody{#5}}
\newcommand{\efor}[3]{\eforgen{\pvx}{#1}{\pvy.\,{#2}}{#3}{\alpha}}

\newcommand{\esub}[3]{[{#3}/{#2}]{#1}}
\newcommand{\tsub}[3]{\subst[#1]{#2}{#3}}
\newcommand{\isnorm}[1]{#1\text{ \textsf{normal}}}

\newcommand{\mycase}{\textbf{Case}\xspace}

\newcommand{\sapp}[2]{#1(#2)}
\newcommand{\sadj}[2]{#1^*(#2)}
\newcommand{\spp}[1]{\sapp{\sigma}{#1}}
\newcommand{\adj}[1]{\sadj{\sigma}{#1}}
\newcommand{\drv}{\mathcal{D}}
\newcommand{\nzvar}{\textit{nz}}
\newcommand{\convvar}{\textit{cnv}}

\newcommand{\pstep}{\textsf{step}}
\newcommand{\modvar}{\textit{mod}}
\newcommand{\monvar}{\mathit{mid}}
\newcommand{\rangevar}{\mathit{Range}}
\newcommand{\testvar}{\mathit{test}}
\newcommand{\elem}[2]{\textsf{Dec}[#1](#2)}
\newcommand{\spc}{\hspace{0.15in}}
\newcommand{\kwmod}{\textsf{mod}}
\newcommand{\emod}[2]{#1~\kwmod~#2}

\newcommand{\ediv}[2]{#1 / #2}

\newcommand{\nim}{\textsc{Nim}}
\newcommand{\cake}{\textsc{CC}}

\newcommand{\churchkleene}{\omega_{\text{CK}}}

\newcommand{\apL}[1]{#1_{\langle{0}\rangle}}
\newcommand{\apR}[1]{#1_{\langle{1}\rangle}}
\newcommand{\dpL}[1]{#1_{[0]}}
\newcommand{\dpR}[1]{#1_{[1]}}

\newcommand{\pvx}{p}
\newcommand{\pvy}{q}
\newcommand{\pvz}{t}
\newcommand{\pvl}{\ell}
\newcommand{\pvr}{r}
\newcommand{\pvrr}{rr}
\newcommand{\pvs}{s}
\newcommand{\pvg}{g}
\newcommand{\sdual}[1]{\pdual{#1}}

\newcommand{\va}{a}
\newcommand{\vb}{b}
\newcommand{\vca}{\overline{a}}
\newcommand{\vcb}{\overline{b}}

\newcommand{\btt}{\text{\rm\tt{tt}}}
\newcommand{\bff}{\text{\rm\tt{ff}}}
\newcommand{\stt}{\top}
\newcommand{\sff}{\bot}

\newcommand{\squarequotes}[2][0]{{%
  {\vphantom{#2}}^{\ulcorner}\kern-\scriptspace
  \mspace{-#1mu}%
  {{}#2}^{\urcorner}%
}}

\newcommand{\sty}{\m{\mathfrak{S}}\xspace}

\newcommand{\rzFst}[1]{\pi_L(#1)}
\newcommand{\rzSnd}[1]{\pi_R(#1)}

\newcommand{\eprojL}[1]{\pi_1{#1}}
\newcommand{\eprojR}[1]{\pi_2{#1}}
\newcommand{\eProjL}[1]{\lstrike\pi_1{#1}\rstrike}
\newcommand{\eProjR}[1]{\lstrike\pi_2{#1}\rstrike}
\newcommand{\edprojL}[1]{\langle\pi_1{#1}\rangle}
\newcommand{\edprojR}[1]{\langle\pi_2{#1}\rangle}
\newcommand{\ebprojL}[1]{[\pi_1{#1}]}
\newcommand{\ebprojR}[1]{[\pi_2{#1}]}

\makeatletter
\let\save@mathaccent\mathaccent
\newcommand*\if@single[3]{%
  \setbox0\hbox{${\mathaccent"0362{#1}}^H$}%
  \setbox2\hbox{${\mathaccent"0362{\kern0pt#1}}^H$}%
  \ifdim\ht0=\ht2 #3\else #2\fi
  }
\newcommand*\rel@kern[1]{\kern#1\dimexpr\macc@kerna}
\newcommand*\widebar[1]{\@ifnextchar^{{\wide@bar{#1}{0}}}{\wide@bar{#1}{1}}}
\newcommand*\wide@bar[2]{\if@single{#1}{\wide@bar@{#1}{#2}{1}}{\wide@bar@{#1}{#2}{2}}}
\newcommand*\wide@bar@[3]{%
  \begingroup
  \def\mathaccent##1##2{%
    \let\mathaccent\save@mathaccent
    \if#32 \let\macc@nucleus\first@char \fi
    \setbox\z@\hbox{$\macc@style{\macc@nucleus}_{}$}%
    \setbox\tw@\hbox{$\macc@style{\macc@nucleus}{}_{}$}%
    \dimen@\wd\tw@
    \advance\dimen@-\wd\z@
    \divide\dimen@ 3
    \@tempdima\wd\tw@
    \advance\@tempdima-\scriptspace
    \divide\@tempdima 10
    \advance\dimen@-\@tempdima
    \ifdim\dimen@>\z@ \dimen@0pt\fi
    \rel@kern{0.6}\kern-\dimen@
    \if#31
      \overline{\rel@kern{-0.6}\kern\dimen@\macc@nucleus\rel@kern{0.4}\kern\dimen@}%
      \advance\dimen@0.4\dimexpr\macc@kerna
      \let\final@kern#2%
      \ifdim\dimen@<\z@ \let\final@kern1\fi
      \if\final@kern1 \kern-\dimen@\fi
    \else
      \overline{\rel@kern{-0.6}\kern\dimen@#1}%
    \fi
  }%
  \macc@depth\@ne
  \let\math@bgroup\@empty \let\math@egroup\macc@set@skewchar
  \mathsurround\z@ \frozen@everymath{\mathgroup\macc@group\relax}%
  \macc@set@skewchar\relax
  \let\mathaccentV\macc@nested@a
  \if#31
    \macc@nested@a\relax111{#1}%
  \else
    \def\gobble@till@marker##1\endmarker{}%
    \futurelet\first@char\gobble@till@marker#1\endmarker
    \ifcat\noexpand\first@char A\else
      \def\first@char{}%
    \fi
    \macc@nested@a\relax111{\first@char}%
  \fi
  \endgroup
}
\makeatother

\providecommand{\hfilll}{\hskip 0pt plus 1filll}
\def\linferenceRuleSidenoteSeparation{~~\hfilll}
\renewcommand*{\irrulename}[1]{\text{\textcolor{gray}{#1}}}

\begin{document}

\title{Constructive Game Logic
\thanks{
This research was sponsored by the AFOSR under grant number FA9550-16-1-0288.
The authors were also funded by the NDSEG Fellowship and Alexander von Humboldt Foundation, respectively.}
}

\author{
Brandon Bohrer\inst{1}\orcidID{0000-0001-5201-9895} \and
Andr\'{e} Platzer\inst{1,2}\orcidID{0000-0001-7238-5710}
}
\authorrunning{B.\ Bohrer and A.\ Platzer}
\institute{Computer Science Department, Carnegie Mellon University, Pittsburgh, USA
\email{\{bbohrer,aplatzer\}@cs.cmu.edu}
\and Fakult\"at f\"ur Informatik, Technische Universit\"at M\"unchen, M\"unchen, Germany
}

\maketitle
\begin{abstract}
Game Logic is an excellent setting to study proofs-about-programs via the interpretation of those proofs as programs, because constructive proofs for games correspond to effective winning strategies to follow in response to the opponent's actions.
We thus develop \emph{Constructive Game Logic}, which extends Parikh's Game Logic (GL) with constructivity and with first-order programs \emph{\`a la} Pratt's first-order dynamic logic (DL).
Our major contributions include:
\begin{inparaenum}
\item a novel realizability semantics capturing the adversarial dynamics of games,
\item a natural deduction calculus and operational semantics describing the computational meaning of strategies via proof-terms, and
\item theoretical results including soundness of the proof calculus w.r.t.\ realizability semantics, progress and preservation of the operational semantics of proofs, and Existential Properties on support of the extraction of computational artifacts from game proofs.
\end{inparaenum}
Together, these results provide the most general account of a Curry-Howard interpretation for any program logic to date, and the first at all for Game Logic.
\end{abstract}


\section{Introduction}
Two of the most essential tools in theory of programming languages are \emph{program logics}, such as Hoare calculi~\cite{DBLP:journals/cacm/Hoare69} and dynamic logics~\cite{DBLP:conf/focs/Pratt76}, and the \emph{Curry-Howard correspondence}~\cite{curry1967combinatory,howard1980formulae}, wherein propositions correspond to types, proofs to functional programs, and proof term normalization to program evaluation.
Their intersection, the Curry-Howard interpretation of program logics, has received surprisingly little study.
We undertake such a study in the setting of Game Logic (\GL) \cite{DBLP:conf/focs/Parikh83}, because this leads to novel insights, because the Curry-Howard correspondence can be explained particularly intuitively for games, and because our first-order \GL is a superset of common logics such as first-order Dynamic Logic (\DL).

Constructivity and program verification have met before: Higher-order constructive logics~\cite{DBLP:journals/iandc/CoquandH88} obey the Curry-Howard correspondence and are used to develop verified functional programs.
Program logics are also often embedded in constructive proof assistants such as Coq~\cite{COQ}, inheriting constructivity from their metalogic.
Both are excellent ways to develop verified software, but we study something else.

We study the computational content of a program logic \emph{itself}.
Every fundamental concept of computation is expected to manifest in all three of logic, type systems, and category theory~\cite{Trinity}.
Because dynamics logics (\DL's) such as $\GL$ have shown that program execution is a first-class construct in modal logic, the theorist has an imperative to explore the underlying notion of computation by developing a constructive \GL with a Curry-Howard interpretation.

The computational content of a proof is especially clear in \GL, which generalizes \DL to programmatic models of zero-sum, perfect-information games between two players, traditionally named Angel and Demon.
Both normal-play and mis\`ere-play games can be modeled in \GL.
In classical \GL, the diamond modality $\ddiamond{\alpha}{\phi}$ and box modality $\dbox{\alpha}{\phi}$ say that Angel and Demon respectively have a strategy to ensure $\phi$ is true at the end of $\alpha$, which is a model of a game.
The difference between classical \GL and \CGL is that classical \GL allows proofs that exclude the middle, which correspond to strategies which branch on undecidable conditions.
\CGL proofs can branch only on decidable properties, thus they correspond to strategies which are \emph{effective} and can be executed by computer.
Effective strategies are crucial because they enable the synthesis of code that implements a strategy.
Strategy synthesis is itself crucial because even simple games can have complicated strategies, and synthesis provides assurance that the implementation correctly solves the game.
A \GL strategy resolves the choices inherent in a game: a diamond strategy specifies every move made by the Angel player, while a box strategy specifies the moves the Demon player will make.

In developing \emph{Constructive Game Logic} (\CGL), adding constructivity is a deep change.
We provide a natural deduction calculus for \CGL equipped with proof terms and an operational semantics on the proofs, demonstrating the meaning of strategies as functional programs and of winning strategies as functional programs that are guaranteed to achieve their objective no matter what counter-strategy the opponent follows.
While the proof calculus of a constructive logic is often taken as ground truth, we go a step further and develop a realizability semantics for \CGL as programs performing winning strategies for game proofs, then prove the calculus sound against it.
We adopt realizability semantics in contrast to the winning-region semantics of classical \GL because it enables us to prove that \CGL satisfies novel properties (\rref{sec:op-theory}).
The proof of our Strategy Property (\rref{thm:stratprop}) constitutes an (on-paper) algorithm that computes a player's (effective) strategy from a proof that they can win a game.
This is the key test of constructivity for \CGL, which would not be possible in classical \GL.
We show that \CGL proofs have \emph{two} computational interpretations: the operational semantics interpret an arbitrary proof (strategy) as a functional program which reduces to a normal-form proof (strategy),
while realizability semantics interpret Angel strategies as programs which defeat arbitrary Demonic opponents.

While \CGL has ample theoretical motivation, the practical motivations from synthesis are also strong.
A notable line of work on \dGL extends first-order \GL to hybrid games to verify safety-critical adversarial cyber-physical systems \cite{DBLP:journals/tocl/Platzer15}.
We have designed \CGL to extend smoothly to hybrid games, where synthesis provides the correctness demanded by safety-critical systems and the synthesis of correct monitors of the external world~\cite{DBLP:journals/fmsd/MitschP16}.


\section{Related Work}

This work is at the intersection of game logic and constructive modal logics.
Individually, they have a rich literature, but little work has been done at their intersection.
Of these, we are the first for \GL and the first with a proofs-as-programs interpretation for a full first-order program logic.

\paragraph{Games in Logic.}
Parikh's propositional \GL~\cite{DBLP:conf/focs/Parikh83} was followed by coalitional \GL~\cite{DBLP:journals/logcom/Pauly02}.
A first-order version of \GL is the basis of differential game logic \dGL~\cite{DBLP:journals/tocl/Platzer15} for hybrid games.
\GL's are unique in their clear delegation of strategy to the \emph{proof} language rather than the \emph{model} language, crucially allowing succinct game specifications with sophisticated winning strategies.
Succinct specifications are important: specifications are \emph{trusted} because proving the \emph{wrong theorem} would not ensure correctness.
Relatives without this separation include Strategy Logic~\cite{DBLP:conf/concur/ChatterjeeHP07}, Alternating-Time Temporal Logic (ATL)~\cite{DBLP:journals/jacm/AlurHK02}, CATL~\cite{DBLP:conf/atal/HoekJW05}, Ghosh's SDGL~\cite{ghosh2008strategies}, Ramanujam's structured strategies~\cite{DBLP:conf/kr/RamanujamS08}, Dynamic-epistemic logics~\cite{DBLP:series/lncs/Benthem15,DBLP:journals/games/BenthemPR11,van2001games}, evidence logics~\cite{DBLP:journals/sLogica/BenthemP11}, and  Angelic Hoare logic~\cite{DBLP:journals/corr/Mamouras16}.

\paragraph{Constructive Modal Logics.}
A major contribution of \CGL is our constructive semantics for games, not to be confused with game semantics~\cite{DBLP:journals/iandc/AbramskyJM00}, which are used to give programs semantics \emph{in terms of} games.
We draw on work in semantics for constructive modal logics, of which two main approaches are intuitionistic Kripke semantics and realizability semantics.

An overview of Intuitionistic Kripke semantics is given by Wijesekera~\citet{DBLP:journals/apal/Wijesekera90}.
Intuitionistic Kripke semantics are parameterized over worlds, but in contrast to classical Kripke semantics, possible worlds represent what is currently \emph{known} of the state.
Worlds are preordered by $w_1 \geq w_2$ when $w_1$ contains at least the knowledge in $w_2$.
Kripke semantics were used in Constructive Concurrent \DL~\cite{DBLP:journals/apal/WijesekeraN05}, where both the world and knowledge of it change during execution.
A key advantage of realizability semantics~\cite{DBLP:journals/mscs/Oosten02,lipton1992constructive} is their explicit interpretation of constructivity as computability by giving a \emph{realizer}, a program which witnesses a fact.
Our semantics combine elements of both: Strategies are represented by realizers, while the game state is a Kripke world.
Constructive set theory~\cite{DBLP:journals/jsyml/AczelG06} aids in understanding which set operations are permissible in constructive semantics.

Modal semantics have also exploited mathematical structures such as:
\begin{inparaenum}[i)]
\item Neighborhood models~\cite{DBLP:conf/lori/BenthemBE17}, topological models for spatial logics~\cite{DBLP:reference/spatial/BenthemB07},  and temporal logics of dynamical systems~\cite{DBLP:journals/lmcs/Fernandez-Duque18}.
\item Categorical~\cite{DBLP:conf/csl/AlechinaMPR01}, sheaf~\cite{Hilken_afirst}, and pre-sheaf~\cite{DBLP:journals/aml/Ghilardi89} models.
\item Coalgebraic semantics for classical Propositional Dynamic Logic (PDL)~\cite{DBLP:journals/corr/abs-1109-3685}.
\end{inparaenum}
While games are known to exhibit algebraic structure~\cite{DBLP:journals/sLogica/Goranko03}, such laws are not essential to this work.
Our semantics are also notable for the seamless interaction between a constructive Angel and a classical Demon.

\CGL is first-order, so we must address the constructivity of operations that inspect game state.
We consider rational numbers so that equality is decidable, but our work should generalize to constructive reals~\cite{bishop1967foundations,bridges2007techniques}.

Intuitionistic modalities also appear in dynamic-epistemic logic (DEL)~\cite{DBLP:journals/logcom/FrittellaGKPS16}, but that work is interested primarily in proof-theoretic semantics while we employ realizability semantics to stay firmly rooted in computation.
Intutionistic Kripke semantics have also been applied to multimodal System K with iteration~\cite{DBLP:journals/fuin/Celani01}, a weak fragment of PDL.

\paragraph{Constructivity and Dynamic Logic.}
With \CGL, we bring to fruition several past efforts to develop constructive dynamic logics.
Prior work on PDL~\cite{degen2006towards} sought an Existential Property for Propositional Dynamic Logic (PDL), but they questioned the practicality of their own implication introduction rule, whose side condition is non-syntactic.
One of our results is a first-order Existential Property, which Degen cited as an open problem beyond the methods of their day~\cite{degen2006towards}.
To our knowledge, only one approach~\citet{kamide2010strong} considers Curry-Howard or functional proof terms for a program logic.
While their work is a notable precursor to ours, their logic is a weak fragment of PDL without tests, monotonicity, or unbounded iteration, while we support not only PDL but the much more powerful first-order \GL.
Lastly, we are preceded by Constructive Concurrent Dynamic Logic, \cite{DBLP:journals/apal/WijesekeraN05} which gives a Kripke semantics for Concurrent Dynamic Logic \cite{DBLP:journals/jacm/Peleg87}, a proper fragment of \GL.
Their work focuses on an epistemic interpretation of constructivity, algebraic laws, and tableaux.
We differ in our use of realizability semantics and natural deduction, which were essential to developing a Curry-Howard interpretation for \CGL.
In summary, we are justified in claiming to have the first Curry-Howard interpretation with proof terms and Existential Properties for an \emph{expressive} program logic, the first constructive game logic, and the only with first-order proof terms.

While constructive natural deduction calculi map most directly to functional programs, proof terms can be generated for any proof calculus, including a well-known interpretation of classical logic as continuation-passing style~\cite{DBLP:conf/popl/Griffin90}.
Proof terms have been developed~\cite{DBLP:conf/cpp/FultonP16} for a Hilbert calculus for \dL, a dynamic logic (\DL) for hybrid systems.
Their work focuses on a provably correct interchange format for classical \dL proofs, not constructive logics.

\section{Syntax}
We define the language of \CGL, consisting of terms, games, and formulas.
The simplest terms are \emph{program variables} $x, y \in \allvars$ where $\allvars$ is the set of variable identifiers.
Globally-scoped mutable program variables contain the state of the game, also called the \emph{position} in game-theoretic terminology.
All variables and terms are rational-valued ($\allrat$); we also write $\allbool$ for the set of Boolean values $\{0,1\}$ for false and true respectively.
\begin{definition}[Terms]
A \emph{term} $f, g$ is a rational-valued  computable function over the game state.
We give a nonexhaustive grammar of terms, specifically those used in our examples:
\[f,g ~\bebecomes~  \cdots \alternative q \alternative x \alternative f + g \alternative f \cdot g \alternative \ediv{f}{g} \alternative \emod{f}{g}\]
where $q \in \mathbb{Q}$ is a rational literal, $x$ a program variable, $f + g$ a sum, $f \cdot g$ a product.
Division-with-remainder is intended for use with integers, but we generalize the standard notion to support rational arguments.
Quotient $\ediv{f}{g}$ is integer even when $f$ and $g$ are non-integer, and thus leaves a rational remainder $\emod{f}{g}$.
Divisors $g$ are assumed to be nonzero.

\label{def:terms}
 \end{definition}
A game in \CGL is played between a constructive player named Angel and a classical player named Demon.
Our usage of the names Angel and Demon differs subtly from traditional \GL usage for technical reasons.
Our Angel and Demon are asymmetric: Angel is ``our'' player, who must play constructively, while the ``opponent'' Demon is allowed to play classically because our opponent need not be a computer.
At any time some player is \emph{active}, meaning their strategy resolves all decisions, and the opposite player is called \emph{dormant}.
Classical \GL identifies Angel with active and Demon with dormant; the notions are distinct in \CGL.

\begin{definition}[Games]
The set of \emph{games} $\alpha,\beta$ is defined recursively as such:
\[\alpha,\beta ~\bebecomes~ \ptest{\phi} \alternative \humod{x}{f} \alternative \prandom{x} \alternative  \pchoice{\alpha}{\beta} \alternative \alpha;\beta \alternative \prepeat{\alpha} \alternative \pdual{\alpha}\]
\end{definition}
In the \emph{test game} $\ptest{\phi},$ the active player wins if they can exhibit a constructive proof that formula $\phi$ currently holds.
If they do not exhibit a proof, the dormant player wins by default and we informally say the active player ``broke the rules''. 
In deterministic assignment games $\humod{x}{f},$ neither player makes a choice, but the program variable $x$ takes on the value of a term $f$.
In nondeterministic assignment games $\prandom{x}$, the active player picks a value for $x : \allrat$.
In the choice game $\alpha \cup \beta,$ the active player chooses whether to play game $\alpha$ or game $\beta$.
In the sequential composition game $\alpha;\beta$, game $\alpha$ is played first, then $\beta$ from the resulting state.
In the repetition game $\prepeat{\alpha},$ the active player chooses after each repetition of $\alpha$ whether to continue playing, but loses if they repeat $\alpha$ infinitely.
Notably, the exact number of repetitions can depend on the dormant player's moves, so the active player need not know, let alone announce, the exact number of iterations in advance.
In the dual game $\pdual{\alpha},$ the active player becomes dormant and vice-versa, then $\alpha$ is played.
We parenthesize games with braces $\{ \alpha \}$ when necessary.
Sequential and nondeterministic composition both associate to the right, i.e., $\alpha \cup \beta \cup \gamma \equiv \{\alpha \cup \{\beta \cup \gamma\}\}$.
This does not affect their semantics as both operators are associative, but aids in reading proof terms.

\begin{definition}[\CGL Formulas]
\label{def:cgl-formula}
The set of \CGL \emph{formulas} $\phi$ (also $\psi, \rho$) is given recursively by the grammar:
\[ \phi ~\bebecomes~ \ddiamond{\alpha}{\phi} \alternative \dbox{\alpha}{\phi} \alternative f \sim g\]
\end{definition}
The defining constructs in \CGL (and \GL) are the modalities $\ddiamond{\alpha}{\phi}$ and $\dbox{\alpha}{\phi}$.
These mean that the active or dormant Angel (i.e., constructive) player has a constructive strategy to play $\alpha$ and achieve postcondition $\phi$.
This paper does not develop the modalities for active and dormant Demon (i.e., classical) players because by definition those cannot be synthesized to executable code.
We assume the presence of interpreted comparison predicates ${\sim} \in \{\leq, <, =, \neq, >, \geq\}$.

The standard connectives of first-order constructive logic can be derived from games and comparisons.
Verum ($\btt$) is defined $1 > 0$ and falsum ($\bff$) is $0 > 1$.
Conjunction $\phi \land \psi$ is defined $\ddiamond{\ptest{\phi}}{\psi},$
disjunction $\phi \lor  \psi$ is defined $\ddiamond{\ptest{\phi} \cup \ptest{\psi}}{\btt},$
implication $\phi \limply \psi$ is defined $\dbox{\ptest{\phi}}{\psi}$,
universal quantification $\lforall{x}{\phi}$ is defined $\dbox{\prandom{x}}{\phi},$ and
existential quantification $\lexists{x}{\phi}$ is defined $\ddiamond{\prandom{x}}{\phi}$.
As usual in logic, equivalence $\phi \lequiv \psi$ can also be defined $(\phi \limply \psi) \land (\psi \limply \phi)$.
As usual in constructive logics, negation $\neg \phi$ is defined $\phi \limply \bff$, and inequality is defined by $f \neq g \equiv \neg(f = g)$.
We will use the derived constructs freely but present semantics and proof rules only for the core constructs to minimize duplication.
Indeed, it will aid in understanding of the proof term language to keep the definitions above in mind, because the proof terms for many first-order programs follow those from first-order constructive logic.

For convenience, we also write derived operators where the dormant player is given control of a single choice before returning control to the active player.
The \emph{dormant choice} $\dchoice{\alpha}{\beta},$ defined $\pdual{\{{\pchoice{\pdual{\alpha}}{\pdual{\beta}}}\}},$ says the dormant player chooses which branch to take, but the active player is in control of the subgames.
We write $\eren{\phi}{x}{y}$ (likewise for $\alpha$ and $f$) for the \emph{renaming} of $x$ for $y$ and vice versa in formula $\phi$, and write $\tsub{\phi}{x}{f}$ for the \emph{substitution} of term $f$ for program variable $x$ in $\phi$, if the substitution is admissible (\rref{def:lem-admit} in \rref{sec:soundness}).

\subsection{Example Games}
\label{sec:example-games}
We demonstrate the meaning and usage of the \CGL constructs via examples, culminating in the two classic games of Nim and cake-cutting.

\paragraph{Nondeterministic Programs.}
Every (possibly nondeterministic) program is also a one-player game.
For example, the program $\humod{n}{0};\prepeat{\{\humod{n}{n+1}\}}$ can nondeterministically sets $n$ to any natural number because Angel has a choice whether to continue after every repetition of the loop, but is not allowed to continue forever.
Conversely, games are like programs where the environment (Demon) is adversarial, and the program (Angel) strategically resolves nondeterminism to overcome the environment.

\paragraph{Demonic Counter.}
Angel's choices often must be \emph{reactive} to Demon's choices.
Consider the game $\humod{c}{10};\prepeat{\{\humod{c}{c-1} \cap \humod{c}{c-2}\}};\ptest{0 \leq c \leq 2}$ where Demon repeatedly decreases $c$ by 1 or 2, and Angel chooses when to stop.
Angel only wins because she can pass the test $0 \leq c \leq 2,$ which she can do by simply repeating the loop until $0 \leq c \leq 2$ holds.
If Angel had to decide the loop duration in advance, Demon could force a rules violation by ``guessing'' the duration and changing his choices of $\humod{c}{c-1}$ vs.\ $\humod{c}{c-2}$.

\paragraph{Coin Toss.}
\newcommand{\guessvar}{\textsf{guess}}
\newcommand{\coinvar}{\textsf{coin}}
Games are perfect-information and do not possess randomness in the probabilistic sense, only (possibilistic) nondeterminism.
This standard limitation is shown by attempting to express a coin-guessing game:
\[\{\humod{\coinvar}{0} \cap \humod{\coinvar}{1}\};\{\humod{\guessvar}{0} \cup \humod{\guessvar}{1}\};\ptest{\guessvar = \coinvar}\]
The Demon player sets the value of a tossed coin, but does so adversarially, not randomly, since strategies in \CGL are \emph{pure} strategies.
The Angel player has perfect knowledge of $\coinvar$  and can set $\guessvar$ equivalently, thus easily passing the test $\guessvar = \coinvar,$ unlike a real coin toss.
Partial information games are interesting future work that could be implemented by limiting the variables visible in a strategy.

\paragraph{Nim.}
Nim  is the standard introductory example of a discrete, 2-player, zero-sum, perfect-information game.
We consider mis\`ere play (last player loses) for a version of Nim that is also known as the \emph{subtraction game}.
The constant $\nim$ defines the game Nim.
\begin{align*}
\nim =\Big\{\big\{&\{\humod{c}{c-1} \cup \humod{c}{c-2} \cup \humod{c}{c-3}\};
                    \ptest{c > 0}\big\};\\
             \big\{&\{\humod{c}{c-1} \cup \humod{c}{c-2} \cup \humod{c}{c-3}\};
                    \ptest{c > 0}\big\}\pdual{}\Big\}^*
\end{align*}
The game state consists of a single counter $c$ containing a natural number, which each player chooses ($\cup$) to reduce by 1, 2, or 3 ($\humod{c}{c-k}$).
The counter is non-negative, and the game repeats as long as Angel wishes, until some player empties the counter, at which point that player is declared the loser ($\ptest{c > 0}$).
\begin{proposition}[Dormant winning region]
\label{prop:nim-demon-wins}
Suppose $c \equiv 1 ~(\textup{mod}~ 4),$
Then the dormant player has a strategy to ensure $c \equiv 1 ~(\textup{mod}~ 4)$ as an invariant.
That is, the following \CGL formula is valid (true in every state): 
\[c > 0 \limply \emod{c}{4} = 1 \limply
\dbox{\prepeat{\nim}}{\,\emod{c}{4} = 1}
\]
\end{proposition}
This implies the dormant player wins the game because the active player violates the rules once $c=1$ and no move is valid.
We now state the winning region for an active player.
\begin{proposition}[Active winning region]
\label{prop:nim-angel-wins}
Suppose $c \in \{0,2,3\} ~(\textup{mod}~ 4)$ initially, and the active player controls the loop duration.
Then the active player can achieve $c \in \{2,3,4\}$:
\[c > 0 \limply \emod{c}{4} \in \{0,2,3\} \limply \ddiamond{\prepeat{\nim}}{\,c \in \{2,3,4\}}\]
\end{proposition}
At that point, the active player will win in one move by setting $c=1$ which forces the dormant player to set $c=0$ and fail the test $\ptest{c > 0}$.

\paragraph{Cake-cutting.}
Another classic 2-player game, from the study of equitable division, is the cake-cutting problem~\cite{DBLP:journals/sLogica/PaulyP03a}:
The active player cuts the cake in two, then the (initially-)dormant player gets first choice of a piece.
This is an optimal protocol for splitting the cake in the sense that the active player is incentivized to split the cake evenly, else the dormant player could take the larger piece.
Cake-cutting is also a simple use case for fractional numbers.
The constant $\cake$  defines the cake-cutting game.
Here $x$ is the relative size (from $0$ to $1$) of the first piece, $y$ is the size of the second piece, $a$ is the size of the active player's piece, and $d$ is the size of dormant player's piece.
\begin{align*}
\cake =\,&\prandom{x}; \ptest{(0 \leq x \leq 1)}; \humod{y}{1-x};\\
            &\{\dchoice{\humod{a}{x};\humod{d}{y}}{\humod{a}{y};\humod{d}{x}}\}
\end{align*}
The game is played only once.
The active player picks the division of the cake, which must be a fraction $0 \leq x \leq 1$.
The dormant player then picks which slice goes to whom.

The active player has a tight strategy to achieve a $0.5$ cake share, as stated in \rref{prop:cake-angel-wins}.
\begin{proposition}[Active winning region]
\label{prop:cake-angel-wins}
The following formula is valid:
\[\ddiamond{\cake}{\,a \geq 0.5}\]
\end{proposition}
The dormant player also has a computable strategy to achieve exactly $0.5$ share of the cake (\rref{prop:cake-demon-wins}).
Division is fair because each player has a strategy to get their fair 0.5 share.
\begin{proposition}[Dormant winning region]
\label{prop:cake-demon-wins}
The following formula is valid:
\[\dbox{\cake}{\,d \geq 0.5}\]
\end{proposition}

\paragraph{Computability and Numeric Types.}
Perfect fair division is only achieved for $a,d \in \allrat$ because rational equality is decidable.
Trichotomy $(a < 0.5 \lor a = 0.5 \lor a > 0.5)$ is a tautology, so the dormant player's strategy can inspect the active player's choice of $a$.
Notably, we intend to support constructive reals in future work, for which exact equality is not decidable and trichotomy is not an axiom.
Future work on real-valued \CGL will need to employ approximate comparison techniques as is typical for constructive reals~\cite{bishop1967foundations,bridges2007techniques,DBLP:series/txtcs/Weihrauch00}.
The examples in this section have been proven%
\iflongversion%
(\rref{app:example-proofs})
\else%
~\cite{br2020constructive}
\fi
using the calculus defined in \rref{sec:proof-calculus}.

\section{Semantics}
\label{sec:rel-sem}
We now develop the semantics of \CGL.
In contrast to classical \GL, whose semantics are well-understood~\cite{DBLP:conf/focs/Parikh83}, the major semantic challenge for \CGL is capturing the competition between a \emph{constructive} Angel and \emph{classical} Demon.
We base our approach on realizability semantics~\cite{DBLP:journals/mscs/Oosten02,lipton1992constructive}, because this approach makes the relationship between constructive proofs and programs particularly clear, and generating programs from \CGL proofs is one of our motivations.

Unlike previous applications of realizability, games feature two agents, and one could imagine a semantics with two realizers, one for each of Angel and Demon.
However, we choose to use only one realizer, for Angel, which captures the fact that only Angel is restricted to a computable strategy, not Demon.
Moreover, a single realizer makes it clear that Angel cannot inspect Demon's strategy, only the game state, and also simplifies notations and proofs.
Because Angel is computable but Demon is classical, our semantics has the flavor both of realizability semantics and of a traditional Kripke semantics for programs.

The semantic functions employ \emph{game states} $\om \in \allstate$ where we write $\allstate$ for the set of all states.
We additionally write $\stt, \sff \in \allstate$ (not to be confused with formulas $\btt$ and $\bff$) for the pseudo-states $\stt$ and $\sff$ indicating that Angel or Demon respectively has won the game early by forcing the other to fail a test.
Each $\om \in \allstate$ maps each $x \in \allvars$ to a value $\om(x) \in \allrat$.
We write $\ssub{\om}{x}{v}$ for the state that agrees with $\om$ except that $x$ is assigned value $v$ where $v \in \allrat$.
\begin{definition}[Arithmetic term semantics]
\label{def:term-sem}
A term $f$ is a computable function of the state, so the interpretation $\tint{f}{\om}$ of term $f$ in state $\om$ is $f(\om)$.
\end{definition}

\subsection{Realizers}
To define the semantics of games, we first define realizers, the programs which implement strategies.
The language of realizers is a higher-order lambda calculus where variables can range over game states, numbers, or realizers which realize a give proposition $\phi$.
Gameplay proceeds in continuation-passing style: invoking a realizer returns another realizer which performs any further moves.
We describe the typing constraints for realizers informally, and say $\myaa$ is a $\ddiamond{\alpha}\phi$-realizer ($\myaa \in \rzfor{\ddiamond{\alpha}{\phi}}$) if it provides strategic decisions exactly when $\ddiamond{\alpha}{\phi}$ demands them.

\begin{definition}[Realizers]
The syntax of realizers $\myaa,\ab,\ac \in \allRz$ (where $\allRz$ is the set of all realizers) is defined coinductively:
\begin{align*}
\myaa,\ab,\ac &\bebecomes x \alternative \rzNil \alternative \rzCons{\myaa}{\ab}\alternative \rzFst{\myaa} \alternative \rzSnd{\myaa}
  \alternative (\rzBLam{\om}{\myaa(\om)}) \alternative (\rzFOLam{x}{\allrat}{\myaa})  \\
            &\alternative (\rzHOLam{x}{\phi}{\myaa}) \alternative \rzApp{\myaa}{v} \alternative \rzApp{\myaa}{\ab} \alternative \rzApp{\myaa}{\om}
\alternative \rzIf{f(\om)}{\myaa}{\ab}
\end{align*}
\end{definition}
\noindent where $x$ is a program (or realizer) variable and $f$ is a term over the state $\om$.
The Roman $\myaa,\ab,\ac$ should not be confused with the Greek $\alpha,\beta,\gamma$ which range over games.
Realizers have access to the game state $\om$, expressed by lambda realizers $(\rzBLam{\om}{\myaa(\om)})$ which, when applied in a state $\nu,$ compute the realizer $\myaa$ with $\nu$ substituted for $\om$.
State lambdas $\lambda$ are distinguished from propositional and first-order lambdas $\Lambda$.
The unit realizer $\rzNil$ makes no choices and is understood as a unit tuple.
Units $\rzNil$ realize $f \sim g$ because \emph{rational} comparisons, in contrast to real comparisons, are decidable.
Conditional strategic decisions are realized by $\rzIf{f(\om)}{\myaa}{\ab}$ for computable function $f : \allstate \to \allbool,$ and execute $\myaa$ if $f$ returns truth, else $\ab$.
Realizer $(\rzBLam{\om}{f(\om)})$ is a $\ddiamond{\alpha\cup\beta}{\phi}$-realizer if $f(\om) \in (\{0\} \times \rzfor{\ddiamond{\alpha}{\phi}}) \cup (\{1\} \times \rzfor{\ddiamond{\beta}{\phi}})$ for all $\omega$.
The first component determines which branch is taken, while the second component is a continuation which must be able to play the corresponding branch.
Realizer $(\rzBLam{\om}{f(\om)})$ can also be a $\ddiamond{\prandom{x}}{\phi}$-realizer, which requires $f(\om) \in \mathbb{Q} \times (\rzfor{\phi})$ for all $\omega$.
The first component determines the value of $x$ while the second component demonstrates the postcondition $\phi$.
The pair realizer $\rzCons{\myaa}{\ab}$ realizes both Angelic tests $\ddiamond{\ptest{\phi}}{\phi}$ and dormant choices $\dbox{\alpha\cup\beta}{\phi}$.
It is identified with a pair of realizers: $\rzCons{\myaa}{\ab} \in \allRz \times \allRz$.

\renewcommand{\rzInd}[2]{\textsf{ind}({#1}.~#2)}
A dormant realizer waits and remembers the active Demon's moves, because they typically inform Angel's strategy once Angel resumes action.
The first-order realizer $(\rzFOLam{x}{\allrat}{\ab})$ is a $\dbox{\prandom{x}}{\phi}$-realizer when  $\tsub{\ab}{x}{v}$ is a $\phi$-realizer for every $v \in \allrat$; Demon tells Angel the desired value of $x$, which informs Angel's continuation $\ab$.
The higher-order realizer $(\rzHOLam{x}{\phi}{\ab})$ realizes $\dbox{\ptest{\phi}}{\psi}$ when $\tsub{\ab}{x}{\ac}$ realizes $\psi$ for every $\phi$-realizer $\ac$.
Demon announces the realizer for $\phi$ which Angel's continuation $\ab$ may inspect.
Tuples are inspected with projections $\rzFst{\myaa}$ and $\rzSnd{\myaa}$.
A lambda is inspected by applying arguments $\rzApp{\myaa}{\om}$ for state-lambdas, $\rzApp{\myaa}{v}$ for first-order, and $\rzApp{\myaa}{\ab}$ for higher-order.
Realizers for sequential compositions $\ddiamond{\alpha;\beta}{\phi}$ (likewise $\dbox{\alpha;\beta}{\phi}$) are $\ddiamond{\alpha}{\ddiamond{\beta}{\phi}}$-realizers: first $\alpha$ is played, and in every case the continuation must play $\beta$ before showing $\phi$.
Realizers for repetitions $\prepeat{\alpha}$ are  streams containing $\alpha$-realizers, possibly infinite by virtue of coinductive syntax.
Active loop realizer $\rzInd{x}{\myaa}$ is the least fixed point of the equation $\ab = \esub{\myaa}{x}{\ab},$ i.e., $x$ is a recursive call which must be invoked only in accordance with some well-order.
We realize dormant loops with $\rzCoind{\myaa}{x}{\ab}{\ac},$ coinductively generated from initial value $\myaa,$ update $\ab$, and post-step $\ac$ with variable $x$ for current generator value.

Active loops must terminate, so $\ddiamond{\prepeat{\alpha}}{\phi}$-realizers are constructed inductively using any well-order on states.
Dormant loops must be played as long as the opponent wishes, so $\dbox{\prepeat{\alpha}}{\phi}$-realizers are constructed coinductively, with the invariant that $\phi$ has a realizer at every iteration.

\subsection{Formula and Game Semantics}
A state $\om$ paired with a realizer $\myaa$ that continues the game is called a \emph{possibility}.
A \emph{region} (written $X,Y,Z$) is a set of possibilities.
We write $\fintR{\phi} \subseteq \rzfor{\phi} \times \allstate$ for the region which realizes formula $\phi$.
A formula $\phi$ is \emph{valid} iff some $\myaa$ uniformly realizes every state, i.e.,  $\{\myaa\} \times \allstate \subseteq \fintR{\phi}$.
A sequent $\seq{\Gamma}{\phi}$ is \emph{valid} iff the formula $\bigwedge \Gamma \limply \phi$ is valid, where $\bigwedge \Gamma$ is the conjunction of all assumptions in $\Gamma$.

The game semantics are region-oriented, i.e., they process possibilities in bulk, though Angel commits to a strategy from the start.
The region $\strategyforR[\alpha]{X} : \powerset{\allRz \times \allstate}$ is the union of all end regions of game $\alpha$ which arise when active Angel commits to an element of $X,$ then Demon plays adversarially.
In $\dstrategyforR[\alpha]{X} : \powerset{\allRz \times \allstate}$ Angel is the \emph{dormant} player, but it is still Angel who commits to an element of $X$ and Demon who plays adversarially.
Recall that pseudo-states $\stt$ and $\sff$ represent early wins by each Angel and Demon, respectively.
The definitions below implicitly assume $\sff, \stt \notin X,$ they extend to the case $\sff \in X$ (likewise $\stt \in X$) using the equations
$\dstrategyforR[\alpha]{(X \cup \{\sff\})} = \dstrategyforR[\alpha]{X} \cup\{\sff\}$ and
$\strategyforR[\alpha]{(X \cup \{\sff\})} = \strategyforR[\alpha]{X} \cup\{\sff\}$.
That is, if Demon has already won by forcing an Angel violation initially, any remaining game can be skipped with an immediate Demon victory, and vice-versa.
The game semantics exploit the \emph{Angelic} projections $\apL{Z}, \apR{Z}$ and \emph{Demonic} projections $\dpL{Z}, \dpR{Z}$, which represent binary decisions made by a constructive Angel and a classical Demon, respectively.
The Angelic projections, which are defined $\apL{Z} = \{(\rzSnd{\myaa}, \om)~|~\rzFst{\myaa}(\om)=0, (\myaa,\om) \in Z\}$ and $\apR{Z} = \{(\rzSnd{\myaa}, \om)~|~\rzFst{\myaa}(\om)=1, (\myaa,\om) \in Z\}$, filter by which branch Angel chooses with $\rzFst{\myaa}(\om) \in \mathbb{B},$ then project the remaining strategy $\rzSnd{\myaa}$.
The Demonic projections, which are defined $\dpL{Z} \equiv \{(\rzFst{\myaa},\om)~|~(\myaa,\om) \in Z\}$ and $\dpR{Z} \equiv \{(\rzSnd{\myaa},\om)~|~(\myaa,\om) \in Z\},$ contain the same states as $Z,$ but project the realizer to tell Angel which branch Demon took.

\begin{definition}[Formula semantics]
$\fintR{\phi} \subseteq \allRz \times \allstate$ is defined as:
\begin{align*}
(\rzNil,\om) \in  \fintR{f \sim g}                 &\text{ iff } \tint{f}{\om} \sim \tint{g}{\om}\\
(\myaa,\om) \in  \fintR{\ddiamond{\alpha}{\phi}}       &\text{ iff } \strategyforR[\alpha]{\{(\myaa,\om)\}} \subseteq (\fintR{\phi} \cup \{\stt\})\\
(\myaa,\om) \in  \fintR{\dbox{\alpha}{\phi}}              &\text{ iff } \dstrategyforR[\alpha]{\{(\myaa,\om)\}} \subseteq (\fintR{\phi} \cup \{\stt\})
\end{align*}
\end{definition}
Comparisons $f \sim g$ defer to the term semantics, so the interesting cases are the game modalities.
Both $\dbox{\alpha}{\phi}$ and $\ddiamond{\alpha}{\phi}$ ask whether Angel wins $\alpha$ by following the given strategy, and differ only in whether Demon vs.\ Angel is the active player, thus in both cases \emph{every} Demonic choice must satisfy Angel's goal, and early Demon wins are counted as Angel losses.
\begin{definition}[Angel game forward semantics]
We inductively define the region $\strategyforR[\alpha]{X} : \powerset{\allRz \times \allstate}$ in which $\alpha$ can end when active Angel plays $X$:
\begin{align*}
\strategyforR[\ptest{\phi}]{X}          &= \{(\rzSnd{\myaa},\om)~|~ (\rzFst{\myaa},\om) \in \fintR{\phi}\text{ for some }(\myaa,\om) \in X~\}&&\\
                                                  &\phantom{\,=\,} \cup \{\sff~|~(\rzFst{\myaa},\om) \notin \fintR{\phi}\text{ for all }(\myaa,\om) \in X~\}&&\\
\strategyforR[\humod{x}{f}]{X}         &= \{(\myaa,\ssub{\om}{x}{\tint{f}{\om}})~|~(\myaa,\om) \in X\}&&\text{}\\
\strategyforR[\prandom{x}]{X}          &= \{(\rzSnd{\myaa},\ssub{\om}{x}{\rzFst{\myaa}(\om)})~|~(\myaa,\om) \in X\}&&\text{}\\
\strategyforR[\alpha;\beta]{X}         &= \strategyforR[\beta]{(\strategyforR[\alpha]{X})} &&\text{}\\
\strategyforR[\alpha\cup\beta]{X}      &= \strategyforR[\alpha]{\apL{X}} \cup \strategyforR[\beta]{\apR{X}}&&\text{}\\
\strategyforR[\prepeat{\alpha}]{X}     &= \bigcap\{\apL{Z}\subseteq \allRz \times \allstate~|~ X \cup (\strategyforR[\alpha]{\apR{Z}}) \subseteq Z\}
&&\\
\strategyforR[\pdual{\alpha}]{X}       &= \dstrategyforR[\alpha]{X}&&\text{}
\end{align*}
\end{definition}
\begin{definition}[Demon game forward semantics]
We inductively define the region $\dstrategyforR[\alpha]{X} : \powerset{\allRz \times \allstate}$ in which $\alpha$ can end when dormant Angel plays $X$:
\begin{align*}
\dstrategyforR[\ptest{\phi}]{X}     &= \{(\rzApp{\myaa}{\ab},\om)~|~(\myaa,\om) \in X, (\ab,\om) \in \fintR{\phi}, \text{ some }\ab \in \allRz\}&&\\
                                                &\phantom{\,=\,} \cup \{\top~|~ (\myaa,\om) \in X,\text{ but no } (\ab,\om) \in \fintR{\phi}\}&&\\
\dstrategyforR[\humod{x}{f}]{X}     &= \{(\myaa,\ssub{\om}{x}{\tint{f}{\om}})~|~(\myaa,\om) \in X\} &&\text{}\\
\dstrategyforR[\prandom{x}]{X}      &= \{(\rzApp{\myaa}{r},\ssub{\om}{x}{r})~|~r \in \allrat\}&&\\
\dstrategyforR[\alpha;\beta]{X}     &= \dstrategyforR[\beta]{(\dstrategyforR[\alpha]{X})} &&\text{}\\
\dstrategyforR[\alpha\cup\beta]{X}  &= \dstrategyforR[\alpha]{\dpL{X}} \cup \dstrategyforR[\beta]{\dpR{X}} &&\text{}\\
\dstrategyforR[\prepeat{\alpha}]{X} &= \bigcap\{\dpL{Z}\subseteq \allRz \times \allstate~|~ X \cup (\dstrategyforR[\alpha]{\dpR{Z}}) \subseteq Z\}&&\\
\dstrategyforR[\pdual{\alpha}]{X}   &= \strategyforR[\alpha]{X} &&\text{}
\end{align*}
\end{definition}

Angelic tests $\ptest{\phi}$ end in the current state $\om$ with remaining realizer $\rzSnd{\myaa}$ if Angel can realize $\phi$ with $\rzFst{\myaa}$, else end in $\sff$.
Angelic deterministic assignments consume no realizer and simply update the state, then end.
Angelic nondeterministic assignments $\prandom{x}$ ask the realizer $\rzFst{\myaa}$ to compute a new value for $x$ from the current state.
Angelic compositions $\alpha;\beta$ first play $\alpha,$ then $\beta$ from the resulting state using the resulting continuation.
Angelic choice games $\alpha \cup \beta$ use the Angelic projections to decide which branch is taken according to $\rzFst{\myaa}$.
The realizer $\rzSnd{\myaa}$ may be reused between $\alpha$ and $\beta,$ since $\rzSnd{\myaa}$ could just invoke $\rzFst{\myaa}$ if it must decide which branch has been taken.
This definition of Angelic choice (corresponding to constructive disjunction) captures the reality that realizers in \CGL, in contrast with most constructive logics, are entitled to observe a game state, but they must do so in computable fashion.

\paragraph{Repetition Semantics.}
In any \GL, the challenge in defining the semantics of repetition games $\prepeat{\alpha}$ is that the number of iterations, while finite, can depend on both players' actions and is thus not known in advance, while the \DL-like semantics of $\prepeat{\alpha}$ as the finite reflexive, transitive closure of $\alpha$ gives an advance-notice semantics.
Classical \GL provides the no-advance-notice semantics as a fixed point~\cite{DBLP:conf/focs/Parikh83}, and we adopt the fixed point semantics as well.
The Angelic choice whether to stop ($\apL{Z}$) or iterate the loop ($\apR{Z}$) is analogous to the case for $\alpha \cup \beta$.

\paragraph{Duality Semantics.}
\label{sec:dual-sem}
To play the dual game $\pdual{\alpha},$ the active and dormant players switch roles, then play $\alpha$.
In \emph{classical} \GL, this characterization of duality is interchangeable with the definition of $\pdual{\alpha}$ as the game that Angel wins exactly when it is impossible for Angel to lose.
The characterizations are \emph{not} interchangeable in \CGL because the Determinacy Axiom (all games have winners) of \GL is not valid in \CGL:
\begin{remark}[Indeterminacy]
Classically equivalent determinacy axiom schemata
\(\lnot\ddiamond{\alpha}{\lnot\phi} \limply \dbox{\alpha}{\phi}\) and \(\ddiamond{\alpha}{\lnot\phi} \lor \dbox{\alpha}{\phi}\) of classical \GL are not valid in \CGL, because they imply double negation elimination.
\end{remark}
\begin{remark}[Classical duality]
In classical \GL, Angelic dual games are characterized by the axiom schema $\ddiamond{\pdual{\alpha}}{\phi} \lequiv \neg \ddiamond{\alpha}{\neg \phi},$ which is not valid in in \CGL.
It is classically interdefinable with $\ddiamond{\pdual{\alpha}} \lequiv \dbox{\alpha}{\phi}$.
\label{rem:cdd}
\end{remark}
\noindent The determinacy axiom is not valid in \CGL, so we take $\ddiamond{\pdual{\alpha}} \lequiv \dbox{\alpha}{\phi}$ as primary.

\subsection{Demonic Semantics}
Demon wins a Demonic test by presenting a realizer $\ab$ as evidence that the precondition holds.
If he cannot present a realizer (i.e., because none exists), then the game ends in $\top$ so Angel wins by default.
Else Angel's higher-order realizer $\myaa$ consumes the evidence of the pre-condition, i.e., Angelic strategies are entitled to depend (computably) on \emph{how} Demon demonstrated the precondition.
Angel can check that Demon passed the test by executing $\ab$.
The Demonic repetition game $\prepeat{\alpha}$ is defined as a fixed-point~\cite{DBLP:journals/tocl/Platzer15} with Demonic projections.
Computationally, a winning invariant for the repetition is the witness of its winnability.

The remaining cases are innocuous by comparison.
Demonic deterministic assignments $\humod{x}{f}$ deterministically store the value of $f$ in $x,$ just as Angelic assignments do.
In demonic nondeterministic assignment $\prandom{x},$ Demon chooses to set $x$ to \emph{any} value.
When Demon plays the choice game $\pchoice{\alpha}{\beta},$ Demon chooses classically between $\alpha$ and $\beta$.
The dual game $\pdual{\alpha}$ is played by Demon becoming dormant and Angel become active in $\alpha$.

\paragraph{Semantics Examples.}
The realizability semantics of games are subtle on a first read, so we provide examples of realizers.
In these examples, the state argument $\omega$ is implicit, and we refer to $\om(x)$ simply as $x$ for brevity.

Recall that $\dbox{\ptest{\phi}}{\psi}$ and $\phi \limply \psi$ are equivalent.
For any $\phi,$ the identity function $(\rzHOLam{x}{\phi}{x})$ is a $\phi \limply \phi$-realizer: for every $\phi$-realizer $x$ which Demon presents, Angel can present the same $x$ as evidence of $\phi$.
This confirms expected behavior per propositional constructive logic: the identity function is the proof of self-implication.

In example formula $\ddiamond{\pdual{\prandom{x}};\{\humod{x}{x} \cup \humod{x}{-x}\}}{x \geq 0},$ Demon gets to set $x,$ then Angel decides whether to negate $x$ in order to make it nonnegative.
It is realized by $\rzFOLam{x}{\allrat}{\rzCons{(\rzIf{x<0}{1}{0})}{\rzNil}}$: Demon announces the value of $x,$ then Angel's strategy is to check the sign of $x$, taking the right branch when $x$ is negative.
Each branch contains a deterministic assignment which consumes no realizer, then the postcondition $x \geq 0$ has trivial realizer $\rzNil$.

Consider the formula $\ddiamond{\prepeat{\{\humod{x}{x+1}\}}}{x > y},$ where  Angel's winning strategy is to repeat the loop until $x > y$, which will occur as $x$ increases.
The realizer is $\rzInd{w}{\rzCons{\rzIf{x > y}{\rzCons{0}{\rzNil}}{\rzCons{1}{w}}}{\rzNil}}$, which says that Angel stops the loop if $x > y$ and proves the postcondition with a trivial strategy.
Else Angel continues the loop, whose body consumes no realizer, and supplies the inductive call $w$ to continue the strategy inductively.

Consider the formula $\dbox{\ptest{x>0};\prepeat{\{\humod{x}{x+1}\}}}{\lexists{y}{(y \leq x \land y > 0)}}$ for a subtle example.
Our strategy for Angel is to record the initial value of $x$ in $y,$ then maintain a proof that $y \leq  x$ as $x$ increases.
This strategy is represented by $\rzHOLam{w}{(x>0)}{
\rzCoind{\rzCons{x}{\rzCons{\rzNil}{w}}}
        {z}
        {\rzCons{\rzFst{z}}{\rzCons{\rzNil}{\rzSnd{\rzSnd{z}}}}}
        {z}}$.
That is, initially Demon announces a proof $w$ of $x>0$.
Angel specifies the initial element of the realizer stream by witnessing $\lexists{y}{(y \leq x \land y > 0)}$ with $\ac_0 = \rzCons{x}{\rzCons{\rzNil}{w}},$ where the first component instantiates $y=x,$
the trivial second component indicates that $y \leq y$ trivially, and the third component reuses $w$ as a proof of $y > 0$.
Demon can choose to repeat the loop arbitrarily.
When Demon demands the $k$'th repetition, $z$ is bound to $\ac_{k-1}$ to compute $\ac_k = \rzCons{\rzFst{z}}{\rzCons{\rzNil}{\rzSnd{\rzSnd{z}}}},$ which plays the next iteration.
That is, at each iteration Angel witnesses $\lexists{y}{(y \leq x \land y > 0)}$ by assigning the same value (stored in $\rzFst{z}$) to $y,$ reproving $y \leq x$ with $\rzNil,$ then reusing the proof (stored in $\rzSnd{\rzSnd{z}})$ that $y > 0$.

\section{Proof Calculus}
\label{sec:proof-calculus}
Having settled on the meaning of a game in~\rref{sec:rel-sem}, we proceed to develop a calculus for proving \CGL formulas syntactically.
The goal is twofold: the practical motivation, as always, is that when verifying a concrete example, the realizability semantics provide a notion of ground truth, but are impractical for proving large formulas.
The theoretical motivation is that we wish to expose the computational interpretation of the modalities $\ddiamond{\alpha}{\phi}$ and $\dbox{\alpha}{\phi}$ as the types of the players' respective winning strategies for game $\alpha$ that has $\phi$ as its goal condition.
Since \CGL is constructive, such a strategy  constructs a proof of the postcondition $\phi$.

To study the computational nature of proofs, we write proof terms explicitly: the main proof judgement $\proves{\Gamma}{M}{\phi}$ says proof term $M$ is a proof of $\phi$ in context $\Gamma$, or equivalently a proof of sequent $(\Gamma \vdash \phi)$.
We write $M,N,O$ (sometimes $A, B,C$) for arbitrary proof terms, and $\pvx,\pvy,\pvl,\pvr,\pvs,\pvg$ for \emph{proof variables}, that is variables that range over proof terms of a given proposition.
In contrast to the assignable \emph{program variables}, the proof variables are given their meaning by substitution and are scoped locally, not globally.
We adapt propositional proof terms such as pairing, disjoint union, and lambda-abstraction to our context of game logic.
To support first-order games, we include first-order proof terms and new terms for features: dual, assignment, and repetition games.

We now develop the calculus by starting with standard constructs and working toward the novel constructs of \CGL.
The assumptions $\pvx$ in $\Gamma$ are named, so that they may appear as variable proof-terms $\pvx$.
We write $\eren{\Gamma}{x}{y}$ and $\eren{M}{x}{y}$ for the renaming of program variable $x$ to $y$ and vice versa in context $\Gamma$ or proof term $M,$ respectively.
Proof rules for state-modifying constructs explicit perform renamings, which both ensures they are applicable as often as possible and also ensures that references to proof variables support an intuitive notion of lexical scope.
Likewise  $\tsub{\Gamma}{x}{f}$ and $\tsub{M}{x}{f}$ are the substitutions of term $f$ for program variable $x$.
We use distinct notation to substitution proof terms for proof variables while avoiding capture: $\esub{M}{\pvx}{N}$ substitutes proof term $N$ for proof variable $\pvx$ in proof term $M$.
Some proof terms such as pairs prove both a diamond formula and a box formula.
We write $\edcons{M}{N}$ and $\ebcons{M}{N}$ respectively to distinguish the terms or $\eCons{M}{N}$ to treat them uniformly.
Likewise we abbreviate $\dmodality{\alpha}{\phi}$ when the same rule works for both diamond and box modalities, using $\pmodality{\alpha}{\phi}$ to denote its dual modality.
The proof terms $\edasgn{y}{x}{\pvx}{M}$ and $\ebasgn{y}{x}{\pvx}{M}$ introduce an auxiliary ghost variable $y$ for the old value of $x$, which improves completeness without requiring manual ghost steps.

\begin{figure}
\centering
\begin{calculuscollections}{\columnwidth}
\begin{calculus}
\cinferenceRule[dchoiceE|{$\langle\cup\rangle${E}}]{}
{
\linferenceRule[formula]
{\proves{\Gamma}{A}{\ddiamond{\alpha\cup\beta}{\phi}}
        &\proves{\Gamma,\pvl:\ddiamond{\alpha}{\phi}}{B}{\psi}
        &\proves{\Gamma,\pvr:\ddiamond{\beta}{\phi}}{C}{\psi}}
{\proves{\Gamma}{\edcase{A}{B}{C}}{\psi}}
}{}
\end{calculus}
\end{calculuscollections}

\begin{calculuscollections}{0.4\columnwidth}
\begin{calculus}
\cinferenceRule[dchoiceIL|{$\langle\cup\rangle${I1}}]{}
{\linferenceRule[formula]
  {\proves{\Gamma}{M}{\ddiamond{\alpha}{\phi}}}
  {\proves{\Gamma}{\edinjL{M}}{\ddiamond{\alpha\cup\beta}{\phi}}}
}{}
\cinferenceRule[dchoiceIR|{$\langle\cup\rangle${I2}}]{}
{\linferenceRule[formula]
  {\proves{\Gamma}{M}{\ddiamond{\beta}{\phi}}}
  {\proves{\Gamma}{\edinjR{M}}{\ddiamond{\alpha\cup\beta}{\phi}}}
}{}
\cinferenceRule[bchoiceI|{$[\cup]${I}}]{}
{\linferenceRule[formula]
  {\proves{\Gamma}{M}{\dbox{\alpha}{\phi}} & \proves{\Gamma}{N}{\dbox{\beta}{\phi}}}
  {\proves{\Gamma}{\ebcons{M}{N}}{\dbox{\alpha\cup\beta}{\phi}}}
}{}
\cinferenceRule[dtestI|{$\langle?\rangle$}{I}]{}
{\linferenceRule[formula]
  {\proves{\Gamma}{M}{\phi} & \proves{\Gamma}{N}{\psi}}
  {\proves{\Gamma}{\edcons{M}{N}}{\ddiamond{\ptest{\phi}}{\psi}}}
}{}
\cinferenceRule[btestI|{$[?]$}{I}]{}
{\linferenceRule[formula]
  {\proves{\Gamma,\pvx:\phi}{M}{\psi}}
  {\proves{\Gamma}{(\eplam{\phi}{M})}{\dbox{\ptest{\phi}}{\psi}}}
}{}
\cinferenceRule[btestE|{$[?]$}{E}]{}
{\linferenceRule[formula]
  {\proves{\Gamma}{M}{\dbox{\ptest{\phi}}{\psi}} & \proves{\Gamma}{N}{\phi}}
  {\proves{\Gamma}{(\eapp{M}{N})}{\psi}}
}{}
\end{calculus}
\end{calculuscollections}%
\qquad%
\begin{calculuscollections}{0.4\columnwidth}
\begin{calculus}
\cinferenceRule[bchoiceEL|{$[\cup]${E1}}]{}
{\linferenceRule[formula]
  {\proves{\Gamma}{M}{\dbox{\alpha\cup\beta}{\phi}}}
  {\proves{\Gamma}{\ebprojL{M}}{\dbox{\alpha}{\phi}}}
}{}
\cinferenceRule[bchoiceER|{$[\cup]${E2}}]{}
{\linferenceRule[formula]
  {\proves{\Gamma}{M}{\dbox{\alpha\cup\beta}{\phi}}}
  {\proves{\Gamma}{\ebprojR{M}}{\dbox{\beta}{\phi}}}
}{}
\cinferenceRule[hyp|{hyp}]{}
{\linferenceRule[formula]
  {}
  {\proves{\Gamma,\pvx:\phi}{\pvx}{\phi}}
}{}
\cinferenceRule[dtestEL|{$\langle?\rangle$}{E1}]{}
{\linferenceRule[formula]
  {\proves{\Gamma}{M}{\ddiamond{\ptest{\phi}}{\psi}}}
  {\proves{\Gamma}{\edprojL{M}}{\phi}}
}{}
\cinferenceRule[dtestER|{$\langle?\rangle$}{E2}]{}
{\linferenceRule[formula]
  {\proves{\Gamma}{M}{\ddiamond{\ptest{\phi}}{\psi}}}
  {\proves{\Gamma}{\edprojR{M}}{\psi}}
}{}
\end{calculus}
\end{calculuscollections}
\caption{\CGL proof calculus: Propositional rules}
\label{fig:cgl-rules-prop}
\end{figure}
The propositional proof rules of \CGL are in \rref{fig:cgl-rules-prop}.
Formula $\dbox{\ptest{\phi}}{\psi}$ is constructive implication, so rule \irref{btestE} with proof term $\eapp{M}{N}$ eliminates $M$ by supplying an $N$ that proves the test condition.
Lambda terms $(\eplam{\phi}{M})$ are introduced by rule \irref{btestI} by extending the context $\Gamma$.
While this rule is standard, it is worth emphasizing that here $\pvx$ is a \emph{proof variable} for which a proof term (like $N$ in \irref{btestE}) may be substituted, and that the \emph{game state} is untouched by \irref{btestI}.
Constructive disjunction (between the branches $\ddiamond{\alpha}{\phi}$ and $\ddiamond{\beta}{\phi})$ is the choice $\ddiamond{\alpha\cup\beta}{\phi}$.
The introduction rules for injections are \irref{dchoiceIL} and \irref{dchoiceIR}, and case-analysis is performed with rule \irref{dchoiceE}, with two branches that prove a common consequence from each disjunct.
The cases $\ddiamond{\ptest{\phi}}{\psi}$ and $\dbox{\alpha\cup\beta}{\phi}$ are conjunctive.
Conjunctions are introduced by \irref{dtestI} and \irref{bchoiceI} as pairs, and eliminated by \irref{dtestEL}, \irref{dtestER}, \irref{bchoiceEL}, and \irref{bchoiceER} as projections.
Lastly, rule \irref{hyp} says formulas in the context hold by assumption.

We now begin considering non-propositional rules, starting with the simplest ones.
\begin{figure}
  \centering
\begin{calculuscollections}{\columnwidth}
\begin{calculus}
\cinferenceRule[drcase|{$\langle*\rangle${C}}]{}
{
\linferenceRule[formula]
{\proves{\Gamma}{A}{\ddiamond{\prepeat{\alpha}}{\phi}} & \proves{\Gamma,\pvs:\phi}{B}{\psi} & \proves{\Gamma,\pvg:\ddiamond{\alpha}{\ddiamond{\prepeat{\alpha}}{\phi}}}{C}{\psi}}
{\proves{\Gamma}{\ercase{A}{B}{C}}{\psi}}
}{}
\cinferenceRule[mon|{M}]{}
{\linferenceRule[formula]
  {\proves{\Gamma}{M}{\ddiamond{\alpha}{\phi}} & \proves{\earen{\Gamma}{\alpha},\pvx:\phi}{N}{\psi}}
  {\proves{\Gamma}{\emon{M}{N}{\pvx}}{\ddiamond{\alpha}{\psi}}}
}{}
\end{calculus}
\end{calculuscollections}

\begin{calculuscollections}{3in}
\begin{calculus}
\cinferenceRule[bunroll|{$[*]${E}}]{}
{\linferenceRule[formula]
  {\proves{\Gamma}{M}{\dbox{\prepeat{\alpha}}{\phi}}}
  {\proves{\Gamma}{\ebunroll{M}}{\phi \land \dbox{\alpha}{\dbox{\prepeat{\alpha}}{\phi}}}}
}{}
\cinferenceRule[dstop|{$\langle*\rangle$S}]{}
{
\linferenceRule[formula]
{\proves{\Gamma}{M}{\phi}}  
{\proves{\Gamma}{\estop{M}}{\ddiamond{\prepeat{\alpha}}{\phi}}}
}{}
\cinferenceRule[dgo|{$\langle*\rangle$G}]{}
{
\linferenceRule[formula]
{\proves{\Gamma}{M}{\phi \lor \ddiamond{\alpha}{\ddiamond{\prepeat{\alpha}}{\phi}}}}
{\proves{\Gamma}{\ego{M}}{\ddiamond{\prepeat{\alpha}}{\phi}}}
}{}
\end{calculus}\hspace{0.3in}%
\begin{calculus}
\cinferenceRule[dualI|{$\lstrike{}^d\rstrike$}I]{}
{\linferenceRule[formula]
  {\proves{\Gamma}{M}{\pmodality{\alpha}{\phi}}}
  {\proves{\Gamma}{\eSwap{M}}{\dmodality{\pdual{\alpha}}{\phi}}}
}{}
\cinferenceRule[broll|{$[*]${R}}]{}
{\linferenceRule[formula]
  {\proves{\Gamma}{M}{\phi \land \dbox{\alpha}{\dbox{\prepeat{\alpha}}{\phi}}}}
  {\proves{\Gamma}{\ebroll{M}}{\dbox{\prepeat{\alpha}}{\phi}}}
}{}
\cinferenceRule[seqI|{$\lstrike{;}\rstrike$}I]{}
{\linferenceRule[formula]
  {\proves{\Gamma}{M}{\dmodality{\alpha}{\dmodality{\beta}{\phi}}}}
  {\proves{\Gamma}{\eSeq{M}}{\dmodality{\alpha;\beta}{\phi}}}
}{}
\end{calculus}
\end{calculuscollections}
\caption{\CGL proof calculus: Some non-propositional rules}
\label{fig:cgl-compos}
\end{figure}
The majority of the rules  in \rref{fig:cgl-compos}, while thoroughly useful in proofs, are computationally trivial.
The repetition rules (\irref{bunroll},\irref{broll}) fold and unfold the notion of repetition as iteration.
The rolling and unrolling terms are named in analogy to the \emph{iso-recursive} treatment of recursive types~\cite{DBLP:conf/tldi/VanderwaartDPCHC03}, where an explicit operation is used to expand and collapse the recursive definition of a type.

Rules \irref{drcase},\irref{dstop},\irref{dgo} are the destructor and injectors for $\ddiamond{\prepeat{\alpha}}{\phi}$,  which are similar to those for $\ddiamond{\alpha\cup\beta}{\phi}$.
The duality rules (\irref{dualI}) say the dual game is proved by proving the game where roles are reversed.
The sequencing rules (\irref{seqI}) say a sequential game is played by playing the first game with the goal of reaching a state where the second game is winnable.

Among these rules, monotonicity \irref{mon} is especially computationally rich.
The notation $\earen{\Gamma}{\alpha}$ says that in the second premiss, the assumptions in $\Gamma$ have all bound variables of $\alpha$ (written $\boundvars{\alpha}$) renamed to fresh variables $\vec{y}$ for completeness.
In practice, $\Gamma$ usually contains some assumptions on variables that are not bound, which we wish to access without writing them explicitly in $\phi$.
Rule \irref{mon} is used to execute programs right-to-left, giving shorter, more efficient proofs.
It can also be used to derive the Hoare-logical sequential composition rule, which is frequently used to reduce the number of case splits.
Note that like every \GL, \CGL is subnormal, so the modal modus ponens axiom K and G\"{o}del generalization (or necessitation) rule G are not sound, and \irref{mon} takes over much of the role they usually serve.
On the surface, \irref{mon} simply says games are monotonic: a game's goal proposition may freely be replaced with a weaker one.
From a computational perspective, \rref{sec:operational} will show  that rule \irref{mon} can be (lazily) eliminated.
Moreover, \irref{mon} is an \emph{admissible} rule, one whose instances can all be derived from existing rules.
When proofs are written right-to-left with \irref{mon}, the normalization relation translates them to left-to-right normal proofs.
Note also that in checking $\emon{M}{N}{\pvx},$ the context $\Gamma$ has the bound variables $\alpha$ renamed freshly to some $\vec{y}$ within $N,$ as required to maintain soundness across execution of $\alpha$.

Next, we consider \emph{first-order} rules, i.e., those which deal with first-order programs that modify \emph{program} variables.
The first-order rules are given in \rref{fig:cgl-first-order}.
In \irref{drandomE}, $\freevars{\psi}$ are the \emph{free variables} of $\psi,$ the variables which can influence its meaning.
\begin{figure}
  \centering
\begin{calculuscollections}{\columnwidth}
\begin{calculus}
\cinferenceRule[asgnI|{$\lstrike{:=}\rstrike$}I]{}
{\linferenceRule[formula]
{\proves{\eren{\Gamma}{x}{y},\pvx:(x=\eren{f}{x}{y})}{M}{\phi}}
{\proves{\Gamma}{\eAsgneq{y}{x}{\pvx}{M}}{\dmodality{\humod{x}{f}}{\phi}}}
}{$y$ fresh}
\cinferenceRule[drandomI|{$\langle{:}*\rangle${I}}]{}
{
\linferenceRule[formula]
{\proves{\eren{\Gamma}{x}{y},\pvx:(x=\eren{f}{x}{y})}{M}{\phi}}
{\proves{\Gamma}{\etcons{f}{M}}{\ddiamond{\prandom{x}}{\phi}}}
}{$y,\pvx$ fresh, $f$ comp.}
\cinferenceRule[drandomE|{$\langle{:}*\rangle${E}}]{}
{
\linferenceRule[formula]
{\proves{\Gamma}{M}{\ddiamond{\prandom{x}}{\phi}} & \proves{\eren{\Gamma}{x}{y},\pvx:\phi}{N}{\psi}}
{\proves{\Gamma}{\eunpack{M}{N}}{\psi}}
}{$y$ fresh, $x \notin \freevars{\psi}$}
\cinferenceRule[brandomI|{$[{:}{*}]${I}}]{}
{
\linferenceRule[formula]
{\proves{\eren{\Gamma}{x}{y}}{M}{\phi}}
{\proves{\Gamma}{(\etlam{\allrat}{M})}{\dbox{\prandom{x}}{\phi}}}
}{y\text{ fresh}}
\cinferenceRule[brandomE|{$[{:}{*}]${E}}]{}
{
      \linferenceRule[formula]
        {\proves{\Gamma}{M}{\dbox{\prandom{x}}{\phi}}}
        {\proves{\Gamma}{(\eapp{M}{f})}{\tsub{\phi}{x}{f}}}
}{$\tsub{\phi}{x}{f}$ admiss.}
\end{calculus}
\end{calculuscollections}
\caption{\CGL proof calculus: first-order games}
\label{fig:cgl-first-order}
\end{figure}
Nondeterministic assignment provides quantification over rational-valued \emph{program} variables.
Rule \irref{brandomI} is universal, with proof term ${(\etlam{\allrat}{M})}$.
While this notation is suggestive, the difference vs.\ the function proof term $(\eplam{\phi}{M})$ is essential: the proof term $M$ is checked (resp.\ evaluated) in a state where the program variable $x$ has changed from its initial value.
For soundness, \irref{brandomI} renames $x$ to fresh program variable $y$ throughout context $\Gamma,$ written $\eren{\Gamma}{x}{y}$.
This means that $M$ can freely refer to all facts of the full context, but they now refer to the state as it was before $x$ received a new value.
Elimination \irref{brandomE} then allows instantiating $x$ to a term $f$.
Existential quantification is introduced by \irref{drandomI} whose proof term $\etcons{f}{M}$ is like a dependent pair plus bound renaming of $x$ to $y$.
The witness $f$ is an arbitrary computable term, as always.
We write $\ietcons{x}{f}{M}$ for short when $y$ is not referenced in $M$.
It is eliminated in \irref{drandomE} by unpacking the pair, with side condition $x \notin \freevars{\psi}$ for soundness.
The assignment rules \irref{asgnI} do not quantify, per se, but always update $x$ to the value of the term $f,$ and in doing so introduce an assumption that $x$ and $f$ (suitably renamed) are now equal.
In \irref{drandomI}  and \irref{asgnI}, program variable $y$ is fresh.

\begin{figure}
\begin{calculuscollections}{\columnwidth}
\begin{calculus}
\cinferenceRule[dloopI|{$\langle*\rangle${I}}]{}
{\linferenceRule[formula]
{\deduce{\proves{\pvx:\conv,\pvy:\met_0 = \met \metgr \metz}{B}{\ddiamond{\alpha}{(\conv\land \met_0 \metgr \met)}}}
        {\quad\ \,\proves{\Gamma}{A}{\conv}}
 & \proves{\pvx:\conv,\pvy:\met = \metz}{C}{\phi}}
{\proves{\Gamma}{\efor{A}{B}{C}}{\ddiamond{\prepeat{\alpha}}{\phi}}}
}{$\met_0$ fresh}
\end{calculus}
\\
\begin{calculus}
\cinferenceRule[bloopI|{$[*]$I}]{}
{\linferenceRule[formula]
  {\deduce{\proves{\pvx:J}{N}{\dbox{\alpha}{J}}}
    {\,\proves{\Gamma}{M}{J}}
  & \proves{\pvx:J}{O}{\phi}}
  {\proves{\Gamma}{(\erep{M}{N}{\pvx:J}{O})}{\dbox{\prepeat{\alpha}}{\phi}}}
}{}
\cinferenceRule[metsplit|split]{}
{\proves{\Gamma}{(\esplit{f}{g})}{f \leq g \lor f > g}}
{}
\end{calculus}%
\hfill%
\begin{calculus}
\cinferenceRule[dloopE|FP]{}
{
\linferenceRule[formula]
{\deduce{\proves{\pvs:\phi}{B}{\psi}}
        {\quad\quad\quad\proves{\Gamma}{A}{\ddiamond{\prepeat{\alpha}}{\phi}}}
 & \proves{\pvg:\ddiamond{\alpha}{\psi}}{C}{\psi}}
{\proves{\Gamma}{\efp{A}{B}{C}}{\psi}}
}{}
\end{calculus}
\end{calculuscollections}
\caption{\CGL proof calculus: loops}
\label{fig:cgl-loops}
\end{figure}
The looping rules in \rref{fig:cgl-loops}, especially \irref{dloopI}, are arguably the most sophisticated in \CGL.
Rule \irref{dloopI} provides a strategy to repeat a game $\alpha$ until the postcondition $\phi$ holds.
This is done by exhibiting a convergence predicate $\conv$ and termination metric $\met$ with terminal value $\metz$ and well-ordering $\metgr$.
Proof term $A$ shows $\conv$ holds initially.
Proof term $B$ guarantees $\met$ decreases with every iteration where $\met_0$ is a fresh metric variable which is equal to $\met$ at the antecedent of $B$ and is never modified.
Proof term $C$ allows any postcondition $\phi$ which follows from convergence $\conv \land \met = \metz$.
Proof term $\efor{A}{B}{C}$ suggests the computational interpretation as a for-loop:
proof $A$ shows the convergence predicate holds in the initial state,  $B$ shows that each step reduces the termination metric while maintaining the predicate, and $C$ shows that the postcondition follows from the convergence predicate upon termination.
The game $\alpha$ repeats until convergence is reached ($\met = \metz$).
By the assumption that metrics are well-founded, convergence is guaranteed in finitely (but arbitrarily) many iterations.

A na\"{i}ve, albeit correct, reading of rule \irref{dloopI} says $\met$ is literally some term $f$.
If lexicographic or otherwise non-scalar metrics should be needed, it suffices to interpret $\conv$ and $\met_0 \metgr \met$ as formulas over several scalar variables.

Rule \irref{dloopE} says $\ddiamond{\prepeat{\alpha}}{\phi}$ is a least pre-fixed-point.
That is, if we wish to show a formula $\psi$ holds now, we show that $\psi$ is any pre-fixed-point, then it must hold as it is no lesser than $\phi$.
Rule \irref{bloopI} is the well-understood induction rule for loops, which applies as well to repeated games.
Premiss $O$ ensures \irref{bloopI} supports any provable postcondition, which is crucial for eliminating \irref{mon} in \rref{lem:progress}.
The elimination form for $\dbox{\prepeat{\alpha}}{\phi}$ is simply \irref{bunroll}.
Like any program logic, reasoning in \CGL consists of first applying program-logic rules to decompose a program until the program has been entirely eliminated, then applying first-order logic principles at the leaves of the proof.
The \emph{constructive} theory of rationals is undecidable because it can express the undecidable~\cite{DBLP:journals/jsyml/Robinson49} \emph{classical} theory of rationals.
Thus facts about rationals require proof in practice.
For the sake of space and since our focus is on program reasoning, we defer an axiomatization of rational arithmetic to future work.
We provide a (non-effective!) rule \irref{QE} which says valid first-order formulas are provable.
\[\cinferenceRule[QE|FO]{}
{\linferenceRule[formula]
{\proves{\Gamma}{M}{\rho}}
{\proves{\Gamma}{\eQE{\phi}{M}}{\phi}}}{\text{exists }\myaa \text{ s.t. }\{\myaa\} \times \allstate \subseteq \fintR{\rho\limply\phi},\ \rho,\phi\text{ F.O.}}
\]
An effective special case of \irref{QE} is \irref{metsplit} (\rref{fig:cgl-loops}), which says all term comparisons are decidable.
Rule \irref{metsplit} can be generalized to decide termination metrics $(\met = \metz \lor \met \metgr \metz)$.
Rule \irref{ghost} says the value of term $f$ can be remembered in fresh ghost variable $x$:
\[\cinferenceRule[ghost|iG]{}
{\linferenceRule[formula]
  {\proves{\Gamma,\pvx:x=f}{M}{\phi}}
  {\proves{\Gamma}{\eghost{x}{f}{\pvx}{M}}{\phi}}
}
{\m{x}\text{ fresh except free in }\m{M, \pvx}\text{ fresh}}
\]
Rule \irref{ghost} can be defined using arithmetic and with quantifiers:
\[\eghost{x}{f}{\pvx}{M} \equiv \eapp{\eapp{(\elam{x}{\allrat}{(\elam{\pvx}{(x=f)}{M})})}{f}}{(\eQE{f=f}{})}\]

\paragraph{What's Novel in the \CGL Calculus?}
\CGL extends first-order reasoning with game reasoning (sequencing~\cite{kamide2010strong}, assignments, iteration, and duality).
The combination of first-order reasoning with game reasoning is synergistic: for example, repetition games are known to be more expressive than repetition systems~\cite{DBLP:journals/tocl/Platzer15}.
We give a new natural-deduction formulation of monotonicity.
Monotonicity is admissible and normalization translates monotonicity proofs into monotonicity-free proofs.
In doing so, normalization shows that right-to-left proofs can be (lazily) rewritten as left-to-right.
Additionally, first-order games are rife with changing state, and soundness requires careful management of the context $\Gamma$.
\iflongversion
In the appendix (\rref{app:example-proofs}), we use our calculus to prove the example formulas.
\else
The extended version~\cite{br2020constructive} uses our calculus to prove the example formulas.
\fi

\section{Theory: Soundness}
\label{sec:soundness}
Full versions of proofs outlined in this paper are given in the 
\iflongversion
appendix (\rref{app:proofs}).
\else
extended version~\cite{br2020constructive}.
\fi
We have introduced a proof calculus for \CGL which can prove winning strategies for $\nim$ and $\cake$.
For any new proof calculus, it is essential to convince ourselves of our soundness, which can be done within several prominent schools of thought.
In proof-theoretic semantics, for example, the proof rules are taken as the ground truth, but are validated by showing the rules obey expected properties such as harmony or, for a sequent calculus, cut-elimination.
While we will investigate proof terms separately (\rref{sec:op-theory}), we are already equipped to show soundness by direct appeal to the realizability semantics (\rref{sec:rel-sem}), which we take as an independent notion of ground truth.
We show soundness of \CGL proof rules against the realizability semantics, i.e., that every provable natural-deduction sequent is valid.
An advantage of this approach is that it explicitly connects the notions of provability and computability!
We build up to the proof of soundness by proving lemmas on structurality, renaming and substitution.

\begin{lemma}[Structurality]
\label{lem:structurality}
The structural rules \irref{weak}, \irref{exchange}, and \irref{contract} are admissible, i.e., the conclusions are provable whenever the premisses are provable.

\begin{calculus}
\cinferenceRule[weak|W]{$x$ fresh}
{\linferenceRule[formula]{\proves{\Gamma}{M}{\phi}}{\proves{\Gamma,\pvx:\psi}{M}{\phi}}}
{}
\end{calculus}
\hspace{0.18in}
\begin{calculus}
\cinferenceRule[exchange|X]{}
{\linferenceRule[formula]{\proves{\Gamma,\pvx:\phi,\pvy:\psi}{M}{\rho}}{\proves{\Gamma,\pvy:\psi,\pvx:\phi}{M}{\rho}}}
{}
\end{calculus}
\hspace{0.18in}
\begin{calculus}
\cinferenceRule[contract|C]{}
{\linferenceRule[formula]{\proves{\Gamma,\pvx:\phi,\pvy:\phi}{M}{\rho}}{\proves{\Gamma,\pvx:\phi}{\esub{M}{\pvy}{\pvx}}{\rho}}}
{}
\end{calculus}
\end{lemma}
\begin{proof}[Proof summary]
Each rule is proved admissible by induction on $M$.
Observe that the only premisses regarding $\Gamma$ are of the form $\Gamma(\pvx) = \phi$, which are preserved under weakening.
Premisses are trivially preserved under exchange because contexts are treated as sets, and preserved modulo renaming by contraction as it suffices to have \emph{any} assumption of a given formula, regardless its name.
The context $\Gamma$ is allowed to vary in applications of the inductive hypothesis, e.g., in rules that bind program variables.
Some rules discard $\Gamma$ in checking the subterms inductively, in which case the IH need not be applied at all.
\end{proof}
\begin{lemma}[Uniform renaming]
  Let $\eren{M}{x}{y}$ be the renaming of program variable $x$ to $y$ (and vice-versa) within $M$, even when neither $x$ nor $y$ is fresh.
  If $\proves{\Gamma}{M}{\phi}$ then $\proves{\eren{\Gamma}{x}{y}}{\eren{M}{x}{y}}{\eren{\phi}{x}{y}}$.
\label{lem:renaming}
\end{lemma}
\begin{proof}[Proof summary]
Straightforward induction on the structure of $M$.
Renaming within proof terms (whose definition we omit as it is quite tedious) follows the usual homomorphisms, from which the inductive cases follow.
In the case that $M$ is a proof variable $z$, then $\left(\eren{\Gamma}{x}{y}\right)(z) =  \eren{\Gamma(z)}{x}{y}$ from which the case follows.
The interesting cases are those which modify program variables, e.g., $\edasgn{w}{z}{\pvx}{M}$.
The bound variable $z$ is renamed to $\eren{z}{x}{y},$ while the auxiliary variable $w$ is $\alpha$-varied if necessary to maintain freshness.
Renaming then happens recursively in $M$.
\end{proof}
Substitution will use proofs of coincidence and bound effect lemmas.
\begin{lemma}[Coincidence]
Only the free variables of an expression influence its semantics.
\label{lem:coincidental}
\end{lemma}
\begin{lemma}[Bound effect]
Only the bound variables of a game are modified by execution.
\label{lem:effects-bound}
\end{lemma}
\begin{proof}[Summary]
By induction on the expression, in analogy to \cite{DBLP:journals/jar/Platzer17}.
\end{proof}
\begin{definition}[Term substitution admissibility]
For simplicity, we say $\tsub{\phi}{x}{f}$ (likewise for context $\Gamma,$ term $f,$ game $\alpha,$ and proof term $M$) is admissible if $\phi$ binds neither $x$ nor free variables of $f$.
\label{def:lem-admit}
\end{definition}
The latter condition can be relaxed in practice~\cite{DBLP:conf/cade/Platzer18} to requiring $\phi$ does not mention $x$ under bindings of free variables.
\begin{lemma}[Arithmetic-term substitution]
 If $\proves{\Gamma}{M}{\phi}$ and the substitutions $\tsub{\Gamma}{x}f,$ $\tsub{M}{x}{f},$ and $\tsub{\phi}{x}{f}$ are admissible, then
 $\proves{\tsub{\Gamma}{x}{f}}{\tsub{M}{x}{f}}{\tsub{\phi}{x}{f}}$.
\label{lem:atsub}
\end{lemma}
\begin{proof}[Summary]
By induction on $M$.
Admissibility holds recursively, and so can be assumed at each step of the induction.
For non-atomic $M$ that bind no variables, the proof follows from the inductive hypotheses.
For $M$ that bind variables, we appeal to \rref{lem:coincidental} and \rref{lem:effects-bound}.
\end{proof}
Just as arithmetic terms are substituted for program variables, proof terms are substituted for proof variables.
\begin{lemma}[Proof term substitution]
  Let $\esub{M}{\pvx}{N}$ substitute $N$ for $\pvx$ in $M$, avoiding capture.
  If $\proves{\Gamma,\pvx:\psi}{M}{\phi}$ and $\proves{\Gamma}{N}{\psi}$ then $\proves{\Gamma}{\esub{M}{\pvx}{N}}{\phi}$.
\label{lem:pt-substitute}
\end{lemma}
\begin{proof}
By induction on $M,$ appealing to renaming, coincidence, and bound effect.
When substituting $N$ for $\pvx$ into a term that binds program variables such as $\edasgn{y}{z}{\pvy}{M},$ we avoid capture by renaming within occurrences of $N$ in the recursive call, i.e., $\esub{\edasgn{y}{z}{\pvy}{M}}{\pvx}{N} = \edasgn{y}{z}{\pvy}{\esub{M}{\pvx}{\eren{N}{y}{z}}},$ preserving soundness by \rref{lem:renaming}.
\end{proof}
Soundness of the proof calculus exploits renaming and substitution.
\begin{theorem}[Soundness of proof calculus]
  If $\proves{\Gamma}{M}{\phi}$ then $(\seq{\Gamma}{\phi})$ is valid.
  As a special case for empty context $\Gemp$, if $\proves{\Gemp}{M}{\phi},$ then $\phi$ is valid.
\label{thm:proof-calculus-sound}
\end{theorem}
\begin{proof}[Proof summary]
By induction on $M$.
Modus ponens case $\eapp{A}{B}$ reduces to \rref{lem:pt-substitute}.
Cases that bind program variables, such as assignment, hold by \rref{lem:atsub} and \rref{lem:renaming}.
Rule \irref{weak} is employed when substituting under a binder.
\end{proof}
We have now shown that the \CGL proof calculus is sound, the \emph{sine qua non} condition of any proof system.
Because soundness was w.r.t.\ a realizability semantics, we have shown \CGL is constructive in the sense that provable formulas correspond to realizable strategies, i.e., imperative programs executed in an adversarial environment.
We will revisit constructivity again in \rref{sec:op-theory} from the perspective of proof terms as \emph{functional} programs.

\section{Operational Semantics}
\label{sec:operational}
The Curry-Howard interpretation of games is not complete without exploring the interpretation of proof simplification as normalization of functional programs.
To this end, we now introduce a structural operational semantics for \CGL proof terms.
This semantics provides a view complementary to the realizability semantics: not only do provable formulas correspond to realizers, but proof terms can be directly executed as functional programs, resulting in a \emph{normal} proof term.
The chief subtlety of our operational semantics is that in contrast to realizer execution, proof simplification is a static operation, and thus does not inspect game state.
Thus the normal form of a proof which branches on the game state is, of necessity, also a proof which branches on the game state.
This static-dynamic phase separation need not be mysterious: it is analogous to the monadic phase separation between a functional program which returns an imperative command vs.\ the execution of the returned command.
While the primary motivation for our operational semantics is to complete the Curry-Howard interpretation, proof normalization is also helpful when implementing software tools which process proof artifacts, since code that consumes a normal proof is in general easier to implement than code that consumes an arbitrary proof.

The operational semantics consist of two main judgments:
$\isnorm{M}$ says that $M$ is a normal form, while $M \stepsto M'$ says that $M$ reduces to term $M'$ in one step of evaluation.
A normal proof is allowed a $\kwcase$ operation at the top-level, either $\edcase{A}{B}{C}$ or $\ercase{A}{B}{C}$.
Normal proofs $M$ without state-casing are called \emph{simple}, written $\issimp{M}$.
The requirement that cases are top-level ensures that proofs which differ only in where the case was applied share a common normal form, and ensures that $\beta$-reduction is never blocked by a case interceding between introduction-elimination pairs.
Top-level case analyses are analogous to case-tree normal forms in lambda calculi with coproducts~\cite{DBLP:conf/lics/AltenkirchDHS01}.
Reduction of proof terms is eager.
\begin{definition}[Normal forms]
We say $M$ is \emph{simple}, written $\issimp{M},$ if eliminators occur only under binders.
We say $M$ is \emph{normal}, written $\isnorm{M},$ if $\issimp{M}$ or $M$ has shape $\edcase{A}{B}{C}$ or $\ercase{A}{B}{C}$
where $A$ is a term such as $(\esplit{f}{g}{M})$ that inspects the state.
Subterms $B$ and $C$ need not be normal since they occur under the binding of $\pvl$ or $\pvr$ (resp.\ $\pvs$ or $\pvg$).
\label{def:normal}
\end{definition}
That is, a normal term has no top-level beta-redexes, and  state-dependent cases are top-level.
We consider rules \irref{broll}, \irref{brandomI}, \irref{btestI}, and \irref{asgnI} binding.
Rules such as \irref{dloopI} have multiple premisses but bind only one.
While \irref{broll} does not introduce a proof variable, it is rather considered binding to prevent divergence, which is in keeping with a coinductive understanding of formula $\dbox{\prepeat{\alpha}}{\phi}$.
If we did not care whether terms diverge, we could have made \irref{broll} non-binding.

For the sake of space, this section focuses on the $\beta$-rules (\rref{fig:opbeta}).
The full calculus, given in%
\iflongversion%
the appendix (\rref{app:pc-full}),
\else%
the extended version~\cite{br2020constructive},
\fi%
includes structural and commuting-conversion rules,
as well as what we call \emph{monotonicity conversion} rules: a proof term $\emon{M}{N}{\pvx}$ is simplified by structural recursion on $M$.
The capture-avoiding substitution of $M$ for $\pvx$ in $N$ is written $\esub{N}{\pvx}{M}$ (\rref{lem:pt-substitute}).
\begin{figure}
\centering
\begin{calculuscollections}{0.5\textwidth}
\begin{calculus}
\cinferenceRule[appBeta|{$\lambda\phi\beta$}]{}{\eapp{(\eplam{\phi}{M})}{N} \stepsto \esub{M}{\pvx}{N}}{}
\cinferenceRule[brandomBeta|{$\lambda\beta$}]{}{\eapp{(\etlam{\allrat}{M})}{f} \stepsto \tsub{M}{x}{f}}{}
\cinferenceRule[projLBeta|{$\pi_1\beta$}]{}{\eProjL{\eCons{M}{N}} \stepsto M}{}
\cinferenceRule[projRBeta|{$\pi_2\beta$}]{}{\eProjR{\eCons{M}{N}} \stepsto N}{}
\end{calculus}
\end{calculuscollections}
\begin{calculuscollections}{0.5\textwidth}
\begin{calculus}
\cinferenceRule[caseBetaL|{case$\beta$L}]{}{\eCase{\eInjL{A}}{B}{C} \stepsto \esub{B}{\ell}{A}}{}
\cinferenceRule[caseBetaR|{{case}$\beta$R}]{}{\eCase{\eInjR{A}}{B}{C} \stepsto \esub{C}{r}{A}}{}
\cinferenceRule[unrollBeta|{{unroll}$\beta$}]{}{\ebunroll{\ebroll{M}} \stepsto M}{}
\end{calculus}
\end{calculuscollections}

\begin{calculuscollections}{0.5\textwidth}
\begin{calculus}
\cinferenceRule[unpackBeta|{{unpack}$\beta$}]{}{\eunpack{\etconsgen{x}{y}{\pvy}{f}{M}}{N} \stepsto \eren{(\eghost{x}{\eren{f}{x}{y}}{\pvy}{\esub{N}{\pvx}{M}})}{y}{x}}{}
\cinferenceRule[fpBeta|{{FP}$\beta$}]{}{\efp{D}{B}{C} \stepsto (\ercase{D}{B}{\esub{C}{\pvg}{(\emon{\pvg}{\efp{z}{B}{C}}{z})}})}{}
\cinferenceRule[repBeta|{{rep}$\beta$}]{}{(\erep{M}{N}{\pvx:J}{O}) \stepsto \ebroll{\edcons{M}{\emon{(\esub{N}{\pvx}{M})}{(\erep{\pvy}{N}{\pvx:J}{O})}{\pvy}}}}{}
\cinferenceRule[forBeta|{{for}$\beta$}]{}{\efor{A}{B}{C} \stepsto}{}
\end{calculus}\\[-0.22in]
\begin{align*}
&\ecaseHead{\esplit{\met}{0}{}}\\
&\ \ \ecaseLeft{\pvl}{\estop{\esub{C}{(\pvx,\pvy)}{(A,\pvl)}}}\\
&\ecaseRight{\pvr}{\eghost{\met_0}{\met}{\textit{\pvrr}}{\ego{(\emon{(\esub{B}{\pvx,\pvy}{A,\edcons{\textit{\pvrr}}{\pvr}})}{(\efor{\eprojL{\pvz}}{B}{C})}{\pvz})}}}\ecaseEnd
  \end{align*}
\end{calculuscollections}
\caption{Operational semantics: $\beta$-rules}
\label{fig:opbeta}
\end{figure}
The propositional cases \irref{appBeta}, \irref{brandomBeta}, \irref{caseBetaL}, \irref{caseBetaR}, \irref{projLBeta}, and \irref{projRBeta} are standard reductions for applications, cases, and projections.
Projection terms $\eprojL{M}$ and $\eprojR{M}$ should not be confused with projection realizers $\rzFst{\myaa}$ and $\rzSnd{\myaa}$.
Rule \irref{unpackBeta} makes the witness of an existential available in its client as a ghost variable.

Rule \irref{fpBeta}, \irref{repBeta}, and \irref{forBeta} reduce introductions and eliminations of loops.
Rule \irref{fpBeta}, which reduces a proof $\efp{A}{B}{C}$ says that if $\prepeat{\alpha}$ has already terminated according to $A$, then $B$ proves the postcondition.
Else the inductive step $C$ applies, but every reference to the IH $\pvg$ is transformed to a recursive application of \irref{dloopE}.
If $A$ uses only \irref{dstop} and \irref{dgo}, then $\efp{A}{B}{C}$ reduces to a simple term, else if $A$ uses \irref{dloopI}, then $\efp{A}{B}{C}$ reduces to a case.
Rule \irref{repBeta} says loop induction $(\erep{M}{N}{\pvx:J}{O})$ reduces to a delayed pair of the ``stop'' and ``go'' cases, where the ``go'' case first shows $\dbox{\alpha}{J},$ for loop invariant $J,$ then expands $J \limply \dbox{\prepeat{\alpha}}{\phi}$ in the postcondition.
Note the laziness of $[\kwroll]$ is essential for normalization: when $(\erep{M}{N}{\pvx:J}{O})$ is understood as a coinductive proof, it is clear that normalization would diverge if \irref{repBeta} were applied indefinitely.
Rule \irref{forBeta} for $\efor{A}{B}{C}$ checks whether the termination metric $\met$ has reached terminal value $\metz$.
If so, the loop $\langle{\kwstop}\rangle$'s and $A$ proves it has converged.
Else, we remember $\met$'s value in a ghost term $\met_0,$ and $\langle{\kwgo}\rangle$ forward, supplying $A$ and $\edcons{\pvr}{\pvrr}$ to satisfy the preconditions of inductive step $B,$ then execute the loop $\efor{\eprojL{t}}{B}{C}$ in the postcondition.
Rule \irref{forBeta} reflects the fact that the exact number of iterations is state dependent.

We discuss the structural, commuting conversion, and monotonicity conversion rules for left injections as an example, with the full calculus in
\iflongversion%
\rref{app:pc-full}.
\else%
~\cite{br2020constructive}.
\fi%
Structural rule \irref{injLS} evaluates term $M$ under an injector.
Commuting conversion rule \irref{injLC} normalizes an injection of a case to a case with injectors on each branch.
Monotonicity conversion rule \irref{injLMon} simplifies a monotonicity proof of an injection to an injection of a monotonicity proof.

\begin{figure}[th!]
\begin{calculuscollections}{\textwidth}
\begin{calculus}
\cinferenceRule[injLS|{$\einjL{}$}S]{}{\linferenceRule[formula]
{M \stepsto M'}
{\eInjL{M} \stepsto \eInjL{M'}}}{}
\cinferenceRule[injLC|{$\eInjL{}$C}]{}
  {\eInjL{\ecasegen{A}{\pvx}{B}{\pvy}{C}}
   \stepsto \ecasegen{A}{\pvx}{\eInjL{B}}{\pvy}{\eInjL{C}}}
{}
\cinferenceRule[injLMon|{$\eInjL{}\circ$}]{}{\emon{\eInjL{M}}{N}{\pvx}   \stepsto \eInjL{(\emon{M}{N}{\pvx})}}{}
\end{calculus}
\end{calculuscollections}
\caption{Operational semantics: structural, commuting conversion, monotonicity rules}
\label{fig:op-rest}
\end{figure}

\section{Theory: Constructivity}
\label{sec:op-theory}
We now complete the study of \CGL's constructivity.
We validate the operational semantics on proof terms by proving that progress and preservation hold, and thus the \CGL proof calculus is sound as a type system for the functional programming language of \CGL proof terms.
\begin{lemma}[Progress]
If $\proves{\Gemp}{M}{\phi},$ then either $M$ is normal or $M \stepsto M'$ for some $M'$.
\label{lem:progress}
\end{lemma}
\begin{proof}[Summary]
By induction on the proof term $M$.
If $M$ is an introduction rule, by the inductive hypotheses the subterms are well-typed.
If they are all simple, then $\issimp{M}$. If some subterm (not under a binder) steps, then $M$ steps by a structural rule.
Else some subterm is an irreducible case expression not under a binder, it lifts by the commuting conversion rule.
If $M$ is an elimination rule, structural and commuting conversion rules are applied as above.
Else by \rref{def:normal} the subterm is an introduction rule, and $M$ reduces with a $\beta$-rule.
Lastly, if $M$ has form $\emon{A}{B}{x}$ and $\issimp{A},$ then by \rref{def:normal} $A$ is an introduction form, thus reduced by some monotonicity conversion rule.
\end{proof}

\begin{lemma}[Preservation]
Let $\stepsto^*$ be the reflexive, transitive closure of the $\stepsto$ relation.
If \m{\proves{\Gemp}{M}{\phi}} and \m{M {\stepsto^*} M',} then \m{\proves{\Gemp}{M'}{\phi}}
\end{lemma}
\begin{proof}[Summary]
Induct on the derivation \m{M {{\stepsto}^*} M'}, then induct on $M \stepsto M'$.
The $\beta$ cases follow by \rref{lem:pt-substitute} (for base constructs), and \rref{lem:pt-substitute} and \rref{lem:renaming} (for assignments).
C-rules and $\circ$-rules lift across binders, soundly by \irref{weak}.
S-rules are direct by IH.
\end{proof}
We gave two understandings of proofs in \CGL, as imperative strategies and as functional programs.
We now give a final perspective: \CGL proofs support synthesis in principle, one of our main motivations.
Formally, the Existential Property (EP) and Disjunction Property (DP) justify synthesis~\cite{degen2006towards} for existentials and disjunctions: whenever an existential or disjunction has a proof, then we can compute some instance or disjunct that has a proof.
We state and prove an EP and DP for \CGL, then introduce a Strategy Property, their counterpart for synthesizing strategies from game modalities.
It is important to our EP that terms are arbitrary computable functions, because more simplistic term languages are often too weak to witness the existentials they induce.
\begin{example}[Rich terms help]
Formulas over polynomial terms can have non-polynomial witnesses.
\end{example}
Let $\phi \equiv (x = y \land x \geq 0) \lor (x = -y \land x < 0)$. 
Then $f = \abs{x}$ witnesses $\lexists[\allrat]{y}{\phi}$.
\begin{lemma}[Existential Property]
If $\proves{\Gamma}{M}{(\lexists[\allrat]{x}{\phi})}$ then there exists a term $f$ and realizer $\ab$ such that for all $(\myaa,\om) \in \cintR{\Gamma},$
we have
$(\rzApp{\ab}{\myaa},\ssub{\om}{x}{f(\om)}) \in \fintR{\phi}$.
\label{lem:term-ep}
\end{lemma}
\begin{proof}
By \rref{thm:proof-calculus-sound}, the sequent $(\seq{\Gamma}{\lexists[\allrat]{x}{\phi}})$ is valid.
Since $(\myaa,\om) \in \cintR{\Gamma}$, then by the definition of sequent validity, there exists a common realizer $\ac$ such that $(\rzApp{\ac}{\myaa},\om) \in \fintR{\lexists[\allrat]{x}{\phi}}$.
Now let $f = \rzFst{\rzApp{\ac}{\myaa}}$ and $\ab = \rzSnd{\rzApp{\ac}{\myaa}}$ and the result is immediate by the semantics of existentials.
\end{proof}

Disjunction strategies can depend on the state, so na\"{i}ve DP does not hold.
\begin{example}[Na\"{i}ve DP]
When $\proves{\Gamma}{M}{(\phi \lor \psi)}$ there need not be $N$ such that ${\proves{\Gamma}{N}{\phi}}$ or ${\proves{\Gamma}{N}{\psi}}$.
\end{example}
  Consider $\phi \equiv x > 0$ and $\psi \equiv x < 1$.
  Then $\proves{\Gemp}{\esplit{x}{0}()}{(\phi \lor \psi)}$, but neither $x < 1$ nor $x > 0$ is valid, let alone provable.

\begin{lemma}[Disjunction Property]
When $\proves{\Gamma}{M}{\phi \lor \psi}$ there exists realizer $\ab$ and computable $f,$ s.t.\ for every $\om$ and $\myaa$ such that $(\myaa,\omega) \in \cintR{\Gamma}$, either $f(\omega)=0$ and $(\rzFst{\ab},\omega) \in \fintR{\phi}{}$, else $f(\omega)=1$ and $(\rzSnd{\ab},\omega) \in \fintR{\psi}$.
\end{lemma}
\begin{proof}
By \rref{thm:proof-calculus-sound}, the sequent $\seq{\Gamma}{\phi \lor \psi}$ is valid.
Since $(\myaa,\om) \in \cintR{\Gamma}$, then by the definition of sequent validity, there exists a common realizer $\ac$ such that 
$(\rzApp{\ac}{\myaa},\om) \in \fintR{\phi \lor \psi}$.
Now let $f = \rzFst{\rzApp{\ac}{\myaa}}$ and $\ab = \rzSnd{\rzApp{\ac}{\myaa}}$ and the result is immediate by the semantics of disjunction.
\end{proof}
Following the same approach, we generalize to a Strategy Property.
In \CGL, strategies are represented by realizers, which implement every computation made throughout the game.
Thus, to show provable games have computable winning strategies, it suffices to exhibit realizers.
\begin{theorem}[Active Strategy Property]
If $\proves{\Gamma}{M}{\ddiamond{\alpha}{\phi}},$ then there exists a realizer $\ab$ such that for all $\om$ and realizers $\myaa$ such that $(\myaa,\om) \in \cintR{\Gamma},$
then $\strategyforR[\alpha]{\{(\rzApp{\ab}{\myaa},\om)\}} \subseteq \fintR{\phi} \cup \{\stt\}$.
\label{thm:stratprop}
\end{theorem}
\begin{theorem}[Dormant Strategy Property]
If $\proves{\Gamma}{M}{\dbox{\alpha}{\phi}},$ then there exists a realizer $\ab$ such that for all $\om$ and realizers $\myaa$ such that $(\myaa,\om) \in \cintR{\Gamma},$
then $\dstrategyforR[\alpha]{\{(\rzApp{\ab}{\myaa},\om)\}} \subseteq \fintR{\phi} \cup \{\stt\}$.
\label{thm:dstratprop}
\end{theorem}
\begin{proof}[Summary]
From proof term $M$ and \rref{thm:proof-calculus-sound}, we have a realizer for formula $\ddiamond{\alpha}{\phi}$ or $\dbox{\alpha}{\phi},$ respectively.
We proceed by induction on $\alpha$: the realizer $\rzApp{\ab}{\myaa}$ contains all realizers applied in the inductive cases composed with their continuations that prove $\phi$ in each base case.
\end{proof}
While these proofs, especially EP and DP, are short and direct, we note that this is by design: the challenge in developing \CGL is not so much the proofs of this section, rather these proofs become simple because we adopted a realizability semantics.
The challenge was in developing the semantics and adapting the proof calculus and theory to that semantics.

\section{Conclusion and Future Work}
In this paper, we developed a Constructive Game Logic \CGL, from syntax and realizability semantics to a proof calculus and operational semantics on the proof terms.
We developed two understandings of proofs as programs: semantically, every proof of game winnability corresponds to a realizer which computes the game's winning strategy, while the language of proof terms is also a functional programming language where proofs reduce to their normal forms according to the operational semantics.
We completed the Curry-Howard interpretation for games by showing Existential, Disjunction, and Strategy properties: programs can be synthesized that decide which instance, disjunct, or moves are taken in existentials, disjunctions, and games.
In summary, we have developed the most comprehensive Curry-Howard interpretation of any program logic to date, for a much more expressive logic than prior work~\cite{kamide2010strong}.
Because \CGL contains constructive Concurrent \DL and first-order \DL as strict fragments, we have provided a comprehensive Curry-Howard interpretation for them in one fell swoop.
The key insights behind \CGL should apply to the many dynamic and Hoare logics used in verification today.

Synthesis is the immediate application of \CGL.
Motivations for synthesis include security games~\cite{DBLP:journals/sLogica/PaulyP03a}, concurrent programs with demonic schedulers (Concurrent Dynamic Logic), and control software for safety-critical cyber-physical systems such as cars and planes.
In general, any kind of software program which must operate correctly in an adversarial environment can benefit from game logic verification.
The proofs of \rref{thm:stratprop} and \rref{thm:dstratprop} constitute an (on-paper) algorithm which perform synthesis of guaranteed-correct strategies from game proofs.
The first future work is to implement this algorithm in code, providing much-needed assurance for software which is often mission-critical or safety-critical.
This paper focused on discrete \CGL with one numeric type simply because any further features would distract from the core features.
Real applications come from many domains which add features around this shared core.

The second future work is to extend \CGL to hybrid games, which provide compelling applications from the domain of adversarial cyber-physical systems.
This future work will combine the novel features of \CGL with those of the classical logic \dGL.
The primary task is to define a constructive semantics for differential equations and to give constructive interpretations to the differential equation rules of \dGL.
Previous work on formalizations of differential equations~\cite{DBLP:conf/itp/MakarovS13} suggests differential equations can be treated constructively.
In principle, existing proofs in \dGL might happen to be constructive, but this does not obviate the present work.
On the contrary, once a game logic proof is shown to fall in the constructive fragment, our work gives a correct synthesis guarantee for it too!

\bibliographystyle{splncs04}
\bibliography{constructive-games,platzer}

\appendix
\newpage
\renewcommand{\G}{\Gamma}
\section{Full Proof Calculus}
\label{app:pc-full}
\begin{figure}[h]
\centering
\begin{calculuscollections}{\columnwidth}
\begin{calculus}
\cinferenceRule[dchoiceE|{$\langle\cup\rangle${E}}]{}
{
\linferenceRule[formula]
{\proves{\G}{A}{\ddiamond{\alpha\cup\beta}{\phi}}
        &\proves{\G,\pvl:\ddiamond{\alpha}{\phi}}{B}{\psi}
        &\proves{\G,\pvr:\ddiamond{\beta}{\phi}}{C}{\psi}}
{\proves{\G}{\edcase{A}{B}{C}}{\psi}}
}{}
\end{calculus}
\end{calculuscollections}

\begin{calculuscollections}{0.4\columnwidth}
\begin{calculus}
\cinferenceRule[dchoiceIL|{$\langle\cup\rangle${I1}}]{}
{\linferenceRule[formula]
  {\proves{\G}{M}{\ddiamond{\alpha}{\phi}}}
  {\proves{\G}{\edinjL{M}}{\ddiamond{\alpha\cup\beta}{\phi}}}
}{}
\cinferenceRule[dchoiceIR|{$\langle\cup\rangle${I2}}]{}
{\linferenceRule[formula]
  {\proves{\G}{M}{\ddiamond{\beta}{\phi}}}
  {\proves{\G}{\edinjR{M}}{\ddiamond{\alpha\cup\beta}{\phi}}}
}{}
\cinferenceRule[bchoiceI|{$[\cup]${I}}]{}
{\linferenceRule[formula]
  {\proves{\G}{M}{\dbox{\alpha}{\phi}} & \proves{\G}{N}{\dbox{\beta}{\phi}}}
  {\proves{\G}{\ebcons{M}{N}}{\dbox{\alpha\cup\beta}{\phi}}}
}{}
\cinferenceRule[dtestI|{$\langle?\rangle$}{I}]{}
{\linferenceRule[formula]
  {\proves{\G}{M}{\phi} & \proves{\G}{N}{\psi}}
  {\proves{\G}{\edcons{M}{N}}{\ddiamond{\ptest{\phi}}{\psi}}}
}{}
\cinferenceRule[btestI|{$[?]$}{I}]{}
{\linferenceRule[formula]
  {\proves{\G,\pvx:\phi}{M}{\psi}}
  {\proves{\G}{(\eplam{\phi}{M})}{\dbox{\ptest{\phi}}{\psi}}}
}{}
\cinferenceRule[btestE|{$[?]$}{E}]{}
{\linferenceRule[formula]
  {\proves{\G}{M}{\dbox{\ptest{\phi}}{\psi}} & \proves{\G}{N}{\phi}}
  {\proves{\G}{(\eapp{M}{N})}{\psi}}
}{}
\end{calculus}
\end{calculuscollections}%
\qquad%
\begin{calculuscollections}{0.4\columnwidth}
\begin{calculus}
\cinferenceRule[bchoiceEL|{$[\cup]${E1}}]{}
{\linferenceRule[formula]
  {\proves{\G}{M}{\dbox{\alpha\cup\beta}{\phi}}}
  {\proves{\G}{\ebprojL{M}}{\dbox{\alpha}{\phi}}}
}{}
\cinferenceRule[bchoiceER|{$[\cup]${E2}}]{}
{\linferenceRule[formula]
  {\proves{\G}{M}{\dbox{\alpha\cup\beta}{\phi}}}
  {\proves{\G}{\ebprojR{M}}{\dbox{\beta}{\phi}}}
}{}
\cinferenceRule[hyp|{hyp}]{}
{\linferenceRule[formula]
  {}
  {\proves{\G,\pvx:\phi}{\pvx}{\phi}}
}{}
\cinferenceRule[dtestEL|{$\langle?\rangle$}{E1}]{}
{\linferenceRule[formula]
  {\proves{\G}{M}{\ddiamond{\ptest{\phi}}{\psi}}}
  {\proves{\G}{\edprojL{M}}{\phi}}
}{}
\cinferenceRule[dtestER|{$\langle?\rangle$}{E2}]{}
{\linferenceRule[formula]
  {\proves{\G}{M}{\ddiamond{\ptest{\phi}}{\psi}}}
  {\proves{\G}{\edprojR{M}}{\psi}}
}{}
\end{calculus}
\end{calculuscollections}
\caption{\CGL proof calculus: Propositional rules}
\label{fig:app-cgl-rules-prop}
\end{figure}

\begin{figure}[h]
  \centering
\begin{calculuscollections}{\columnwidth}
\begin{calculus}
\cinferenceRule[drcase|{$\langle*\rangle${C}}]{}
{
\linferenceRule[formula]
{\proves{\G}{A}{\ddiamond{\prepeat{\alpha}}{\phi}} & \proves{\G,\pvs:\phi}{B}{\psi} & \proves{\G,\pvg:\ddiamond{\alpha}{\ddiamond{\prepeat{\alpha}}{\phi}}}{C}{\psi}}
{\proves{\G}{\ercase{A}{B}{C}}{\psi}}
}{}
\cinferenceRule[bunroll|{$[*]${E}}]{}
{\linferenceRule[formula]
  {\proves{\G}{M}{\dbox{\prepeat{\alpha}}{\phi}}}
  {\proves{\G}{\ebunroll{M}}{\phi \land \dbox{\alpha}{\dbox{\prepeat{\alpha}}{\phi}}}}
}{}
\cinferenceRule[dualI|{$\lstrike{}^d\rstrike$}I]{}
{\linferenceRule[formula]
  {\proves{\G}{M}{\pmodality{\alpha}{\phi}}}
  {\proves{\G}{\eSwap{M}}{\dmodality{\pdual{\alpha}}{\phi}}}
}{}
\cinferenceRule[mon|{M}]{}
{\linferenceRule[formula]
  {\proves{\G}{M}{\ddiamond{\alpha}{\phi}} & \proves{\earen{\G}{\alpha},\pvx:\phi}{N}{\psi}}
  {\proves{\G}{\emon{M}{N}{\pvx}}{\ddiamond{\alpha}{\psi}}}
}{}
\end{calculus}%
\hfill%
\begin{calculus}
\cinferenceRule[dstop|{$\langle*\rangle$S}]{}
{
\linferenceRule[formula]
{\proves{\G}{M}{\phi}}  
{\proves{\G}{\estop{M}}{\ddiamond{\prepeat{\alpha}}{\phi}}}
}{}
\cinferenceRule[dgo|{$\langle*\rangle$G}]{}
{
\linferenceRule[formula]
{\proves{\G}{M}{\ddiamond{\alpha}{\ddiamond{\prepeat{\alpha}}{\phi}}}}
{\proves{\G}{\ego{M}}{\ddiamond{\prepeat{\alpha}}{\phi}}}
}{}
\cinferenceRule[broll|{$[*]${R}}]{}
{\linferenceRule[formula]
  {\proves{\G}{M}{\phi \land \dbox{\alpha}{\dbox{\prepeat{\alpha}}{\phi}}}}
  {\proves{\G}{\ebroll{M}}{\dbox{\prepeat{\alpha}}{\phi}}}
}{}
\cinferenceRule[seqI|{$\lstrike{;}\rstrike$}I]{}
{\linferenceRule[formula]
  {\proves{\G}{M}{\dmodality{\alpha}{\dmodality{\beta}{\phi}}}}
  {\proves{\G}{\eSeq{M}}{\dmodality{\alpha;\beta}{\phi}}}
}{}
\end{calculus}
\end{calculuscollections}
\caption{\CGL proof calculus: Non-propositional rules}
\label{fig:app-cgl-compos}
\end{figure}

\begin{figure}
  \centering
\begin{calculuscollections}{\columnwidth}
\begin{calculus}
\cinferenceRule[drandomI|{$\langle{:}*\rangle${I}}]{}
{
\linferenceRule[formula]
{\proves{\eren{\G}{x}{y},\pvx:(x=\eren{f}{x}{y})}{M}{\phi}}
{\proves{\G}{\etcons{f}{M}}{\ddiamond{\prandom{x}}{\phi}}}
}{$y,\pvx$ fresh, $f$ comp.}
\cinferenceRule[asgnI|{$\lstrike{:=}\rstrike$}I]{}
{\linferenceRule[formula]
{\proves{\eren{\G}{x}{y},\pvx:(x=\eren{f}{x}{y})}{M}{\phi}}
{\proves{\G}{\eAsgneq{y}{x}{\pvx}{M}}{\dmodality{\humod{x}{f}}{\phi}}}
}{$y$ fresh}
\cinferenceRule[drandomE|{$\langle{:}*\rangle${E}}]{}
{
\linferenceRule[formula]
{\proves{\G}{M}{\ddiamond{\prandom{x}}{\phi}} & \proves{\eren{\G}{x}{y},\pvx:\phi}{N}{\psi}}
{\proves{\G}{\eunpack{M}{N}}{\psi}}
}{$y$ fresh, $x \notin \freevars{\psi}$}
\cinferenceRule[brandomI|{$[{:}{*}]${I}}]{}
{
\linferenceRule[formula]
{\proves{\eren{\G}{x}{y}}{M}{\phi}}
{\proves{\G}{(\etlam{\allrat}{M})}{\dbox{\prandom{x}}{\phi}}}
}{y\text{ fresh}}
\cinferenceRule[brandomE|{$[{:}{*}]${E}}]{}
{
      \linferenceRule[formula]
        {\proves{\G}{M}{\dbox{\prandom{x}}{\phi}}}
        {\proves{\G}{\eapp{M}{f}}{\tsub{\phi}{x}{f}}}
}{$\tsub{\phi}{x}{f}$ admiss.}
\end{calculus}
\end{calculuscollections}
\caption{\CGL proof calculus: first-order programs}
\label{fig:app-cgl-first-order}
\end{figure}

\begin{figure}
\begin{calculuscollections}{\columnwidth}
\begin{calculus}
\cinferenceRule[dloopI|{$\langle*\rangle${I}}]{}
{\linferenceRule[formula]
{\deduce{\proves{\pvx:\conv,\pvy:\met_0 = \met \metgr \metz}{B}{\ddiamond{\alpha}{(\conv\land \met_0 \metgr \met)}}}
        {\proves{\G}{A}{\conv}}
 &  \proves{\pvx:\conv,\pvy:\met = \metz}{C}{\phi}}
{\proves{\G}{\efor{A}{B}{C}}{\ddiamond{\prepeat{\alpha}}{\phi}}}
}{$\met_0$ fresh}
\end{calculus}
\\
\begin{calculus}
\cinferenceRule[bloopI|{loop}]{}
{\linferenceRule[formula]
  {\proves{\G}{M}{J} & \proves{\pvx:J}{N}{\dbox{\alpha}{J}} & \proves{\pvx:J}{O}{\phi}}
  {\proves{\G}{(\erep{M}{N}{\pvx:J}{O})}{\dbox{\prepeat{\alpha}}{\phi}}}
}{}
\cinferenceRule[metsplit|split]{}
{
\proves{\G}{(\esplit{f}{g})}{f \leq g \lor f > g}
}
{}
\end{calculus}%
\hfill%
\begin{calculus}
\cinferenceRule[dloopE|FP]{}
{
\linferenceRule[formula]
{\proves{\G}{A}{\ddiamond{\prepeat{\alpha}}{\phi}}
        &\proves{\pvs:\phi}{B}{\psi} & \proves{\pvg:\ddiamond{\alpha}{\psi}}{C}{\psi}}
{\proves{\G}{\efp{A}{B}{C}}{\psi}}
}{}
\end{calculus}
\end{calculuscollections}
\caption{\CGL proof calculus: loops}
\label{fig:app-cgl-loops}
\end{figure}

\begin{figure}
\begin{calculus}
\cinferenceRule[QE|FO]{}
{\linferenceRule[formula]
{\proves{\G}{M}{\rho}}
{\proves{\G}{\eQE{\phi}{M}}{\phi}}
 }{exists $\myaa, \{\myaa\} \times \allstate \subseteq \fintR{\rho\limply\phi},$ $\rho,\phi$ F.O.}
\cinferenceRule[dec|Dec]{}
{\linferenceRule[formula]
 {\proves{\G}{M}{\rho}}
 {\proves{\G}{(\elem{\phi \lor \psi}{M})}{\phi \lor \psi}}
 }{exists $\myaa, \{\myaa\} \times \allstate \subseteq \fintR{\rho \limply \phi \lor \psi},$ $\rho, \phi, \psi$ F.O.}
\end{calculus}
\caption{\CGL proof calculus: first-order arithmetic}
\label{fig:cgl-arith}
\end{figure}

\[\cinferenceRule[ghost|iG]{}
{\linferenceRule[formula]
  {\proves{\G,\pvx:x=f}{M}{\phi}}
  {\proves{\G}{\eghost{x}{f}{\pvx}{M}}{\phi}}
}
{\m{x}\text{ fresh except free in }\m{M,} \m{\pvx}\text{ fresh}}
\]

\newpage~

\section{Full operational semantics}
\label{app:op-full}
We recite the full operational semantics on proof terms here for reference.
We begin by reciting the language of proof terms.
\begin{definition}[Proof terms]
The proof terms $M,N$ (sometimes $A,B,C$) are given by the grammar:
\begin{align*}
M,N &\bebecomes~
\etlam{\allrat}{M} \alternative \eapp{M}{f} \alternative \eplam{\phi}{M}  \alternative  \eapp{M}{N} \alternative \edinjL{M} \alternative \edinjR{M} \\
& \alternative \ercase{A}{B}{C} \alternative \edcase{A}{B}{C} \\
&\alternative  \efor{A}{B}{C} \alternative\efp{A}{B}{C} \\
 & \alternative \erep{M}{N}{\pvx:J}{O}  \alternative \ebunroll{M}  \alternative \estop{M} \alternative \ego{M}\\
& \alternative \etcons{f}{M} \alternative  \eProjL{M} \alternative \eProjR{M} \alternative \eCons{M}{N}  \alternative \eSeq{M}  \alternative \eSwap{M}\\ 
& \alternative \eAsgneq{y}{x}{\pvx}{M} \alternative \eunpack{M}{N}   \alternative \emon{M}{N}{\pvx} \alternative \eQE{\phi}{M} \alternative \pvx
\end{align*}
where $\pvx, \pvy,\pvl,\pvr,\pvs,$ and $\pvg$ are \emph{proof variables}, that is variables that range over proof terms of a given proposition.
In contrast to the assignable \emph{program variables}, the proof variables are given their meaning by substitution and are scoped locally, not globally.
\end{definition}
We proceed with the $\beta$ rules.
\begin{figure}[h]
\centering
\begin{calculuscollections}{0.5\textwidth}
\begin{calculus}
\cinferenceRule[appBeta|{$\lambda\phi\beta$}]{}{\eapp{(\eplam{P}{M})}{N} \stepsto \esub{M}{x}{N}}{}
\cinferenceRule[brandomBeta|{$\lambda\beta$}]{}{\eapp{(\etlam{\allrat}{M})}{f} \stepsto \tsub{M}{x}{f}}{}
\cinferenceRule[projLBeta|{$\pi_1\beta$}]{}{\eProjL{\eCons{M}{N}} \stepsto M}{}
\cinferenceRule[projRBeta|{$\pi_2\beta$}]{}{\eProjR{\eCons{M}{N}} \stepsto N}{}
\end{calculus}
\end{calculuscollections}
\begin{calculuscollections}{0.5\textwidth}
\begin{calculus}
\cinferenceRule[caseBetaL|{case$\beta$L}]{}{\eCase{\eInjL{A}}{B}{C} \stepsto \esub{B}{\ell}{A}}{}
\cinferenceRule[caseBetaR|{{case}$\beta$R}]{}{\eCase{\eInjR{A}}{B}{C} \stepsto \esub{C}{r}{A}}{}
\cinferenceRule[bunrollBeta|{{unroll}$\beta$}]{}{\ebunroll{\ebroll{M}} \stepsto M}{}
\cinferenceRule[QEAllBeta|{FO$\forall\beta$}]{}{\eQE{\lforall{x}{\phi}}{M} \stepsto (\etlam{\allrat}{\eQE{\phi}{M}})}{}
\cinferenceRule[QEAndBeta|{FO$\wedge\beta$}]{}{\eQE{\phi \land \psi}{M} \stepsto \econs{\eQE{\phi}{M}}{\eQE{\psi}{M}}}{}
\end{calculus}
\end{calculuscollections}

\begin{calculuscollections}{0.5\textwidth}
\begin{calculus}
\cinferenceRule[QEExistsBeta|{FO$\exists\beta$}]{ \text{ for satisfying instance }\ensuremath{f}}
{        \eQE{\lexists{x}{\phi}}{M} 
\stepsto \etcons{f} 
                {\eQE{\phi}{\esub{(\tsub{M}{x}{\big(\eren{f}{x}{y}\big)})}
                                 {\pvx}
                                 {\eQE{\eren{f}{x}{y}
                                      =\eren{f}{x}{y}}
                                      {}}}}}
{}
\cinferenceRule[QEOrBeta|{FO$\lor\beta$}]{$\fint{\phi\lor\psi}{\myaa}=\allstate$}
{\eQE{\phi \lor \psi}{M} \stepsto \ecase{\rzFst{\myaa}}{\einjL{\eQE{\phi}{M,\ell}}}{\einjR{\eQE{\psi}{M,r}}}}
{}
\cinferenceRule[unpackBeta|{{unpack}$\beta$}]{}{\eunpack{\etconsgen{x}{y}{\pvy}{f}{M}}{N} \stepsto \eren{(\eghost{x}{\eren{f}{x}{y}}{\pvy}{\esub{N}{\pvx}{M}})}{y}{x}}{}
\cinferenceRule[fpBeta|{{FP}$\beta$}]{}{\efp{A}{B}{C} \stepsto \ercase{A}{B}{\esub{C}{\pvg}{(\emon{\pvg}{\efp{z}{B}{C}}{z})}}}{}
\cinferenceRule[repBeta|{{rep}$\beta$}]{}{(\erep{M}{N}{\pvx:J}{O}) \stepsto \ebroll{\edcons{M}{\emon{(\esub{N}{\pvx}{M})}{(\erep{\pvy}{N}{\pvx:J}{O})}{\pvy}}}}{}
\cinferenceRule[forBeta|{{for}$\beta$}]{}{\efor{A}{B}{C} \stepsto}{}
\end{calculus}\\[-0.22in]
\begin{align*}
&\ecaseHead{\esplit{\met}{0}{}}\\
&\ \ \ecaseLeft{\pvl}{\estop{\esub{C}{(\pvx,\pvy)}{(A,\pvl)}}}\\
&\ecaseRight{\pvr}{\eghost{\met_0}{\met}{\textit{\pvrr}}{\ego{(\emon{(\esub{B}{\pvx,\pvy}{A,\edcons{\textit{\pvrr}}{\pvr}})}{(\efor{\eprojL{\pvz}}{B}{C})}{\pvz})}}}\ecaseEnd
  \end{align*}
\end{calculuscollections}
\caption{Operational semantics: $\beta$-rules}
\label{fig:app-opbeta}
\end{figure}
The first-order simplification rules (e.g., \irref{QEAllBeta}) are not discussed in the main paper for the sake of simplicitiy.
These rules enable a simpler notion of normal forms.
First-order non-syntactic proofs can be used to provide a proof of arbitrary first-order facts.
However, these rules show that non-syntactic proofs need not be thought of as a normal form because they can be simplified into normal proofs by inspection on the underlying realizer.

\begin{figure}
\centering
\begin{calculuscollections}{\textwidth}
\begin{calculus}
\cinferenceRule[rlamMon|{$\lambda\circ$}]{}{\emon{(\etlam{\allrat}{M})}{N}{\pvy} \stepsto (\etlam{\allrat}{(\esub{N}{\pvy}{M})})}{}
\cinferenceRule[bconsMon|{$[\wedge]\circ$}]{}{\emon{\ebcons{A}{B}}{N}{\pvy}\stepsto \ebcons{\emon{A}{(\eren{N}{\va'}{\va})}{\pvx}}{\emon{B}{(\eren{N}{\vb'}{\vb})}{\pvx}}}{}
\cinferenceRule[dconsMon|{$\langle\wedge\rangle\circ$}]{}{\emon{\edcons{A}{B}}{N}{\pvx} \stepsto \edcons{A}{\esub{N}{\pvx}{B}}}{}
\cinferenceRule[injLMon|{$\eInjL{}\circ$}]{}{\emon{\eInjL{M}}{N}{\pvx}   \stepsto \eInjL{(\emon{M}{N}{\pvx})}}{}
\cinferenceRule[bswapMon|{$[\pdual{}]\circ$}]{}{\emon{\ebswap{M}}{N}{\pvx} \stepsto \ebswap{(\emon{M}{N}{\pvx})}}{}
\cinferenceRule[dswapMon|{$\langle\pdual{}\rangle\circ$}]{}{\emon{\edswap{M}}{N}{\pvx} \stepsto \edswap{(\emon{M}{N}{\pvx})}}{}
\end{calculus}
\begin{calculus}
\cinferenceRule[lamMon|{$\lambda\phi\circ$}]{}{\emon{(\eplam{P}{M})}{N}{\pvy} \stepsto (\eplam{P}{(\esub{N}{\pvy}{M})})}{}
\cinferenceRule[tconsMon|{${{:}*}\circ$}]{}{\emon{\etconsgen{x}{y}{\pvy}{f}{M}}{N}{\pvx} \stepsto \etconsgen{x}{y}{\pvy}{f}{\esub{N}{\pvx}{M}}}{}
\cinferenceRule[bseqMon|{$[\iota\circ]$}]{}{\emon{\ebseq{M}}{N}{\pvx} \stepsto \ebseq{(\emon{M}{(\emon{\pvy}{N}{\pvx})}{\pvy})}}{}
\cinferenceRule[dseqMon|{$\langle\iota\circ\rangle$}]{}{\emon{\edseq{M}}{N}{\pvx} \stepsto \edseq{(\emon{M}{(\emon{\pvy}{N}{\pvx})}{\pvy})}}{}
\cinferenceRule[injRMon|{$\eInjR{}\circ$}]{}{\emon{\eInjR{M}}{N}{\pvx}   \stepsto \eInjR{(\emon{M}{N}{\pvx})}}{}
\end{calculus}
\end{calculuscollections}

\begin{calculuscollections}{\textwidth}
\begin{calculus}
\cinferenceRule[dasgnMon|{$\langle{:=}\circ\rangle$}]{}{\emon{\edasgn{y}{x}{\pvy}{M}}{N}{\pvx} \stepsto \edasgn{y}{x}{\pvy}{\esub{N}{\pvx}{M}}}{}
\cinferenceRule[basgnMon|{$[{:=}\circ]$}]{}{\emon{\ebasgn{y}{x}{\pvy}{M}}{N}{\pvx} \stepsto \ebasgn{y}{x}{\pvy}{\esub{N}{\pvx}{M}}}{}
\cinferenceRule[caseMon|{case$\circ$}]{}{\emon{\eCase{A}{B}{C}}{D}{\pvx} \stepsto \ecase{A}{\emon{B}{D}{\pvx}}{\emon{C}{D}{\pvx}}}{}
\cinferenceRule[brollMon|{$[*]\circ$}]{}{\emon{\ebroll{M}}{N}{\pvx}  \stepsto \ebroll{\edcons{\emon{(\eprojL{M})}{(\eren{N}{\vec{y}}{\vec{x}})}{\pvx}}{\emon{(\eprojR{M})}{(\emon{\pvz}{N}{\pvx})}{\pvz}}}}{}
\cinferenceRule[stopMon|stop$\circ$]{}{\emon{\estop{M}}{N}{\pvx} \stepsto \estop{(\esub{(\eren{N}{\va'}{\va})}{\pvx}{M})}}
{}
\cinferenceRule[goMon|go$\circ$]{}{\emon{\ego{M}}{N}{\pvx} \stepsto \ego{(\emon{M}{\emon{\pvy}{N}{\pvx}}{\pvy})}}{}
\end{calculus}
\end{calculuscollections}
\caption{Operational semantics: monotonicity rules.
When rule {stop$\circ$} is applied to program $\prepeat{\alpha}$ then $\va = \boundvars{\alpha}$ and $\va'$ is the vector of  corresponding ghost variables}
\label{fig:opmon}
\end{figure}

\begin{figure}[h]
\centering
\begin{calculuscollections}{\textwidth}
\begin{calculus}
\cinferenceRule[projLC|{$\eprojL{}$C}]{}
{\eprojL{\ecase{A}{B}{C}} \stepsto \ecase{A}{\eprojL{B}}{\eprojL{C}}}
{}
\cinferenceRule[projRC|{$\eprojR{}$C}]{}
{\eprojR{\ecase{A}{B}{C}} \stepsto \ecase{A}{\eprojR{B}}{\eprojR{C}}}
{}

\cinferenceRule[bconsCL|{$\wedge$C1}]{}
{\ebcons{\ecase{A}{B}{C}}{N} \stepsto \ecase{A}{\ebcons{B}{N}}{\ebcons{C}{N}}}
{}
\cinferenceRule[bconsCR|{$\wedge$C2}]{}
{\ebcons{M}{\ecase{A}{B}{C}} \stepsto \ecase{A}{\ebcons{M}{B}}{\ebcons{M}{C}}}
{}
\cinferenceRule[dconsCL|{$\wedge$}C1]{}
{\edcons{\ecase{A}{B}{C}}{N} \stepsto \ecase{A}{\edcons{B}{N}}{\edcons{C}{N}}}
{}
\cinferenceRule[dconsCR|{$\wedge$}C2]{}
{\edcons{M}{\ecase{A}{B}{C}} \stepsto \ecase{A}{\edcons{M}{B}}{\edcons{M}{C}}}
{}
\end{calculus}
\end{calculuscollections}
\caption{Operational semantics: commuting conversion rules (contd.)}
\end{figure}

\begin{figure}[h]
\begin{calculuscollections}{\textwidth}
\begin{calculus}
\cinferenceRule[stopC|stopC]{}
{\estop{\ecase{A}{B}{C}} \stepsto \ecase{A}{\estop{B}}{\estop{C}}}
{}
\cinferenceRule[goC|goC]{}
{\ego{\ecase{A}{B}{C}} \stepsto \ecase{A}{\ego{B}}{\ego{C}}}
{}
\cinferenceRule[injLC|{$\eInjL{}$C}]{}
{\eInjL{\ecasegen{A}{\pvx}{B}{\pvy}{C}}
   \stepsto \ecasegen{A}{\pvx}{\einjL{B}}{\pvy}{\einjL{C}}}
{}
\cinferenceRule[injRC|{$\einjR{}$C}]{}
{\eInjR{\ecasegen{A}{\pvx}{B}{\pvy}{C}}
   \stepsto \ecasegen{A}{\pvx}{\einjR{B}}{\pvy}{\einjR{C}}}
{}
\cinferenceRule[drcaseC|caseC]{}
{
\m{\begin{aligned}
&\ercase{\ecase{A}{B}{C}}{D}{E}\\
\stepsto &\ecase{A}{\ercase{B}{D}{E}}{\ercase{C}{D}{E}}
\end{aligned}}
}
{}
\cinferenceRule[caseC|caseC]{}
{
\m{\begin{aligned}
&\edcase{\ecase{A}{B}{C}}{D}{E}\\
\stepsto &\ecase{A}{\edcase{B}{D}{E}}{\edcase{C}{D}{E}}
\end{aligned}}}
{}
\end{calculus}
\end{calculuscollections}
\caption{Operational semantics: commuting conversion rules (contd.)}
\end{figure}

\begin{figure}[h]
\begin{calculuscollections}{\textwidth}
\begin{calculus}
\cinferenceRule[bunrollC|unrollC]{}
{\ebunroll{\ecase{A}{B}{C}} \stepsto \ecase{A}{\ebunroll{B}}{\ebunroll{C}}}
{}
\cinferenceRule[repC|repC]{}
{
\m{
  \begin{aligned}
&\erep{\ecase{A}{B}{C}}{N}{\pvx:J}{O} \\
\stepsto & \ecase{A}{(\erep{B}{N}{\pvx:J}{O})}{(\erep{C}{N}{\pvx:J}{O})}
  \end{aligned}
}}
{}
\cinferenceRule[forC|{forC}]{}
{
\m{\begin{aligned}
 &\efor{\ecase{A}{B}{C}}{N}{O}\\
\stepsto & \ecase{A}{\efor{B}{N}{O}}{\efor{C}{N}{O}}
\end{aligned}}}
{}
\cinferenceRule[fpC|FPC]{}
{
\m{\begin{aligned}
 &\efp{\ecase{A}{B}{C}}{D}{E}\\
\stepsto &\ecase{A}{\efp{B}{D}{E}}{\efp{C}{D}{E}}
\end{aligned}}
}
{}

\end{calculus}
\end{calculuscollections}
\caption{Operational semantics: commuting conversion rules (contd.)}
\end{figure}

\begin{figure}[h]
\begin{calculuscollections}{\textwidth}
\begin{calculus}
\cinferenceRule[dseqC|{$\iota$C}]{}
{\edseq{\ecase{A}{B}{C}} \stepsto \ecase{A}{\edseq{B}}{\edseq{C}}}
{}
\cinferenceRule[bseqC|{$\iota$C}]{}
{\ebseq{\ecase{A}{B}{C}} \stepsto \ecase{A}{\ebseq{B}}{\ebseq{C}}}
{}
\cinferenceRule[dswapC|yieldC]{}
{\edswap{\ecase{A}{B}{C}} \stepsto \ecase{A}{\edswap{B}}{\edswap{C}}}
{}
\cinferenceRule[bswapC|yieldC]{}
{\ebswap{\ecase{A}{B}{C}} \stepsto \ecase{A}{\ebswap{B}}{\ebswap{C}}}
{}
\cinferenceRule[appCL|app1C]{}
{\eapp{\ecase{A}{B}{C}}{N} \stepsto \ecase{A}{\eapp{B}{N}}{\eapp{C}{N}}}
{}
\cinferenceRule[appCR|app2C]{}
{\eapp{M}{\ecase{A}{B}{C}} \stepsto \ecase{A}{\eapp{M}{B}}{\eapp{M}{C}}}
{}
\cinferenceRule[brandomC|appC]{}
{\eapp{\ecase{A}{B}{C}}{f} \stepsto \ecase{A}{\eapp{B}{f}}{\eapp{C}{f}}}
{}
\cinferenceRule[monC|{$\circ$}C]{}
{\emon{\ecase{A}{B}{C}}{N}{\pvx} \stepsto \ecase{A}{(\emon{B}{N}{\pvx})}{(\emon{C}{N}{\pvx})}}
{}
\cinferenceRule[tconsC|{tconsC}]{}
{
\m{\begin{aligned}
&\etcons{f}{\ecase{A}{B}{C}} \\
\stepsto &\ecase{A}{\etcons{f}{B}}{\etcons{f}{C}}
\end{aligned}}}
{}
\cinferenceRule[unpackC|{${:*}$}C]{}
{
\m{\begin{aligned}
&\eunpack{\ecase{A}{B}{C}}{N} \\
\stepsto &\ecase{A}{\eunpack{B}{N}}{\eunpack{C}{N}}
\end{aligned}}
}
{}
\end{calculus}
\end{calculuscollections}
\caption{Operational semantics: commuting conversion rules (contd.)}
\end{figure}

\begin{figure}
\centering
\begin{calculuscollections}{\textwidth}
\begin{calculus}
\cinferenceRule[projLS|{$\pi_1$}S]{}{\linferenceRule[formula]
{M \stepsto M'}
{\eProjL{M} \stepsto \eProjL{M'}}}{}
\cinferenceRule[projRS|{$\pi_2$}S]{}{\linferenceRule[formula]
{M \stepsto M'}
{\eProjR{M} \stepsto \eProjR{M'}}}{}
\cinferenceRule[repS|repS]{}{\linferenceRule[formula]
{M \stepsto M'}
{(\erep{M}{N}{\pvx:J}{O}) \stepsto (\erep{M'}{N}{\pvx:J}{O})}}{}
\cinferenceRule[bunrollS|{unroll}]{}{\linferenceRule[formula]
{M \stepsto M'}
{\ebunroll{M} \stepsto \ebunroll{M'}}}{}
\cinferenceRule[brandomS|{[{:}{*}]}{S}]{}{\linferenceRule[formula]
{M \stepsto M'}
{\eapp{M}{f} \stepsto \eapp{M'}{f}}}{}
\cinferenceRule[appSL|appS1]{}{\linferenceRule[formula] 
{M \stepsto M'}
{\eapp{M}{N} \stepsto \eapp{M'}{N}}}{}
\cinferenceRule[bseqS|{$[\iota]$}S]{}{\linferenceRule[formula]
{M \stepsto M'}
{\eSeq{M} \stepsto \eSeq{M'}}}{}
\cinferenceRule[dseqS|{$\langle\iota\rangle$}S]{}{\linferenceRule[formula]
{M \stepsto M'}
{\eSeq{M} \stepsto \eSeq{M'}}}{}
\end{calculus}
\begin{calculus}
\cinferenceRule[monS|{$\circ$}S]{}{\linferenceRule[formula]
{M \stepsto M'}
{\emon{M}{N}{\pvx} \stepsto \emon{M'}{N}{\pvx}}}{}
\cinferenceRule[injLS|{$\einjL{}$}S]{}{\linferenceRule[formula]
{M \stepsto M'}
{\eInjL{M} \stepsto \eInjL{M'}}}{}
\cinferenceRule[injRS|{$\einjR{}$}S]{}{\linferenceRule[formula]
{M \stepsto M'}
{\eInjR{M} \stepsto \eInjR{M'}}}{}
\cinferenceRule[bconsSL|{$[\wedge$]}S1]{}{\linferenceRule[formula]
{M \stepsto M'}
{\ebcons{M}{N} \stepsto \ebcons{M'}{N}}}{}
\cinferenceRule[dconsSL|{$\langle\wedge\rangle$}S1]{}{\linferenceRule[formula]
{M \stepsto M'}
{\edcons{M}{N} \stepsto \edcons{M'}{N}}}{}
\cinferenceRule[bconsSR|{$[\wedge]$}S2]{}{\linferenceRule[formula]
{\isnorm{M} & N \stepsto N'}
{\ebcons{M}{N} \stepsto \ebcons{M}{N'}}}{}
\cinferenceRule[dconsSR|{$\langle\wedge\rangle$}S2]{}{\linferenceRule[formula]
{\isnorm{M} & N \stepsto N'}
{\edcons{M}{N} \stepsto \edcons{M}{N'}}}{}
\cinferenceRule[appSR|{app{S}2}]{}{\linferenceRule[formula]%
{\isnorm{M} & N \stepsto N'}
{\eapp{M}{N} \stepsto \eapp{M}{N'}}}{}
\cinferenceRule[bswapS|{$[\pdual{}]$}S]{}{\linferenceRule[formula]
{M \stepsto M'}
{\ebswap{M} \stepsto \ebswap{M'}}}{}
\cinferenceRule[dswapS|{$\langle\pdual{}\rangle$}S]{}{\linferenceRule[formula]
{M \stepsto M'}
{\edswap{M} \stepsto \edswap{M'}}}{}
\end{calculus}

\begin{calculus}
\cinferenceRule[forS|forS]{}{\linferenceRule[formula]
{A \stepsto A'}	
{\efor{A}{B}{C} \stepsto \efor{A}{B}{C}}}{}
\cinferenceRule[fpS|FPS]{}{\linferenceRule[formula]
{A \stepsto A'}	
{\efp{A}{B}{C} \stepsto \efp{A'}{B}{C}}}{}
\cinferenceRule[caseS|caseS]{}{\linferenceRule[formula]
{A \stepsto A'}
{\eCase{A}{B}{C} \stepsto \eCase{A'}{B}{C}}}{}
\cinferenceRule[unpackS|{${:*}$}S]{}{\linferenceRule[formula]
{M \stepsto M'}
{\eunpack{M}{N} \stepsto \eunpack{M'}{N}}}{}
\end{calculus}
\end{calculuscollections}
\caption{Operational semantics: structural rules}
\label{fig:op-structural}
\end{figure}

\newpage~\newpage~\newpage~\newpage~\newpage

\newpage
\section{Example Proofs}
\label{app:example-proofs}
Having developed a proof calculus, we are now equipped to verify the example games of \rref{sec:example-games}.
We give proof terms for each player's winning region in Nim in \rref{fig:example-proof-term-nim}, likewise for cake-cutting in \rref{fig:example-proof-term-cake}.
While we expect to provide additional tool support for practical proving, the natural deduction proofs in this case are already surprisingly simple.
Moreover, the examples serve to demonstrate the typical usage of each construct and to confirm that the rules are applicable to typical practical examples.
By convention, in the examples, we assume the proof variable for the fact introduced by assigning to a variable named $x$ is itself also named $x$.

\begin{figure}
\centering
\begin{align*}
\textsf{Mid} &\equiv c>0 \land \emod{c}{4} \in \{0,2,3\}\ \ \met\equiv \ediv{c}{4}\ \ \varphi \equiv (c > 0 \land \emod{c}{4} \in \{0,2,3\})\\
\textsf{aNim} &\equiv
\elam{\nzvar}{c>0}{\elam{\modvar_0}{(\emod{c}{4} \in \{1,2,3\})}{}}\\
&\spc\eforHead{\convvar}{\edcons{\nzvar}{\modvar_0}}{\modvar}{
      \pstep
   }\{\nim\}\\
\pstep &\equiv \edOpenseq \edOpenseq \ecaseHead{\elem{(\emod{c}{4}=3)}{}}\\
    &\spc\spc\ecaseLeft{\ell_1}{\edinjR{\edinjL{\iedasgn{c}{\eQE{\textsf{Mid}}{\convvar,\ell_1,c}}}}}\\
    &\spc\ecaseRight{r_1}{}\ecaseHead{\elem{(\emod{c}{4}=2)}{}}\\
    &\spc\spc\spc\ecaseLeft{\ell_2}{\edinjL{\iedasgn{c}{\eQE{\textsf{Mid}}{\convvar,r_1,\ell_2,c}}}}\\
    &\spc\spc\ecaseRight{r_2}{\edinjR{\edinjR{\iedasgn{c}{\eQE{\textsf{Mid}}{\convvar,r_1,r_2,c}}}}}\ecaseEnd\ecaseEnd\\
 &\emonInfix{\monvar:\textsf{Mid}}\\
   &\spc\edOpencons\eQE{c>0}{\monvar}  \edSepcons \edOpenswap\ebOpenseq \\
   &\spc\,\,\ebOpencons \iebasgn{c}{(\elam{\testvar}{c > 0}{\eQE{\conv\land\met_0 \metgr \met}{\monvar,c}})}\\
   &\spc\ebSepcons\ebOpencons\iebasgn{c}{(\elam{\testvar}{c > 0}{\eQE{\conv\land\met_0 \metgr \met}{\monvar,c}})}\\
   &\spc\ebSepcons\iebasgn{c}{(\elam{\testvar}{c > 0}{\eQE{\conv\land\met_0 \metgr \met}{\monvar,c}})}
    \ebClosecons\ebClosecons\ebCloseseq\edCloseswap\edClosecons\edCloseseq\edCloseseq
\end{align*}
\begin{align*}
\textsf{Mid} &\equiv c \geq 0 \land (\emod{c}{4}=2 \lor \emod{c}{4}=3 \lor \emod{c}{4}=0)\\
\met &\equiv \ediv{c}{4}\ \ \ \ \varphi \equiv (c > 0 \land \emod{c}{4} = 1)\\
\textsf{dNim} &\equiv  \elam{\nzvar}{c>0}{\elam{\modvar_0}{(\emod{c}{4}=1)}{}}\\
              &\spc \erep{\edcons{\nzvar}{\modvar}}{\pstep}{\convvar:\varphi}{\convvar}\\
\pstep &\equiv  \ebOpenseq  (\ebOpenseq \\
&\spc\,\,\ebOpencons\iebasgn{c}{(\elam{\testvar}{c > 0}{\eQE{\textsf{Mid}}{\convvar,c,\testvar}})}\\
&\spc\ebSepcons\ebOpencons\iebasgn{c}{(\elam{\testvar}{c > 0}{\eQE{\textsf{Mid}}{\convvar,c,\testvar}})}\\
&\spc\ebSepcons\iebasgn{c}{(\elam{\testvar}{c > 0}{\eQE{\textsf{Mid}}{\convvar,c,\testvar}})}\ebClosecons\ebClosecons)\\
&\emonInfix{\monvar:\textsf{Mid}}\\
&\spc\edOpenswap\ecaseHead{\monvar}\\
&\spc\spc\ecaseLeft{\ell_1}{\edinjL{\iedasgn{c}{\edcons{\eQE{c>0}{\monvar,\ell_1,c}}{\eQE{\conv\land\met<\met_0}{\monvar,\ell_1,c}}}}}\\
&\spc\ecaseRight{r_1}{}\ecaseHead{r_1}\\
&\spc\spc\spc\ecaseLeft{\ell_2}{\edinjR{\edinjL{\iedasgn{c}{\edcons{\eQE{c>0}{\monvar,r_1,\ell_2,c}}{\eQE{\conv\land\met_0 \metgr \met}{\monvar,r_1,\ell_2,c}}}}}}\\
&\spc\spc\ecaseRight{r_2}{\edinjR{\edinjR{\iedasgn{c}{\edcons{\eQE{c>0}{\monvar,r_1,r_2,c}}{\eQE{\conv\land\met_0 \metgr \met}{\monvar,r_1,r_2,c}}}}}}\ecaseEnd\ecaseEnd
\edCloseswap
\end{align*}
\caption{Proof terms for Nim example}
\label{fig:example-proof-term-nim}
\end{figure}
\begin{figure}
\begin{align*}
\textsf{Mid} &\equiv x = 0.5 \land y = 0.5\\
\textsf{aCake} &\equiv\spc\ietcons{x}{0.5}{\edcons{\eQE{0 < x < 1}{}}{\iedasgn{y}{\edcons{\eQE{\textsf{Mid}}{x,y}}{\eQE{\textsf{Mid}}{x,y}}}}}\\
&\emonInfix{\monvar:\textsf{Mid}}\\
&\spc\edOpenswap\ebOpencons\ebswap{\edseq{\iedasgn{a}{\iedasgn{d}{\eQE{a\geq 0.5}{a,x}}}}}\\
&\spc\spc\ebSepcons\ebswap{\edseq{\iedasgn{a}{\iedasgn{d}{\eQE{a\geq 0.5}{a,x}}}}}\ebClosecons\edCloseswap\\
\textsf{dCake} &\equiv
\elam{x}{\allrat}{}
\elam{\rangevar}{(0\leq x \leq 1)}{}\\
&\spc(\iebasgn{y}{\eQE{x + y = 1}{\rangevar,y}})\\
&\emonInfix{\monvar:(x + y = 1)} \ebOpenswap\\
&\spc\ecaseHead{(\esplit{x}{0.5}{})}\\
&\spc\spc\ecaseLeft{\ell}{\einjL{\edswap{\iebasgn{a}{\iebasgn{d}{\eQE{d \geq 0.5}{\ell,d,m,\rangevar}}}}}}\\
&\spc\spc\ecaseRight{r}{\einjR{\edswap{\iebasgn{a}{\iebasgn{d}{\eQE{d \geq 0.5}{r,d}}}}}}\ecaseEnd\ebCloseswap
\end{align*}
\caption{Proof terms for cake-cutting example}
\label{fig:example-proof-term-cake}
\end{figure}

\textsf{aNim} and \textsf{dNim} both give examples of typical convergence proofs, with a subtle difference.
In \textsf{aNim}, the invariant is $\emod{c}{4} \neq 0$, i.e., $(\emod{c}{4} = 0) \limply \bff,$ so the invariant does not tell Angel which value $\emod{c}{4}$ assumes, only that it would be a contradiction to take value 0.
Thus Angel must inspect the state using the law of excluded middle (which is constructive for equality tests on natural numbers).
For each case Angel picks the branch that ensures $\emod{c}{4}=0$. 
The uses of \irref{mon} reduces the number of cases from 9 total to 3 for Angel and 3 for demon: Under the assumption (named $\monvar$) that $c > 0 \land \emod{c}{4}=0,$ then all possible Demon choices restore the invariant, which we prove by \irref{QE} (automation, in practice).
The \textsf{dNim} shows a similar strategy where the dormant Angel mirrors Demon's choices, but for the sake of demonstrating both techniques, the monotonicity formula is a constructive disjunction, ie.., Demon explicitly tells Angel which branch they took, which Angel then mirrors.

In \textsc{aCake}, Angel splits the cake evenly, which is their optimal strategy.
Even assuming Demon is entitled to nonconstructive strategies, whatever slice they pick must be one of the two, both of which have size 0.5.
In \textsc{dCake}, Demon splits the cake according to an arbitrary strategy, then Angel inspects the value of $x$ to decide which slice to take.
If $x \geq 0.5$ then Angel takes it, else Angel takes $y$.
Angel has an optimal strategy because rationals can be compared exactly.
If we operated over constructive reals, Angel would not have exact optimality because reals could not be compared exactly.
They would, however, be allowed to get arbitrarily close to optimal.
\newpage~\newpage

\newpage
\section{Theory Proofs}
\label{app:proofs}
We give proofs regarding semantics, soundness, and progress and preservation.

\subsection{Static semantics}
We repeat the definitions of free variable $\freevars{e}$, bound variable $\boundvars{\alpha}$, and must-bound variable $\mustboundvars{\alpha}$ computations which will be used throughout the proofs.
Free variables are those which might influence the expression.
Bound variables are those which could be modified during a game, and must-bound variables are those which are modified on every path of a game.
The definitions are standard~\cite{DBLP:conf/cade/Platzer18}.

\begin{align*}
  \freevars{f \sim g} &= \freevars{f} \sim \freevars{g}\\
  \freevars{\ddiamond{\alpha}{\phi}} &= \freevars{\alpha} \cup (\freevars{\phi} \setminus \mustboundvars{\alpha})\\
  \freevars{\dbox{\alpha}{\phi}} &= \freevars{\alpha} \cup (\freevars{\phi} \setminus \mustboundvars{\alpha})\\\hline
  \freevars{\ptest{\phi}} &= \freevars{\phi}\\
  \freevars{\humod{x}{f}} &= \freevars{f}\\
  \freevars{\prandom{x}} &= \emptyset\\
  \freevars{\alpha;\beta} &= \freevars{\alpha} \cup (\freevars{\beta} \setminus \mustboundvars{\alpha})\\
  \freevars{\alpha\cup\beta} &= \freevars{\alpha} \cup \freevars{\beta}\\
  \freevars{\prepeat{\alpha}} &= \freevars{\alpha}\\
  \freevars{\pdual{\alpha}} &= \freevars{\alpha}
\end{align*}

\begin{align*}
  \boundvars{\ptest{\phi}} &= \emptyset\\
  \boundvars{\humod{x}{f}} &= \{x\}\\
  \boundvars{\prandom{x}} &= \{x\}\\
  \boundvars{\alpha;\beta} &= \boundvars{\alpha} \cup \boundvars{\beta}\\
  \boundvars{\alpha\cup\beta} &= \boundvars{\alpha} \cup \boundvars{\beta}\\
  \boundvars{\prepeat{\alpha}} &= \boundvars{\alpha}\\
  \boundvars{\pdual{\alpha}} &= \boundvars{\alpha}
\end{align*}

\begin{align*}
  \mustboundvars{\ptest{\phi}} &= \emptyset\\
  \mustboundvars{\humod{x}{f}} &= \{x\}\\
  \mustboundvars{\prandom{x}} &= \{x\}\\
  \mustboundvars{\alpha;\beta} &= \mustboundvars{\alpha} \cup \mustboundvars{\beta}\\
  \mustboundvars{\alpha\cup\beta} &= \mustboundvars{\alpha} \cap \mustboundvars{\beta}\\
  \mustboundvars{\prepeat{\alpha}} &= \emptyset\\
  \mustboundvars{\pdual{\alpha}} &= \mustboundvars{\alpha}
\end{align*}

\subsection{Repetition}
In classical \GL, the challenge in defining the semantics of repetition games $\prepeat{\alpha}$ is that the number of iterations, while finite, can depend on both players' actions and is thus not known in advance, while the \DL-like semantics of $\prepeat{\alpha}$ as the finite reflexive, transitive closure of $\alpha$ gives an advance-notice semantics.
Classical \GL provides the no-advance-notice semantics as a fixed point~\cite{DBLP:conf/focs/Parikh83}, and we adopt the fixed point semantics as well.
Surprisingly, the fixed point semantics and inflationary semantics agree with each other for \CGL because Angel must fix a computable strategy in advance.
This results in a semantics with only Demonic nondeterminism, reminiscent of Kripke semantics for \DL.
\begin{definition}[Inflationary \CGL semantics]
We define the $k$-step Angelic and Demonic inflations recursively for $k \in \mathbb{N},$ i.e.,
\begin{align*}
\istrat[\alpha]{X}{0}  &= X & \istrat[\alpha]{X}{k+1} &= \strategyforR[\alpha]{\apR{(\istrat[\alpha]{X}{k})}}\\
\idstrat[\alpha]{X}{0} &= X & \idstrat[\alpha]{X}{k+1} &= \dstrategyforR[\alpha]{\apR{(\idstrat[\alpha]{X}{k})}}
\end{align*}
\end{definition}
\begin{remark}[Pre-inflation and post-inflation]
Pre-inflation and post-inflation agree in the following sense:
\begin{align*}
\istrat[\alpha]{(\strategyforR[\alpha]{\apR{X}})}{k} &= \strategyforR[\alpha]{\apR{(\istrat[\alpha]{X}{k})}}
&
\idstrat[\alpha]{(\dstrategyforR[\alpha]{\dpR{X}})}{k} &= \dstrategyforR[\alpha]{\dpR{(\istrat[\alpha]{X}{k})}}
\end{align*}
\label{rem:app-inflate}
\end{remark}
\begin{proof}
Each claim is by induction on $k$.

Angel, case 0:
$\istrat[\alpha]{(\strategyforR[\alpha]{\apR{X}})}{0}
= \strategyforR[\alpha]{\apR{X}}
= \strategyforR[\alpha]{(\apR{(\istrat[\alpha]{X}{0})})}$.

Angel, case $k+1$:\\
Assume IH:  $\istrat[\alpha]{(\strategyforR[\alpha]{\apR{X}})}{k} = \strategyforR[\alpha]{(\apR{(\istrat[\alpha]{X}{k})})}$\\
Show: $\istrat[\alpha]{(\strategyforR[\alpha]{\apR{X}})}{k+1}
=      \istrat[\alpha]{(\strategyforR[\alpha]{\apR{(\strategyforR[\alpha]{\apR{X}})}})}{k}
=_{IH}  \strategyforR[\alpha]{\apR{(\istrat[\alpha]{\apR{(\strategyforR[\alpha]{X})}}{k})}}
= \strategyforR[\alpha]{\apR{(\istrat[\alpha]{X}{k+1})}}$
as desired, completing the induction for Angel.

Demon, case 0:
$\idstrat[\alpha]{(\dstrategyforR[\alpha]{\dpR{X}})}{0}
= \dstrategyforR[\alpha]{\dpR{X}}
= \dstrategyforR[\alpha]{\dpR{(\idstrat[\alpha]{X}{0})}}$.

Demon, case $k+1$:\\
Assume IH:  $\idstrat[\alpha]{(\dstrategyforR[\alpha]{\dpR{X}})}{k} = \dstrategyforR[\alpha]{(\dpR{(\idstrat[\alpha]{X}{k})})}$\\
Show: $\idstrat[\alpha]{(\dstrategyforR[\alpha]{\dpR{X}})}{k+1}
=        \idstrat[\alpha]{(\dstrategyforR[\alpha]{(\dpR{(\dstrategyforR[\alpha]{\dpR{X}})})})}{k}
=_{IH}   \dstrategyforR[\alpha]{\dpR{(\idstrat[\alpha]{\dpR{(\dstrategyforR[\alpha]{X})}}{k})}}
= \dstrategyforR[\alpha]{\dpR{(\idstrat[\alpha]{X}{k+1})}}$
as desired, completing the induction for Demon.
\end{proof}
\begin{lemma}[Alternative semantics]
The \CGL fixed-point definition of repetition agrees with the inflationary definition:
\begin{align*}
\strategyforR[\prepeat{\alpha}]{X}   &= \bigcup_{k\in\mathbb{N}} \apL{(\istrat[\alpha]{X}{k})}
&
\dstrategyforR[\prepeat{\alpha}]{X} &= \bigcup_{k\in\mathbb{N}} \dpL{(\idstrat[\alpha]{X}{k})}
\end{align*}
\label{lem:app-altern}
\end{lemma}
\begin{proof}
We show the set equality first for Angel by showing both directions of inclusion, then likewise for Demon.

For Angel, let $L = \strategyforR[\prepeat{\alpha}]{X} = \bigcap\{\apL{Z}~|~X \cup \strategyforR[\alpha]{\apR{Z}} \subseteq Z\}$.
Let $R = \bigcup_{k\in\mathbb{N}} \apL{\istrat[\alpha]{X}{k}}$.
First show $L \subseteq R$.
Because $L$ was defined inductively, it suffices to show any show two cases:
Case 1: $X \subseteq L$.
Then $\apL{X} = \apL{\istrat[\alpha]{X}{0}} \subseteq R$.

Case 2: $\strategyforR[\alpha]{\apR{Z}} \subseteq Z$.
By IH, assume some $\apL{Z} \subseteq R$.
Let $Z^i = Z \cap \istrat[\alpha]{X}{k},$ then $Z = \bigcup_{i \in \mathbb{N}}{Z^i}$ and suffices to show (since $\apL{\cdot}$ is monotone)
that $\strategyforR[\alpha]{\apR{Z^i}} \subseteq R$.
This holds because  (by \rref{lem:app-mono} first, then \rref{rem:app-inflate})
$\strategyforR[\alpha]{\apR{Z^i}} \subseteq \strategyforR[\alpha]{\apR{\istrat[\alpha]{X}{k}}} = \istrat[\alpha]{X}{1+k} = Z_{i+i}$.
Then since $\bigcup_{k\in\mathbb{N}}{Z^{i+1}} = \bigcup_{k\in\mathbb{N}}{Z^i} = Z$ it follows that each $\strategyforR[\alpha]{\apR{Z^i}} \subseteq R,$ as desired.
\end{proof}

It may be surprising at first that \rref{rem:app-inflate} holds because unlike in DL,
the formula  $\ddiamond{\alpha}{\ddiamond{\prepeat{\alpha}}{\phi}}$ need not hold when the commuted formula
$\ddiamond{\prepeat{\alpha}}{\ddiamond{\alpha}{\phi}}$ does.
However, \rref{rem:app-inflate} does not refer to two different games, but to two characterizations of the executions taken through a game $\prepeat{\alpha}$ by the same realizers given in $X$.
Likewise, \label{lem:app-inflate} may be surprising because the game $\alpha$ is also disequal to the game $\bigcup_{k\in\mathbb{N}}{\alpha^k}$ and because fixed point constructions of winning regions in classical \GL can take infinitely many repetitions to converge.
However, \label{lem:app-inflate} only claims something simpler: once the strategies have been fixed (by $X$), it suffices to consider finite executions according to those strategies, which is equally true in classical \GL and \CGL.

We differ from the classical definition in that we require the termination condition to be computable ($\rzFst{\rzApp{\myaa}{\ab}}$), which gives a constructive interpretation as shown by the elimination rule (\irref{dloopE} in \rref{sec:proof-calculus}).
We postulate that the closure ordinal of \CGL is exactly the Church-Kleene ordinal $\churchkleene,$ which is only known to be a lower bound for classical GL's~\cite{DBLP:journals/tocl/Platzer15}.
The Angelic choice whether to stop ($\apL{Z}$) or iterate the loop ($\apR{Z}$) is analogous to the case for $\alpha \cup \beta$.

\subsection{Soundness}
Let $\eren{e}{x}{y}$ denote the  renaming of $x$ and $y$ in $e,$ (where $e$ is term $f$ formula $\phi$, game $\alpha$, or realizer $\myaa$).
\begin{lemma}[Renaming is Self-dual]
\label{lem:app-ren-selfdual}
Renaming the same variables twice cancels, because renaming is by transposition.
\begin{itemize}
\item $\eren{{\eren{f}{x}{y}}}{x}{y} = f$
\item $\eren{{\eren{\myaa}{x}{y}}}{x}{y} = \myaa$
\item $\eren{{\eren{\phi}{x}{y}}}{x}{y} = \phi$
\item $\eren{{\eren{\alpha}{x}{y}}}{x}{y} = \alpha$
\item $\sren{{\sren{\om}{x}{y}}}{x}{y} = \om$
\end{itemize}
\end{lemma}
The first claim holds for any reasonable representation of computable functions, and can be proved by induction on the example term constructors.
The next three claims are by induction, where the induction on formulas and games is simultaneous.
The last claim is direct.

\begin{lemma}[Renaming]
\label{lem:app-tren}
Renaming commutes with the interpretation functions.
\begin{itemize}
\item $\tint{\eren{f}{x}{y}}{\om}  = \tint{f}{\sren{\om}{x}{y}}$
\item $\fint{\eren{\phi}{x}{y}}{}     = \{ (\sren{\myaa}{x}{y},\sren{\om}{x}{y})~|~(\myaa,\om) \in\fintR{\phi}\}$
\item $\strategyforR[\eren{\alpha}{x}{y}]{\sren{X}{x}{y}} = \eren{(\strategyforR[\alpha]{X})}{x}{y}$
\item $\dstrategyforR[\eren{\alpha}{x}{y}]{\sren{X}{x}{y}} = \eren{(\dstrategyforR[\alpha]{X})}{x}{y}$
\item $\eren{\myaa}{x}{y}(\om) = \myaa(\sren{\om}{x}{y})$
\end{itemize}
\end{lemma}
\begin{proof}
For terms, this property holds of any reasonable representation of computable functions.
We give an inductive proof which handles the polynomial cases from the paper.
Formula and program cases are by simultaneous induction on formulas and programs.
Realizers are by a separate induction.

We give cases for term connectives.
Any term language which supports the renaming and substitution properties could be used; these cases are only examples.

\mycase $q$
Have $\tint{\eren{q}{x}{y}}{\om} = \tint{q}{\om} = q = \tint{q}{\sren{\om}{x}{y}}$

\mycase $x$
Have $\tint{\eren{z}{x}{y}}{\om} = \om(\eren{z}{x}{y}) = \sren{\om}{x}{y}(z) = \tint{z}{\sren{\om}{x}{y}}$

\mycase $f + g$
Have $\tint{\eren{(f + g)}{x}{y}}{\om}
=\tint{\eren{f}{x}{y}}{\om} + \tint{\eren{g}{x}{y}}{\om}
=\tint{f}{\sren{\om}{x}{y}} + \tint{g}{\sren{\om}{x}{y}}
=\tint{f + g}{\sren{\om}{x}{y}}$

\mycase $f \cdot g$
Have $\tint{\eren{(f \cdot g)}{x}{y}}{\om}
=\tint{\eren{f}{x}{y}}{\om} \cdot \tint{\eren{g}{x}{y}}{\om}
=\tint{f}{\sren{\om}{x}{y}} \cdot \tint{g}{\sren{\om}{x}{y}}
=\tint{f \cdot g}{\sren{\om}{x}{y}}$

We give the formula cases.

\mycase $f > g$
Have 
\begin{align*}
&\phantom{= } \fintR{\eren{f > g}{x}{y}}\\
&= \{(\rzNil,\om)~|~\tint{\eren{f}{x}{y}}{\om} > \tint{\eren{g}{x}{y}}{\om}\}\\
&= \{(\rzNil,\om)~|~\tint{f}{\sren{\om}{x}{y}} > \tint{g}{\sren{\om}{x}{y}}\}\\
&= \{(\rzNil,\sren{\om}{x}{y})~|~\tint{f}{\om} > \tint{g}{\om}\}\\
&= \sren{\fintR{f > g}}{x}{y}.
\end{align*}
The cases for $<, \leq, =, \neq \geq$ are symmetric.

\mycase $\ddiamond{\alpha}{\phi}$
Have
\begin{align*}
&\phantom{= } \fintR{\eren{\ddiamond{\alpha}{\phi}}{x}{y}}\\
&=  \{(\myaa,\om)~|~\strategyforR[\eren{\alpha}{x}{y}]{\{(\myaa,\om)\}} \subseteq \fint{\eren{\phi}{x}{y}} \cup \{\stt\}\}\\
&=  \{(\myaa,\om)~|~\sren{\strategyforR[\alpha]{\{(\eren{\myaa}{x}{y},\eren{\om}{x}{y})\}}}{x}{y} \subseteq \sren{\fint{\phi}}{x}{y} \cup \{\stt\}\}\\
&=  \{(\sren{\myaa}{x}{y},\sren{\om}{x}{y})~|~\sren{\strategyforR[\alpha]{\{(\myaa,\om)\}}}{x}{y} \subseteq \sren{\fint{\phi}}{x}{y} \cup \{\stt\}\}\\
&=  \{(\sren{\myaa}{x}{y},\sren{\om}{x}{y})~|~\strategyforR[\alpha]{\{(\myaa,\om)\}} \subseteq \fint{\phi} \cup \{\stt\}\}\\
&= \sren{\fint{\ddiamond{\alpha}{\phi}}}{x}{y}
\end{align*}.

\mycase $\dbox{\alpha}{\phi}$
Have
\begin{align*}
&\phantom{= }\fintR{\eren{\dbox{\alpha}{\phi}}{x}{y}}\\
&=  \{(\myaa,\om)~|~\dstrategyforR[\eren{\alpha}{x}{y}]{\{(\myaa,\om)\}} \subseteq \fint{\eren{\phi}{x}{y}}  \cup \{\stt\} \}\\
&=  \{(\myaa,\om)~|~\sren{\dstrategyforR[\alpha]{\{(\eren{\myaa}{x}{y},\eren{\om}{x}{y})\}}}{x}{y} \subseteq \sren{\fint{\phi}}{x}{y}  \cup \{\stt\}\}\\
&=  \{(\sren{\myaa}{x}{y},\sren{\om}{x}{y})~|~\sren{\dstrategyforR[\alpha]{\{(\myaa,\om)\}}}{x}{y} \subseteq \sren{\fint{\phi}}{x}{y}  \cup \{\stt\}\}\\
&=  \{(\sren{\myaa}{x}{y},\sren{\om}{x}{y})~|~\dstrategyforR[\alpha]{\{(\myaa,\om)\}} \subseteq \fint{\phi} \cup \{\stt\}\}\\
&= \sren{\fint{\dbox{\alpha}{\phi}}}{x}{y}
\end{align*}.

We give the game cases for Angel.

\mycase $\humod{x}{f}$
Have $\strategyforR[\eren{\{\humod{x}{f}\}}{x}{y}]{\sren{X}{x}{y}}
= \strategyforR[\humod{y}{(\eren{f}{x}{y})}]{\sren{X}{x}{y}}
= \{(\myaa,\subst[\om]{y}{\tint{\eren{f}{x}{y}}{\om}})~|~(\myaa,\om) \in \sren{X}{x}{y}\}
= \{(\myaa,\subst[\om]{y}{\tint{f}{\sren{\om}{x}{y}}})~|~(\myaa,\om) \in \sren{X}{x}{y}\}
= \{(\sren{\myaa}{x}{y},\subst[\om]{x}{\tint{f}{\om}})~|~(\myaa,\om) \in  X\}
= \sren{\strategyforR[\humod{x}{f}]{X}}{x}{y}$

\mycase $\prandom{x}$
Have $\strategyforR[\eren{\{\prandom{x}\}}{x}{y}]{\sren{X}{x}{y}}
=\strategyforR[\prandom{y}]{\sren{X}{x}{y}}
=\{(\myaa,\subst[\om]{y}{\eren{\myaa}{x}{y}(\om)})~|~(\myaa,\om) \in \sren{X}{x}{y}\}
=\{(\myaa,\subst[\om]{y}{\myaa(\sren{\om}{x}{y})})~|~ (\myaa,\om) \in \sren{X}{x}{y}\}
=\{(\sren{\myaa}{x}{y},\sren{\subst[\om]{x}{\myaa(\subst[\om]{x}{y})}}{x}{y})~|~(\myaa,\om) \in X\}
=\sren{\strategyforR[\prandom{x}]{X}}{x}{y}$

\mycase $\ptest{\phi}$
Have $\strategyforR[\ptest{\eren{\phi}{x}{y}}]{\sren{X}{x}{y}}
= \{(\rzSnd{\myaa}, \om)~|~ (\myaa,\om) \in \sren{X}{x}{y}, (\rzFst{\myaa},\om) \in \fintR{\eren{\phi}{x}{y}} \}
= \{(\rzSnd{\myaa}, \om)~|~ (\myaa,\om) \in \sren{X}{x}{y}, (\sren{\rzFst{\myaa}}{x}{y},\sren{\om}{x}{y}) \in \fintR{\phi} \}
= \{(\rzSnd{\myaa}, \om)~|~ (\sren{\myaa}{x}{y},\sren{\om}{x}{y}) \in X, (\sren{\rzFst{\myaa}}{x}{y},\sren{\om}{x}{y}) \in \fintR{\phi} \}
= \{(\sren{\rzSnd{\myaa}}{x}{y}, \sren{\om}{x}{y})~|~ (\myaa,\om) \in X, (\rzFst{\myaa},\om) \in \fintR{\phi} \}
= \sren{\strategyforR[\ptest{\phi}]{X}}{x}{y}$

\mycase $\alpha \cup \beta$
Have $\strategyforR[\eren{\alpha \cup \beta}{x}{y}]{\eren{X}{x}{y}}
= \strategyforR[\eren{\alpha}{x}{y}]{\eren{\apL{X}}{x}{y}} \cup \strategyforR[\eren{\beta}{x}{y}]{\eren{\apR{X}}{x}{y}}
= \sren{\strategyforR[\alpha]{\apL{X}}}{x}{y} \cup \sren{\strategyforR[\beta]{\apR{X}}}{x}{y}
= \sren{\strategyforR[\alpha\cup\beta]{X}}{x}{y}$

\mycase $\alpha; \beta$
Have $\strategyforR[\alpha;\beta]{\eren{X}{x}{y}}
= \strategyforR[\alpha]{\strategyforR[\beta]{\eren{X}{x}{y}}}
= \strategyforR[\alpha]{\sren{\strategyforR[\beta]{X}}{x}{y}}
= \sren{\strategyforR[\alpha]{\strategyforR[\beta]{X}}}{x}{y}$

\mycase $\prepeat{\alpha}$
Have 
\begin{align*}
&\ \ \ \ \strategyforR[\prepeat{\eren{\alpha}{x}{y}}]{\sren{X}{x}{y}} \\
&= \bigcap{\apL{Z}~|~\sren{X}{x}{y} \cup \strategyforR[\eren{\alpha}{x}{y}]{\apR{Z}} \subseteq Z}\\
&= \bigcap{\apL{Z}~|~\sren{X}{x}{y} \cup \sren{\strategyforR[\alpha]{\sren{\apR{Z}}{x}{y}}}{x}{y} \subseteq Z}\\
&= \bigcap{\apL{Z}~|~X \cup \strategyforR[\alpha]{\sren{\apR{Z}}{x}{y}} \subseteq \sren{Z}{x}{y}}\\
&= \bigcap{\sren{\apL{Z}}{x}{y}~|~X \cup \strategyforR[\alpha]{\apR{Z}} \subseteq Z}\\
&= \bigcap{\apL{\sren{Z}{x}{y}}~|~X \cup \strategyforR[\alpha]{\apR{Z}} \subseteq Z}\\
&= \sren{\strategyforR[\prepeat{\alpha}]{X}}{x}{y}
\end{align*}

\mycase $\pdual{\alpha}$
$\strategyforR[\pdual{{\eren{\alpha}{x}{y}}}]{\sren{X}{x}{y}}
= \dstrategyforR[\eren{\alpha}{x}{y}]{\sren{X}{x}{y}}
= \sren{\dstrategyforR[\alpha]{X}}{x}{y}
= \sren{\strategyforR[\pdual{\alpha}]{X}}{x}{y}$

We give the game cases for Demon.

\mycase $\humod{x}{f}$
Have
$\dstrategyforR[\eren{\humod{x}{f}}{x}{y}]{\sren{X}{x}{y}}
= \dstrategyforR[\humod{y}{\eren{f}{x}{y}}]{\sren{X}{x}{y}}
= \{(\rzNil,\subst[\om]{y}{\tint{\eren{f}{x}{y}}{\om}})~|~ (\rzNil,\om)\in \sren{X}{x}{y}\}
= \{(\rzNil,\subst[\om]{y}{\tint{f}{\eren{\om}{x}{y}}})~|~ (\rzNil,\om)\in \sren{X}{x}{y}\}
= \{(\rzNil,\subst[\om]{y}{\tint{f}{\eren{\om}{x}{y}}})~|~ (\rzNil,\sren{\om}{x}{y}) \in X\}
= \{(\rzNil,\subst[\sren{\om}{x}{y}]{y}{\tint{f}{\om}})~|~ (\rzNil,\om) \in X\}
= \{(\rzNil,\sren{\subst[\om]{x}{\tint{f}{\om}}}{x}{y})~|~ (\rzNil,\om) \in X\}
= \sren{\dstrategyforR[\humod{x}{f}]{X}}{x}{y}$

\mycase $\prandom{x}$
Have
$\dstrategyforR[\eren{\prandom{x}}{x}{y}]{\sren{X}{x}{y}}
=\dstrategyforR[\prandom{y}]{\sren{X}{x}{y}}
=\{\subst[\om]{y}{r}~|~\om \in \sren{X}{x}{y},\text{ exists }r\in\allrat\}
=\{\sren{\subst[\om]{x}{r}}{x}{y}~|~ \om \in X,\text{ exists }r\in\allrat\}
=\sren{\dstrategyforR[\prandom{x}]{X}}{x}{y}$

\mycase $\ptest{\phi}$
Have
$\dstrategyforR[\ptest{\eren{\phi}{x}{y}}]{\sren{X}{x}{y}}
= \{(\rzApp{\myaa}{\ab}, \om)~|~(\myaa,\om) \in \sren{X}{x}{y}, (\ab,\om) \in \fintR{\eren{\phi}{x}{y}}\}
= \{(\rzApp{\myaa}{\ab}, \om)~|~(\sren{\myaa}{x}{y},\sren{\om}{x}{y}) \in X, (\sren{\ab}{x}{y},\sren{\om}{x}{y}) \in \fintR{\phi}\}
= \{(\sren{\rzApp{\myaa}{\ab}}{x}{y}, \sren{\om}{x}{y})~|~(\myaa,\om) \in X, (\ab,\om) \in \fintR{\phi}\}
= \sren{\dstrategyforR[\ptest{\phi}]{X}}{x}{y}$

\mycase $\alpha \cup \beta$
Have
$\dstrategyforR[\eren{\alpha \cup \beta}{x}{y}]{\sren{X}{x}{y}} 
= \dstrategyforR[\eren{\alpha}{x}{y}]{\dpL{{\sren{X}{x}{y}}}} \cup \dstrategyforR[\eren{\beta}{x}{y}]{\dpR{{\sren{X}{x}{y}}}}
= \dstrategyforR[\eren{\alpha}{x}{y}]{\sren{{\dpL{X}}}{x}{y}} \cup \dstrategyforR[\eren{\beta}{x}{y}]{\sren{{\dpR{X}}}{x}{y}}
= \sren{\dstrategyforR[\alpha]{\dpL{X}}}{x}{y} \cup \sren{\dstrategyforR[\beta]{\dpR{X}}}{x}{y}
= \sren{\dstrategyforR[\alpha]{\dpL{X}} \cup \dstrategyforR[\beta]{\dpR{X}}}{x}{y}
= \sren{\dstrategyforR[\alpha\cup\beta]{X}}{x}{y}$

\mycase $\alpha; \beta$
Have
$ \dstrategyforR[\eren{\alpha}{x}{y};\eren{\beta}{x}{y}]{\sren{X}{x}{y}} 
= \dstrategyforR[\eren{\alpha}{x}{y}]{\dstrategyforR[\eren{\beta}{x}{y}]{\sren{X}{x}{y}}}
= \dstrategyforR[\eren{\alpha}{x}{y}]{\sren{\dstrategyforR[\beta]{X}}{x}{y}}
= \sren{\dstrategyforR[\alpha]{\dstrategyforR[\beta]{X}}}{x}{y}$

\mycase $\prepeat{\alpha}$
Have
\begin{align*}
&\ \ \ \ \dstrategyforR[\prepeat{\eren{\alpha}{x}{y}}]{\sren{X}{x}{y}}\\
&= \bigcap{\dpL{Z}~|~\sren{X}{x}{y} \cup \dstrategyforR[\eren{\alpha}{x}{y}]{\dpR{Z}} \subseteq Z}\\
&= \bigcap{\dpL{Z}~|~\sren{X}{x}{y} \cup \sren{\dstrategyforR[\alpha]{\sren{\dpR{Z}}{x}{y}}}{x}{y} \subseteq Z}\\
&= \bigcap{\dpL{Z}~|~X \cup \dstrategyforR[\alpha]{\sren{\dpR{Z}}{x}{y}} \subseteq \sren{Z}{x}{y}}\\
&= \bigcap{\sren{\dpL{Z}}{x}{y}~|~X \cup \dstrategyforR[\alpha]{\dpR{Z}} \subseteq Z}\\
&= \bigcap{\dpL{\sren{Z}{x}{y}}~|~X \cup \dstrategyforR[\alpha]{\dpR{Z}} \subseteq Z}\\
&= \sren{\dstrategyforR[\prepeat{\alpha}]{X}}{x}{y}
\end{align*}

\mycase $\pdual{\alpha}$
$\dstrategyforR[\pdual{{\eren{\alpha}{x}{y}}}]{\sren{X}{x}{y}}
= \strategyforR[\eren{\alpha}{x}{y}]{\sren{X}{x}{y}}
= \sren{\strategyforR[\alpha]{X}}{x}{y}
= \sren{\dstrategyforR[\pdual{\alpha}]{X}}{x}{y}$

We give the cases for realizers.

 \mycase $z$:
 Have $\eren{z}{x}{y}(\om) = z(\sren{\om}{x}{y})$ by definition of renaming on realizer variables.

 \mycase $f$:
 Have $\eren{f}{x}{y}(\om) = f(\sren{\om}{x}{y}$ by IH on terms.

 \mycase $\rzFst{\myaa}$:
 Have $\eren{\rzFst{\myaa}}{x}{y}(\om)
 = \rzFst{\eren{\myaa}{x}{y}}(\om)
 = \rzFst{\eren{\myaa}{x}{y}(\om)}
 = \rzFst{\myaa(\sren{\om}{x}{y})}
 = \rzFst{\myaa}(\sren{\om}{x}{y})$

 \mycase $\rzSnd{\myaa}$:
 Have $\eren{\rzSnd{\myaa}}{x}{y}(\om)
 = \rzSnd{\eren{\myaa}{x}{y}}(\om)
 = \rzSnd{\eren{\myaa}{x}{y}(\om)}
 = \rzSnd{\myaa(\sren{\om}{x}{y})}
 = \rzSnd{\myaa}(\sren{\om}{x}{y})$

 \mycase $\rzNil$:
 Have $\eren{\rzNil}{x}{y}(\om) = \rzNil(\om) = \text{unit} = \rzNil(\sren{\om}{x}{y})$

 \mycase $\rzCons{\myaa}{\ab}$
 Have $\eren{(\rzCons{\myaa}{\ab})}{x}{y}(\om)
 = \rzCons{\eren{\myaa}{x}{y}}{\eren{\ab}{x}{y}}(\om)
 = (\eren{\myaa}{x}{y}(\om),\eren{\ab}{x}{y}(\om)
 = (\myaa(\sren{\om}{x}{y}),\ab(\sren{\om}{x}{y}))
 = \rzCons{\myaa}{\ab}(\sren{\om}{x}{y}) $

\mycase $(\rzBLam{\nu}{\myaa})$
Have $\eren{(\rzBLam{\nu}{\myaa})}{x}{y}(\om)
=(\rzBLam{\nu}{\tsub{\myaa}{\nu}{\sren{\nu}{x}{y}}}(\om))
=\tsub{\myaa}{\nu}{\sren{\om}{x}{y}}
=(\rzBLam{\nu}{\myaa}(\sren{\om}{x}{y}))$

\mycase $(\rzFOLam{z}{\allrat}{\myaa})$
Have $\eren{(\rzFOLam{z}{\allrat}{\myaa})}{x}{y}(\om)
=(\rzFOLam{z}{\allrat}{\eren{\myaa}{x}{y}(\om)})
=(\rzFOLam{z}{\allrat}{\myaa(\sren{\om}{x}{y})})
=(\rzFOLam{z}{\allrat}{\myaa(\om)}(\sren{\om}{x}{y}))$

\mycase $(\rzHOLam{z}{\phi}{\myaa})$
Have that $\eren{(\rzHOLam{z}{\phi}{\myaa})}{x}{y}(\om)
=(\rzHOLam{z}{\phi}{\eren{\myaa}{x}{y}(\om)})
=(\rzHOLam{z}{\phi}{\myaa(\sren{\om}{x}{y})})
=(\rzHOLam{z}{\phi}{\myaa(\om)}(\sren{\om}{x}{y}))$
\end{proof}

\begin{lemma}[Suitable substitutes]
The realizer $\myaa$ is suitable for $\alpha$ iff $\sigma(\myaa)$ is suitable for $\sigma(\alpha)$.
\end{lemma}
\begin{proof}
By induction on $\alpha$ with separate cases for Demon and Angel suitability.
\end{proof}

\begin{definition}[Substitution adjoints]
We write $\sigma$ for any substitution $\tsub{\cdot}{x}{f}$.
For any region $X,$ we write $\adj{X}$ for the adjoint region which applies substitution $\sigma$ to every possibility in $X$.
We write $\spp{e}$ for the result of substituting $\sigma$ in expression $e$, which could be a term $f,$ formula $\phi,$ game $\alpha,$ or realizer $\myaa$.
\end{definition}

\begin{lemma}[Substitution distributes over unions]
For all regions $A$ and $B$ have
$\adj{(A \cup B)} = \adj{A} \cup \adj{B}$
\end{lemma}
\begin{proof}
Follows from $\spp{\myaa}(\om) = \myaa(\adj{\om})$.
\end{proof}

\begin{lemma}[Expression substitution]
If the substitution $\sigma$ is admissible, then:
\begin{itemize}
\item $\spp{\myaa}(\om) = \myaa(\adj{\om})$
\item $\tint{\sapp{\sigma}{f}}{\om} = \tint{f}{\adj{\om}}$
\item $(\myaa,\om) \in \dstrategyforR[\spp{\alpha}]{X} \text{ iff } (\adj{\myaa},\adj{\om}) \in \dstrategyforR[\alpha]{\adj{X}}$
\item $(\myaa,\om) \in \strategyforR[\spp{\alpha}]{X} \text{ iff } (\adj{\myaa},\adj{\om}) \in \strategyforR[\alpha]{\adj{X}}$
\item $(\myaa,\om) \in \fintR{\sapp{\sigma}{\phi}} \text{ iff } (\adj{\myaa},\adj{\om}) \in \fintR{\phi}$
\end{itemize}
\label{lem:app-tsub}
\end{lemma}
\begin{proof}
We give the realizer cases.

\mycase $\rzNil$
Have $\spp{\rzNil}(\om) = \rzNil(\om) = \text{unit} = \rzNil(\adj{\om})$

\mycase $x$
Have $\spp{x}(\om) = x(\om) = x(\adj{\om})$ by definition of substitution on realizer variables.

\mycase $f$
Have $\spp{f}(\om) = \tint{\spp{f}}{\om} = \tint{f}{\adj{\om}} = f(\adj{\om})$ by term IH.

\mycase $\rzFst{\myaa}$
Have $\spp{\rzFst{\myaa}}(\om)
= \rzFst{\spp{\myaa}}(\om)
= \rzFst{\spp{\myaa}(\om)}
= \rzFst{\myaa(\adj{\om})}
= \rzFst{\myaa}(\adj{\om})$

\mycase $\rzSnd{\myaa}$
Have $\spp{\rzSnd{\myaa}}(\om)
= \rzSnd{\spp{\myaa}}(\om)
= \rzSnd{\spp{\myaa}(\om)}
= \rzSnd{\myaa(\adj{\om})}
= \rzSnd{\myaa}(\adj{\om})$

\mycase $\rzCons{\myaa}{\ab}$
We have that $\spp{\rzCons{\myaa}{\ab}}(\om) 
=  \rzCons{\spp{\myaa}}{\spp{\ab}}(\om)
= (\spp{\myaa}(\om),\spp{\ab}(\om))
= (\myaa(\adj{\om}),\ab(\adj{\om}))
=  \rzCons{\myaa}{\ab}(\adj{\om})$

\mycase $(\rzBLam{x}{\myaa})$
We have $\spp{\rzBLam{x}{\myaa}}(\om)
=(\rzBLam{x}{\spp{\myaa}})(\om)
= \tsub{\spp{\myaa}}{x}{\om}
= \tsub{\myaa}{x}{\adj{\om}}
= (\rzBLam{x}{\myaa})(\adj{\om})$

\mycase $(\rzFOLam{x}{\allrat}{\myaa})$
We reason by equality: $\spp{\rzFOLam{x}{\allrat}{\myaa}}(\om)
=(\rzFOLam{x}{\allrat}{\spp{\myaa}})(\om)
=(\rzFOLam{x}{\allrat}{\tsub{\spp{\myaa}}{y}{\om}})
=(\rzFOLam{x}{\allrat}{\tsub{\myaa}{y}{\adj{\om}}}
=(\rzFOLam{x}{\allrat}{\myaa})(\adj{\om})$

\mycase $\rzHOLam{x}{\phi}{\myaa}$
We have $\spp{\rzHOLam{x}{\phi}{\myaa}}(\om)
=(\rzHOLam{x}{\phi}{\spp{\myaa}})(\om)
=(\rzHOLam{x}{\phi}{\tsub{\spp{\myaa}}{y}{\om}})
=(\rzHOLam{x}{\phi}{\tsub{\myaa}{y}{\adj{\om}}}
=(\rzHOLam{x}{\phi}{\myaa})(\adj{\om})$

For terms, we give an inductive proof handling the polynomial cases from the paper.
The claim for terms will hold for any reasonable definition of computable functions.

\mycase $q$
Have $\tint{\spp{q}}{\om} = q = \tint{q}{\adj{\om}}$

\mycase $x$
Have $\tint{\spp{x}}{\om} = \om(\spp{x}) = \om(x) = \adj{\om}(x) = \tint{x}{\adj{\om}}$
in the case $x \notin \boundvars{\sigma}$
or
$\tint{\spp{x}}{\om} = \tint{f_i}{\om} =  \adj{\om}(x) = \tint{x}{\adj{\om}}$
if $x \in \boundvars{\sigma}$.

\mycase$f + g$
Have $\tint{\spp{f + g}}{\om} 
= \tint{\spp{f} + \spp{g}}{\om} 
= \tint{\spp{f}}{\om} + \tint{\spp{g}}{\om}
= \tint{f}{\adj{\om}} + \tint{g}{\adj{\om}}
= \tint{f + g}{\adj{\om}}$

\mycase$f \cdot g$
Have $\tint{\spp{f \cdot g}}{\om} 
= \tint{\spp{f} \cdot \spp{g}}{\om} 
= \tint{\spp{f}}{\om} \cdot \tint{\spp{g}}{\om}
= \tint{f}{\adj{\om}} \cdot \tint{g}{\adj{\om}}
= \tint{f \cdot g}{\adj{\om}}$

We give the Angel cases.

\mycase $\humod{x}{f}$
Have $(\myaa,\om) \in \strategyforR[\spp{\humod{x}{f}}]{X}$
iff $(\myaa,\om) \in \strategyforR[\humod{x}{\spp{f}}]{X}$
iff $\om = \subst[\nu]{x}{\tint{\spp{f}}{\om}}$ for some $(\myaa,\nu) \in X$
iff $\om = \subst[\nu]{x}{\tint{\spp{f}}{\om}}$ for some $(\myaa,\adj{\nu}) \in \adj{X}$
iff $\om = \subst[\nu]{x}{\tint{f}{\adj{\om}}}$ for some $(\myaa,\adj{\nu}) \in \adj{X}$
iff $\adj{\om} = \subst[\adj{\nu}]{x}{\tint{f}{\adj{\om}}},$ for some $(\myaa,\nu) \in \adj{X}$
iff $\adj{\om} \in \strategyforR[\humod{x}{f}]{\adj{X}}$
by IH on $f$ and since by admissibility of $\sigma$ have $x \notin \boundvars{\sigma}$.

\mycase $\prandom{x}$
Have $(\ab,\om) \in \strategyforR[\spp{\prandom{x}}]{X}$
iff $(\ab,\om) \in \strategyforR[\prandom{x}]{X}$
iff $\om = \subst[\nu]{x}{\rzApp{\myaa}{\om}},$ for some $(\rzCons{\myaa}{\ab},\nu) \in X$
iff $\om = \subst[\nu]{x}{\rzApp{\myaa}{\om}},$ for some $(\rzCons{\myaa}{\ab},\adj{\nu}) \in \adj{X}$
iff $\om = \subst[\nu]{x}{\rzApp{\adj{\myaa}}{\adj{\om}}},$ for some $(\rzCons{\myaa}{\ab},\adj{\nu}) \in \adj{X}$
iff $\adj{\om} = \subst[\adj{\nu}]{x}{\rzApp{\adj{\myaa}}{\adj{\om}}},$ for some $(\rzCons{\myaa}{\ab},\adj{\nu}) \in \adj{X}$
iff $(\adj{\ab},\adj{\om}) = \subst[\adj{\nu}]{x}{\rzApp{\adj{\myaa}}{\adj{\om}}},$ for some $(\rzCons{\myaa}{\ab},\adj{\nu}) \in \adj{X}$
iff $(\adj{\ab},\adj{\om}) \in \strategyforR[\humod{x}{f}]{\adj{X}}$
by IH on $f$ and since by admissibility of $\sigma$ have $x \notin \boundvars{\sigma}$.

\mycase $\ptest{\phi}$
Have $(\ab,\om) \in \strategyforR[\spp{\ptest{\phi}}]{X}$
iff $(\myaa,\om) \in \fintR{\spp{\phi}}$ and $(\rzCons{\myaa}{\ab}, \om) \in X$
iff $(\adj{\myaa},\adj{\om}) \in \adj{\fintR{\phi}}$ and $(\adj{\rzCons{\myaa}{\ab}}, \adj{\om}) \in \adj{X}$
iff $(\adj{\myaa},\adj{\om}) \in \strategyforR[\ptest{\phi}]{\adj{X}}$

\mycase $\alpha;\beta$
Have
\begin{align*}
&\phantom{\text{iff}}\, (\myaa,\om) \in \strategyforR[\spp{\alpha;\beta}]{X}\\
&\text{iff}\,(\myaa,\om) \in \strategyforR[\spp{\beta}]{\strategyforR[\spp{\alpha}]{X}}\\
&\text{iff}\,(\adj{\myaa},\adj{\om}) \in \strategyforR[\beta]{\{(\adj{\myaa},\adj{\nu})~|~ (\myaa,\nu) \in \strategyforR[\spp{\alpha}]{X}\}}\\
&\text{iff}\,(\adj{\myaa},\adj{\om}) \in \strategyforR[\beta]{\strategyforR[\alpha]{\adj{X}}}\\
&\text{iff}\,(\adj{\myaa},\adj{\om}) \in \strategyforR[\alpha;\beta]{\adj{X}}
\end{align*}
using the IH's on $\alpha$ and $\beta$.

\mycase$\alpha\cup\beta$
Have
\begin{align*}
&\phantom{\text{iff }}\ (\myaa,\om) \in \strategyforR[\spp{\alpha\cup\beta}]{X}\\
&\text{iff } (\myaa,\om) \in \strategyforR[\spp{\alpha}]{\apL{X}} \cup \strategyforR[\spp{\beta}]{\apR{X}}\\
&\text{iff } (\adj{\myaa},\adj{\om}) \in \strategyforR[\alpha]{(\adj{\apL{X}})} \cup \strategyforR[\beta]{(\adj{\apR{X}})}\\
&\text{iff } (\adj{\myaa},\adj{\om}) \in \strategyforR[\alpha]{(\apL{\adj{X}})} \cup \strategyforR[\beta]{(\apR{\adj{X}})}\\
&\text{iff } (\adj{\myaa},\adj{\om}) \in \strategyforR[\alpha\cup\beta]{\adj{X}}
\end{align*}
where the main step uses the IH's on $\alpha,\beta,$ and $\rzFst{\myaa}$.

\mycase$\prepeat{\alpha}$
Have
\begin{align*}
&\phantom{\text{iff}}\om \in \strategyforR[\spp{\prepeat{\alpha}}]{X}\\
&\text{iff}\,(\myaa,\om) \in \bigcap{\{\apL{Z}~|~X \cup \strategyforR[\spp{\alpha}]{\apR{Z}} \subseteq Z\}}\\
&\text{iff}\,(\adj{\myaa},\adj{\om}) \in \bigcap{\{\adj{\apL{Z}}~|~\adj{X} \cup \adj{\strategyforR[\spp{\alpha}]{\apR{Z}}} \subseteq \adj{Z}\}}\\
&\text{iff}\,(\adj{\myaa},\adj{\om}) \in \bigcap{\{\adj{\apL{Z}}~|~\adj{X} \cup \adj{\strategyforR[\spp{\alpha}]{\apR{Z}}} \subseteq \adj{Z}\}}\\
&\text{iff}\,(\adj{\myaa},\adj{\om}) \in \bigcap{\{\adj{\apL{Z}}~|~\adj{X} \cup \strategyforR[\alpha]{\adj{\apR{Z}}} \subseteq \adj{Z}\}}\\
&\text{iff}\,(\adj{\myaa},\adj{\om}) \in \bigcap{\{\apL{\adj{Z}}~|~\adj{X} \cup \strategyforR[\alpha]{\apR{\adj{Z}}} \subseteq \adj{Z}\}}\\
&\text{iff}\,(\adj{\myaa},\adj{\om}) \in \bigcap{\{\apL{Z}~|~\adj{X} \cup \strategyforR[\alpha]{\apR{Z}} \subseteq Z\}}\\
&\text{iff}\,(\adj{\myaa},\adj{\om}) \in \strategyforR[\prepeat{\alpha}]{\adj{X}}
\end{align*}

\mycase$\pdual{\alpha}$
Have $(\myaa,\om) \in \strategyforR[\spp{\pdual{\alpha}}]{X}$
iff $(\myaa,\om) \in \dstrategyforR[\spp{\alpha}]{X}$
iff $(\adj{\myaa},\adj{\om}) \in \dstrategyforR[\alpha]{\adj{X}}$
iff $(\adj{\myaa},\adj{\om}) \in \strategyforR[\pdual{\alpha}]{\adj{X}}$

We give the Demon cases.

\mycase $\humod{x}{f}$
Have $(\myaa,\om) \in \dstrategyforR[\spp{\humod{x}{f}}]{X}$
iff $(\myaa,\om) \in \dstrategyforR[\humod{x}{\spp{f}}]{X}$
iff $\om = \subst[\nu]{x}{\tint{\spp{f}}{\nu}},$ for some $(\myaa,\nu)\in X$
iff $\adj{\om} = \adj{\subst[\nu]{x}{\tint{\spp{f}}{\nu}}},$ for some $(\myaa,\nu)\in \adj{X}$
iff $\adj{\om} = \adj{\subst[\nu]{x}{\tint{f}{\adj{\nu}}}},$ for some $(\myaa,\nu)\in \adj{X}$
iff $\adj{\om} = \subst[\adj{\nu}]{x}{\tint{f}{\adj{\nu}}},$ for some $(\myaa,\nu) \in \adj{X}$
iff $(\adj{\myaa},\adj{\om}) \in \dstrategyforR[\humod{x}{f}]{\adj{X}}$
by IH on $f$ and since by admissibility of $\sigma$ have $x \notin \boundvars{\sigma}$.

\mycase $\prandom{x}$
Have $(\myaa,\om) \in \dstrategyforR[\spp{\prandom{x}}]{X}$
iff $(\myaa,\om) \in \dstrategyforR[\prandom{x}]{X}$
iff $\om = \subst[\nu]{x}{v},$ for some $(\myaa,\nu)\in X, v \in \allrat$
iff $\adj{\om} = \adj{\subst[\nu]{x}{v}},$ for some  $(\myaa,\nu)\in \adj{X}, v \in \allrat$
iff $\adj{\om} = \subst[\adj{\nu}]{x}{v},$ for some $(\myaa,\nu) \in \adj{X}$
iff $(\adj{\myaa},\adj{\om}) \in \dstrategyforR[\prandom{x}]{\adj{X}}$
by IH on $f$ and since by admissibility of $\sigma$ have $x \notin \boundvars{\sigma}$.

\mycase $\ptest{\phi}$
Have $(\rzApp{\myaa}{\ab},\om) \in \dstrategyforR[\spp{\ptest{\phi}}]{X}$
iff $(\myaa,\om) \in X$ for all  $(\ab,\om) \in \fintR{\spp{\phi}}$	
iff $(\adj{\myaa},\adj{\om}) \in \adj{X}$ for all $(\adj{\ab},\adj{\om}) \in \fintR{\phi}$
iff $(\adj{\myaa},\adj{\om}) \in \dstrategyforR[\ptest{\phi}]{\adj{X}}$

\mycase $\alpha;\beta$
\begin{align*}
&\phantom{\text{iff }} (\ab,\om) \in \dstrategyforR[\spp{\alpha;\beta}]{X}\\
&\text{iff } (\ab,\om) \in \dstrategyforR[\spp{\beta}]{\dstrategyforR[\spp{\alpha}]{\spp{\rzSnd{\myaa}}}{X}}\\
&\text{iff }  (\adj{\ab},\adj{\om}) \in \dstrategyforR[\beta]{\{(\adj{\myaa},\adj{\nu})~|~ (\myaa,\nu) \in \dstrategyforR[\spp{\alpha}]{X}\}}\\
&\text{iff } (\adj{\ab},\adj{\om}) \in \dstrategyforR[\beta]{\{(\adj{\myaa},\adj{\nu})~|~ (\adj{\myaa},\adj{\nu}) \in \dstrategyforR[\alpha]{\adj{X}}\}}\\
&\text{iff } (\adj{\ab},\adj{\om}) \in \dstrategyforR[\beta]{\dstrategyforR[\alpha]{\adj{X}}}\\
&\text{iff } (\adj{\ab},\adj{\om}) \in \dstrategyforR[\alpha;\beta]{\adj{X}}
\end{align*}
using the IH's on $\alpha$ and $\beta$.

\mycase$\alpha\cup\beta$
Have 
\begin{align*}
&\text{iff }(\myaa,\om) \in \dstrategyforR[\spp{\alpha\cup\beta}]{X}\\
&\text{iff }(\myaa,\om) \in \dstrategyforR[\spp{\alpha}]{\dpL{X}} \cup \dstrategyforR[\spp{\beta}]{\dpR{X}}\\
&\text{iff }(\adj{\myaa},\adj{\om}) \in \dstrategyforR[\alpha]{(\adj{\dpL{X}})} \cup \dstrategyforR[\beta]{(\adj{\dpR{X}})}\\
&\text{iff }(\adj{\myaa},\adj{\om}) \in \dstrategyforR[\alpha]{(\dpL{\adj{X}})} \cup \dstrategyforR[\beta]{(\dpR{\adj{X}})}\\
&\text{iff }(\adj{\myaa},\adj{\om}) \in \dstrategyforR[\alpha\cup\beta]{(\adj{X})}
\end{align*}
where the main step uses the IH's on $\alpha,\beta.$

\mycase$\prepeat{\alpha}$
Have
\begin{align*}
    &\phantom{\text{iff }} (\myaa,\om) \in \bigcap{\{\dpL{Z}~|~X \cup \dstrategyforR[\spp{\alpha}]{\dpR{Z}} \subseteq Z\}}\\
&\text{iff } (\adj{\myaa},\adj{\om}) \in \bigcap{\{\adj{\dpL{Z}}~|~\adj{X} \cup \adj{\dstrategyforR[\spp{\alpha}]{\dpR{Z}}} \subseteq \adj{Z}\}}\\
&\text{iff } (\adj{\myaa},\adj{\om}) \in \bigcap{\{\adj{\dpL{Z}}~|~\adj{X} \cup \adj{\dstrategyforR[\spp{\alpha}]{\dpR{Z}}} \subseteq \adj{Z}\}}\\
&\text{iff } (\adj{\myaa},\adj{\om}) \in \bigcap{\{\adj{\dpL{Z}}~|~\adj{X} \cup \dstrategyforR[\alpha]{\adj{\dpR{Z}}} \subseteq \adj{Z}\}}\\
&\text{iff } (\adj{\myaa},\adj{\om}) \in \bigcap{\{\dpL{\adj{Z}}~|~\adj{X} \cup \dstrategyforR[\alpha]{\dpR{\adj{Z}}} \subseteq \adj{Z}\}}\\
&\text{iff } (\adj{\myaa},\adj{\om}) \in \bigcap{\{\dpL{Z}~|~\adj{X} \cup \dstrategyforR[\alpha]{\dpR{Z}} \subseteq Z\}}\\
&\text{iff } (\adj{\myaa},\adj{\om}) \in \dstrategyforR[\prepeat{\alpha}]{\adj{X}}
\end{align*}

\mycase$\pdual{\alpha}$
Have $(\myaa,\om) \in \dstrategyforR[\spp{\pdual{\alpha}}]{X}$
iff $(\myaa,\om) \in \strategyforR[\spp{\alpha}]{X}$
iff $(\adj{\myaa},\adj{\om}) \in \strategyforR[\alpha]{\adj{X}}$
iff $(\adj{\myaa},\adj{\om}) \in \dstrategyforR[\pdual{\alpha}]{\adj{X}}$

We give the formula cases.

\mycase $f > g$
Have $(\rzNil,\om) \in \fintR{\spp{f > g}}$
iff  $\tint{\spp{f}}{\om} < \tint{\spp{g}}{\om}$
iff  $\tint{f}{\adj{\om}} < \tint{g}{\adj{\om}}$
iff  $(\adj{\rzNil}, \adj{\om}) \in \fintR{f < g}$.
The cases for $\leq, <, =, \neq, \geq$ are symmetric.

\mycase $\ddiamond{\alpha}{\phi}$
Have
\begin{align*}
&\phantom{\text{iff }}(\myaa,\om) \in \fintR{\spp{\ddiamond{\alpha}{\phi}}}\\
&\text{iff }\strategyforR[\spp{\alpha}]{\{(\myaa,\om)\}} \subseteq \{(\myaa,\om)~|~(\adj{\myaa},\adj{\om}) \in \fintR{\phi} \cup \{\stt\}\}\\
&\text{iff }\adj{\strategyforR[\alpha]{\{(\adj{\myaa},\adj{\om})\}}} \subseteq \{(\myaa,\om)~|~(\adj{\myaa},\adj{\om}) \in \fintR{\phi} \cup \{\stt\}\}\\
&\text{iff }\strategyforR[\alpha]{\{(\adj{\myaa},\adj{\om})\}} \subseteq \fintR{\phi} \cup \{\stt\}\\
&\text{iff }(\adj{\myaa},\adj{\om}) \in \fintR{\ddiamond{\alpha}{\phi}}
\end{align*}

\mycase $\dbox{\alpha}{\phi}$
Have
\begin{align*}
&\phantom{\text{iff }}(\myaa,\om) \in \fintR{\spp{\dbox{\alpha}{\phi}}}\\
&\text{iff }\dstrategyforR[\spp{\alpha}]{\{(\myaa,\om)\}} \subseteq \{(\myaa,\om)~|~(\adj{\myaa},\adj{\om}) \in \fintR{\phi}  \cup \{\stt\}\}\\
&\text{iff }\adj{\dstrategyforR[\alpha]{\{(\adj{\myaa},\adj{\om})\}}} \subseteq \{(\myaa,\om)~|~(\adj{\myaa},\adj{\om}) \in \fintR{\phi}  \cup \{\stt\}\}\\
&\text{iff }\adj{\dstrategyforR[\alpha]{\{(\adj{\myaa},\adj{\om})\}}} \subseteq \{(\myaa,\om)~|~(\adj{\myaa},\adj{\om}) \in \fintR{\phi}  \cup \{\stt\}\}\\
&\text{iff }\dstrategyforR[\alpha]{\{(\adj{\myaa},\adj{\om})\}} \subseteq \fintR{\phi}  \cup \{\stt\}\\
&\text{iff }(\adj{\myaa},\adj{\om}) \in \fintR{\dbox{\alpha}{\phi}}
\end{align*}
\end{proof}

\begin{lemma}[Coincidence]
If $\om = \tom$ and $\myaa = \taa$ on $ V \supseteq \freevars{e}$ (where $e$ is $f$ or $\phi$ or $\alpha$ or $\myaa$):
\begin{itemize}
\item $\myaa(\om) = \myaa(\tom)$
\item $\tint{f}{\om} = \tint{f}{\tom}$
\item $\strategyforR[\alpha]{\{(\myaa,\om)\}}$ =  $\strategyforR[\alpha]{\{(\taa,\tom)\}}$ on $\mustboundvars{\alpha} \cup V$ 
\item $\dstrategyforR[\alpha]{\{(\myaa,\om)\}}$ =  $\dstrategyforR[\alpha]{\{(\taa,\tom)\}}$ on $\mustboundvars{\alpha} \cup V$ 
\item $\om \in \fint{\phi}{}$ iff $\tom \in \fint{\phi}{}$
\end{itemize}
\label{lem:app-coincide}
\end{lemma}
\begin{proof}
First claim is by induction on $\myaa$.
Second claim holds for any reasonable representation of computable functions, we give an example inductive proof for polynomial cases.
Remaining claims are by induction simultaneously on $\alpha, \phi$.
In each case assumption (1) is the assumption $\om = \tom$ on $\freevars{e}$

We give the cases for realizers.

\mycase $\rzNil$
Have $\rzNil(\om) = \text{unit} = \rzNil(\tom)$

\mycase $x$  where $x$ is a variable over realizers.
We have $x(\om) = x(\tom)$ since $\om$ and $\tom$ agree on the free variables of $x$.
That is, they agree on all variables so $\om = \tom$.

\mycase $f$
Have $f(\om) = f(\tom)$ by IH on terms.

\mycase $\rzFst{\myaa}$
Have $\rzFst{\myaa}(\om) = \rzFst{\myaa(\om)} = \rzFst{\taa(\tom)} = \rzFst{\taa}(\tom)$ by IH.

\mycase $\rzSnd{\myaa}$
Have $\rzSnd{\myaa}(\om) = \rzSnd{\myaa(\om)} = \rzSnd{\taa(\tom)} = \rzSnd{\taa}(\tom)$ by IH.

\mycase $\rzCons{\myaa}{\ab}$
Have $\rzCons{\myaa}{\ab}(\om)
= (\myaa(\om),\ab(\om))
= (\myaa(\tom),\ab(\tom))
= (\taa(\tom),\tab(\tom))
= \rzCons{\taa}{\tab}(\tom)$ by IHs.

\mycase $(\rzBLam{x}{\myaa})$
Have 
$(\rzBLam{x}{\myaa})(\om)
= \tsub{\myaa}{x}{\om}
= \tsub{\taa}{x}{\tom}
= (\rzBLam{x}{\taa}(\tom))
$ by (1),
$\om = \tom$ on $\freevars{\rzBLam{x}{\myaa}} = \freevars{\myaa}$. 

\mycase $(\rzFOLam{x}{\allrat}{\myaa})$
Have $(\rzFOLam{x}{\allrat}{\myaa})(v)
= \tsub{\myaa}{x}{v}
= \tsub{\taa}{x}{v}
= (\rzFOLam{x}{\allrat}{\taa})(v)$

\mycase $(\rzHOLam{x}{\phi}{\myaa})$
Have $(\rzHOLam{x}{\phi}{\myaa})(v)
= \tsub{\myaa}{x}{\ab}
= \tsub{\taa}{x}{\tab}
= (\rzHOLam{x}{\phi}{\taa})(v)$

We give (example) term cases.

\mycase $q$
Have $\tint{q}{\om} = q = \tint{q}{\tom} $

\mycase $x$
Have $\tint{x}{\om} = \om(x) = \tom(x) = \tint{x}{\tom}$ by (1).

\mycase$f + g$
Have $\tint{f + g }{\om} = \tint{f}{\om} + \tint{g}{\om}  = \tint{f}{\tom} + \tint{g}{\tom} = \tint{f + g}{\tom}$ by the IHs.

\mycase$f \cdot g$
Have $\tint{f \cdot g }{\om} = \tint{f}{\om} \cdot \tint{g}{\om}  = \tint{f}{\tom} \cdot \tint{g}{\tom} = \tint{f \cdot g}{\tom}$ by the IHs.

We give the Angel cases.

\mycase $\humod{x}{f}$
    $\strategyforR[\humod{x}{f}]{\{(\myaa,\om)\}}$
= $\{(\myaa,\subst[\om]{x}{\tint{f}{\om}}\}$
= $\{(\myaa,\subst[\om]{x}{\tint{f}{\om}}\}$
= $\{(\myaa,\subst[\om]{x}{\tint{f}{\tom}}\}$
which equals = $\{(\taa,\subst[\tom]{x}{\tint{f}{\tom}}\}$ on $V \cup \mustboundvars{\alpha} = V \cup \{x\}$.

\mycase $\prandom{x}$
Have   $\strategyforR[\prandom{x}]{\{(\myaa,\om)\}}$
= $\{(\rzSnd{\myaa},\subst[\om]{x}{\rzApp{\rzFst{\myaa}}{\om}}\}$
= $\{(\rzSnd{\myaa},\subst[\om]{x}{\rzApp{\rzFst{\myaa}}{\om}}\}$
Then $\rzApp{\rzFst{\myaa}}{\om} = \rzApp{\rzFst{\taa}}{\tom}$ by assumption so
$\{(\rzSnd{\myaa},\subst[\om]{x}{\rzApp{\rzFst{\myaa}}{\om}}\}$
= $\{(\rzSnd{\myaa},\subst[\om]{x}{\rzApp{\rzFst{\myaa}}{\om}}\}$ on $V \cup \mustboundvars{\alpha} = V \cup \{x\}$.

\mycase $\ptest{\phi}$
Have $\strategyforR[\ptest{\phi}]{\{(\myaa,\om)\}}$ 
= $\{(\rzSnd{\myaa},\om)\}$ if $(\rzFst{\myaa},\om) \in \fintR{\phi}$ else $\emptyset$.
= $\{(\rzSnd{\myaa},\om)\}$ if $(\rzFst{\myaa},\om) \in \fintR{\phi}$ else $\emptyset$.
= (IH) $\{(\rzSnd{\myaa},\om)\}$ if $(\rzFst{\taa},\tom) \in \fintR{\phi}$ else $\emptyset$.
= (1) $\{(\rzSnd{\taa},\tom)\}$ if $(\rzFst{\taa},\tom) \in \fintR{\phi}$ else $\emptyset$.
=$\strategyforR[\ptest{\phi}]{\{(\taa,\tom)\}}$

\mycase $\alpha;\beta$
Have $\strategyforR[\alpha;\beta]{\{(\myaa,\om)\}}$ 
= $\strategyforR[\beta]{\strategyforR[\alpha]{\{(\myaa,\om)\}}}$ 
Then by IH on $\alpha$ have $\strategyforR[\alpha]{\{(\myaa,\om)\}} = \strategyforR[\alpha]{\{(\taa,\tom)\}}$ on $\mustboundvars{\alpha} \cup V$.
Then by \rref{lem:app-partition} and IH on $\beta$ have $\strategyforR[\beta]{\strategyforR[\alpha]{\{(\myaa,\om)\}}} = \strategyforR[\beta]{\strategyforR[\alpha]{\{(\taa,\tom)\}}}$ on $\mustboundvars{\alpha;\beta} = \mustboundvars{\alpha} \cup \mustboundvars{\beta} \cup V$.

\mycase$\alpha\cup\beta$
Have $\strategyforR[\alpha\cup\beta]{\{(\myaa,\om)\}}$ 
= $\strategyforR[\alpha]{\apL{\{(\myaa,\om)\}}} \cup \strategyforR[\beta]{\apR{\{(\myaa,\om)\}}}$ 
In each case, one projection is singleton and the other empty.
In each case, $\strategyforR[\alpha\cup\beta]{\{(\myaa,\om)\}} = \{(\myaa,\om)\}$  on $\mustboundvars{\alpha} \cup V$ or $\mustboundvars{\beta} \cup V$, then since $\mustboundvars{\alpha\cup\beta} = \mustboundvars{\alpha}\cap\mustboundvars{\beta},$ in either case they agree on $\mustboundvars{\alpha\cup\beta}\cup V$.

\mycase$\prepeat{\alpha}$
Let $C(\myaa,\om) = \bigcap\{\apL{Z}~|~\{(\myaa,\om)\} \cup \strategyforR[\alpha]{\apR{Z}} \subseteq Z\}$.
$(\myaa,\om) \in \strategyforR[\prepeat{\alpha}]{\{(\myaa,\om)\}}$ 
iff $(\myaa,\om) \in \bigcap\{\apL{Z}~|~ \{(\myaa,\om)\} \cup \strategyforR[\alpha]{\apR{Z}} \subseteq Z \}$.
Consider any such $Z$:
$\apL{\{(\myaa,\om)\}} \cup \strategyforR[\alpha]{\apR{Z}}$ 
= $\apL{\{(\taa,\tom)\}} \cup \strategyforR[\alpha]{\apR{Z}},$ on $\mustboundvars{\prepeat{\alpha}} \cup V = V,$ the first conjunct by (1) and the second by the IH.
Then by definition of least fixed point, $C$ is one such $Z$, satisfying the set inclusion with exact equality, so
$C(\myaa,\om) $ iff $C(\taa,\tom),$ i.e.,
$\strategyforR[\prepeat{\alpha}]{\{(\myaa,\om)\}}$ 
= $\strategyforR[\prepeat{\alpha}]{\{(\taa,\tom)\}}$ 

\mycase$\pdual{\alpha}$
Have    $\strategyforR[\pdual{\alpha}]{\{(\myaa,\om)\}}$ 
= $\dstrategyforR[\alpha]{\{(\myaa,\om)\}}$ 
= (IH) $\dstrategyforR[\alpha]{\{(\taa,\tom)\}}$ 
on $V \cup \mustboundvars{\alpha} = V \cup \mustboundvars{\pdual{\alpha}}$.
which equals
  $\strategyforR[\pdual{\alpha}]{\{(\taa,\tom)\}}$ 
yieliding the case.

We give the Demon cases.

\mycase $\humod{x}{f}$
Have $\dstrategyforR[\humod{x}{f}]{\{(\myaa,\om)\}}$
= $\{(\myaa,\subst[\om]{x}{\tint{f}{\om}}\}$
= $\{(\myaa,\subst[\om]{x}{\tint{f}{\om}}\}$
= $\{(\myaa,\subst[\om]{x}{\tint{f}{\tom}}\}$
which equals = $\{(\taa,\subst[\tom]{x}{\tint{f}{\tom}}\}$ on $V \cup \mustboundvars{\alpha} = V \cup \{x\}$.

\mycase $\prandom{x}$
Have    $\dstrategyforR[\prandom{x}]{\{(\myaa,\om)\}}$
= $\{(\rzApp{\myaa}{v},\subst[\om]{x}{v})~|~v \in \allrat\}$
Then $\rzApp{\myaa}{v} = \rzApp{\taa}{v}$ by assumption so
$\{(\rzApp{\myaa}{v},\subst[\om]{x}{v})\}$
= $\{(\rzApp{\myaa}{v},\subst[\om]{x}{v})\}$ on $V \cup \mustboundvars{\alpha} = V \cup \{x\}$.

\mycase $\ptest{\phi}$
Have $\dstrategyforR[\ptest{\phi}]{\{(\myaa,\om)\}}$
= $\{(\rzApp{\myaa}{\ab},\om)~| (\ab,\om) \in \fintR{\phi}\}$
likewise
$ \dstrategyforR[\ptest{\phi}]{\{(\taa,\tom)\}}$
= $\{(\rzApp{\taa}{\tab},\tom)~| (\tab,\tom) \in \fintR{\phi}\}$
then by (IH) and (1)  $(\tab,\tom) \in \fintR{\phi}$ iff $(\ab,\om) \in \fintR{\phi}$
and $\myaa$ = $\taa$ and $\om = \tom$ on $V \cup \mustboundvars{\ptest{\phi}} = V$.

\mycase $\alpha;\beta$
Have $\dstrategyforR[\alpha;\beta]{\{(\myaa,\om)\}}$
= $\dstrategyforR[\beta]{\dstrategyforR[\alpha]{\{(\myaa,\om)\}}}$
Then by IH on $\alpha$ have $\dstrategyforR[\alpha]{\{(\myaa,\om)\}} = \dstrategyforR[\alpha]{\{(\taa,\tom)\}}$ on $\mustboundvars{\alpha} \cup V$.
Then by \rref{lem:app-partition} and IH on $\beta$ have $\dstrategyforR[\beta]{\dstrategyforR[\alpha]{\{(\myaa,\om)\}}} = \dstrategyforR[\beta]{\dstrategyforR[\alpha]{\{(\taa,\tom)\}}}$ on $\mustboundvars{\alpha;\beta} = \mustboundvars{\alpha} \cup \mustboundvars{\beta} \cup V$.

\mycase$\alpha\cup\beta$
Have $\dstrategyforR[\alpha\cup\beta]{\{(\myaa,\om)\}}$
= $\dstrategyforR[\alpha]{\dpL{\{(\myaa,\om)\}}} \cup \dstrategyforR[\beta]{\dpR{\{(\myaa,\om)\}}}$
In each case, one projection is singleton and the other empty.
In each case, $\dstrategyforR[\alpha\cup\beta]{\{(\myaa,\om)\}} = \{(\myaa,\om)\}$  on $\mustboundvars{\alpha} \cup V$ or $\mustboundvars{\beta} \cup V$, then since $\mustboundvars{\alpha\cup\beta} = \mustboundvars{\alpha}\cap\mustboundvars{\beta},$ in either case they agree on $\mustboundvars{\alpha\cup\beta}\cup V$.

\mycase$\prepeat{\alpha}$
Let $C(\myaa,\om) = \bigcap\{\dpL{Z}~|~\{(\myaa,\om)\} \cup \dstrategyforR[\alpha]{\dpR{Z}} \subseteq Z\}$.
$(\myaa,\om) \in \dstrategyforR[\prepeat{\alpha}]{\{(\myaa,\om)\}}$
iff $(\myaa,\om) \in \bigcap\{\dpL{Z}~|~ \{(\myaa,\om)\} \cup \dstrategyforR[\alpha]{\dpR{Z}} \subseteq Z \}$.
Consider any such $Z$:
$\dpL{\{(\myaa,\om)\}} \cup \dstrategyforR[\alpha]{\dpR{Z}}$
= $\dpL{\{(\taa,\tom)\}} \cup \dstrategyforR[\alpha]{\dpR{Z}},$ on $\mustboundvars{\prepeat{\alpha}} \cup V = V,$ the first conjunct by (1) and the second by the IH.
Then by definition of least fixed point, $C$ is one such $Z$, satisfying the set inclusion with exact equality, so
$C(\myaa,\om) $ iff $C(\taa,\tom),$ i.e.,
$\dstrategyforR[\prepeat{\alpha}]{\{(\myaa,\om)\}}$ 
= $\dstrategyforR[\prepeat{\alpha}]{\{(\taa,\tom)\}}$

\mycase$\pdual{\alpha}$
Have $\dstrategyforR[\pdual{\alpha}]{\{(\myaa,\om)\}}$
= $\strategyforR[\alpha]{\{(\myaa,\om)\}}$
= (IH) $\strategyforR[\alpha]{\{(\taa,\tom)\}}$
on $V \cup \mustboundvars{\alpha} = V \cup \mustboundvars{\pdual{\alpha}}$.
which equals
  $\dstrategyforR[\pdual{\alpha}]{\{(\taa,\tom)\}}$ 
yieliding the case.

We give the formula cases

\mycase $f > g$
Have $(\rzNil,\om) \in \fintR{f > g}$
iff $\tint{f}{\om} > \tint{g}{\om}$
iff $\tint{f}{\tom} > \tint{g}{\tom}$
iff $(\rzNil,\tom) \in \fintR{f > g}$.
The cases for $\leq, <, =, \neq, \geq$ are symmetric.

\mycase $\ddiamond{\alpha}{\phi}$
Have  $(\myaa,\om) \in \fintR{\ddiamond{\alpha}{\phi}}$
iff $\strategyforR[\alpha]{\{(\myaa,\om)\}} \subseteq \fint{\phi}  \cup \{\stt\}$
iff$^*$ $\strategyforR[\alpha]{\{(\taa,\tom)\}} \subseteq \fint{\phi}  \cup \{\stt\}$
iff $(\taa,\tom) \in \fintR{\ddiamond{\alpha}{\phi}}$
where the starred step holds by the IH because $\myaa = \taa$ and $\om = \tom$ on $\freevars{\ddiamond{\alpha}{\phi}} = \freevars{\alpha} \cup (\freevars{\phi} - \mustboundvars{\alpha}),$ thus $ \strategyforR[\alpha]{\{(\myaa,\om)\}} = \strategyforR[\alpha]{\{(\taa,\tom)\}}$ on $\mustboundvars{\alpha} \cup V \supseteq \freevars{\phi}$ giving the equivalence by the IH on $\phi$.

\mycase $\dbox{\alpha}{\phi}$
Have $(\myaa,\om) \in \fintR{\dbox{\alpha}{\phi}}$
iff $\dstrategyforR[\alpha]{\{(\myaa,\om)\}} \subseteq \fintR{\phi}  \cup \{\stt\}$
iff$^*$ $\dstrategyforR[\alpha]{\{(\taa,\tom)\}} \subseteq \fintR{\phi}  \cup \{\stt\}$
iff  $(\taa,\tom) \in \fintR{\dbox{\alpha}{\phi}}$
where the starred step holds by the IH because $\myaa = \taa$ and $\om = \tom$ on $\freevars{\ddiamond{\alpha}{\phi}} = \freevars{\alpha} \cup (\freevars{\phi} - \mustboundvars{\alpha}),$ thus $ \strategyforR[\alpha]{\{(\myaa,\om)\}} = \strategyforR[\alpha]{\{(\taa,\tom)\}}$ on $\mustboundvars{\alpha} \cup V \supseteq \freevars{\phi}$ giving the equivalence by the IH on $\phi$.
\end{proof}

\begin{lemma}[Partition]
For all suitable regions $X$ and all games $\alpha,$
\begin{itemize}
\item $\strategyforR[\alpha]{X} = \bigcup_{(\myaa,\om) \in X} \strategyforR[\alpha]{\{(\myaa,\om)\}}$
\item $\dstrategyforR[\alpha]{X} = \bigcup_{(\myaa,\om) \in X} \dstrategyforR[\alpha]{\{(\myaa,\om)\}}$
\end{itemize}
\label{lem:app-partition}
\end{lemma}
\begin{proof}
Straightforward induction on $\alpha$.
\end{proof}

\begin{lemma}[Bound Effect]
If $X = \{(\myaa,\om)\}$ and $(\ab,\nu) \in \dstrategyforR[\alpha]{X}$ or $(\ab,\nu) \in \strategyforR[\alpha]{X}$ then $\om = \nu$ on $\boundvars{\alpha}^\complement,$ the complement of set $\boundvars{\alpha}$.
\end{lemma}
\begin{proof}
Angel Cases:

\mycase $\humod{x}{f}$
Have $(\ab,\nu) \in \strategyforR[\humod{x}{f}]{X}$ 
iff $\ab = \myaa$ and $\nu  =\subst[\om]{x}{\tint{f}{\om}}$
and $\boundvars{\humod{x}{f}} = \{x\}$ and $\nu = \om$ on $\{x\}^\complement.$

\mycase $\prandom{x}$
Have $(\ab,\nu) \in \strategyforR[\prandom{x}]{X}$ 
iff $\ab = \rzSnd{\myaa}$ and $\nu  =\subst[\om]{x}{\rzApp{\rzFst{\myaa}}{\om}}$
and $\boundvars{\prepeat{x}} = \{x\}$ and $\nu = \om$ on $\{x\}^\complement.$

\mycase $\ptest{\phi}$
Have $(\ab,\nu) \in \strategyforR[\ptest{\phi}]{X}$ 
iff $\nu = \om$ and $\ab = \rzSnd{\myaa}$ and $(\rzFst{\myaa},\om) \in \fintR{\phi}$ and $(\myaa,\om) \in X$
then $\om = \om$ on $\emptyset^\complement$ and $\boundvars{\ptest{\phi}} = \emptyset$.

\mycase $\alpha;\beta$
Have $(\ab,\nu) \in \strategyforR[\alpha;\beta]{X}$ 
iff $(\ab,\nu) \in \strategyforR[\beta]{\strategyforR[\alpha]{X}}$ 
iff exists $(\ac,\mu) \in \strategyforR[\alpha]{X}$ s.t.\ $(\ab,\nu) \in \strategyforR[\beta]{\{(\ab,\nu)\}}$, then result holds by \rref{lem:app-partition} and IH twice since $\boundvars{\alpha;\beta} = \boundvars{\alpha} \cup \boundvars{\beta}$.

\mycase$\alpha\cup\beta$
Have $(\ab,\nu )\in \strategyforR[\alpha\cup\beta]{X}$
iff $(\ab,\nu) \in \strategyforR[\alpha]{\apL{X}} \cup  \strategyforR[\beta]{\apR{X}}$ .
In each case one branch is empty and the other singleton, so by IH on the singleton side, $\om = \nu$  on $\boundvars{\alpha}^\complement$ or $\boundvars{\beta}^\complement$ and thus in both cases on $\boundvars{\alpha\cup\beta}^\complement = \boundvars{\alpha}^\complement \cap \boundvars{\beta}^\complement$.

\mycase$\prepeat{\alpha}$
Let $C = \bigcap \{\apL{Z}~|~ X \cup \strategyforR[\alpha]{\apR{Z}} \subseteq Z\}$.
$(\ab,\nu) \in \strategyforR[\prepeat{\alpha}]{X}$
iff (LFP) $\apL{\{(\ab,\nu)\}} \subseteq X \cup \strategyforR[\alpha]{\apR{C}}$
then by IH $\nu = \om$ on $\boundvars{\alpha}^\complement = \boundvars{\prepeat{\alpha}}^\complement$.

\mycase$\pdual{\alpha}$
$(\ab,\nu) \in \strategyforR[\pdual{\alpha}]{X}$
iff $(\ab,\nu) \in \dstrategyforR[\alpha]{X}$
so by IH $\nu = \om$ on $\boundvars{\pdual{\alpha}}^\complement = \boundvars{\alpha}^\complement$.

We give the Demon cases.

\mycase $\humod{x}{f}$
Have $(\ab,\nu) \in \dstrategyforR[\humod{x}{f}]{X}$ 
iff $\ab = \myaa$ and $\nu  =\subst[\om]{x}{\tint{f}{\om}}$
and $\boundvars{\humod{x}{f}} = \{x\}$ and $\nu = \om$ on $\{x\}^\complement.$

\mycase $\prandom{x}$
Have $(\ab,\nu) \in \dstrategyforR[\prandom{x}]{X}$ 
iff $\ab = \rzApp{\myaa}{v}$ and $\nu  =\subst[\om]{x}{v}{\om},$ some $v \in \allrat$.
So $\boundvars{\prandom{x}} = \{x\}$ and $\nu = \om$ on $\{x\}^\complement.$

\mycase $\alpha;\beta$
Have $(\ab,\nu) \in \dstrategyforR[\alpha;\beta]{X}$ 
iff $(\ab,\nu) \in \dstrategyforR[\beta]{\strategyforR[\alpha]{X}}$ 
iff exists $(\ac,\mu) \in \dstrategyforR[\alpha]{X}$ s.t.\ $(\ab,\nu) \in \dstrategyforR[\beta]{\{(\ab,\nu)\}}$, then result holds by \rref{lem:app-partition} and IH twice since $\boundvars{\alpha;\beta} = \boundvars{\alpha} \cup \boundvars{\beta}$.

\mycase$\alpha\cup\beta$
Have $(\ab,\nu )\in \strategyforR[\alpha\cup\beta]{X}$ 
iff $(\ab,\nu) \in \strategyforR[\alpha]{\dpL{X}} \cup  \strategyforR[\beta]{\dpR{X}}$ .
In each case one branch is empty and the other singleton, so by IH on the singleton side, $\om = \nu$  on $\boundvars{\alpha}^\complement$ or $\boundvars{\beta}^\complement$ and thus in both cases on $\boundvars{\alpha\cup\beta}^\complement = \boundvars{\alpha}^\complement \cap \boundvars{\beta}^\complement$.

\mycase$\prepeat{\alpha}$
Let $C = \bigcap \{\dpL{Z}~|~ X \cup \dstrategyforR[\alpha]{\apR{Z}} \subseteq Z\}$.
$(\ab,\nu) \in \dstrategyforR[\prepeat{\alpha}]{X}$ 
iff (LFP) $\dpL{\{(\ab,\nu)\}} \subseteq X \cup \dstrategyforR[\alpha]{\apR{C}}$
then by IH $\nu = \om$ on $\boundvars{\alpha}^\complement = \boundvars{\prepeat{\alpha}}^\complement$.

\mycase$\pdual{\alpha}$
Have $(\ab,\nu) \in \dstrategyforR[\pdual{\alpha}]{X}$ 
iff $(\ab,\nu) \in \strategyforR[\alpha]{X}$
so by IH $\nu = \om$ on $\boundvars{\pdual{\alpha}}^\complement = \boundvars{\alpha}^\complement$.
\end{proof}

\begin{lemma}[Monotonicity]
If $X \subseteq Y$ then $\dstrategyforR[\alpha]{X} \subseteq \dstrategyforR[\alpha]{Y}$ and $\dstrategyforR[\alpha]{X} \subseteq \dstrategyforR[\alpha]{Y}$.
\label{lem:app-mono}
\end{lemma}
\begin{proof}
By induction on $\alpha$.
We begin with the Angel cases.

\mycase $\humod{x}{f}$
Have $\strategyforR[\humod{x}{f}]{X}
= \{(\myaa,\subst[\om]{x}{\tint{f}{\om}})~|~ (\myaa,\om) \in X\}
\subseteq \{(\myaa,\subst[\om]{x}{\tint{f}{\om}})~|~ (\myaa,\om)\in Y\}
= \strategyforR[\humod{x}{f}]{Y}$

\mycase $\prandom{x}$
Have $\strategyforR[\prandom{x}]{X}
= \{(\rzSnd{\myaa},\subst[\om]{x}{\rzApp{\rzFst{\myaa}}{\om}})~|~ (\myaa,\om) \in X\}
\subseteq \{(\rzSnd{\myaa},\subst[\om]{x}{\rzApp{\rzFst{\myaa}}{\om}})~|~ (\myaa,\om) \in Y\}
= \strategyforR[\prandom{x}]{Y}$

\mycase $\ptest{\phi}$
Have $\strategyforR[\ptest{\phi}]{X}
= \{(\rzSnd{\myaa},\om)~|~(\myaa,\om) \in X\text{ and } (\rzFst{\myaa},\om) \in \fintR{\phi}\}
\subseteq \{(\rzSnd{\myaa},\om)~|~(\myaa,\om) \in Y\text{ and } (\rzFst{\myaa},\om) \in \fintR{\phi}\}
= \strategyforR[\ptest{\phi}]{Y}$.

\mycase $\alpha \cup \beta$
Have $\strategyforR[\alpha \cup \beta]{X}
= \strategyforR[\alpha ]{\apL{X}} \cup \strategyforR[\beta]{\apR{X}}
\subseteq
= \strategyforR[\alpha ]{\apL{Y}} \cup \strategyforR[\beta]{\apR{Y}}
= \strategyforR[\alpha \cup \beta]{Y}$

\mycase $\alpha; \beta$
Have $\strategyforR[\alpha; \beta]{X}
=\strategyforR[\beta]{\strategyforR[\alpha]{X}}
\subseteq \strategyforR[\beta]{\strategyforR[\alpha]{Y}}
= \strategyforR[\alpha; \beta]{Y}$
since by the second IH have  $\strategyforR[\beta]{X} \subseteq \strategyforR[\beta]{Y}$.

\mycase $\prepeat{\alpha}$
Have
$\strategyforR[\prepeat{\alpha}]{X}
= \bigcap\{\apL{Z}~|~ X \cup \strategyforR[\alpha]{\apR{Z}} \subseteq Z\}
\subseteq  \bigcap\{\apL{Z}~|~ Y \cup \strategyforR[\alpha]{\apR{Z}} \subseteq Z\}
= \strategyforR[\prepeat{\alpha}]{Y}$.
Since by assumption $X \subseteq Y$, then $X \cup \strategyforR[\alpha]{\apR{Z}}$ is no larger than $Y \cup \strategyforR[\alpha]{\apR{Z}}$, so the $\bigcap Z$ for $Y$ contains only larger terms than $\bigcap Z$ of $X$ never smaller.

\mycase $\pdual{\alpha}$
$\strategyforR[\pdual{\alpha}]{X}
= \dstrategyforR[\alpha]{X}
\subseteq \dstrategyforR[\alpha]{Y}
= \strategyforR[\pdual{\alpha}]{Y}$

We give the Demon cases.

\mycase $\humod{x}{f}$
Have $\dstrategyforR[\humod{x}{f}]{X}
= \{(\myaa,\subst[\om]{x}{\tint{f}{\om}})~|~ (\myaa,\om) \in X\}
\subseteq \{(\myaa,\subst[\om]{x}{\tint{f}{\om}})~|~ (\myaa,\om)\in Y\}
= \dstrategyforR[\humod{x}{f}]{Y}$

\mycase $\prandom{x}$
Have $\dstrategyforR[\prandom{x}]{X}
= \{(\rzApp{\myaa}{v},\subst[\om]{x}{v})~|~ (\myaa,\om) \in X\text{, for some }v\in\allrat\}
\subseteq \{(\rzApp{\myaa}{v},\subst[\om]{x}{v})~|~(\myaa,\om) \in Y\text{, for some }v\in\allrat\}
= \dstrategyforR[\prandom{x}]{Y}$

\mycase $\ptest{\phi}$
Have $\dstrategyforR[\ptest{\phi}]{X}
= \{(\rzApp{\myaa}{\ab},\om)~|~(\myaa,\om) \in X, (\ab,\om) \in \fintR{\phi}\}
\subseteq  \{(\rzApp{\myaa}{\ab},\om)~|~(\myaa,\om) \in Y, (\ab,\om) \in \fintR{\phi}\}
= \dstrategyforR[\ptest{\phi}]{Y}$

\mycase $\alpha \cup \beta$
Have $\dstrategyforR[\alpha \cup \beta]{X}
= \dstrategyforR[\alpha]{\dpL{X}} \cup \dstrategyforR[\beta]{\dpR{X}}
= \dstrategyforR[\alpha]{\dpL{Y}} \cup \dstrategyforR[\beta]{\dpR{Y}}
 \dstrategyforR[\alpha \cup \beta]{Y}$

\mycase $\alpha; \beta$
Have $\dstrategyforR[\alpha; \beta]{X}
= \dstrategyforR[\beta]{\dstrategyforR[\alpha]{X}}
\subseteq 
\dstrategyforR[\beta]{\dstrategyforR[\alpha]{Y}}
= \dstrategyforR[\alpha; \beta]{Y}$.
since by the second IH $\dstrategyforR[\beta]{X} \subseteq \dstrategyforR[\beta]{Y}$.

\mycase $\prepeat{\alpha}$
Have $\dstrategyforR[\prepeat{\alpha}]{X}
= \bigcap\{\dpL{Z}~|~ X \cup \dstrategyforR[\alpha]{\dpR{Z}} \subseteq Z\}
\subseteq  \bigcap\{\dpL{Z}~|~ Y \cup \dstrategyforR[\alpha]{\dpR{Z}} \subseteq Z\}
= \dstrategyforR[\prepeat{\alpha}]{Y}$.
Since by assumption $X \subseteq Y$, then $X \cup \dstrategyforR[\alpha]{\dpR{Z}}$ is no larger than 
${Y \cup \dstrategyforR[\alpha]{\dpR{Z}}}$, so the $\bigcap Z$ for $Y$ contains only larger terms than $\bigcap Z$ of $X,$ never smaller.

\mycase $\pdual{\alpha}$
Have $\dstrategyforR[\pdual{\alpha}]{X}
= \strategyforR[\alpha]{X}
\subseteq \strategyforR[\alpha]{Y}
= \dstrategyforR[\pdual{\alpha}]{Y}$
\end{proof}

\begin{theorem}[Soundness]
If $\proves{\G}{M}{\phi},$ then the sequent $(\seq{\G}{\phi})$ is valid.
\label{thm:app-soundness}
\end{theorem}
\begin{proof}
By structural induction on the proof term $M$.
In every case, we use $\drv, \drv_1, \drv_2,$ etc.\ to refer to the subderivations and the IH's on those subderivations.

\mycase $\infer{\proves{\G}{\edprojL{M}}{\phi}}{\proves{\G}{M}{\ddiamond{\ptest{\phi}}{\psi}}}$

By $\drv,$ assume $(\ab,\om) \in \cint{\G}$
and $(\rzApp{\ac}{\ab}, \om) \in \fintR{\ddiamond{\ptest{\phi}}{\psi}}$
iff $\strategyforR[\ptest{\phi}]{\{(\rzApp{\ac}{\ab}, \om)\}} \subseteq \fintR{\psi}   \cup \{\stt\}$
iff $(\rzFst{\rzApp{\ac}{\ab}},\om) \in \fintR{\phi}$ and $(\rzSnd{\rzApp{\ac}{\ab}},\om) \in \fintR{\psi}$
Now define $\myaa(\ab) = \rzBLam{\nu}{\rzFst{\rzApp{\ac}{\ab}}(\nu)}$ and
$(\myaa,\om) \in \fintR{\phi},$ then since this held for all $(\ab,\om) \in \cint{\G}$ and $\ac,$ we have that $\seq{\G}{\phi}$ is valid.

\mycase $\infer{\proves{\G}{\edprojR{M}}{\psi}}{\proves{\G}{M}{\ddiamond{\ptest{\phi}}{\psi}}}$

By $\drv,$ assume $(\ab,\om) \in \cint{\G}$
and $(\rzApp{\ac}{\ab}, \om) \in \fintR{\ddiamond{\ptest{\phi}}{\psi}}$
iff $\strategyforR[\ptest{\phi}]{\{(\rzApp{\ac}{\ab}, \om)\}} \subseteq \fintR{\psi}   \cup \{\stt\}$
iff $(\rzFst{\rzApp{\ac}{\ab}},\om) \in \fintR{\phi}$ and $(\rzSnd{\rzApp{\ac}{\ab}},\om) \in \fintR{\psi}$
Now define $\myaa(\ab) = \rzBLam{\nu}{\rzSnd{\rzApp{\ac}{\ab}}(\nu)}$ and
$(\myaa,\om) \in \fintR{\psi},$ then since this held for all $(\ab,\om) \in \cint{\G}$ and $\ac,$ we have that $\seq{\G}{\psi}$ is valid.

\mycase $\infer{\proves{\G}{\ebprojL{M}}{\dbox{\alpha}{\phi}}}{\proves{\G}{M}{\dbox{\alpha \cup \beta}{\phi}}}$

By $\drv,$ assume $(\ab,\om) \in \cint{\G}$
and $(\rzApp{\ac}{\ab},\om) \in \fintR{\dbox{\alpha \cup \psi}{\phi}} = 
\dstrategyforR[\alpha\cup\beta]{(\rzApp{\ac}{\ab},\om)} \subseteq \fintR{\phi}  \cup \{\stt\}= 
\dstrategyforR[\alpha]{\dpL{(\rzApp{\ac}{\ab},\om)}} \cup  \dstrategyforR[\beta]{\dpR{(\rzApp{\ac}{\ab},\om)}} \subseteq \fintR{\phi}  \cup \{\stt\} \text{ implies }
\dstrategyforR[\alpha]{\dpL{(\rzApp{\ac}{\ab},\om)}} \subseteq \fintR{\phi}  \cup \{\stt\}$ 
Now define $\myaa(\ab) = \rzFst{\rzApp{\ac}{\ab}}$ and
$(\myaa,\om) \in \fintR{\dbox{\alpha}{\phi}},$ 
i.e., $\seq{\G}{\dbox{\alpha}{\phi}}$ is valid.

\mycase $\infer{\proves{\G}{\ebprojR{M}}{\dbox{\beta}{\phi}}}{\proves{\G}{M}{\dbox{\alpha \cup \beta}{\phi}}}$

By $\drv,$ assume $(\ab,\om) \in \cint{\G}$
and $(\rzApp{\ac}{\ab},\om) \in \fintR{\dbox{\alpha \cup \psi}{\phi}} = 
\dstrategyforR[\alpha\cup\beta]{(\rzApp{\ac}{\ab},\om)} \subseteq \fintR{\phi} = 
\dstrategyforR[\alpha]{\dpL{(\rzApp{\ac}{\ab},\om)}} \cup  \dstrategyforR[\beta]{\dpR{(\rzApp{\ac}{\ab},\om)}} \subseteq \fintR{\phi} \text{ implies }
\dstrategyforR[\beta]{\dpR{(\rzApp{\ac}{\ab},\om)}} \subseteq \fintR{\phi}$ 
Now define $\myaa(\ab) = \rzSnd{\rzApp{\ac}{\ab}}$ and
$(\myaa,\om) \in \fintR{\dbox{\beta}{\phi}},$ 
i.e., $\seq{\G}{\dbox{\beta}{\phi}}$ is valid.

\mycase $\infer{\proves{\G}{\edcons{M}{N}}{\ddiamond{\ptest{\phi}}{\psi}}}{\proves{\G}{M}{\phi} & \proves{\G}{N}{\psi}}$

Assume $(\ab,\om) \in \cint{\G}$ and $(\rzApp{\ac_1}{\ab},\om) \in \fintR{\phi}$ and $(\rzApp{\ac_2}{\ab},\om) \in \fintR{\psi}$ by $\drv_1$  and $\drv_2,$ 
thus $\strategyforR[\ptest{\phi}]{(\rzCons{\rzApp{\ac_1}{\ab}}{\rzApp{\ac_2}{\ab}},\om)} \subseteq \fintR{\psi}  \cup \{\stt\}$
thus $(\rzCons{\rzApp{\ac_1}{\ab}}{\rzApp{\ac_2}{\ab}},\om) \in \fintR{\ddiamond{\ptest{\phi}}{\psi}} $ thus 
letting $\myaa(\ab) = \rzCons{\rzApp{\ac_1}{\ab}}{\rzApp{\ac_2}{\ab}},$ have $\seq{\G}{\ddiamond{\ptest{\phi}}{\psi}}$ valid.

\mycase $\infer{\proves{\G}{\ebcons{M}{N}}{\dbox{\alpha \cup \beta}{\phi}}}{\proves{\G}{M}{\dbox{\alpha}{\phi}} & \proves{\G}{N}{\dbox{\beta}{\phi}}}$

Assume $(\ab,\om) \in \cint{\G}$ 
and $(\rzApp{\ac_1}{\ab},\om) \in \fintR{\dbox{\alpha}{\phi}}$ 
and $(\rzApp{\ac_2}{\ab},\om) \in \fintR{\dbox{\beta}{\phi}}$ by $\drv_1$  and $\drv_2,$ 
thus $\dstrategyforR[\alpha]{(\rzApp{\ac_1}{\ab},\om)}
\cup \dstrategyforR[\beta]{(\rzApp{\ac_2}{\ab},\om)}
 \subseteq  \fintR{\phi}$.
Then, letting $\myaa(\ab) = \rzCons{\rzApp{\ac_1}{\ab}}{\rzApp{\ac_2}{\ab}},$ have
$\dstrategyforR[\alpha\cup\beta]{(\myaa,\om)} 
= \dstrategyforR[\alpha]{\dpL{\{(\myaa,\om)\}}} \cup \dstrategyforR[\beta]{\dpR{\{(\myaa,\om)\}}}
= \dstrategyforR[\alpha]{(\rzApp{\ac_1}{\ab},\om)} \cup \dstrategyforR[\beta]{(\rzApp{\ac_2}{\ab},\om)}
 \subseteq \fintR{\phi}$.
Thus have $\seq{\G}{\dbox{\alpha\cup\beta}}{\phi}$ valid.

\mycase $\infer{\proves{\G}{\etcons{f}{M}}{\ddiamond{\prandom{x}}{\phi}}}{\proves{\G}{M}{\subst[\phi]{x}{f}}}$

Assume $(\ab,\om) \in \cint{\G},$
by \rref{lem:app-tsub} then $(\eren{\rzApp{\ac}{\ab}}{x}{y}, \subst[\om]{x}{\eren{f}{x}{y}}) \in \cint{\eren{\G}{x}{y}}$ so by $\drv_1$
Let $\myaa(\ab) = \eren{\rzCons{f}{\rzApp{\ac}{\ab}}}{x}{y}$
have $(\rzApp{\ac}{\eren{\ab}{x}{y}}, \subst[\om]{x}{\eren{f}{x}{y}}) \in \fintR{\phi}$
Suffice to show $(\rzApp{\ac}{\eren{\ab}{x}{y}}, \subst[\om]{x}{\eren{f}{x}{y}}) \in \strategyforR[\prandom{x}]{\{(\myaa,\om)\}}
= \{(\rzSnd{\myaa}, \ssub{\om}{x}{\rzApp{\rzFst{\myaa}}{\om}})\}$ which holds by definition of $\myaa,$ then by the previous line
$\strategyforR[\prandom{x}]{\{(\myaa,\om)\}} \subseteq \fintR{\phi}  \cup \{\stt\}$ as desired.

\mycase $\infer{\proves{\G}{\edinjL{M}}{\ddiamond{\alpha\cup\beta}{\phi}}}{\proves{\G}{M}{\ddiamond{\alpha}{\phi}}}$

By $\drv,$ assume $(\ab,\om) \in \cint{\G}$ have $(\rzApp{\ac}{\ab},\om) \in \fintR{\ddiamond{\alpha}{\phi}}$
so $\strategyforR[\alpha]{\{(\rzApp{\ac}{\ab},\om)\}} \subseteq \fintR{\phi}  \cup \{\stt\}$.
Now let $f(\nu) = 0$ and $\myaa(\ab) = \rzCons{f}{\rzApp{\ac}{\ab}},$ then 
$\strategyforR[\alpha\cup\beta]{\{(\myaa,\om)\}} = \strategyforR[\alpha]{\{(\rzApp{\ac}{\ab},\om)\}}$ and by transitivity 
$\strategyforR[\alpha\cup\beta]{\{(\myaa,\om)\}} \subseteq \fintR{\phi}  \cup \{\stt\},$ that is 
$(\myaa,\om) \in \fintR{\ddiamond{\alpha\cup\beta}{\phi}}$ as desired, so $\seq{\G}{\ddiamond{\alpha\cup\beta}{\phi}}$ is valid.

\mycase $\infer{\proves{\G}{\edinjR{M}}{\ddiamond{\alpha\cup\beta}{\phi}}}{\proves{\G}{M}{\ddiamond{\beta}{\phi}}}$

By $\drv,$ assume $(\ab,\om) \in \cint{\G}$ have $(\rzApp{\ac}{\ab},\om) \in \fintR{\ddiamond{\beta}{\phi}}$
so $\strategyforR[\beta]{\{(\rzApp{\ac}{\ab},\om)\}} \subseteq \fintR{\phi}  \cup \{\stt\}$.
Now let $f(\nu) = 1$ and $\myaa(\ab) = \rzCons{f}{\rzApp{\ac}{\ab}},$ then 
$\strategyforR[\alpha\cup\beta]{\{(\myaa,\om)\}} = \strategyforR[\beta]{\{(\rzApp{\ac}{\ab},\om)\}}$ and by transitivity 
$\strategyforR[\alpha\cup\beta]{\{(\myaa,\om)\}} \subseteq \fintR{\phi}  \cup \{\stt\},$ that is 
$(\myaa,\om) \in \fintR{\ddiamond{\alpha\cup\beta}{\phi}}$ as desired, so $\seq{\G}{\ddiamond{\alpha\cup\beta}{\phi}}$ is valid.

\mycase $\infer{\proves{\G}{\edcase{A}{B}{C}}{\psi}}{\proves{\G}{A}{\ddiamond{\alpha\cup\beta}{\phi}} & \proves{\G,\pvl:\ddiamond{\alpha}{\phi}}{B}{\psi} & \proves{\G,\pvr:\ddiamond{\beta}{\phi}}{C}{\psi}}$

By $\drv_1$ assume $(\ab,\om) \in \cint{\G}$ and have $(\rzApp{\ac}{\ab},\om) \in \fintR{\ddiamond{\alpha\cup\beta}{\phi}}$ so
$\strategyforR[\alpha\cup\beta]{\{(\rzApp{\ac}{\ab},\om)\}} \subseteq \fintR{\phi}  \cup \{\stt\}$.

In the first case, $\rzFst{\rzApp{\ac}{\ab}}(\om) = 0$, so $\strategyforR[\alpha]{\{(\rzSnd{\rzApp{\ac}{\ab}},\om)\}} \subseteq \fintR{\phi}  \cup \{\stt\}$.
Now let $\ab_1 = \rzCons{\ab}{\rzSnd{\rzApp{\ac}{\ab}}}$ then $\drv_2$ is applicable because $(\ab_1,\om) \in \cintR{\G,\pvl:\ddiamond{\alpha}{\phi}},$ so we have some $\ac_1$ such that $(\rzApp{\ac_1}{\ab_1},\om) \in \fintR{\psi}$

In the second case, $\rzFst{\rzApp{\ac}{\ab}}(\om) = 1$, so $\strategyforR[\beta]{\{(\rzSnd{\rzApp{\ac}{\ab}},\om)\}} \subseteq \fintR{\phi}  \cup \{\stt\}$.
Now let $\ab_2 = \rzCons{\ab}{\rzSnd{\rzApp{\ac}{\ab}}}$ then $\drv_3$ is applicable because $(\ab_2,\om) \in \cintR{\G,\pvr:\ddiamond{\beta}{\phi}},$ so we have some $\ac_2$ such that $(\rzApp{\ac_2}{\ab_2},\om) \in \fintR{\psi}$

Lastly, define $\myaa(\ab)(\om) = \textsf{if}(f(\om)=0)\{\rzApp{\ac_1}{\ab_1}(\om)\}\textsf{else}\{\rzApp{\ac_2}{\ab_2}(\om)\}$ and
$(\rzApp{\myaa}{\ab},\om) \in \fintR{\psi}$ as desired by the previous two cases.

\mycase $\linferenceRule[formula]
  {\proves{\G}{M}{J} & \proves{\pvx:J}{N}{\dbox{\alpha}{J}} & \proves{\pvx:J}{O}{\phi}}
  {\proves{\G}{(\erep{M}{N}{\pvx:J}{O})}{\dbox{\prepeat{\alpha}}{\phi}}}$

Assume some $\ac_1$ such that for all $(\ab_1,\om) \in \cint{\G}$  have $(\rzApp{\ac_1}{\ab_1},\om) \in \fintR{J}$ by $\drv_1$,
assume some $\ac_2$ such that for all $(\ab_2,\om) \in \cint{\pvx:J}$ have $(\rzApp{\ac_2}{\ab_2},\om) \in \fintR{\dbox{\alpha}{J}}$ by $\drv_2,$ and
assume some $\ac_3$ such that for all $(\ab_3,\om) \in \cint{\pvx:J}$ have $(\rzApp{\ac_3}{\ab_3}, \om) \in \fintR{\phi}$.
Let $C \equiv \bigcap\{Z \subseteq \allRz \times \allstate~|~X \cup\dstrategyforR[\alpha]{\dpR{Z}} \subseteq Z \}$,
where $X = \{(\rzApp{\myaa}{\ab}, \om)\},$
letting $\myaa(\ab) = \rzCons{\rzApp{\ac_1}{\ab}}{\rzApp{\ac_2}{\ab}}$.
Assume $(\ab,\om) \in \cint{\G}$.
Consider arbitrary $(\ad,\nu) \in C,$ suffices to show that $(\ad,\nu) \in \fintR{J}$.
Because $C$ is inductively defined as the intersection across $Z,$ suffices to induct on the membership of $(\ad,\nu) $ in $C$.
In the base case, $(\ad,\nu) \in \dpL{X} = \dpL{\{(\myaa,\om)\}}$ so $\ad = \rzFst{\myaa} = \rzApp{\ac_1}{\ab}$ and $\nu = \om$.
By $\drv_1$ then have $(\ad,\nu) = (\rzApp{\ac_1}{\ab},\om) \in \fintR{J},$ which then yields
$(\rzApp{\ac_3}{(\rzApp{\ac_1}{\ab})}, \om) \in \fintR{\phi}$ as desired.
In the inductive case assume $Z \subseteq \fintR{J}$ and $(\ad,\nu) \in \dpL{\dstrategyforR[\alpha]{\dpR{Z}}}$ and show $(\ad,\nu) \in \fint{J}$.
By the IH have $\ac_2$ s.t.\ for all $(\ab,\om) \in Z$ have  $(\rzApp{\ac_2}{\ab},\om) \in \fintR{\dbox{\alpha}{J}},$ so
$\dstrategyforR[\alpha]{\{(\rzApp{\ac_2}{\ab},\om)\}} \subseteq \fintR{J}$.
Then $\dpR{\{(\myaa(\ab),\om)\}} = \{(\rzApp{\ac_2}{\ab},\om)\}$ so
$\dstrategyforR[\alpha]{\dpR{\{(\myaa(\ab),\om)\}}} \subseteq \fintR{J}$.
Then $\dpL{\dstrategyforR[\alpha]{\dpR{\{(\myaa(\ab),\om)\}}}} \subseteq \dstrategyforR[\alpha]{\dpR{\{(\myaa(\ab),\om)\}}} \subseteq \fintR{J} \subseteq \fintR{\phi}$ (by $\drv_3$) as desired.

\mycase $\infer{\proves{\G}{\eapp{M}{N}}{\psi}}{\proves{\G}{M}{\dbox{\ptest{\phi}}{\psi}} & \proves{\G}{N}{\phi}}$

Assume (1) some $\ac_1$ such that for $(\ab,\om) \in \cint{\G}$ have $(\rzApp{\ac_1}{\ab},\om) \in \fintR{\dbox{\ptest{\phi}}{\psi}}$ by $\drv_1$.
Assume (2) some $\ac_2$ such that for $(\ab,\om) \in \cint{\G}$ have $(\rzApp{\ac_2}{\ab},\om) \in \fintR{\phi}$ by $\drv_2$.
Assume (3) some $(\ab,\om) \in \cint{\G}$.
By (1), then $\dstrategyforR[\ptest{\phi}]{(\rzApp{\ac_1}{\ab},\om)} \subseteq \fint{\psi}  \cup \{\stt\}$.
That is for all $\ad$ s.t.\ $(\ad,\om) \in \fint{\phi}$ then $(\rzApp{\rzApp{\ac_1}{\ab}}{\ad},\om) \in \fint{\psi}$.
By (2) then let $\ad = \rzApp{\ac_2}{\ab}$ and have $(\rzApp{(\rzApp{\ac_1}{\ab})}{(\rzApp{\ac_2}{\ab})},\om) \in \fint{\psi}$.
Then letting $\myaa(\ab) = \rzApp{(\rzApp{\ac_1}{\ab})}{(\rzApp{\ac_2}{\ab})}$ have $(\rzApp{\myaa}{\ab},\om) \in \fintR{\psi}$ as desired.

\mycase $\infer{\proves{\G}{\eplam{\phi}{M}}{\dbox{\ptest{\phi}}{\psi}}}{\proves{\G,\pvx:\phi}{M}{\psi}}$

Assume (1) some $\ac$ such that for $(\ab_2,\om) \in \cint{(\G,\pvx:\phi)}$ have $(\rzApp{\ac}{\ab_2},\om) \in \fintR{\psi}$ by $\drv$.
Assume (2) some $(\ab,\om) \in \cint{\G}$.
Let $\myaa(\ab)(\ab_1) =  \rzApp{\ac}{\rzCons{\ab}{\ab_1}}$
Suffice to show that assuming (3)  $(\ab_1,\om) \in \fintR{\phi}$ then  $\rzApp{\rzApp{\myaa}{\ab}}{\ab_1} \in \fintR{\psi}$.
This holds by expanding the definition of $\myaa,$ then applying (1), whose assumptions hold by (2) and (3).

\mycase 
$\linferenceRule[formula]
  {\proves{\eren{\G}{x}{y}}{M}{\phi}}
  {\proves{\G}{(\etlam{\allrat}{M})}{\dbox{\prandom{x}}{\phi}}}$

Assume (1) some $\ac$ such that for $(\ab_1,\om) \in \cint{\eren{\G}{x}{y}}$ have $(\rzApp{\ac}{\ab_1},\om) \in \fintR{\phi}$ by $\drv$.
Assume (2) some $(\ab,\om) \in \cint{\G}$.
By side condition, $y$ is fresh.
Let $\myaa(\ab)(v) = \rzApp{\ac}{\ssub{\ab}{x}{v}}$.
Suffice to show $\fintR{\phi} \supseteq \dstrategyforR[\prandom{x}]{\{(\myaa(\ab), \om)\}} = \{(\myaa(\ab)(v), \subst[\om]{x}{v})~|~v\in \allrat\}$.
Fix some $v \in \allrat,$ then by freshness of $y$ and \rref{lem:app-tren} have 
$(\eren{\ab}{x}{y},\eren{\om}{x}{y}) \in \cintR{\eren{\G}{x}{y}}$.
Then by (1) have $(\rzApp{\ac}{\eren{\ab}{x}{y}},\eren{\om}{x}{y}) \in \fintR{\phi}$.
Then by \rref{lem:app-coincide} note $(\rzApp{\ac}{\eren{\ab}{x}{y}},\eren{\om}{x}{y})$ agrees with $\{(\myaa(\ab)(v), \subst[\om]{x}{v})~|~v\in\allrat\}$
on the free variables of $\phi$, since $\myaa(\ab)(v) = \rzApp{\ac}{\ssub{\ab}{x}{v}} = \rzApp{\ac}{\ssub{\ab}{x}{y}}$ for \emph{some} assignment of the fresh $y$, that is each $v$ fixed earlier matches some value of $y$ when taking the union.
Furthermore $\eren{\om}{x}{y} = \subst[\om]{x}{v}$  for some valuation of $y$ for the same reason.
Then by \rref{lem:app-coincide} can conclude $\{(\myaa(\ab)(v), \subst[\om]{x}{v})~|~v\in\allrat\} \subseteq \fintR{\phi}$ as desired.

\mycase $\infer{\proves{\G}{\ebseq{M}}{\dbox{\alpha;\beta}{\phi}}}{\proves{\G}{M}{\dbox{\alpha}{\dbox{\beta}{\phi}}}}$

Assume (1) some $\ac$ such that for $(\ab,\om) \in \cint{\G}$ have $(\rzApp{\ac}{\ab},\om) \in \fintR{\dbox{\alpha}{\dbox{\beta}{\phi}}}$ by $\drv$.
Assume (2) some $(\ab,\om) \in \cint{\G}$.
Then by (1) have
$(\rzApp{\ac}{\ab},\om) \in \fintR{\dbox{\alpha}{\dbox{\beta}{\phi}}}$
iff $\dstrategyforR[\alpha]{\{(\rzApp{\ac}{\ab},\om)\}} \subseteq \fintR{\dbox{\beta}{\phi}}$
iff $\dstrategyforR[\alpha]{\{(\rzApp{\ac}{\ab},\om)\}} \subseteq \{(\ad,\nu)~|~\dstrategyforR[\beta]{\{(\ad,\nu)\}} \subseteq \fintR{\phi}\}$
iff $\left(\bigcup_{(\ad,\nu)\in\dstrategyforR[\alpha]{\{(\rzApp{\ac}{\ab},\om)\}}} \dstrategyforR[\beta]{\{(\ad,\nu)\}}\right) \subseteq \fintR{\phi}$
iff (\rref{lem:app-partition}) $ \dstrategyforR[\beta]{\dstrategyforR[\alpha]{\{(\rzApp{\ac}{\ab},\om)\}}} \subseteq \fintR{\phi}$
iff $(\rzApp{\ac}{\ab},\om) \in \fintR{\dbox{\alpha;\beta}{\phi}}$ as desired.

\mycase $\infer{\proves{\G}{\edseq{M}}{\ddiamond{\alpha;\beta}{\phi}}}{\proves{\G}{M}{\ddiamond{\alpha}{\ddiamond{\beta}{M}}}}$

Assume (1) some $\ac$ such that for $(\ab,\om) \in \cint{\G}$ have $(\rzApp{\ac}{\ab},\om) \in \fintR{\ddiamond{\alpha}{\dbox{\beta}{\phi}}}$ by $\drv$.
Assume (2) some $(\ab,\om) \in \cint{\G}$.
Then by (1) have
\begin{align*}
&\phantom{\text{iff}}    (\rzApp{\ac}{\ab},\om) \in \fintR{\ddiamond{\alpha}{\dbox{\beta}{\phi}}}\\
&\text{iff} \strategyforR[\alpha]{\{(\rzApp{\ac}{\ab},\om)\}} \subseteq (\fintR{\dbox{\beta}{\phi}}\cup\{\top\})\\
&\text{iff} \strategyforR[\alpha]{\{(\rzApp{\ac}{\ab},\om)\}} \subseteq \{(\ad,\nu)~|~\strategyforR[\beta]{\{(\ad,\nu)\}} \subseteq \fintR{\phi} \cup\{\top\})\}\\
&\text{iff} \left(\bigcup_{(\ad,\nu)\in\strategyforR[\alpha]{\{(\rzApp{\ac}{\ab},\om)\}}} \strategyforR[\beta]{\{(\ad,\nu)\}}\right) \subseteq \fintR{\phi} \cup \{\top\}\\
&\text{iff (\rref{lem:app-partition})} \strategyforR[\beta]{\strategyforR[\alpha]{\{(\rzApp{\ac}{\ab},\om)\}}} \subseteq \fintR{\phi} \cup\{\top\}\\
&\text{iff}(\rzApp{\ac}{\ab},\om) \in \fintR{\ddiamond{\alpha;\beta}{\phi}}
\end{align*}
 as desired.

\mycase $\infer{\proves{\G}{\edswap{M}}{\ddiamond{\pdual{\alpha}}{\phi}}}{\proves{\G}{M}{\dbox{\alpha}{\phi}}}$

Assume (1) some $\ac$ such that for $(\ab,\om) \in \cint{\G}$ have $(\rzApp{\ac}{\ab},\om) \in \fintR{\dbox{\alpha}{\phi}}$ by $\drv$.
Assume (2) some $(\ab,\om) \in \cint{\G}$, then by (1) have
$(\rzApp{\ac}{\ab},\om) \in \fintR{\dbox{\alpha}{\phi}}$
iff $\dstrategyforR[\alpha]{\{(\rzApp{\ac}{\ab},\om)\}} \subseteq \fintR{\phi} \cup \{\bot\}$
iff $(\sdual{(\dstrategyforR[\alpha]{\{(\rzApp{\ac}{\ab},\om)\}})} \subseteq \fintR{\phi} \cup \{\top\}$
iff $\strategyforR[\pdual{\alpha}]{\{(\rzApp{\ac}{\ab},\om)\}} \subseteq \fintR{\phi} \cup \{\top\}$
iff $(\rzApp{\ac}{\ab},\om) \in \fintR{\ddiamond{\pdual{\alpha}}{\phi}}$ as desired.

\mycase $\infer{\proves{\G}{\ebswap{M}}{\dbox{\pdual{\alpha}}{\phi}}}{\proves{\G}{M}{\ddiamond{\alpha}{\phi}}}$

Assume (1) some $\ac$ such that for $(\ab,\om) \in \cint{\G}$ have $(\rzApp{\ac}{\ab},\om) \in \fintR{\ddiamond{\alpha}{\phi}}$ by $\drv$.
Assume (2) some $(\ab,\om) \in \cintR{\G}$, then by (1) have
$(\rzApp{\ac}{\ab},\om) \in \fintR{\ddiamond{\alpha}{\phi}}$
iff $\strategyforR[\alpha]{\{(\rzApp{\ac}{\ab},\om)\}} \subseteq \fintR{\phi} \cup \{\top\}$
iff $(\sdual{(\strategyforR[\alpha]{\{(\rzApp{\ac}{\ab},\om)\}})} \subseteq \fintR{\phi} \cup \{\bot\}$
iff $\dstrategyforR[\pdual{\alpha}]{\{(\rzApp{\ac}{\ab},\om)\}} \subseteq \fintR{\phi} \cup \{\bot\}$
iff $(\rzApp{\ac}{\ab},\om) \in \fintR{\dbox{\pdual{\alpha}}{\phi}}$ as desired.

\mycase $\infer{\proves{\G}{\emon{M}{N}{x}}{\ddiamond{\alpha}{\psi}}}{\proves{\G}{M}{\ddiamond{\alpha}{\phi}} & \proves{x:\phi}{N}{\psi}}$

Assume (1) some $\ac_1$ such that for $(\ab_1,\om) \in \cintR{\G}$ have  $(\rzApp{\ac_1}{\ab_1},\om) \in \fintR{\ddiamond{\alpha}{\phi}}$ by $\drv_1$.
Assume (2) some $\ac_2$ such that for $(\ab_2,\nu)  \in \cintR{x:\phi}$ have $(\rzApp{\ac_2}{\ab_2},\om) \in \fintR{\psi}$ by $\drv_2$.
Assume (3) some $(\ab,\om) \in \cintR{\G}$.
Let $\myaa(\ab) =  \rzApp{\ac_2} \circ \rzApp{\ac_1}{\ab},$ that is the realizer which applies $\ac_2$ as the continuation after playing according to $\ac_1$.
By (1) have $(\rzApp{\ac_1}{\ab},\om) \in \fintR{\ddiamond{\alpha}{\phi}}$
iff $\strategyforR[\alpha]{\{(\rzApp{\ac_1}{\ab},\om)\}} \subseteq \fintR{\phi} \cup \{\top\}$
Consider arbitrary $(\ab_2,\nu)$ in $\strategyforR[\alpha]{\{(\rzApp{\ac_1}{\ab},\om)\}}$.
If $(\ab_2,\nu) = \top$ then we're done because $\strategyforR[\alpha]{\{(\rzApp{\ac_1}{\ab},\om)\}} \subseteq \fintR{\psi} \cup \{\top\}$ as well, else
we can apply (2) and get $(\rzApp{\ac_2}{\ab_2},\nu) \in \fintR{\psi}$.
Now recall it suffices to show $\strategyforR[\alpha]{(\rzApp{\myaa}{\ab},\om)} \subseteq \fintR{\psi} \cup \{\top\}$.
By definition of $\myaa,$ have $\strategyforR[\alpha]{(\rzApp{\myaa}{\ab},\om)} = \{(\rzApp{\ac_2}{\ab_2},\nu)~|~(\ab_2,\nu) \in \strategyforR[\alpha]{\{(\rzApp{\ac_1}{\ab},\om)\}}\}$, so by (2)  have $\strategyforR[\alpha]{(\rzApp{\myaa}{\ab},\om)} \subseteq \fintR{\psi}$ as desired.

\mycase $\infer{\proves{\G}{\estop{M}}{\ddiamond{\prepeat{\alpha}}{\phi}}}{\proves{\G}{M}{\phi}}$

Assume (1) some $\ac$ such that for $(\ab,\om) \in \cintR{\G}$ have $(\rzApp{\ac_1}{\ab},\om) \in \fintR{\phi}$ by $\drv$.
Assume (2) some $(\ab,\om) \in \cintR{\G}$.
Let $\myaa(\ab) = \rzCons{(\rzBLam{\nu}{0})}{\rzApp{\ac_1}{\ab}}$.
Then by (1) have $(\rzApp{\ac}{\ab},\om) \in \fintR{\phi}$
and $\apL{\{(\myaa(\ab),\om)\}} = \{(\rzApp{\ac}{\ab},\om)\}$
so $\apL{\{(\myaa(\ab),\om)\}} \subseteq \fintR{\phi}$.

Suffices to show $(\myaa(\ab), \om) \in \fintR{\ddiamond{\prepeat{\alpha}}{\phi}}$
iff $\strategyforR[\prepeat{\alpha}]{\{(\myaa(\ab), \om)\}} \subseteq \fintR{\phi}$ 
iff $\bigcap\{\apL{Z}\subseteq \allRz \times \allstate~|~ \{(\myaa(\ab), \om)\} \cup (\strategyforR[\alpha]{\apR{Z}}) \subseteq Z\} \subseteq \fintR{\phi}\cup\{\top\}$
which holds if $\apL{\{(\myaa(\ab), \om)\}} \subseteq \fintR{\phi},$ which we showed above, so we're done.

\mycase $\infer{\proves{\G}{\ego{M}}{\ddiamond{\prepeat{\alpha}}{\phi}}}{\proves{\G}{M}{\ddiamond{\alpha}{\ddiamond{\prepeat{\alpha}}{\phi}}}}$

Assume (1) some $\ac_{both}$ such that for $(\ab,\om) \in \cintR{\G}$ have $(\rzApp{\ac_1}{\ab},\om) \in \fintR{\ddiamond{\alpha}{\ddiamond{\prepeat{\alpha}}{\phi}}}$ by $\drv$.
Split $\ac_{both}$ into $\ac_2 \circ \ac_1$ for $\prepeat{\alpha}$ and $\alpha$.
Assume (2) some $(\ab,\om) \in \cintR{\G}$.
Let $C(X) = \bigcap\{Z \subseteq \allRz \times \allstate~|~ X \cup (\strategyforR[\alpha]{\apR{Z}}) \subseteq Z\}$
Note because $C$ is a fixed point, $C(X) = X \cup (\strategyforR[\alpha]{\apR{C}})$ holds as an exact equality.
Then by (1) have $(\rzApp{\ac}{\ab},\om) \in \fintR{\ddiamond{\alpha}{\ddiamond{\prepeat{\alpha}}{\phi}}}$
iff $\strategyforR[\alpha]{\{(\rzApp{\ac}{\ab},\om)\}} \subseteq \fintR{\ddiamond{\prepeat{\alpha}}{\phi}} \cup \{\top\}$
iff $\strategyforR[\alpha]{\{(\rzApp{\ac}{\ab},\om)\}} \subseteq \{(\ad,\nu)~|~\strategyforR[\prepeat{\alpha}]{\{(\ad,\nu)\}} \subseteq \fintR{\phi} \cup \{\top\}\}$
iff (2) for all $(\rzApp{\ac_1}{\ab},\mu) \in \strategyforR[\alpha]{\{(\rzApp{\ac}{\ab},\om)\}}$
  and all $(\rzApp{\ac_2}{\rzApp{\ac_1}{\ab}},\nu) \in \strategyforR[\prepeat{\alpha}]{\{(\rzApp{\ac_1}{\ab},\mu)\}}$
  have $(\rzApp{\ac_2}{\rzApp{\ac_1}{\ab}},\nu) \in \fintR{\phi} \cup \{\top\}$.
Let $\myaa(\ab) = \rzCons{(\rzBLam{\nu}{1})}{\ac_2 \circ \rzApp{\ac_1}{\ab}}$.
So $\apL{\{(\myaa(\ab),\om)\}} = \emptyset $
and $\apR{\{(\myaa(\ab),\om)\}} = \{(\ac_2 \circ \rzApp{\ac_1}{\ab},\om)\}$.
Then $\strategyforR[\alpha]{\{(\ac_2 \circ \rzApp{\ac_1}{\ab},\om)\}} =
\{(\rzApp{\ac_1}{\ab},\mu) \in \strategyforR[\alpha]{\{(\rzApp{\ac}{\ab},\om)\}}\}$, fulfilling the first assumption of (2).
Then $\strategyforR[\prepeat{\alpha}]{\{(\rzApp{\ac_1}{\ab},\mu)~|~\cdots\}} =
(\rzApp{\ac_2}{\rzApp{\ac_1}{\ab}},\nu) \in \strategyforR[\prepeat{\alpha}]{\{(\rzApp{\ac_1}{\ab},\mu)\}}$ fulfilling second assumption of (2),
thus we conclude $(\rzApp{\ac_2}{\rzApp{\ac_1}{\ab}},\nu) \in \fintR{\phi} \cup \{\top\}$
By definition of $\myaa,$ we have that $\strategyforR[\prepeat{\alpha}]{(\strategyforR[\alpha]{\{(\myaa(\ab),\om)\}})}$ is the set of all such 
$(\rzApp{\ac_2}{\rzApp{\ac_1}{\ab}},\nu)$ so 
$\strategyforR[\prepeat{\alpha}]{(\strategyforR[\alpha]{\{(\myaa(\ab),\om)\}})} \subseteq \fintR{\phi} \cup \{\top\}$
as desired.

\mycase $\linferenceRule[formula]
{\proves{\G}{A}{\ddiamond{\prepeat{\alpha}}{\phi}} & \proves{\G,\ell:\phi}{B}{\psi} & \proves{\G,r:\ddiamond{\alpha}{\ddiamond{\prepeat{\alpha}}{\phi}}}{C}{\psi}}
{\proves{\G}{\ercase{A}{B}{C}}{\psi}}$

Assume (1) some $\ac_1$ such that for $(\ab,\om) \in \cintR{\G}$ have $(\rzApp{\ac_1}{\ab},\om) \in \fintR{\ddiamond{\prepeat{\alpha}}{\phi}}$ by $\drv_1$.
Assume (2) some $\ac_2$ such that for $(\ab_2,\om) \in \cintR{\G,\pvl:\phi}$ have $(\rzApp{\ac_2}{\ab_2},\om) \in \fintR{\psi}$ by $\drv_2$.
Assume (3) some $\ac_3$ such that for $(\ab_3,\om) \in \cintR{\G,\pvr:\ddiamond{\alpha}{\ddiamond{\prepeat{\alpha}}{\phi}}}$ have $(\rzApp{\ac_3}{\ab_3},\om) \in \fintR{\psi}$ by $\drv_3$.
Assume (4) some $(\ab,\om) \in \cintR{\G}$, 
To handle case $\rzFst{\rzApp{\ac_1}{\ab}}(\om) = 0$  let  $\ab_2 = \rzCons{\ab}{\rzSnd{\rzApp{\ac_1}{\ab}}}$.
To handle case $\rzFst{\rzApp{\ac_1}{\ab}}(\om) = 1$ let  $\ab_3 = \rzCons{\ab}{\rzSnd{\rzApp{\ac_1}{\ab}}}$.
Note (5a) when $\rzFst{\rzApp{\ac_1}{\ab}}(\om) = 0$  then  $(\rzSnd{\rzApp{\ac_1}{\ab}},\om) \in \fintR{\phi}$.
and (5b) when $\rzFst{\rzApp{\ac_1}{\ab}}(\om) = 0$  then  $(\rzSnd{\rzApp{\ac_1}{\ab}},\om) \in \fintR{\ddiamond{\alpha}{\ddiamond{\prepeat{\alpha}}{\phi}}}$ by (1) and the definition of diamond repetition realizers.
Let $\myaa(\ab)(\om) = \textsf{if}(\rzFst{\rzApp{\ac_1}{\ab}}(\om) = 0)\{\rzApp{\ac_2}{\ab_2(\om)}\}\textsf{else}\{\rzApp{\ac_3}{\ab_3(\om)}\}$

By (4) and (1) have  $(\rzApp{\ac_1}{\ab},\om) \in \fintR{\ddiamond{\prepeat{\alpha}}{\phi}}$
iff $\strategyforR[\prepeat{\alpha}]{(\rzApp{\ac_1}{\ab},\om)} \subseteq \fintR{\phi} \cup \{\top\}$.
By \rref{lem:app-inflate} we have either
a) $\strategyforR[\prepeat{\alpha}]{\{(\rzApp{\ac_1}{\ab},\om)\}} = \apL{(\rzApp{\ac_1}{\ab},\om)} \subseteq \fintR{\phi} \cup \{\top\}$, or
b) $\strategyforR[\prepeat{\alpha}]{\{(\rzApp{\ac_1}{\ab},\om)\}} = \strategyforR[\prepeat{\alpha}]{\strategyforR[\alpha]{\apR{(\rzApp{\ac_1}{\ab},\om)}}} \subseteq \fintR{\phi} \cup \{\top\}$.

In case (a) then $\strategyforR[\prepeat{\alpha}]{\{(\myaa(\ab),\om)\}} = \rzApp{\ac_2}{\ab_2(\om)} \subseteq \fintR{\psi}$ by (2), which is applicable because 
$(\rzCons{\ab}{\rzSnd{\rzApp{\ac_1}{\ab}}},\om) \in \cintR{\G,\pvl:\phi}$ by (5a).
In case (b) then $\strategyforR[\prepeat{\alpha}]{\{(\myaa(\ab),\om)\}} = \rzApp{\ac_3}{\ab_3(\om)} \subseteq \fintR{\psi}$ by (3), which is applicable because 
$(\rzCons{\ab}{\rzSnd{\rzApp{\ac_1}{\ab}}},\om) \in \cintR{\G,\pvl:\ddiamond{\alpha}{\ddiamond{\prepeat{\alpha}}{\phi}}}$ by (5b).

In either case $\strategyforR[\prepeat{\alpha}]{\{(\myaa(\ab),\om)\}} \subseteq \fintR{\psi}$ as desired.

\mycase $\infer{\proves{\G}{\ebroll{M}}{\dbox{\prepeat{\alpha}}{\phi}}}{\proves{\G}{M}{\phi \land \dbox{\alpha}{\dbox{\prepeat{\alpha}}{\phi}}}}$

Assume (1) some $\ac$ such that for $(\ab,\om) \in \cintR{\G}$ have $(\rzApp{\ac}{\ab},\om) \in \fintR{\phi \land \dbox{\alpha}{\dbox{\prepeat{\alpha}}{\phi}}}$ by $\drv$.
Assume (2) some $(\ab,\om) \in \cintR{\G}$, so by (1) have 
$(\rzApp{\ac}{\ab},\om) \in \fintR{\phi \land \dbox{\alpha}{\dbox{\prepeat{\alpha}}{\phi}}}$
iff (3a) $(\rzFst{\rzApp{\ac}{\ab}},\om) \in \fintR{\phi}$ and (3b) $(\rzSnd{\rzApp{\ac}{\ab}},\om) \in \fintR{\dbox{\alpha}{\dbox{\prepeat{\alpha}}{\phi}}}$

Suffice to show $(\rzApp{\ac}{\ab},\om) \in \fintR{\dbox{\prepeat{\alpha}}{\phi}}$
iff $\dstrategyforR[\prepeat{\alpha}]{\{(\rzApp{\ac}{\ab},\om)\}} \subseteq \fintR{\phi} \cup \{\bot\}$
iff (by \rref{lem:app-inflate}) (G1) $\dpL{\{(\rzApp{\ac}{\ab},\om)\}} \subseteq \fintR{\phi} \cup \{\bot\}$
and (G2) $\dstrategyforR{\dstrategyforR[\prepeat{\alpha}]{\{(\rzApp{\ac}{\ab},\om)\}}} \subseteq \fintR{\phi} \cup \{\bot\}$.
Goal (G1) is direct by (3a) and (G2) holds from (3b) since 
$(\rzSnd{\rzApp{\ac}{\ab}},\om) \in \fintR{\dbox{\alpha}{\dbox{\prepeat{\alpha}}{\phi}}}$
iff  $\dstrategyforR[\prepeat{\alpha}]{\dstrategyforR[\alpha]{\{(\rzSnd{\rzApp{\ac}{\ab}},\om)\}}} \subseteq \fintR{\phi} \cup \{\bot\}$.

\mycase $\infer{\proves{\G}{\ebunroll{M}}{\phi \land \dbox{\alpha}{\dbox{\prepeat{\alpha}}{\phi}}}}{\proves{\G}{M}{\dbox{\prepeat{\alpha}}{\phi}}}$

Assume (1) some $\ac$ such that for $(\ab,\om) \in \cintR{\G}$ have $(\rzApp{\ac}{\ab},\om) \in \fintR{\dbox{\prepeat{\alpha}}{\phi}}$ by $\drv$.
Assume (2) some $(\ab,\om) \in \cintR{\G}$, so by (1) have 
$(\rzApp{\ac}{\ab},\om) \in \fintR{\dbox{\prepeat{\alpha}}{\phi}}$
iff  (3a) $(\rzFst{\rzApp{\ac}{\ab}},\om) \in \fintR{\phi}$ and 
(3b) $(\rzSnd{\rzApp{\ac}{\ab}},\om) \in \fintR{\dbox{\alpha}{\dbox{\prepeat{\alpha}}{\phi}}}$

Suffice to show $(\rzApp{\ac}{\ab},\om) \in \fintR{\phi \land \dbox{\alpha}{\dbox{\prepeat{\alpha}}{\phi}}}$
iff (G1) $\dstrategyforR[\prepeat{\alpha}]{\{(\rzFst{\rzApp{\ac}{\ab}},\om)\}} \subseteq \fintR{\phi} \cup \{\bot\}$
and (G2) $\dstrategyforR{\dstrategyforR[\prepeat{\alpha}]{\{(\rzSnd{\rzApp{\ac}{\ab}},\om)\}}} \subseteq \fintR{\phi} \cup \{\bot\}$
Goal (G1) is direct by (3a) and (G2) holds from (3b) since 
$(\rzSnd{\rzApp{\ac}{\ab}},\om) \in \fintR{\dbox{\alpha}{\dbox{\prepeat{\alpha}}{\phi}}}$
iff $\dstrategyforR[\prepeat{\alpha}]{\dstrategyforR[\alpha]{\{(\rzSnd{\rzApp{\ac}{\ab}},\om)\}}} \subseteq \fintR{\phi} \cup \{\bot\}$.

\mycase $\linferenceRule[formula]
       {\proves{\eren{\G}{x}{y},\pvx:(x=\eren{f}{x}{y})}{M}{\phi}}
        {\proves{\G}{\ebasgneq{y}{x}{\pvx}{M}}{\dbox{\humod{x}{f}}{\phi}}}$

Assume (1) some $\ac$ such that for $(\ab_1,\nu) \in \cintR{\eren{\G}{x}{y},\pvx:(x=\eren{f}{x}{y}}$ 
 have $(\rzApp{\ac}{\ab_1},\nu) \in \fintR{\phi}$ by $\drv$.
Assume (2) some $(\ab,\om) \in \cintR{\G}$.
Let $v = \tint{\eren{f}{x}{y}}{\eren{\om}{x}{y}} = \tint{f}{\om}$ by \rref{lem:app-tren} and
$\nu = \ssub{\eren{\om}{x}{y}}{x}{v}$.
Let  $\ab_1(\om) = \ab(\subst[\om]{x}{\tint{f}{\om}})$ and $\ab_2(\om) = \ab_1(\eren{\om}{x}{y})$
Then by \rref{lem:app-tsub} have $(\ab_1,\sren{\om}{x}{y}) \in \cintR{\eren{\G}{x}{y}}$
Then by \rref{lem:app-tren} $(\ab_2, \nu) \in \cintR{\eren{\G}{x}{y}}$.
Also by definition of $\nu$ have $(\rzCons{\rzNil}{\rzNil},\nu) \in \fintR{x=\eren{f}{x}{y}}$
so let $\ab_1 = \rzCons{\ab_2}{\rzCons{\rzNil}{\rzNil}}$ and then by (1) have $(\rzApp{\ac}{\ab_1},\nu) \in \fintR{\phi}$.
Let $\myaa = \rzApp{\ac}{\ab_1}$.
Suffices to show $\dstrategyforR[\humod{x}{f}]{\{(\myaa,\om)\}} \subseteq \fintR{\phi} \cup \{\top\},$ the latter case occurring only when $(\myaa,\om) = \top$.
Then $\dstrategyforR[\humod{x}{f}]{\{(\myaa,\om)\}} = \{(\myaa,\subst[\om]{x}{v})\}$ which agrees with $(\rzApp{\ac}{\ab_1},\nu)$ except on fresh $y,$ so by \rref{lem:app-coincide} have $(\myaa,\subst[\om]{x}{v}) \in \fintR{\phi}$ as desired.

\mycase  $\linferenceRule[formula]
       {\proves{\eren{\G}{x}{y},\pvx:(x=\eren{f}{x}{y})}{M}{\phi}}
        {\proves{\G}{\edasgneq{y}{x}{\pvx}{M}}{\ddiamond{\humod{x}{f}}{\phi}}}$

Assume (1) some $\ac$ such that for $(\ab_1,\nu) \in \cintR{\eren{\G}{x}{y},\pvx:(x=\eren{f}{x}{y})}$ 
 have $(\rzApp{\ac}{\ab_1},\nu) \in \fintR{\phi}$ by $\drv$.
Assume (2) some $(\ab,\om) \in \cintR{\G}$.
Let $v = \tint{\eren{f}{x}{y}}{\eren{\om}{x}{y}} = \tint{f}{\om}$ by \rref{lem:app-tren} and
$\nu = \ssub{\eren{\om}{x}{y}}{x}{v}$.
Let  $\ab_1(\om) = \ab(\subst[\om]{x}{\tint{f}{\om}})$ and $\ab_2(\om) = \ab_1(\eren{\om}{x}{y})$
Then by \rref{lem:app-tsub} have $(\ab_1,\sren{\om}{x}{y}) \in \cintR{\eren{\G}{x}{y}}$
Then by \rref{lem:app-tren} $(\ab_2, \nu) \in \cintR{\eren{\G}{x}{y}}$.
Also by definition of $\nu$ have $(\rzCons{\rzNil}{\rzNil},\nu) \in \fintR{x=\eren{f}{x}{y}}$
so let $\ab_1 = \rzCons{\ab_2}{\rzCons{\rzNil}{\rzNil}}$ and then by (1) have $(\rzApp{\ac}{\ab_1},\nu) \in \fintR{\phi}$.
Let $\myaa = \rzApp{\ac}{\ab_1}$.
Suffices to show $\strategyforR[\humod{x}{f}]{\{(\myaa,\om)\}} \subseteq \fintR{\phi} \cup \{\top\},$ the latter case occurring only when $(\myaa,\om) = \top$.
Then $\strategyforR[\humod{x}{f}]{\{(\myaa,\om)\}} = \{(\myaa,\subst[\om]{x}{v})\}$ which agrees with $(\rzApp{\ac}{\ab_1},\nu)$ except on fresh $y,$ so by \rref{lem:app-coincide} have $(\myaa,\subst[\om]{x}{v}) \in \fintR{\phi}$ as desired.

\mycase  $\linferenceRule[formula]{\proves{\G}{M}{\ddiamond{\prandom{x}}{\phi}} & \proves{\eren{\G}{x}{y},\pvx:\phi}{N}{\psi}}
{\proves{\G}{\eunpack{M}{N}}{\psi}}$ for  $y$ fresh and  $x \notin \freevars{\psi}$

Assume (1) some $\ac_1$ such that for $(\ab,\om) \in \cintR{\G}$ have $(\rzApp{\ac_1}{\ab},\om) \in \fintR{\ddiamond{\prandom{x}}{\phi}}$ by $\drv_1$.
Assume (2) some $\ac_2$ such that for $(\ab_2,\nu) \in \cintR{\eren{\G}{x}{y},\pvx:\phi}$ have $(\rzApp{\ac_2}{\ab_2},\nu) \in \fintR{\psi}$ by $\drv_2$.
Assume (3) some $(\ab,\om) \in \cintR{\G}$, then by (1) have $(\rzApp{\ac_1}{\ab},\om) \in \fintR{\ddiamond{\prandom{x}}{\phi}}$
iff $\strategyforR[\prandom{x}]{\{(\rzApp{\ac_1}{\ab},\om)\}} \subseteq \fintR{\phi} \cup \{\top\},$ the latter occurring only when $(\ab,\om) = \top,$ in which case we are done.
Else $(\rzSnd{\rzApp{\ac_1}{\ab}}, \subst[\om]{x}{\rzFst{\rzApp{\ac_1}{\ab}}}) \in \fintR{\phi}$ and by freshness of $y$ and by \rref{lem:app-coincide} then have
(4a) $(\rzSnd{\rzApp{\ac_1}{\ab}}, \subst[\eren{\om}{x}{y}]{x}{\rzFst{\rzApp{\ac_1}{\ab}}}) \in \fintR{\phi}$
Also (4b) $(\eren{\ab}{x}{y}, \ssub{\eren{\om}{x}{y}}{x}{\rzFst{\rzApp{\ac_1}{\ab}}}) \in \cintR{\eren{\G}{x}{y}}$ by \rref{lem:app-tren} on (3) and by \rref{lem:app-coincide} because $x \notin \freevars{\eren{\G}{x}{y}}$.
Then from (4a) and (4b) have $(\rzCons{\eren{\ab}{x}{y}}{\rzSnd{\rzApp{\ac_1}{\ab}}}, \ssub{\eren{\om}{x}{y}}{x}{\rzFst{\rzApp{\ac_1}{\ab}}}) \in \cintR{\eren{\G}{x}{y},\pvx:\phi}$ so by (2) have 
(5) $(\rzApp{\ac_2}{\rzCons{\eren{\ab}{x}{y}}{\rzSnd{\rzApp{\ac_1}{\ab}}}},\ssub{\eren{\om}{x}{y}}{x}{\rzFst{\rzApp{\ac_1}{\ab}}}) \in \fintR{\psi}$.
Lastly note $\om = \ssub{\eren{\om}{x}{y}}{x}{\rzFst{\rzApp{\ac_1}{\ab}}}$ on $\freevars{\psi} \subseteq \{x,y\}^\complement,$ so let $\myaa = \rzApp{\ac_2}{\rzCons{\eren{\ab}{x}{y}}{\rzSnd{\rzApp{\ac_1}{\ab}}}}$ then $(\myaa,\om) \in \fintR{\psi}$ as desired.

\mycase $\infer{\proves{\G}{\efp{A}{B}{C}}{\psi}}
        {\proves{\G}{A}{\ddiamond{\prepeat{\alpha}}{\phi}} 
        &\proves{\pvx:\phi}{B}{\psi} & \proves{\pvy:\ddiamond{\alpha}{\psi}}{C}{\psi}}$

Assume (1) some $\ac_1$ such that for $(\ab,\om) \in \cintR{\G}$ have $(\rzApp{\ac}{\ab},\om) \in \fintR{\ddiamond{\prepeat{\alpha}}{\phi}}$ by $\drv_1$.
Assume (2) some $\ac_2$ such that for $(\ab_2,\nu) \in \cintR{\pvx:\phi}$ have $(\rzApp{\ac_2}{\ab_2},\nu) \in \fintR{\psi}$ by $\drv_2$.
Assume (3) some $\ac_3$ such that for $(\ab_3,\nu) \in \cintR{\pvy:\ddiamond{\alpha}{\psi}}$ have $(\rzApp{\ac_3}{\ab_3},\nu) \in \fintR{\psi}$ by $\drv_3$.
Assume (4) some $(\ab,\om) \in \cintR{\G}$, so by (1) have $(\rzApp{\ac}{\ab},\om) \in \fintR{\ddiamond{\prepeat{\alpha}}{\phi}}$ 
iff $\apL{\istrat[\alpha]{\{(\rzApp{\ac}{\ab},\om)\}}{k}} \subseteq \fintR{\phi}$ for all $k \in \mathbb{N}$.
Fix any $k$ and any such $(\ad,\nu) \in \apL{\istrat[\alpha]{\{(\rzApp{\ac}{\ab},\om)\}}{k}}$
We define $\myaa$ by induction on $k$.
This definition can be considered constructive in the sense that we could instrument the realizer $\rzApp{\ac}{\ab}$ to compute the results of $\myaa$ ``as it goes'':

Case $(\ad,\nu) \in \apL{\{(\rzApp{\ac}{\ab},\om)\}},$ then $(\ad,\nu) = (\rzSnd{\rzApp{\ac}{\ab}},\om)$ and $\rzFst{\rzApp{\ac}{\ab}} = 0$.
Then let $\myaa(\om) = (\rzApp{\ac_2}{\ab_2},\om)$ where $\ab_2 = \ad$, then by (2) $(\myaa,\om) \in \fintR{\psi}$.

Case $(\ad,\nu) \in \apL{\dstrategyforR[\alpha]{\idstrat[\alpha]{\{(\rzApp{\ac}{\ab},\om)\}}{k}}}$.
By (1) had $(\ad,\nu) \in \fintR{\phi}$, so by (2) have $(\rzApp{\ac_2}{\ad},\nu) \in \fintR{\psi}$.
Then by applying (3) $k$ times, have $(\ac_3^k(\rzApp{\ac_2}{\ad}),\om) \in \fintR{\psi}$
So it suffices to let $\myaa = \ac_3^k(\rzApp{\ac_2}{\ad})$.

\mycase $\linferenceRule[formula]
{\proves{\G}{A}{\conv}
& \proves{\pvx:\conv,\pvy:(\met_0 = \met \metgr \metz)}{B}{\ddiamond{\alpha}{(\conv\land \met_0 \metgr \met)}}
& \proves{\pvx:\conv,\pvy:(\met = \metz)}{C}{\phi}}
{\proves{\G}{\efor{A}{B}{C}}{\ddiamond{\prepeat{\alpha}}{\phi}}}$

Assume (1) some $\ac_1$ such that for $(\ab,\om) \in \cintR{\G}$ have $(\rzApp{\ac}{\ab},\om) \in \fintR{\conv}$ by $\drv_1$.
Assume (2) some $\ac_2$ such that for $(\ab_2,\nu) \in \cintR{\pvx:\conv,\pvy:(\met_0 = \met \metgr \metz)}$ have $(\rzApp{\ac_2}{\ab_2},\nu) \in \fintR{\ddiamond{\alpha}{(\conv\land \met_0 \metgr \met)}}$ by $\drv_2$.
Assume (3) some $\ac_3$ such that for $(\ab_3,\nu) \in \cintR{\pvx:\conv,\pvy:(\met=\metz)}$ have $(\rzApp{\ac_3}{\ab_3},\nu) \in \fintR{\phi}$ by $\drv_3$.
Assume (4) some $(\ab,\om) \in \cintR{\G}$, so by (1) have $(\rzApp{\ac}{\ab},\om) \in \fintR{\conv}$.

Begin by showing $(\myaa,\om) \in \fintR{\ddiamond{\prepeat{\alpha}}{(\met = \metz \land \conv)}},$ that is
$\strategyforR[\prepeat{\alpha}]{(\myaa,\om)} \subseteq \fintR{(\met = \metz \land \conv)} \cup \{\top\},$ that is
$\bigcup_{k\in\mathbb{N}} \apL{\istrat[\alpha]{(\myaa,\om)}{k}} \subseteq \fintR{(\met = \metz \land \conv)} \cup \{\top\}$
We define $\myaa$ recursively by $\myaa(\om) = \textsf{if}(\met \metgr \metz)\{(\myaa \circ \rzApp{\ac_2}{\ab_2})(\om)\}\textsf{else}\{\rzCons{\rzNil}{\rzApp{\ac}{\ab}}\}$ where $\ab_2 = \rzCons{\rzApp{\ac}{\ab}}{\rzNil}$.
The recursion is well-founded by the fact that every application of $\rzApp{\ac_2}{\ab_2}$ reduces the metric $\met,$ which is N\"otherian.

Now fix arbitrary $k\in\mathbb{N}$ and $(\ad,\nu) \in \apL{\istrat[\alpha]{(\myaa,\om)}{k}},$ then suffices to show $(\ad,\nu) \in \fintR{(\met = \metz \land \conv)} \cup \{\top\}$.
Case 0, then $(\ad,\nu) = (\rzCons{\rzNil}{\rzApp{\ac}{\ab}},\om)$ and since the test $\met \metgr \metz$ was negative, thus $(\rzNil,\om) \in \fintR{\met = \metz}$ and by (1)
$(\rzCons{\rzNil}{\rzApp{\ac}{\ab}},\om) \in \fintR{\met = \metz \land \conv} \cup \{\top\}$ as desired.

Case $1+k,$ then $(\ad,\nu) \in \apL{\istrat[\alpha]{\strategyforR[\alpha]{\{(\myaa,\om)\}}}{k}}$.
Have $\myaa(\om) = (\myaa \circ \rzApp{\ac_2}{\ab_2})(\om)$ and since the test $\met \metgr \metz$ was positive, thus $(\rzNil,\om) \in \fintR{\met \metgr \metz}$ and by (1) have
$(\ab_2,\om) \in \fintR{\met \metgr \metz \land \conv}$ so by
(2) have  $(\rzApp{\ac_2}{\ab_2},\om) \in \fintR{\ddiamond{\alpha}{(\conv\land \met_0 \metgr \met)}} \cup \{\top\}$
For all $(\ad_1,\nu_1) \in \strategyforR[\alpha]{\{(\myaa,\om)\}}$ have $(\ad_1,\nu_1) \in \fintR{(\conv\land \met_0 \metgr \met)} \cup \{\top\}$.
In case $\top,$ we are already finished, else $\ad_1=\myaa$ recursively, and applying the IH on smaller $\met$ and $k$ have $\apL{\istrat[\alpha]{(\myaa,\nu_1)}{k}} \subseteq \fintR{\met = \metz \land \conv} \cup \{\top\}$ as desired.

Now (by commuting $\met = \metz$ and $\conv$) $\drv_3$ yields $\fintR{\met = \metz \land \conv} \subseteq \fintR{\phi}$ so that finally
$(\myaa',\om) \in \fintR{\ddiamond{\prepeat{\alpha}}{\phi}}$ for some $\myaa'$.
\end{proof}

\subsection{Proof Theory}
The following theorems concern proof terms, normal forms, progress, and preservation.

\begin{lemma}[Normal forms]
\label{lem:app-normal}
If $\proves{\G}{M}{\rho}$ and $\isnorm{M},$ then either $M$ is $\ecase{A}{B}{C}$ for simple $A$ and normal $B$ and $C,$ else $M$ starts with the canonical introduction rule(s) for the top-level connective of $\rho$ and all subterms are well-typed and simple, except those under binders.

For example, $M$ has form:
\begin{itemize}
\item {$\eplam{\phi}{N}$ if $\rho \equiv \phi \limply \psi$}
\item {$\econs{A}{B}$ if $\rho \equiv \ddiamond{\ptest{\phi}}{\psi}$}
\item {$\einjL{N}$ or $\einjR{N}$ if $\rho \equiv \ddiamond{\alpha\cup\beta}{\phi}$}
\item {$\ego{N}$ or $\estop{N}$ if $\rho \equiv \ddiamond{\prepeat{\alpha}}{\phi}$}
\item {$\ebroll{N}$ if $\rho \equiv \dbox{\prepeat{\alpha}}{\phi}$}
\end{itemize}
\end{lemma}

\begin{proof}
\end{proof}

\begin{lemma}[Progress]
If $\proves{\Gemp}{M}{\phi}$ then either $\isnorm{M}$ or exists $M'$ such that $M \stepsto M'$.
\end{lemma}
\begin{proof}

\mycase \irref{brandomI}, \irref{drandomI}, \irref{btestI}, or \irref{asgnI}:
In these cases, $M$ is already normal by definition.


\mycase \irref{dtestI}, term $\edcons{M}{N}:$
By first IH1, $\isnorm{M}$ or $M \stepsto M'$.
If $M \stepsto M'$ then $\edcons{M}{N} \stepsto \edcons{M'}{N}$ by \irref{dconsSL}.
Else, by IH2, $\isnorm{N}$ or $N \stepsto N'$.
If $N \stepsto N'$ then $\edcons{M}{N} \stepsto \edcons{M}{N'}$ by \irref{dconsSR}.
Else, $\isnorm{M}$ and $\isnorm{N}$.
If $M \equiv \ecase{A}{B}{C}$ then $\edcons{\ecase{A}{B}{C}}{N} \stepsto \ecase{A}{\edcons{B}{N}}{\edcons{C}{N}}$ by \irref{dconsCL}.
Else if $N \equiv \ecase{A}{B}{C}$ then $\edcons{M}{\ecase{A}{B}{C}} \stepsto \ecase{A}{\edcons{M}{B}}{\edcons{M}{C}}$ by \irref{dconsCR}.
Else $\issimp{M}$ and $\issimp{N}$ so $\issimp{\econs{M}{N}}$ and $\isnorm{\econs{M}{N}}$.

\mycase \irref{bchoiceI} is symmetric.

\mycase \irref{dchoiceIL}, term $\einjL{M}$:
By IH, either $M \stepsto M'$ or $\isnorm{M}$.
If $M \stepsto M',$ then $\einjL{M} \stepsto \einjL{M'}$ by \irref{injLS}.
Else, either $\issimp{M}$ or $M \equiv \ecase{A}{B}{C}$.
If $M \equiv \ecase{A}{B}{C}$ then $\einjL{\ecase{A}{B}{C}} \stepsto \ecase{A}{\einjL{B}}{\einjL{C}}$ by \irref{injLC}.
Else $\issimp{M}$ so $\issimp{\einjL{M}}$.
Diamond case $\estop{M}$ is symmetric.

\mycase \irref{dchoiceIR}, term $\einjR{M}$:
By IH, either $M \stepsto M'$ or $\isnorm{M}$.
If $M \stepsto M',$ then $\einjR{M} \stepsto \einjR{M'}$ by \irref{injRS}.
Else, either $\issimp{M}$ or $M \equiv \ecase{A}{B}{C}$.
If $M \equiv \ecase{A}{B}{C}$ then $\einjR{\ecase{A}{B}{C}} \stepsto \ecase{A}{\einjR{B}}{\einjR{C}}$ by \irref{injRC}.
Else $\issimp{M}$ so $\issimp{\einjR{M}}$.
Diamond case $\ego{M}$ is symmetric.

\mycase \irref{seqI} (both box and diamond): 
By IH, either $M \stepsto M'$ or $\isnorm{M}$.
If $M \stepsto M',$ then $\eSeq{M} \stepsto \eSeq{M'}$ by \irref{bseqS} or \irref{dseqS}.
Else, either $\issimp{M}$ or $M \equiv \ecase{A}{B}{C}$.
If $M \equiv \ecase{A}{B}{C}$ then $\eSeq{\ecase{A}{B}{C}} \stepsto \ecase{A}{\eSeq{B}}{\eSeq{C}}$ by \irref{bseqC} or \irref{dseqC}.
Else $\issimp{M}$ so $\eSeq{M}$ is normal.

\mycase \irref{dualI} (both box and diamond):
By IH, either $M \stepsto M'$ or $\isnorm{M}$.
If $M \stepsto M',$ then $\eSwap{M} \stepsto \eSwap{M'}$ by  \irref{bswapS} or \irref{dswapS}.
Else, either $\issimp{M}$ or $M \equiv \ecase{A}{B}{C}$.
If $M \equiv \ecase{A}{B}{C}$ then $\eSeq{\ecase{A}{B}{C}} \stepsto \ecase{A}{\eSwap{B}}{\eSwap{C}}$ by \irref{bswapC} or \irref{dswapC}.
Else $\issimp{M}$ so $\eSwap{M}$ is normal.


\mycase \irref{dloopI}, term $\efor{A}{B}{C}$.
Abbreviate term $M = \emon{\esub{B}{\pvx,\pvy}{A,\edcons{\pvr}{\textit{\pvrr}}}}
                      {\efor{\eprojL{\pvz}}{B}{C}}{\pvz}$.
By IH, either $A \stepsto A'$ or $\isnorm{A}$.
If $A \stepsto A'$ $\efor{A}{B}{C} \stepsto \efor{A'}{B}{C}$ by \irref{forS}.
else $\isnorm{A}$ so $\issimp{A}$ or $A \equiv \ecase{D}{E}{F}$.
If $\issimp{A}$ then
$\efor{A}{B}{C}$ $\stepsto$ $\ecaseHead{\esplit{\met}{0}}$ $\ecaseLeft{\pvl}{
\estop{\esub{C}{(\pvx,\pvy)}{(A,\pvl)}}
}$
$\ecaseRight{\pvr}
   {\eghost{\met_0}{\met}{\textit{\pvrr}}
           {\ego{M}}}
\ecaseEnd$ by \irref{forBeta}.
Otherwise we have $\efor{\ecase{D}{E}{F}}{B}{C} \stepsto \ecase{D}{\efor{E}{B}{C}}{\efor{F}{B}{C}}$ by \irref{forC}.

\mycase \irref{drcase}, term $\efp{A}{B}{C}$.
By IH, either $A \stepsto A'$ or $\isnorm{A}$.
If $A \stepsto A'$ then $\efp{A}{B}{C} \stepsto \efp{A'}{B}{C}$ by \irref{fpS}.
Else, since $\isnorm{A}$ then $A \equiv \ecase{D}{E}{F}$ or $\issimp{A}$.
If $A \equiv \ecase{D}{E}{F}$ then $\efp{\ecase{D}{E}{F}}{B}{C} \stepsto \ecase{D}{\efp{E}{B}{C}}{\efp{F}{B}{C}}$ by \irref{fpC}.
If $\issimp{A}$ then 
$\efp{A}{B}{C} \stepsto (\ercase{A}{B}{\esub{C}{\pvy}{\emon{\pvr}{\efp{\pvz}{B}{C}}{\pvz}}})$ by \irref{fpBeta}.

\mycase \irref{btestE}, term $\eapp{M}{N}$.
By IH, either $M \stepsto M'$ or $\isnorm{M}$.
If $M \stepsto M'$ then $\eapp{M}{N} \stepsto \eapp{M'}{N}$ by \irref{appSL}.
Else if $M \equiv \ecase{A}{B}{C}$ then $\eapp{\ecase{A}{B}{C}}{N} \stepsto \ecase{A}{\eapp{B}{N}}{\eapp{C}{N}}$ by \irref{appCL}.
Else $\issimp{M}$. By IH, either $N \stepsto N'$ or $\isnorm{N}$.
If $N \stepsto N'$ then $\eapp{M}{N} \stepsto \eapp{M}{N'}$ by \irref{appSR}.
Else if $N \equiv \ecase{A}{B}{C}$ then $\eapp{M}{\ecase{A}{B}{C}} \stepsto \ecase{A}{\eapp{M}{B}}{\eapp{M}{C}}$ by \irref{appCR}.
Else $\issimp{N}$.
By \rref{lem:app-normal} $M = (\eplam{\phi}{D})$ then $\eapp{(\eplam{\phi}{D})}{N} \stepsto \esub{D}{\pvx}{N}$ by \irref{appBeta}.

\mycase \irref{brandomE}, term $\eapp{M}{f}$.
By IH, $M \stepsto M'$ or $\isnorm{M}$.
If $M \stepsto M'$ then $\eapp{M}{f} \stepsto \eapp{M'}{f}$ by \irref{bunrollS}.
Else if $M \equiv \ecase{A}{B}{C}$ then $\eapp{\ecase{A}{B}{C}}{f} \stepsto \ecase{A}{\eapp{B}{f}}{\eapp{C}{f}}$ by \irref{brandomC}.
Else $\issimp{M}$ and by \rref{lem:app-normal} then $M \equiv (\etlam{\allrat}{D})$ so $\eapp{(\etlam{\allrat}{D})}{f} \stepsto \tsub{D}{x}{f}$ by \irref{brandomBeta}.

\mycase \irref{dtestEL}, term $\eprojL{M}$.
By IH, $M \stepsto M'$ or $\isnorm{M}$.
If $M \stepsto M'$ then $\eprojL{M} \stepsto \eprojL{M'}$ by \irref{projLS}.
Else if $M \equiv \ecase{A}{B}{C}$ then $\eprojL{\ecase{A}{B}{C}} \stepsto \ecase{A}{\eprojL{B}}{\eprojL{C}}$ by \irref{projLC}.
Else $\issimp{M}$ so by \rref{lem:app-normal} have $M \equiv \econs{A}{B}$ so $\eprojL{\econs{A}{B}} \stepsto A$ by \irref{projLBeta}.

\mycase \irref{dtestER}, term $\eprojR{M}$.
By IH, $M \stepsto M'$ or $\isnorm{M}$.
If $M \stepsto M'$ then $\eprojR{M} \stepsto \eprojR{M'}$ by \irref{projRS}.
Else if $M \equiv \ecase{A}{B}{C}$ then $\eprojR{\ecase{A}{B}{C}} \stepsto \ecase{A}{\eprojR{B}}{\eprojR{C}}$ by \irref{projRC}.
Else $\issimp{M}$ so by \rref{lem:app-normal} have $M \equiv \econs{A}{B}$ so $\eprojR{\econs{A}{B}} \stepsto B$ by \irref{projRBeta}.

\mycase \irref{bunroll}, term $\ebunroll{M}$.
By IH, $M \stepsto M'$ or $\isnorm{M}$.
If $M \stepsto M'$ then $\ebunroll{M} \stepsto \ebunroll{M'}$ by \irref{bunrollS}.
Else if $M \equiv \ecase{A}{B}{C}$ then $\ebunroll{\ecase{A}{B}{C}} \stepsto \ecase{A}{\ebunroll{B}}{\ebunroll{C}}$ by \irref{bunrollC}.
Else $\issimp{M}$ so by \rref{lem:app-normal} have $M \equiv \ebroll{A}$ so $\ebunroll{\ebroll{A}} \stepsto A$ by \irref{bunrollBeta}.

\mycase \irref{drcase}, term $\ercase{A}{B}{C}$.
By IH, $A \stepsto A'$ or $\isnorm{A}$.
If $A \stepsto A'$ then $\ercase{A}{B}{C} \stepsto \ercase{A'}{B}{C}$ by \irref{caseS}.
Else if $M \equiv \ecase{A}{B}{C}$ then $\ercase{\ecase{A}{B}{C}}{D}{E} \stepsto \ecase{A}{\ercase{B}{D}{E}}{\ercase{C}{D}{E}}$ by \irref{caseC}.
Else $\issimp{M}$ so by \rref{lem:app-normal} have $M \equiv \estop{A}$ or $M \equiv \ego{A}$ so $\ercase{\estop{A}}{B}{C} \stepsto \esub{B}{\ell}{A}$ and
$\ercase{\ego{A}}{B}{C} \stepsto \esub{C}{r}{A}$ by \irref{caseBetaL} and \irref{caseBetaR}.

\mycase \irref{bchoiceI}, term $\ebcons{M}{N}$.
By IH, $M \stepsto M'$ or $\isnorm{M}$.
If $M \stepsto M'$ then $\ebcons{M}{N} \stepsto \ebcons{M'}{N}$ by \irref{bconsSL}.
Else if $M \equiv \ecase{A}{B}{C}$ then $\ebcons{\ecase{A}{B}{C}}{N} \stepsto \ecase{A}{\ebcons{B}{N}}{\ebcons{C}{N}}$ by \irref{bconsCL}.
Else $\issimp{M}$.
By IH, $N \stepsto N'$ or $\isnorm{N}$.
If $N \stepsto N'$ then $\ebcons{M}{N} \stepsto \ebcons{M}{N'}$ by \irref{bconsSR}.
Else if $N \equiv \ecase{A}{B}{C}$ then $\ebcons{M}{\ecase{A}{B}{C}} \stepsto \ecase{A}{\ebcons{M}{B}}{\ebcons{M}{C}}$ by \irref{bconsCR}.
Else $\issimp{N}$, so $\isnorm{\ebcons{M}{N}}$.

\mycase \irref{dchoiceE}, term $\ecase{A}{B}{C}$.
By IH, $A \stepsto A'$ or $\isnorm{A}$.
If $A \stepsto A'$ then $\ecase{A}{B}{C} \stepsto \ecase{A'}{B}{C}$ by \irref{caseS}.
Else $\isnorm{A}$.
If $A \equiv \ecase{D}{E}{F}$ then $\ecase{\ecase{D}{E}{F}}{B}{C} \stepsto \ecase{D}{\ecase{E}{B}{C}}{\ecase{F}{B}{C}}$ by \irref{caseC}.
Else $\issimp{A}$ so by \rref{lem:app-normal}, $A \equiv \einjL{M}$ or $A \equiv \einjR{M}$.
In the first case, $\ecase{\einjL{M}}{B}{C} \stepsto \esub{B}{\pvl}{M}$ by \irref{caseBetaL} else  $\ecase{\einjR{M}}{B}{C} \stepsto \esub{C}{\pvr}{M}$ by \irref{caseBetaR}.

\mycase \irref{bloopI}, term $(\erep{M}{N}{\pvx:J}{O})$
By IH, $M \stepsto M'$ or $\isnorm{M}$.
If $M \stepsto M'$ then $(\erep{M}{N}{\pvx:J}{O}) \stepsto (\erep{M'}{N}{\pvx:J}{O})$ by \irref{repS}.
Else if $M \equiv \ecase{A}{B}{C}$ then $\erep{\ecase{A}{B}{C}}{N}{\pvx:J}{O} \stepsto \ecase{A}{(\erep{B}{N}{\pvx:J}{O})}{(\erep{C}{N}{\pvx:J}{O})}$ by \irref{repC}.
Else $\issimp{M}$ so $(\erep{M}{N}{\pvx:J}{O}) \stepsto \ebroll{\edcons{M}{\emon{(\esub{N}{\pvx}{M})}{(\erep{\pvy}{N}{\pvx:J}{O})}{\pvy}}}$ by \irref{repBeta}.

 \mycase \irref{drandomE}, term $\eunpack{M}{N}$.
 By IH, $M \stepsto M'$ or $\isnorm{M}$.
 If $M \stepsto M'$ then $\eunpack{M}{N} \stepsto \eunpack{M}{N}$ by \irref{unpackS}.
 Else if $M \equiv \ecase{A}{B}{C}$ then $\eunpack{\ecase{A}{B}{C}}{N} \stepsto \ecase{A}{\eunpack{B}{N}}{\eunpack{C}{N}}$ by \irref{unpackC}.
 Else $\issimp{M}$ and by \rref{lem:app-normal} then $M \equiv \etcons{f}{M'}$ so
$\eunpack{\etconsgen{x}{y}{\pvy}{f}{M}}{N} \stepsto \eren{(\eghost{x}{\eren{f}{x}{y}}{\pvy}{\esub{N}{\pvx}{M}})}{y}{x}$ by \irref{unpackBeta}.

\mycase \irref{hyp}, term $x$.
By inversion on $\proves{\Gemp}{x}{\phi}$ then $x \in \Gemp,$ so the case holds by absurdity.

\mycase \irref{mon}, term $\emon{M}{N}{\pvx}$.
By IH, $M \stepsto M'$ or $\isnorm{M}$.
If $M \stepsto M'$ then $\emon{M}{N}{\pvx} \stepsto \emon{M'}{N}{\pvx}$ by \irref{monS}.
else if $M \equiv \ecase{A}{B}{C}$ then $\emon{\ecase{A}{B}{C}}{N}{\pvx} \stepsto \ecase{A}{\emon{B}{N}{\pvx}}{\emon{C}{N}{\pvx}}$ by \irref{monC}.
Else $\isnorm{M}$.
We proceed by cases on $M$. By \rref{lem:app-normal}, it suffices to show the cases for introduction forms.
The termination argument is by lexicographic induction on the postcondition formula  $\phi$ and on the derivation $M$.

\mycase \irref{dtestI}, term $(\eplam{\phi}{M})$: Then by \irref{lamMon} have $\emon{(\eplam{\phi}{M})}{N}{\pvx} \stepsto (\eplam{\phi}{\emon{M}{N}{\pvx}})$

\mycase \irref{brandomI}, term $(\etlam{\allrat}{M})$: Then by \irref{rlamMon} have $\emon{(\etlam{\allrat}{M})}{N}{\pvx} \stepsto (\etlam{\allrat}{(\emon{M}{N}{\pvx})})$

\mycase \irref{bchoiceI}, term $\ebcons{A}{B}$: Then by \irref{bconsMon} have $\emon{\ebcons{A}{B}}{N}{\pvx}\stepsto\ebcons{\emon{A}{N}{\pvx}}{\emon{B}{N}{\pvx}}$.

\mycase \irref{dtestI}, term $\edcons{A}{B}$: Then by \irref{dconsMon} have $\emon{\edcons{A}{B}}{N}{\pvx}\stepsto\edcons{A}{\esub{N}{\pvx}{B}}$

\mycase \irref{drandomI}, term $\etcons{f}{M}$: Then by \irref{tconsMon} $\emon{\etcons{f}{M}}{N}{\pvx}\stepsto \etcons{f}{\emon{M}{\esub{N}{f}{\pvy}}{\pvx}}$

\mycase \irref{dualI} (diamond), term $\edswap{M}$: Then by \irref{dswapMon} have $\emon{\edswap{M}}{N}{\pvx} \stepsto \edswap{\emon{M}{N}{\pvx}}$

\mycase \irref{dualI} (box), term $\ebswap{M}$: Then by \irref{bswapMon} have $\emon{\ebswap{M}}{N}{\pvx} \stepsto \ebswap{\emon{M}{N}{\pvx}}$

\mycase \irref{dchoiceIL}, term $\einjL{M}$: Then by \irref{injLMon} have $\emon{(\einjL{M})}{N}{\pvx} \stepsto \einjL{(\emon{M}{N}{\pvx})}$

\mycase \irref{dchoiceIR}, term $\einjR{M}$: Then by \irref{injRMon} have $\emon{(\einjR{M})}{N}{\pvx} \stepsto \einjR{(\emon{M}{N}{\pvx})}$

\mycase \irref{broll}, term $\ebroll{M}$: Then by \irref{brollMon} have
$\emon{\ebroll{M}}{N}{\pvx}  \stepsto \ebroll{\edcons{\emon{\eprojL{M}}{(\eren{N}{\vec{y}}{\vec{x}})}{\pvx}}{\emon{\eprojR{M}}{\emon{\pvz}{N}{\pvx}}{\pvz}}}$

\mycase \irref{dstop}, term $\estop{M}$: 
Then by \irref{injLMon} have $\emon{\estop{M}}{N}{\pvx} \stepsto \estop{\emon{M}{N}{\pvx}}$

\mycase \irref{dgo}, term $\ego{M}$:
Then by \irref{injRMon} have $\emon{\ego{M}}{N}{\pvx} \stepsto \ego{\emon{M}{N}{\pvx}}$

\mycase \irref{asgnI} (diamond), term $\edasgneq{y}{x}{\pvx}{M}$: 
Then by \irref{dasgnMon} have $\emon{\edasgneq{x}{y}{\pvx}{M}}{N}{\pvx} \stepsto \edasgneq{x}{y}{\pvx}{\esub{N}{\pvx}{M}}$

\mycase \irref{asgnI} (box), term $\ebasgneq{y}{x}{\pvx}{M}$: 
Then by \irref{basgnMon} have $\emon{\ebasgneq{x}{y}{\pvx}{M}}{N}{\pvx} \stepsto \ebasgneq{x}{y}{\pvx}{\esub{N}{\pvx}{M}}$

\mycase \irref{seqI} (box), term $\ebseq{M}$: Then by \irref{bseqMon} have $\emon{(\ebseq{M})}{N}{\pvx} \stepsto \ebseq{(\emon{M}{\emon{\pvy}{\esub{N}{\pvx}{\pvz}}{\pvz}}{\pvy})}$

\mycase \irref{seqI} (diamond), term $\edseq{M}$: Then by \irref{dseqMon} have $\emon{(\edseq{M})}{N}{\pvx} \stepsto \edseq{(\emon{M}{\emon{\pvy}{\esub{N}{\pvx}{\pvz}}{\pvz}}{\pvy})}$


\end{proof}

\begin{lemma}[Structurality]
Rules \irref{weak}, \irref{exchange}, and \irref{contract} are admissible.

\begin{calculuscollections}{\textwidth}
\begin{calculus}
\cinferenceRule[weak|W]{$x$ fresh}
{\linferenceRule[formula]{\proves{\G}{M}{\phi}}{\proves{\G,\pvx:\psi}{M}{\phi}}}
{}
\cinferenceRule[exchange|X]{}
{\linferenceRule[formula]{\proves{\G,\pvx:\phi,\pvy:\psi}{M}{\rho}}{\proves{\G,\pvy:\psi,\pvx:\phi}{M}{\rho}}}
{}
\cinferenceRule[contract|C]{}
{\linferenceRule[formula]{\proves{\G,\pvx:\phi,\pvy:\phi}{M}{\rho}}{\proves{\G,\pvx:\phi}{\esub{M}{\pvy}{\pvx}}{\rho}}}
{}
\end{calculus}
\end{calculuscollections}
\end{lemma}
\begin{proof}
To show each rule admissible, prove that when the premisses are derivable, the conclusion is derivable.
Each proof is by induction on the proof term $M$.

Observe that the only premisses regarding $\G$ are of the form $\G(x) = \phi$.
Such premisses are preserved when extending $\G$ with a fresh proof variable $\pvx:\phi$, as needed by weakening.
Premisses are also preserved under exchange, because contexts are understood as sets to begin with, i.e., the premiss $\G(\pvx)=\phi$ is ignorant of the position of $x$.
Premisses are preserved modulo renaming under contraction because the assumption $\G(\pvx)=\phi$ is as good as the assumption $\G(\pvy)=\phi$.

We note briefly that the inductive hypothesis can always be applied when needed.
The context $\G$ is kept universally quantified in the induction over $M$.
When proving $\proves{\G}{M}{\phi},$ we sometimes have assumptions in an empty or singleton context; 
for weakening we need not apply the IH at all in these cases, and for exchange and contraction we simply vary the context as we apply the inductive hypothesis.
We also vary the context in the inductive hypothesis of the assignment rule.
\end{proof}

\begin{lemma}[Proof term renaming]
\label{lem:app-ptren}
  If $\proves{\G}{M}{\phi}$ then $\proves{\eren{\G}{x}{y}}{\eren{M}{x}{y}}{\eren{\phi}{x}{y}}$.
  If $x$ or $y$ is fresh in both $M$ and $\phi$ and $\proves{\eren{\G}{x}{y}}{M}{\phi}$ then
$\proves{\G}{M}{\phi}$.
\end{lemma}
\begin{proof}
Induction on $M$ analogous to the proof of \rref{lem:app-ptsub}.
\end{proof}

\begin{lemma}[Proof term-term substitution]
\label{lem:app-pttsub}
  If $\proves{\G}{M}{\phi}$ and $\tsub{M}{x}{f}$ is admissible then $\proves{\tsub{\G}{x}{f}}{\tsub{M}{x}{f}}{\tsub{\phi}{x}{f}}$.
\end{lemma}
\begin{proof}
Induction on $M$ analogous to the proof of \rref{lem:app-ptsub}.
\end{proof}

\begin{lemma}[Proof term substitution]
\label{lem:app-ptsub}
  If $\proves{\G,\pvx:\psi}{M}{\phi}$ and $\proves{\G}{N}{\psi}$ then $\proves{\G}{\esub{M}{\pvx}{N}}{\phi}$.
\end{lemma}
By induction on the derivation $\proves{\G,\pvx:\psi}{M}{\phi}$, generalizing the context.

\mycase \irref{asgnI} (diamond):
In this case to avoid clash let's say the substituted proof variable was $\pvz$ not $\pvx$.
By premiss, have $\proves{\eren{\G}{x}{y},\pvz:\eren{\psi}{x}{y},\pvx:(x=\eren{f}{x}{y})}{M}{\phi}$.
Now regroup the context as $(\esub{\G}{x}{y},\pvx:(x=\eren{f}{x}{y}), \pvz:\eren{\psi}{x}{y})$.
By \rref{lem:app-ptren} on $N$ and weakening have $\proves{\eren{\G}{x}{y},\pvx:(x=\eren{f}{x}{y})}{\eren{N}{x}{y}}{\eren{\psi}{x}{y}}$.
Then we can apply the IH on $M$ but varying the context to $\eren{\G}{x}{y},\pvx:(x=\eren{f}{x}{y})$ and varying $N$ to $\eren{N}{x}{y},$ which yields
$\proves{\eren{\G}{x}{y},\pvx:(x=\eren{f}{x}{y})}{\esub{M}{\pvz}{\eren{N}{x}{y}}}{\phi}$, then by \irref{asgnI} have
$\proves{\G}{\edasgn{y}{x}{\pvx}{\esub{M}{\pvz}{\eren{N}{x}{y}}}}{\dbox{\humod{x}{f}}{\phi}}$.
Note to avoid capture, the substitution case for assignment does a renaming, i.e., $\esub{\edasgn{y}{x}{\pvx}{M}}{\pvz}{N} = \edasgn{y}{x}{\pvx}{\esub{M}{\pvz}{\eren{N}{x}{y}}}$

\mycase \irref{drandomI}:
By premiss have $\proves{\eren{\G}{x}{y},\pvz:\eren{\psi}{x}{y},\pvx:(x=\eren{f}{x}{y})}{M}{\phi}$ for some state-computable computable function $f$.
Now regoup the context as $(\esub{\G}{x}{y},\pvx:(x=\eren{f}{x}{y}),\pvz:\eren{\psi}{x}{y})$.
By \rref{lem:app-ptren} on $N$ and weakening have $\proves{\eren{\G}{x}{y},\pvx:(x=\eren{f}{x}{y})}{\eren{N}{x}{y}}{\eren{\psi}{x}{y}}$.
Then apply IH on $M,$ varying context to $\eren{\G}{x}{y},\pvx:(x=\eren{f}{x}{y})$  and varying $N$ to $\eren{N}{x}{y},$ which yields
$\proves{\eren{\G}{x}{y},\pvx:(x=\eren{f}{x}{y})}{\esub{M}{\pvz}{\eren{N}{x}{y}}}{\phi}$, then by \irref{drandomI} have
$\proves{\G}{\etcons{f}{\esub{M}{\pvz}{N}}}{\ddiamond{\prandom{x}}{\phi}}$.
The case holds since $\esub{\etcons{f}{M}}{\pvz}{N} = \etcons{f}{\esub{M}{\pvz}{N}}$.
Note the formula $(x = \eren{f}{x}{y})$ is more subtle than it first appears, because $f$ is an arbitrary computable function, not strictly an element of the term language, however renaming and formulas generalize just fine from terms to computable functions.

\mycase \irref{drandomE}, postcondition $\rho,$ term $\eunpack{A}{B}$, call the substituted variable $\pvz$.
Assume the side condition that $x \notin \freevars{\rho}$.
By premiss have $\drv_1$ is $\proves{\G,\pvz:\psi}{A}{\ddiamond{\prandom{x}}{\phi}}$ and
$\drv_2$ is $\proves{\eren{(\G,\pvz:\psi)}{x}{z},\pvx:\phi}{B}{\rho}$.
By IH on $\drv_1$ have $\proves{\G}{\esub{A}{\pvz}{\psi}}{\ddiamond{\prandom{x}}{\phi}}$.
In $\drv_2$ reorder the context as $(\eren{\G}{x}{z},\pvx:\phi),\pvz:\eren{\psi}{x}{z},$ then observe by \rref{lem:app-ptren} have $\proves{\eren{\G}{x}{z}}{\eren{N}{x}{z}}{\eren{\psi}{x}{z}}$ and by weakening, $\proves{\eren{\G}{x}{z},\pvx:\phi}{\eren{N}{x}{z}}{\eren{\psi}{x}{z}}$ so by IH on $\drv_2$ have
$\proves{\eren{\G}{x}{z},\pvx:\phi}{\esub{B}{\pvz}{\eren{N}{x}{z}}}{\rho},$ so by \irref{drandomE} have
$\proves{\G}{\eunpack{\esub{A}{\pvz}{N}}{\esub{B}{\pvz}{\eren{N}{x}{z}}}}{\rho}$.
Note then that this case of substitution renames for $N$'s within $B,$ i.e., $\esub{\eunpack{A}{B}}{\pvz}{N} =\eunpack{\esub{A}{\pvz}{N}}{\esub{B}{\pvz}{\eren{N}{x}{z}}}$ as desired.

\mycase \irref{dloopI}, variable $\pvz$:
By the premisses, 
have $\drv_1$ is $\proves{\G,\pvz:\psi}{A}{\conv}$ and 
have $\drv_2$ is $\proves{\pvx:\conv,\pvy:(\met_0 = \met \metgr \metz)}{B}{\ddiamond{\alpha}{(\conv \land \met_0 \metgr )}}$ for fresh $\met_0$ and 
have $\drv_3$ is $\proves{\pvx:\conv,\pvy:(\met=\metz)}{C}{\phi}$.
By IH on $\drv_1$ have $\proves{\G}{\esub{A}{\pvz}{N}}{\conv},$ then reusing $\drv_2$ as is, apply \irref{dloopI} giving
$\proves{\G}{\efor{\esub{A}{\pvz}{N}}{B}{C}}{\ddiamond{\prepeat{\alpha}}{\phi}}$.
This completes the case since $\esub{(\efor{A}{B}{C})}{\pvz}{N} = \efor{\esub{A}{\pvz}{N}}{B}{C}$.
Note substitution does not substitute in the $B$ or $C$ branches, which cannot access $\pvz$ to begin with.

\mycase \irref{dloopE}, postcondition $\rho$:
By the premisses, have $\drv_1$ is $\proves{\G,\pvz:\psi}{A}{\ddiamond{\prepeat{\alpha}}{\phi}}$,
$\drv_2$ is $\proves{\pvx:\phi}{B}{\rho}$, and
$\drv_3$ is $\proves{\pvy:\ddiamond{\alpha}{\rho}}{C}{\rho}$.
By IH on $\drv_1$ have $\proves{\G}{\esub{A}{\pvz}{N}}{\ddiamond{\prepeat{\alpha}}{\phi}}$, then apply \irref{dloopE} reusing $\drv_2$ and $\drv_3$ as-is: yielding
$\proves{\G}{\efp{\esub{A}{\pvz}{N}}{B}{C}}{\rho}$, which completes the case since $\esub{\efp{A}{B}{C}}{\pvz}{N} = \efp{\esub{A}{\pvz}{N}}{B}{C}$.
Note that it does not substitute in the $B$ and $C$ branches because they cannot access $\pvz$ in the first place.

\mycase \irref{dstop}, postcondition $\ddiamond{\prepeat{\alpha}}{\phi}$:
By premiss, $\proves{\G,\pvx:\psi}{M}{\phi}$.
By IH, $\proves{\G}{\esub{M}{\pvx}{N}}{\phi}$ and by \irref{dstop},
$\proves{\G}{\estop{\esub{M}{\pvx}{N}}}{\ddiamond{\prepeat{\alpha}}{\phi}}$, 
then the case holds because $\esub{\estop{M}}{\pvx}{N}= \estop{\esub{M}{\pvx}{N}}$.

\mycase \irref{dgo}, postcondition $\ddiamond{\prepeat{\alpha}}{\phi}$:
By premiss, $\proves{\G,\pvx:\psi}{M}{\phi}$.
By IH, $\proves{\G}{\esub{M}{\pvx}{N}}{\ddiamond{\alpha}{\ddiamond{\prepeat{\alpha}}{\phi}}}$ and by \irref{dgo},
$\proves{\G}{\ego{\esub{M}{\pvx}{N}}}{\ddiamond{\prepeat{\alpha}}{\phi}}$, 
then the case holds because $\esub{\ego{M}}{\pvx}{N}= \ego{\esub{M}{\pvx}{N}}$.

\mycase \irref{drcase}, postcondition $\rho,$ assume without loss of generality $\pvx$ has been uniformly renamed so $\pvx \notin \{\pvs,\pvg\}$.
By $\drv_1$, $\proves{\G,\pvx:\psi}{A}{\ddiamond{\prepeat{\alpha}}{\phi}}$.
By $\drv_2$, $\proves{\G,\pvx:\psi,\pvs:\phi}{B}{\rho}$.
By $\drv_3$, $\proves{\G,\pvx:\psi,\pvg:\ddiamond{\alpha}{\ddiamond{\prepeat{\alpha}}{\phi}}}{C}{\rho}$.
By weakening on $N$ and IH1, $\proves{\G}{\esub{A}{\pvx}{N}}{\ddiamond{\prepeat{\alpha}}{\phi}}$,
by weakening on $N$ and IH2, $\proves{\G,\pvs:\psi}{\esub{B}{\pvx}{N}}{\rho}$,
and by IH3 $\proves{\G,\pvg:\ddiamond{\alpha}{\ddiamond{\prepeat{\alpha}}{\psi}}}{\esub{C}{\pvx}{N}}{\rho}$,
then by \irref{drcase},
$\proves{\G}{\ercase{\esub{A}{\pvx}{N}}{\esub{B}{\pvx}{N}}{\esub{C}{\pvx}{N}}}{\rho}$.

\mycase \irref{dchoiceIL}, postcondition $\ddiamond{\alpha\cup\beta}{\phi}$:
By premiss, $\proves{\G,\pvx:\psi}{M}{\ddiamond{\alpha}{\phi}}$.
By IH, $\proves{\G}{\esub{M}{\pvx}{N}}{\ddiamond{\alpha}{\phi}}$, and by \irref{dchoiceIL}
$\proves{\G}{\edinjL{\esub{M}{\pvx}{N}}}{\ddiamond{\alpha\cup\beta}{\phi}}$,
then the case holds because $\esub{\edinjL{M}}{\pvx}{N} = \edinjL{\esub{M}{\pvx}{N}}$.

\mycase \irref{dchoiceIR}, postcondition $\ddiamond{\alpha\cup\beta}{\phi}$:
By premiss, $\proves{\G,\pvx:\psi}{M}{\ddiamond{\beta}{\phi}}$.
By IH, $\proves{\G}{\esub{M}{\pvx}{N}}{\ddiamond{\beta}{\phi}}$, and by \irref{dchoiceIR}
$\proves{\G}{\edinjR{\esub{M}{\pvx}{N}}}{\ddiamond{\alpha\cup\beta}{\phi}}$,
then the case holds because $\esub{\edinjR{M}}{\pvx}{N} = \edinjR{\esub{M}{\pvx}{N}}$.

\mycase \irref{dchoiceE}, postcondition $\phi,$ term $\ecase{A}{B}{C}$:
By premisses, $\drv_1$ is $\G,\pvx:\psi \vdash A :\ddiamond{\alpha\cup\beta}{\rho}$, 
$\drv_2$ is $\G,\pvx:\psi,\ell:\ddiamond{\alpha}{\rho} \vdash B:\phi$, and
$\drv_3$ is $\G,\pvx:\psi,r:\ddiamond{\beta}{\rho}\vdash C:\phi$.
By weakening, $\G,\pvl:\ddiamond{\alpha}{\rho}\vdash N:\psi$ and $\G,\pvr:\ddiamond{\beta}{\rho}\vdash N:\psi$.
By IH then $\G\vdash \esub{A}{\pvx}{N}:\ddiamond{\alpha\cup\beta}{\rho}$
and $\G,\pvl:\ddiamond{\alpha}{\rho} \vdash \esub{B}{\pvx}{N}:\phi$
and $\G,\pvr:\ddiamond{\beta}{\rho} \vdash \esub{C}{\pvx}{N}: \phi$.
Then by \irref{dchoiceE}, have $\G \vdash \ecase{\esub{A}{\pvx}{N}}{\esub{B}{\pvx}{N}}{\esub{C}{\pvx}{N}} : \psi$,
then since $\pvx \notin \{\pvl,\pvr\}$ (by renaming if necessary) have $\esub{\ecase{A}{B}{C}}{\pvx}{N} = \ecase{\esub{A}{\pvx}{N}}{\esub{B}{\pvx}{N}}{\esub{C}{\pvx}{N}}$ which completes the case.

\mycase \irref{dtestI}, postcondition $\ddiamond{\ptest{\rho}}{\phi},$ term $\edcons{A}{B}$:
By premisses have $\drv_1$ that $\proves{\G,\pvx:\psi}{A}{\rho}$ and $\drv_2$ that $\proves{\G,\pvx:\psi}{B}{\phi}$.
By IH have $\G \vdash \esub{A}{\pvx}{N}:\rho$ and $\G \vdash\esub{B}{\pvx}{N}:\phi$. 
By  \irref{dtestI}  have $\G \vdash \edcons{\esub{A}{\pvx}{N}}{\esub{B}{\pvx}{N}} : \ddiamond{\ptest{\rho}}{\phi}$.
Then since $\esub{\edcons{A}{B}}{\pvx}{N}$ equals $\edcons{\esub{A}{\pvx}{N}}{\esub{B}{\pvx}{N}}$ this completes the case.

\mycase \irref{dtestEL}, postcondition $\rho,$ term $\eprojL{M}$:
By premiss, have $\proves{\G,\pvx:\psi}{M}{\ddiamond{\ptest{\rho}}{\phi}}$.
By IH, have $\proves{\G}{\esub{M}{\pvx}{N}}{\ddiamond{\ptest{\rho}}{\phi}}$, then by \irref{dtestEL} have
$\proves{\G}{\eprojL{\esub{M}{\pvx}{N}}}{\rho}$, then $\esub{\eprojL{M}}{\pvx}{N} = \eprojL{\esub{M}{\pvx}{N}}$ which completes the case.

\mycase \irref{dtestER}, postcondition $\phi,$ term $\eprojR{M}$:
By premiss, have $\proves{\G,\pvx:\psi}{M}{\ddiamond{\ptest{\rho}}{\phi}}$.
By IH, have $\proves{\G}{\esub{M}{\pvx}{N}}{\ddiamond{\ptest{\rho}}{\phi}}$, then by \irref{dtestER} have
$\proves{\G}{\eprojR{\esub{M}{\pvx}{N}}}{\phi}$, then $\esub{\eprojR{M}}{\pvx}{N} = \eprojR{\esub{M}{\pvx}{N}}$ which completes the case.

\mycase \irref{seqI}:
By premiss, have $\proves{\G,\pvx:\psi}{M}{\dmodality{\alpha}{\dmodality{\beta}{\phi}}}$.
By IH, have $\G \vdash \esub{M}{\pvx}{N} : \dmodality{\alpha}{\dmodality{\beta}{\phi}}$, then by \irref{seqI} have
$\G \vdash \eSeq{\esub{M}{\pvx}{N}}:\dmodality{\alpha;\beta}{\phi}$, 
then $\esub{\eSeq{M}}{\pvx}{N} = \eSeq{(\esub{M}{\pvx}{N})}$ which completes the case.

\mycase \irref{dualI}:
By premiss, have $\proves{\G,\pvx:\psi}{M}{\dmodality{\alpha}{\phi}}$.
By IH, have $\G \vdash \esub{M}{\pvx}{N} : \dmodality{\alpha}{\phi}$, then by \irref{dualI} have
$\G \vdash \eSwap{\esub{M}{\pvx}{N}}:\pmodality{\pdual{\alpha}}{\phi}$, 
then $\esub{\eSwap{M}}{\pvx}{N} = \eSwap{(\esub{M}{\pvx}{N})}$ which completes the case.

\mycase \irref{asgnI} (box):
In this case to avoid confusion let's say the substituted proof variable was $\pvz$ not $\pvx$.
By premiss, have $\proves{\eren{\G}{x}{y},\pvz:\eren{\psi}{x}{y},\pvx:(x=\eren{f}{x}{y})}{M}{\phi}$.
Now regroup the context as $(\eren{\G}{x}{y},\pvx:(x=\eren{f}{x}{y}), \pvz:\eren{\psi}{x}{y})$.
By \rref{lem:app-ptren} and weakening have $\proves{\eren{\G}{x}{y},\pvx:(x=\subst[f]{x}{y})}{\eren{N}{x}{y}}{\subst[\psi]{x}{y}}$.
Then we can apply the IH on $M$ but varying the context to $\subst[\G]{x}{y},\pvx:(x=\eren{f}{x}{y})$ and varying $N$ to $\eren{N}{x}{y},$ which yields
$\proves{\eren{\G}{x}{y},\pvx:(x=\eren{f}{x}{y})}{(\esub{M}{\pvz}{(\eren{N}{x}{y})})}{\phi}$, then by \irref{asgnI} have
$\proves{\G}{\ebasgn{y}{x}{\pvx}{(\esub{M}{\pvz}{(\eren{N}{x}{y})})}}{\dbox{\humod{x}{f}}{\phi}}$.
Note that to avoid capture, this case of substitution renames, i.e., 
$\esub{\ebasgn{y}{x}{\pvx}{M}}{\pvz}{N} 
= \ebasgn{y}{x}{\pvx}{\esub{M}{\pvz}{\eren{N}{x}{y}}}$

\mycase \irref{brandomI}, here call the substituted variable $\pvz$:
By premiss, have $\proves{\sren{(\G,\pvz:\psi)}{x}{y}}{M}{\phi}$.
Regroup the context as $(\sren{\G}{x}{y}),\pvz:\sren{\psi}{x}{y}$.
By \rref{lem:app-ptren} have $\proves{\sren{\G}{x}{y}}{\eren{N}{x}{y}}{\sren{\phi}{x}{y}}$.
Then apply IH for $M$ and varying to $\sren{\G}{x}{y}$ and $\eren{N}{x}{y}$, yielding $\proves{\sren{\G}{x}{y}}{\esub{M}{\pvz}{\eren{N}{x}{y}}}{\phi}$, then
by \irref{brandomI} have $\proves{\G}{(\elam{x}{\allrat}{\esub{M}{\pvz}{\eren{N}{x}{y}}})}{\dbox{\prandom{x}}{\phi}}$.
Note that to avoid capture, this case of substitution renames, i.e.,  $\esub{(\etlam{\allrat}{M})}{\pvz}{N} = \elam{x}{\allrat}{(\esub{M}{\pvz}{(\eren{N}{x}{y})})}$.

\mycase \irref{brandomE}:
By premiss, have $\proves{\G,\pvz:\psi}{M}{\dbox{\prandom{x}}{\phi}}$.
By IH,  have $\proves{\G}{\esub{M}{\pvz}{N}}{\dbox{\prandom{x}}{\phi}}$.
By \irref{brandomE} have $\proves{\G}{\eapp{(\esub{M}{\pvz}{N})}{f}}{\tsub{\phi}{x}{f}}$, then the case holds since $\esub{(\eapp{M}{f})}{\pvz}{N} = \eapp{(\esub{M}{\pvz}{N})}{f}$

\mycase \irref{bloopI}:
By premisses, have $\drv_1$ is $\proves{\G,\pvz:\psi}{A}{J}$ and $\drv_2$ is $\proves{\pvy:J}{B}{\dbox{\alpha}{J}}$ and $\drv_3$ is $\proves{\pvy:J}{C}{\phi}$.
By IH on $\drv_1$ have $\proves{\G}{\esub{A}{\pvz}{N}}{J},$ then applying \irref{bloopI} without changing $\drv_2$ or $\drv_3$, have $\proves{\G}{\erep{(\esub{A}{\pvz}{N})}{B}{\pvy:J}{C}}{\dbox{\prepeat{\alpha}}{\phi}}$.
This concludes the case since $\esub{(\erep{A}{B}{\pvy})}{\pvz}{N} = \erep{(\esub{A}{\pvz}{N})}{B}{\pvy},$ i.e., no substitution in $B$ or $C$ since they cannot access $\pvz$ anyway.

\mycase \irref{broll}, postcondition $\dbox{\prepeat{\alpha}}{\phi}$:
By premiss, $\proves{\G,\pvx:\psi}{M}{\phi \land \dbox{\alpha}{\dbox{\prepeat{\alpha}}{\phi}}}$.
By IH, $\proves{\G}{\esub{M}{\pvx}{N}}{\phi \land \dbox{\alpha}{\dbox{\prepeat{\alpha}}{\phi}}}$ and by \irref{broll},
$\proves{\G}{\ebroll{\esub{M}{\pvx}{N}}}{\dbox{\prepeat{\alpha}}{\phi}}$, then the case holds because $\esub{\ebroll{M}}{\pvx}{N} = \ebroll{\esub{M}{\pvx}{N}}$.

\mycase \irref{bunroll}, postcondition $\phi \land \dbox{\alpha}{\dbox{\prepeat{\alpha}}{\phi}}$:
By premiss, $\proves{\G,\pvx:\psi}{M}{\dbox{\prepeat{\alpha}}{\phi}}$.
By IH, $\proves{\G}{\esub{M}{\pvx}{N}}{\dbox{\prepeat{\alpha}}{\phi}}$ and by \irref{bunroll},
$\proves{\G}{\ebunroll{\esub{M}{\pvx}{N}}}{\phi \land \dbox{\alpha}{\dbox{\prepeat{\alpha}}{\phi}}}$, 
then the case holds because $\esub{\ebunroll{M}}{\pvx}{N}= \ebunroll{\esub{M}{\pvx}{N}}$.

\mycase \irref{bchoiceI}, postcondition $\dbox{\alpha\cup\beta}{\phi},$ term $\ebcons{A}{B}$:
By premisses have $\drv_1$ that $\proves{\G,\pvx:\psi}{A}{\dbox{\alpha}{\phi}}$ and $\drv_2$ that $\proves{\G,\pvx:\psi}{B}{\dbox{\beta}{\phi}}$.
By IH have $\proves{\G}{\esub{A}{\pvx}{N}}{\dbox{\alpha}{\phi}}$ and $\proves{\G}{\esub{B}{\pvx}{N}}{\dbox{\beta}{\phi}}$. 
By \irref{bchoiceI} have $\proves{\G}{\ebcons{\esub{A}{\pvx}{N}}{\esub{B}{\pvx}{N}}}{\dbox{\alpha\cup\beta}{\phi}}$.
Then since $\esub{\ebcons{A}{B}}{\pvx}{N} = \ebcons{\esub{A}{\pvx}{N}}{\esub{B}{\pvx}{N}}$ this completes the case.

\mycase \irref{bchoiceEL}, postcondition $\dbox{\alpha}{\phi},$ term $\ebprojL{M}$:
By premiss, have $\proves{\G,\pvx:\psi}{M}{\dbox{\alpha\cup\beta}{\phi}}$.
By IH, have $\proves{\G}{\esub{M}{\pvx}{N}}{\dbox{\alpha\cup\beta}{\phi}}$, then by \irref{bchoiceEL} have
$\proves{\G}{\ebprojL{\esub{M}{\pvx}{N}}}{\dbox{\alpha}{\phi}}$, then $\esub{\ebprojL{M}}{\pvx}{N} = \ebprojL{\esub{M}{\pvx}{N}}$ which completes the case.

\mycase \irref{bchoiceER}, postcondition $\dbox{\beta}{\phi},$ term $\ebprojR{M}$:
By premiss, have $\proves{\G,\pvx:\psi}{M}{\dbox{\alpha\cup\beta}{\phi}}$.
By IH, have $\proves{\G}{\esub{M}{\pvx}{N}}{\dbox{\alpha\cup\beta}{\phi}}$, then by \irref{bchoiceER} have
$\proves{\G}{\ebprojR{\esub{M}{\pvx}{N}}}{\dbox{\beta}{\phi}}$, then $\esub{\ebprojR{M}}{\pvx}{N} = \ebprojR{\esub{M}{\pvx}{N}}$ which completes the case.

\mycase \irref{btestI}, postcondition $\dbox{\ptest{\rho}}{\phi},$ term $\elam{\pvy}{\rho}{M}$:
By premiss, have $\proves{\G,\pvx:\psi,\pvy:\rho}{M}{\phi}$.
By weakening, have $\proves{\G,\pvy:\rho}{N}{\psi}$.
By IH, have $\proves{\G,\pvy:\rho}{M}{\phi}$ so by \irref{btestI} have $\elam{\pvy}{\rho}{(\esub{M}{\pvx}{N})}$.
Then (by renaming if necessary) $\esub{(\elam{\pvy}{\rho}{M})}{\pvx}{N} = \elam{\pvy}{\rho}{(\esub{M}{\pvx}{N})}$ as desired.

\mycase \irref{btestE}, postcondition $\phi,$ term $\eapp{A}{B}$:
By premiss, have $\proves{\G,\pvx:\psi}{A}{\dbox{\ptest{\rho}}{\phi}}$ and $\proves{\G,\pvx:\psi}{B}{\rho}$.
By IH, have $\proves{\G}{\esub{A}{\pvx}{N}}{\dbox{\ptest{\rho}}{\phi}}$ and $\proves{\G}{\esub{B}{\pvx}{N}}{\rho},$ then
$\esub{\eapp{A}{B}}{\pvx}{N} = \eapp{\esub{A}{\pvx}{N}}{\esub{B}{\pvx}{N}}$ as desired.

\mycase \irref{seqI} (box):
By premiss, have $\proves{\G,\pvx:\psi}{M}{\dbox{\alpha}{\dbox{\beta}{\phi}}}$.
By IH, have $\proves{\G}{\esub{M}{\pvx}{N}}{\dbox{\alpha}{\dbox{\beta}{\phi}}}$, then by \irref{seqI} have
$\proves{\G}{\ebseq{\esub{M}{\pvx}{N}}}{\dbox{\alpha;\beta}{\phi}}$, 
then $\esub{\ebseq{M}}{\pvx}{N} = \ebseq{\esub{M}{\pvx}{N}}$ which completes the case.

\mycase \irref{dualI} (box):
By premiss, have $\proves{\G,\pvx:\psi}{M}{\ddiamond{\alpha}{\phi}}$.
By IH, have $\proves{\G}{\esub{M}{\pvx}{N}}{\ddiamond{\alpha}{\phi}}$, then by \irref{dualI} have
$\proves{\G}{\edswap{\esub{M}{\pvx}{N}}}{\dbox{\pdual{\alpha}}{\phi}}$, 
then $\esub{\edswap{M}}{\pvx}{N} = \edswap{\esub{M}{\pvx}{N}}$ which completes the case.

\mycase \irref{mon}, postcondition $\ddiamond{\alpha}{\rho}$, term $\emon{A}{B}{\pvy}$:
Let $\vec{x} = \boundvars{\alpha}$ and $\vec{y}$ be a vector of fresh variables of the same length.
By $\drv_1$ have $\proves{\G,\pvx:\psi}{A}{\ddiamond{\alpha}{\phi}}$ 
and by $\drv_2$ have $\proves{\eren{(\G,\pvx:\psi)}{\vec{x}}{\vec{y}},\pvy:\phi}{B}{\rho}$.
Then by IH on $\drv_1$, have (1) $\proves{\G}{\esub{A}{\pvx}{N}}{\ddiamond{\alpha}{\phi}}$.
By \rref{lem:app-rename} ($|\vec{x}|$ times) on $N$ and weakening, have $\proves{\eren{\G}{\vec{x}}{\vec{y}},\pvy:\phi}{\eren{N}{\vec{x}}{\vec{y}}}{\eren{\phi}{\vec{x}}{\vec{y}}},$ then because ($(\eren{(\G,\pvx:\psi)}{\vec{x}}{\vec{y}}),\pvy:\phi  = (\eren{(\G)}{\vec{x}}{\vec{y}},\pvy:\phi),\pvx:\eren{\psi}{\vec{x}}{\vec{y}}$), we can apply the IH on $\drv_2,$ giving (2) $\proves{\eren{(\G)}{\vec{x}}{\vec{y}},\pvy:\phi}{\esub{B}{\pvx}{\eren{N}{\vec{x}}{\vec{y}}}}{\rho}$.
Then apply \irref{mon} to (1) and (2), giving $\proves{\G}{\emon{(\esub{A}{\pvx}{N})}{\esub{B}{\pvx}{\eren{N}{\vec{x}}{\vec{y}}}}{\pvy}}{\ddiamond{\alpha}{\rho}}$.
Note that the substutition avoids capture in the $B$ branch, that is $\esub{(\emon{A}{B}{\pvy})}{\pvx}{N} = \emon{(\esub{A}{\pvx}{N})}{\esub{B}{\pvx}{\eren{N}{\vec{x}}{\vec{y}}}}{\pvy}$, so that the case holds.

\mycase \irref{hyp}, postcondition $\phi$:
That is, the proof term is some variable $\pvy$.
If $\pvy = \pvx$, then $\esub{\pvy}{\pvx}{N} = N$ and $\proves{\G}{N}{\psi}$ by assumption. 
Moreover $\phi = \psi$ since $\G,\pvx:\phi = \G,\pvx:\psi$ assigns only one type to $\pvx$.
Thus trivially $\proves{\G}{\esub{\pvy}{\pvx}{N}}{\phi}$ as desired.
Else $\pvy \neq \pvx,$ so $\esub{\pvy}{\pvx}{N} = \pvy$ and $\pvy \in \G,$ so that $\proves{\G}{\pvy}{\phi}$ by \irref{hyp} and thus $\proves{\G}{\esub{\pvy}{\pvx}{N}}{\phi}$ as well.

\begin{lemma}[Preservation]
If $\proves{\Gemp}{M}{\phi}$ and $M \stepsto M'$ then $\proves{\Gemp}{M'}{\phi}$
\end{lemma}
\begin{proof}
By induction on the evaluation trace $M \stepsto M'$.

We split the cases of the proof according to the four main kinds of rules: Beta rules, structural rules, monotonicity-conversion rules, and commuting-conversion rules.
While the cases within each class are not necessarily symmetric, they usually share a similar proof approach.

We give the beta rules first.

\mycase \irref{appBeta}, rule $\eapp{(\eplam{\psi}{M})}{N} \stepsto \esub{M}{\pvx}{N}$, postcondition $\phi$:
By inversion, $\proves{\Gemp}{M}{\psi\limply\phi}$ and $\proves{\Gemp}{N}{\psi}$.
Then by \rref{lem:app-ptsub}, have $\proves{\Gemp}{\esub{M}{\pvx}{N}}{\phi}$ as desired.

\mycase \irref{brandomBeta}, rule $\eapp{(\etlam{\allrat}{M})}{f} \stepsto \tsub{M}{x}{f}$, postcondition $\tsub{\phi}{x}{f}$:
By inversion $\proves{\eren{\Gemp}{x}{y}}{M}{\phi}$,
Then by \rref{lem:app-pttsub} have $\proves{\tsub{\eren{\Gemp}{x}{y}}{x}{f}}{\tsub{M}{x}{f}}{\tsub{\phi}{x}{f}}$.
Then note $\tsub{\eren{\Gemp}{x}{y}}{x}{f}=\eren{\Gemp}{x}{y}$ for $y$ fresh in $\tsub{M}{x}{f}$ and in $\tsub{\phi}{x}{f}$, so by \rref{lem:app-ptren}, $\proves{\Gemp}{\tsub{M}{x}{f}}{\tsub{\phi}{x}{f}}$.

\mycase \irref{projLBeta}, rule $\eprojL{\econs{A}{B}} \stepsto A$, postcondition $\phi$:
By inversion twice, $\proves{\Gemp}{A}{\phi}$ as desired.

\mycase \irref{projRBeta}:
By inversion twice, $\proves{\Gemp}{B}{\phi}$ as desired.

\mycase \irref{caseBetaL}, rule $\ecase{\einjL{A}}{B}{C} \stepsto \esub{B}{\pvl}{A},$ postcondition $\phi$:
By inversion $\proves{\Gemp}{A}{\ddiamond{\alpha}{\psi}}$, for some $\alpha,\psi$.
By inversion, $\proves{\pvl:\ddiamond{\alpha}{\psi}}{B}{\phi},$ then by \rref{lem:app-ptsub} $\proves{\Gemp}{\esub{B}{\pvl}{A}}{\phi}$ as desired.

\mycase \irref{caseBetaR}, rule $\ecase{\einjR{A}}{B}{C} \stepsto \esub{C}{\pvr}{A},$ postcondition $\phi$:
By inversion $\proves{\Gemp}{A}{\ddiamond{\beta}{\psi}}$, for some $\beta,\psi$.
By inversion, $\proves{\pvr:\ddiamond{\beta}{\psi}}{C}{\phi},$ then by \rref{lem:app-ptsub} $\proves{\Gemp}{\esub{C}{\pvr}{A}}{\phi}$ as desired.

\mycase \irref{unpackBeta}, rule $\eunpack{\etcons{f}{M}}{N} \stepsto \esub{N}{x,y}{f,M}$, postcondition $\phi$:
By inversion, $\proves{\Gemp}{M}{\tsub{\psi}{x}{f}}$ and 
$\proves{\eren{\Gemp}{x}{z},y:\psi}{N}{\phi}$.
Then note $x$ is free in $\eren{\Gemp}{x}{z}$.
By \rref{lem:app-pttsub} have $\proves{\tsub{\eren{\Gemp}{x}{z}}{x}{f},y:\tsub{\psi}{x}{f}}{\tsub{N}{x}{f}}{\tsub{\phi}{x}{f}}$
which simplifies to 
$\proves{\eren{\Gemp}{x}{z},y:\tsub{\psi}{x}{f}}{\tsub{N}{x}{f}}{\phi}$ 
(by trivial context identities and because $x$ was fresh in $\phi$ by side condition) and then simplifies again by \rref{lem:app-ptren}
(using side condition $z \notin \freevars{\tsub{N}{x}{f}}$) to 
$\proves{y:\tsub{\psi}{x}{f}}{\tsub{N}{x}{f}}{\phi}$
Then by \rref{lem:app-ptsub} have 
$\proves{\Gemp}{\esub{\tsub{N}{x}{f}}{y}{M}}{\phi}$ 
where $\esub{\tsub{N}{x}{f}}{y}{M}$ is a funny way of writing $\esub{N}{x,y}{f,M}$, as desired.

\mycase \irref{fpBeta}, rule $\efp{A}{B}{C} \stepsto \ercase{A}{B}{\esub{C}{\pvg}{\emon{\pvg}{\efp{\pvz}{B}{C}}{\pvz}}}$, postcondition $\psi$:
By inversion, $\drv_1$ is $\proves{\Gemp}{A}{\ddiamond{\prepeat{\alpha}}{\phi}}$,
$\drv_2$ is (1)$\proves{\pvs:\phi}{B}{\psi}$, and
$\drv_3$ is $\proves{\pvg:\ddiamond{\alpha}{\psi}}{C}{\psi}$.
On the other branch, note
$\proves{\pvz:\ddiamond{\prepeat{\alpha}}{\phi}}{\efp{\pvz}{B}{C}}{\psi}$, by \irref{dloopE} then by \irref{hyp} and by $\drv_2$ and $\drv_3$.
Then by \irref{mon}, have $\proves{\pvg:\ddiamond{\alpha}{\ddiamond{\prepeat{\alpha}}{\phi}}}{\emon{\pvg}{\efp{\pvz}{B}{C}}{\pvz}}{\ddiamond{\alpha}{\psi}}$ by the previous line and by weakening and \irref{hyp}.
Next, apply \rref{lem:app-ptsub} to $\drv_3$, (first weakening it with $\G$), getting (2) $\proves{\pvg:\ddiamond{\alpha}{\ddiamond{\prepeat{\alpha}}{\phi}}}{\esub{C}{\pvg}{\emon{\pvg}{\efp{\pvz}{B}{C}}{\pvz}}}{\psi}$.
Now apply rule \irref{dchoiceE} (which is the eliminator for $\ddiamond{\prepeat{\alpha}}{\phi}$ as well) to facts (1) and (2), which yields
$\proves{\Gemp}{\ercase{A}{B}{\esub{C}{\pvg}{\emon{\pvg}{\efp{\pvz}{B}{C}}{\pvz}}}}{\psi}$ as desired.

\mycase \irref{repBeta}, rule $(\erep{M}{N}{\pvx:J}{O}) \stepsto \ebroll{\edcons{M}{\emon{(\esub{N}{\pvx}{M})}{(\erep{\pvy}{N}{\pvx:J}{O})}{\pvy}}}$:
By inversion have $\drv_1$ is $\proves{\Gemp}{M}{J}$ and
$\drv_2$ is $\proves{\pvx:J}{N}{\dbox{\alpha}{J}}$ and
$\drv_3$ is $\proves{\pvx:J}{O}{\phi}$.

Note (1) $\proves{\pvy:J}{(\erep{\pvy}{N}{\pvx:J}{O})}{\dbox{\prepeat{\alpha}}{\phi}}$ by \irref{bloopI} since 
$\proves{\pvy:J}{\pvy}{J}$ by \irref{hyp} and $\proves{\pvx:J}{N}{\dbox{\alpha}{J}}$ by $\drv_2$.
Then (2) $\proves{\Gemp}{\esub{N}{\pvx}{M}}{\dbox{\alpha}{J}}$ by \rref{lem:app-ptsub} and $\drv_1$ and $\drv_2$ and weakening.
Then from (1) and (2) by \irref{mon} have $\proves{\Gemp}{\emon{(\esub{N}{\pvx}{M})}{(\erep{\pvy}{N}{\pvx:J}{O})}{\pvy}}{\dbox{\alpha}{\dbox{\prepeat{\alpha}}{\phi}}}$
Then by \irref{dtestI} have $\proves{\Gemp}{\edcons{M}{\emon{(\esub{N}{\pvx}{M})}{(\erep{\pvy}{N}{\pvx:J}{O})}{\pvy}}}{\phi \land \dbox{\alpha}{\dbox{\prepeat{\alpha}}{\phi}}}$
so by \irref{broll} have $\proves{\Gemp}{\ebroll{\edcons{M}{\emon{(\esub{N}{\pvx}{M})}{(\erep{\pvy}{N}{\pvx:J}{O})}{\pvy}}}}{\dbox{\prepeat{\alpha}}{\phi}}$ as desired.

\mycase \irref{bunrollBeta}, rule $\ebunroll{\ebroll{M}} \stepsto M$:
By inversion twice, $\proves{\Gemp}{M}{\phi\land\dbox{\alpha}{\dbox{\prepeat{\alpha}}{\phi}}}$ as desired.

\mycase \irref{caseBetaL}, rule $\edcase{\einjL{A}}{B}{C} \stepsto \esub{B}{\ell}{A}$:
By inversion, $\proves{\Gemp}{A}{\phi}$ and by substitution $\proves{\Gemp}{\esub{B}{\pvl}{A}}{\psi}$ as desired.
Case for diamond loops is symmetric.

\mycase \irref{caseBetaR}, rule $\edcase{\einjR{A}}{B}{C} \stepsto \esub{C}{r}{A}$:
By inversion, $\proves{\Gemp}{A}{\ddiamond{\alpha}{\ddiamond{\prepeat{\alpha}}{\phi}}}$ and by substitution $\proves{\Gemp}{\esub{C}{\pvr}{A}}{\psi}$ as desired.
Case for diamond loops is symmetric.

\mycase \irref{forBeta}, rule
\[\begin{aligned}
&{\efor{A}{B}{C} \stepsto}\\
&\ecaseHead{\esplit{\met}{0}{}}\\
&\ \ \ecaseLeft{\pvl}{\estop{\esub{C}{(\pvx,\pvy)}{(\pvl,A)}}}\\
&\ecaseRight{\pvr}{\eghost{\met_0}{\met}{\textit{\pvrr}}{\ego{(\emon{(\esub{B}{\pvx,\pvy}{A,\edcons{\pvr}{\textit{\pvrr}}})}{(\efor{\eprojL{\pvz}}{B}{C})}{\pvz})}}}\ecaseEnd
  \end{aligned}\]

 By inversion, $\drv_1$ is $\proves{\Gemp}{A}{\conv}$ and
 $\drv_2$ is $\proves{\pvx:\conv,\pvy:(\met_0 = \met \metgr \metz)}{B}{\ddiamond{\alpha}{(\conv\land \met_0 \metgr \met)}}$ and 
 $\drv_3$ is $\proves{\pvx:\conv,\pvy:(\met=\metz)}{C}{\phi}$.
 By \irref{metsplit}, (1) $\proves{\Gemp}{\esplit{\met}{\metz}}{(\met = \metz \lor \met \metgr \metz)}$.
 Consider the first branch:
 $\proves{\ell:\met \leq 0}{\edcons{\ell}{A}}{\met = \metz \land \conv}$ by \irref{dtestI} and \irref{hyp} and weakening.
 Then by \irref{dstop} have (left) $\proves{\ell:\met = \metz}{\estop{\esub{C}{(\pvx,\pvy)}{(A,\pvl)}}}{\ddiamond{\prepeat{\alpha}}{\phi}}$
Consider the second branch:
By \irref{dtestI} have (x) $\proves{pvr:\met_0=\met, \pvrr: \met \metgr \metz}{\edcons{\pvr}{\textit{\pvrr}}}{\met_0 = \met \metgr \metz}$.
then by (x) and  $\drv_2$ and weakening and \rref{lem:app-ptsub} have
(3a) $\proves{\pvr: \met \geq 0, \pvr:\met_0=\met}{\esub{B}{x,y}{A,\edcons{\pvr}{\pvrr}}}{\ddiamond{\alpha}{(\met = \metz \lor \met \metgr \metz)}}$.
Separately, (3b) $\proves{\pvz:\conv\land \met_0 \metgr \met}{\efor{\eprojL{\pvz}}{B}{C}}{\ddiamond{\prepeat{\alpha}}{(\met = \metz \land \conv)}}$.
Then combining (3a) and (3b) with \irref{mon}, have 
$\proves{\pvrr:\met_0=\met,\pvr: \met \geq 0}
{\emon{\esub{B}{x,y}{A,\edcons{\pvr}{\pvrr}}}{\efor{\eprojL{\pvz}}{B}{C}}{\pvz}}
{\ddiamond{\alpha}{\ddiamond{\prepeat{\alpha}}{(\phi)}}}$.
By \irref{dgo} have $\proves{\pvrr:\met_0=\met,\pvr \met \geq 0}{\ego{\emon{\esub{B}{x,y}{A,\edcons{\pvr}{\pvrr}}}{\efor{\eprojL{\pvz}}{B}{C}}{\pvz}}}{\ddiamond{\prepeat{\alpha}}{\phi}}$.
By \irref{ghost} have (3)
\begin{align}
 & \pvr: \met \geq 0\\
&\vdash \eghost{\met_0}{\met}{\textit{\pvrr}}{\ego{\emon{\esub{B}{x,y}{A,\edcons{\pvr}{\pvrr}}}{\efor{\eprojL{\pvz}}{B}{C}}{\pvz}}}\\
&: \ddiamond{\prepeat{\alpha}}{\phi}
\end{align}

Lastly, combining (1), (2), and (3), then by \irref{dchoiceE} have
\begin{align}
  \Gemp\\
&\vdash \ecaseHead{\esplit{\met}{0}}\\
&\ecaseLeft{\ell}{\estop{\esub{C}{(\pvx,\pvy)}{(A,\pvl)}}}\\
&\ecaseRight{r}{\eghost{\met_0}{\met}{\textit{\pvrr}}{\ego{\emon{\esub{B}{x,y}{A,\edcons{\pvr}{\pvrr}}}{\efor{\eprojL{\pvz}}{B}{C}}{\pvz}}}} \\
&: \ddiamond{\prepeat{\alpha}}{\phi}
\end{align}
 as desired.

Structural rules:
We do not write out the premisses of the stepping rules explicitly in each case, because they are very similar.
For unary operations applied to a term $M,$ as well as the left structural rule of any binary operator, we assume $M \stepsto M',$ or for the right structural rule of any binary operator, we assume $\issimp{M}$ for the left operand and $N \stepsto N'$ for the right operand.


\mycase \irref{bconsSL}:
By the IH, $\proves{\Gemp}{M'}{\dbox{\alpha}{\phi}},$ so by case assumption and by \irref{bchoiceI} then $\proves{\Gemp}{\ebcons{M'}{N}}{\dbox{\alpha\cup\beta}{\phi}}$.

\mycase \irref{bconsSR}, $\ebcons{M}{N} \stepsto \ebcons{M}{N'}$, postcondition $\dbox{\alpha\cup\beta}{\phi}$:
By the IH, $\proves{\Gemp}{N'}{\dbox{\beta}{\phi}},$ so by case assumption and by \irref{bchoiceI} then $\proves{\Gemp}{\ebcons{M}{N'}}{\dbox{\alpha\cup\beta}{\phi}}$.

\mycase \irref{dconsSL} (diamond):
By the IH, $\proves{\Gemp}{M'}{\phi},$ so by case assumption and by \irref{dtestI} then
$\proves{\Gemp}{\edcons{M'}{N}}{\ddiamond{\ptest{\phi}}{\psi}}$.

\mycase \irref{dconsSR} (diamond):
By the IH, $\proves{\Gemp}{N'}{\psi},$ so by case assumption and by \irref{dtestI} then 
$\proves{\Gemp}{\edcons{M'}{N}}{\ddiamond{\ptest{\phi}}{\psi}}$.

\mycase \irref{projLS}:
By the IH, $\proves{\Gemp}{M'}{\ddiamond{\ptest{\phi}}{\psi}}$ so by \irref{dtestEL} then
$\proves{\Gemp}{\edprojL{M'}}{\phi}$.
The case for $\ebprojL{M'}$ is symmetric.

\mycase \irref{projRS}:
By the IH, $\proves{\Gemp}{M'}{\ddiamond{\ptest{\phi}}{\psi}}$ so by \irref{dtestER} then
$\proves{\Gemp}{\edprojR{M'}}{\psi}$.
The case for $\ebprojR{M'}$ is symmetric.

\mycase \irref{injLS}:
By the IH, $\proves{\Gemp}{M'}{\ddiamond{\alpha}{\phi}}$ so by \irref{dchoiceIL} then
$\proves{\Gemp}{\einjL{M'}}{\ddiamond{\alpha\cup\beta}{\phi}}$.

\mycase \irref{injRS}:
By the IH, $\proves{\Gemp}{M'}{\ddiamond{\beta}{\phi}}$ so by \irref{dchoiceIR} then
$\proves{\Gemp}{\einjR{M'}}{\ddiamond{\alpha\cup\beta}{\phi}}$.

\mycase \irref{repS}:
By the IH, $\proves{\Gemp}{M'}{J}$ so by \irref{bloopI} then
$\proves{\Gemp}{\erep{M'}{N}{\pvx:J}{O}}{\dbox{\prepeat{\alpha}}{\phi}}$.

\mycase \irref{appSL}:
By the IH, $\proves{\Gemp}{M'}{\phi \limply \psi}$ so by \irref{btestE} then $\proves{\Gemp}{\eapp{M'}{N}}{\psi}$.

\mycase \irref{appSR};
By the IH, $\proves{\Gemp}{N'}{\phi}$ so by \irref{btestE} then $\proves{\Gemp}{\eapp{M}{N'}}{\psi}$

\mycase \irref{bseqS}:
By the IH, $\proves{\Gemp}{M'}{\dbox{\alpha}{\dbox{\beta}{\phi}}}$ so by \irref{seqI} then $\proves{\Gemp}{\ebseq{M'}}{\dbox{\alpha;\beta}{\phi}}$

\mycase \irref{dseqS}:
By the IH, $\proves{\Gemp}{M'}{\ddiamond{\alpha}{\ddiamond{\beta}{\phi}}}$ so by \irref{seqI} then $\proves{\Gemp}{\edseq{M'}}{\ddiamond{\alpha;\beta}{\phi}}$

\mycase \irref{bswapS}:
By the IH, $\proves{\Gemp}{M'}{\ddiamond{\alpha}{\phi}}$ so by \irref{dualI} then $\proves{\Gemp}{\ebswap{M'}}{\dbox{\pdual{\alpha}}{\phi}}$.

\mycase \irref{dswapS}:
By the IH, $\proves{\Gemp}{M'}{\dbox{\alpha}{\phi}}$ so by \irref{dualI} then $\proves{\Gemp}{\edswap{M'}}{\ddiamond{\pdual{\alpha}}{\phi}}$

\mycase \irref{monS}:
By the IH, $\proves{\Gemp}{M'}{\ddiamond{\alpha}{\phi}}$ so by \irref{mon} then $\proves{\Gemp}{\emon{M'}{N}{\pvx}}{\ddiamond{\alpha}{\psi}}$

\mycase \irref{fpS}:
By the IH, $\proves{\Gemp}{A'}{\ddiamond{\prepeat{\alpha}}{\phi}}$ so by \irref{dloopE} then $\proves{\Gemp}{\efp{A'}{B}{C}}{\psi}$.

\mycase \irref{caseS}:
By the IH, $\proves{\Gemp}{A'}{\ddiamond{\alpha\cup\beta}{\phi}}$ so by \irref{dchoiceE} then $\proves{\Gemp}{\ecase{A'}{B}{C}}{\psi}$.

\mycase \irref{drcase} (diamond loops):
By the IH, $\proves{\Gemp}{A'}{\ddiamond{\prepeat{\alpha}}{\phi}}$ so by \irref{drcase} then $\proves{\Gemp}{\ercase{A'}{B}{C}}{\psi}$.



\mycase \irref{unpackS}:
By the IH, $\proves{\Gemp}{M'}{\ddiamond{\prandom{x}}{\phi}}$ so by \irref{drandomE} then $\proves{\Gemp}{\eunpack{M'}{N}}{\psi}$.

\mycase \irref{forS}:
By the IH, $\proves{\Gemp}{M'}{\conv}$ so by \irref{dloopI} then $\proves{\Gemp}{\efor{M'}{N}{O}}{\ddiamond{\prepeat{\alpha}}{\phi}}$.

\mycase \irref{bunrollS}:
By the IH, $\proves{\Gemp}{M'}{\dbox{\prepeat{\alpha}}{\phi}}$ so by \irref{bunroll} then $\proves{\Gemp}{\ebunroll{M'}}{\phi \land \dbox{\alpha}{\dbox{\prepeat{\alpha}}{\phi}}}$.


We give the monotonicity conversion rules next.
In each case, let $\vec{x}$ denote the bound variables of the full program and $\vec{y}$ be a fresh variable vector of the same size.
For subprograms $\alpha,\beta$ we write $\va, \vb$ for the vectors of bound variables of each subprogram, $\vca,\vcb$ for their relative complements $\boundvars{\beta} - \boundvars{\alpha}$ and vice-versa, and we affix primes to indicate a vector of fresh variables.

\mycase \irref{lamMon}, rule $\emon{(\eplam{\phi}{M})}{N}{\pvy} \stepsto (\eplam{\phi}{(\esub{N}{\pvy}{M})})$:
Note $\eren{\Gemp}{\vec{x}}{\vec{y}} = \Gemp$.
By premisses, have $\drv_1$ is $\proves{\Gemp}{\eplam{\phi}{M}}{\dbox{\ptest{\phi}}{\psi}}$ and
$\drv_2$ is $\proves{\eren{\Gemp}{\vec{x}}{\vec{y}},\pvy:\psi}{N}{\rho}$, i.e., $\proves{\pvy:\psi}{N}{\rho}$.
By inversion on $\drv_1$ have $\proves{\pvx:\phi}{M}{\psi}$
then $\proves{\pvx:\phi}{\esub{N}{\pvy}{M}}{\rho}$ by \rref{lem:app-ptsub} and weakening, then by \irref{dtestI} have 
$\proves{\Gemp}{\eplam{\phi}{(\esub{N}{\pvy}{M})}}{\dbox{\ptest{\phi}}{\rho}}$.

\mycase \irref{rlamMon}, rule $\emon{(\etlam{\allrat}{M})}{N}{\pvy} \stepsto (\etlam{\allrat}{\esub{M}{\pvy}{N}})$:
By premisses, have $\drv_1$ is $\proves{\Gemp}{(\etlam{\allrat}{M})}{\dbox{\prandom{x}}{\phi}}$ and
$\drv_2$ is $\proves{\eren{\Gemp}{\vec{x}}{\vec{y}},\pvy:\phi}{N}{\psi}$.
By inversion on $\drv_1$ have $\proves{\eren{\Gemp}{\vec{x}}{\vec{y}}}{M}{\phi}$ since $\vec{x} = x$.
then $\proves{\eren{\Gemp}{\vec{x}}{\vec{y}}}{\esub{M}{\pvy}{N}}{\rho}$ by \rref{lem:app-ptsub} and weakening, then by \irref{drandomI} have 
$\proves{\Gemp}{(\etlam{\allrat}{(\esub{M}{\pvy}{N})})}{\dbox{\prandom{x}}{\psi}}$.

\mycase \irref{bconsMon}, rule $\emon{\ebcons{A}{B}}{N}{\pvx} \stepsto \ebcons{\emon{A}{N}{\pvx}}{\emon{B}{N}{\pvx}}$:
By premisses, have $\drv_1$ is $\proves{\Gemp}{\ebcons{A}{B}}{\dbox{\alpha\cup\beta}{\phi}}$
and $\drv_2$ is $\proves{\eren{\Gemp}{\vec{x}}{\vec{y}},\pvx:\phi}{N}{\psi}$.
By inversion on $\drv_1$ have (1) $\proves{\Gemp}{A}{\dbox{\alpha}{\phi}}$
and (2) $\proves{\Gemp}{B}{\dbox{\beta}{\phi}}$.
By \rref{lem:app-ptren} on $\drv_2$ have (3a) $\proves{\eren{\Gemp}{\va}{\va'},\pvx:\phi}{\eren{N}{\va'}{\va}}{\psi}$ since $\vca'$ are fresh in $\psi$ and $\phi$,
and  (3b) $\proves{\eren{\Gemp}{\vb}{\vb'},\pvx:\phi}{\eren{N}{\vb'}{\vb}}{\psi}$ since $\vb'$ are fresh in $\psi$ and $\phi$.
By \irref{mon} on (1) and (3a) have $\proves{\Gemp}{\emon{A}{\eren{N}{\va'}{\va}}{\pvx}}{\dbox{\alpha}{\psi}}$
and by \irref{mon} on (2) and (3b) $\drv_2$ have $\proves{\Gemp}{\emon{B}{\eren{N}{\vb'}{\vb}}{\pvx}}{\dbox{\beta}{\psi}}$,
lastly by \irref{bchoiceI} have $\proves{\Gemp}{\ebcons{\emon{A}{\eren{N}{\va'}{\va}}{\pvx}}{\emon{B}{\eren{N}{\vb'}{\vb}}{\pvx}}}{\dbox{\alpha\cup\beta}{\psi}}$.

\mycase \irref{dconsMon}, rule $\emon{\edcons{A}{B}}{N}{\pvx} \stepsto \edcons{A}{\esub{N}{\pvx}{B}}$:
Note $\eren{\Gemp}{\vec{x}}{\vec{y}} = \G$.
By premisses, have $\drv_1$ is $\proves{\Gemp}{\edcons{A}{B}}{\ddiamond{\ptest{\phi}}{\psi}}$
and $\drv_2$ is $\proves{\pvx:\psi}{\rho}$.
By inversion on $\drv_1$ have (1) $\proves{\Gemp}{A}{\phi}$ and (2) $\proves{\Gemp}{B}{\psi}$.
By \rref{lem:app-ptsub}, have $\proves{\Gemp}{\esub{B}{\pvx}{N}}{\rho}$.
By \irref{dtestI} have $\proves{\Gemp}{\edcons{A}{\esub{B}{\pvx}{N}}}{\ddiamond{\ptest{\phi}}{\rho}}$.

\mycase \irref{tconsMon}:
By premisses, have $\drv_1$ is $\proves{\Gemp}{\etcons{f}{M}}{\ddiamond{\prandom{x}}{\phi}}$, 
and $\drv_2$ is $\proves{\eren{\Gemp}{\vec{x}}{\vec{y}},\pvy:\phi}{N}{\psi}$.
then by inversion on $\drv_1$ have 
(1) $\proves{\eren{\Gemp}{\vec{x}}{\vec{y}},\pvx:(x=\eren{f}{x}{y})}{M}{\phi}$.
By weakening on $\drv_2$ have  $\proves{\eren{\Gemp}{\vec{x}}{\vec{y}},\pvy:\phi,\pvx:(x=\eren{f}{x}{y})}{N}{\psi}$, then by \rref{lem:app-pttsub} 
have $\proves{\eren{\Gemp}{\vec{x}}{\vec{y}},\pvx:(x=\eren{f}{x}{y})}{\esub{N}{\pvx}{M}}{\psi},$ and lastly by \irref{drandomI} have
$\proves{\Gemp}{\etcons{f}{\esub{N}{\pvx}{M}}}{\ddiamond{\prandom{x}}{\psi}}$.

\mycase \irref{dswapMon}:
By premisses, have $\drv_1$ is $\proves{\Gemp}{\edswap{M}}{\ddiamond{\pdual{\alpha}}{\phi}}$, 
and $\drv_2$ is $\proves{\eren{\Gemp}{\vec{x}}{\vec{y}},\pvy:\phi}{N}{\psi}$.
By inversion on $\drv_1$ have (1) $\proves{\Gemp}{M}{\dbox{\alpha}{\phi}}$.
By \irref{mon} have $\proves{\Gemp}{\emon{M}{N}{\pvx}}{\dbox{\alpha}{\psi}}$ so
by \irref{dualI} have $\proves{\Gemp}{\edswap{(\emon{M}{N}{\pvx})}}{\ddiamond{\pdual{\alpha}}{\psi}}$ as desired.

\mycase \irref{bswapMon}:
By premisses, have $\drv_1$ is $\proves{\Gemp}{\ebswap{M}}{\dbox{\pdual{\alpha}}{\phi}}$, 
and $\drv_2$ is $\proves{\eren{\Gemp}{\vec{x}}{\vec{y}},\pvy:\phi}{N}{\psi}$.
By inversion on $\drv_1$ have (1) $\proves{\Gemp}{M}{\ddiamond{\alpha}{\phi}}$.
By \irref{mon} have $\proves{\Gemp}{\emon{M}{N}{\pvx}}{\ddiamond{\alpha}{\psi}}$ so
by \irref{dualI} have $\proves{\Gemp}{\ebswap{(\emon{M}{N}{\pvx})}}{\dbox{\pdual{\alpha}}{\psi}}$ as desired.

\mycase \irref{injLMon}:
By premisses, have $\drv_1$ is $\proves{\Gemp}{\edinjL{M}}{\ddiamond{\alpha\cup\beta}{\phi}}$, 
then by inversion have (1) $\proves{\Gemp}{M}{\ddiamond{\alpha}{\phi}}$.
By \irref{mon} have $\proves{\Gemp}{\emon{M}{N}{\pvx}}{\ddiamond{\alpha}{\psi}}$ so
by \irref{dchoiceIL} have $\proves{\Gemp}{\einjL{(\emon{M}{N}{\pvx})}}{\dbox{\alpha\cup\beta}{\psi}}$ as desired.

\mycase \irref{injRMon}:
By premisses, have $\drv_1$ is $\proves{\Gemp}{\edinjR{M}}{\ddiamond{\alpha\cup\beta}{\phi}}$, 
then by inversion have (1) $\proves{\Gemp}{M}{\ddiamond{\beta}{\phi}}$.
By \irref{mon} have $\proves{\Gemp}{\emon{M}{N}{\pvx}}{\ddiamond{\beta}{\psi}}$ so
by \irref{dchoiceIR} have $\proves{\Gemp}{\einjR{(\emon{M}{N}{\pvx})}}{\dbox{\alpha\cup\beta}{\psi}}$ as desired.

\mycase \irref{dasgnMon} (diamond):
By premisses, have $\drv_1$ is $\proves{\Gemp}{\edasgn{y}{x}{\pvx}{M}}{\ddiamond{\humod{x}{f}}{\phi}}$, 
and $\drv_2$ is $\proves{\eren{\Gemp}{\vec{x}}{\vec{y}},\pvy:\phi}{N}{\psi}$.
then by inversion on $\drv_1$ have
(1) $\proves{\eren{\Gemp}{\vec{x}}{\vec{y}},\pvx:(x=\eren{f}{x}{y})}{M}{\phi}$.
By weakening on $\drv_2$ have  $\proves{\eren{\Gemp}{\vec{x}}{\vec{y}},\pvy:\phi,\pvx:(x=\eren{f}{x}{y})}{N}{\psi}$, then by \rref{lem:app-pttsub} 
have $\proves{\eren{\Gemp}{\vec{x}}{\vec{y}},\pvx:(x=\eren{f}{x}{y})}{\esub{N}{\pvx}{M}}{\psi}$, and lastly by \irref{drandomI} have
$\proves{\Gemp}{\edasgn{y}{x}{\pvx}{\esub{N}{\pvx}{M}}}{\ddiamond{\humod{x}{f}}{\psi}}$.

\mycase \irref{basgnMon} (box):
By premisses, have $\drv_1$ is $\proves{\Gemp}{\ebasgn{y}{x}{\pvx}{M}}{\dbox{\humod{x}{f}}{\phi}}$,
and $\drv_2$ is $\proves{\eren{\Gemp}{\vec{x}}{\vec{y}},\pvy:\phi}{N}{\psi}$.
then by inversion on $\drv_1$ have
(1) $\proves{\eren{\Gemp}{\vec{x}}{\vec{y}},\pvx:(x=\eren{f}{x}{y})}{M}{\phi}$.
By weakening on $\drv_2$ have  $\proves{\eren{\Gemp}{\vec{x}}{\vec{y}},\pvy:\phi,\pvx:(x=\eren{f}{x}{y})}{N}{\psi}$, then by \rref{lem:app-pttsub}
have $\proves{\eren{\Gemp}{\vec{x}}{\vec{y}},\pvx:(x=\eren{f}{x}{y})}{\esub{N}{\pvx}{M}}{\psi}$, and lastly by \irref{drandomI} have
$\proves{\Gemp}{\ebasgn{y}{x}{\pvx}{\esub{N}{\pvx}{M}}}{\dbox{\humod{x}{f}}{\psi}}$.

\mycase \irref{bseqMon}:
By premisses have $\drv_1$ is $\proves{\Gemp}{\ebseq{M}}{\dbox{\alpha;\beta}{\phi}}$, and
$\drv_2$ is $\proves{\eren{\Gemp}{\vec{x}}{\vec{y}},\pvx:\phi}{N}{\psi}$.
By inversion on $\drv_1$ have (1) $\proves{\Gemp}{M}{\dbox{\alpha}{\dbox{\beta}{\phi}}}$.
From $\drv_2$ have(2) $\proves{\eren{\Gemp}{\va}{\va'},\pvx:\phi}{\eren{N}{\vca'}{\vca}}{\psi}$ by \rref{lem:app-ptren} since $\vca'$ are fresh in $\phi$ and $\psi$.
Then by \irref{mon} on (2) have $\proves{\pvy:\dbox{\beta}{\phi}}{\emon{\pvy}{\eren{N}{\vca'}{\vca}}{\pvx}}{\dbox{\beta}{\psi}}$, then by
\irref{mon} again have $\proves{\Gemp}{\emon{M}{\emon{\pvy}{(\eren{N}{\vca'}{\vca})}{\pvx}}{\pvy}}{\dbox{\alpha}{\dbox{\beta}{\psi}}}$ and finally
by \irref{seqI} have $\proves{\Gemp}{\ebseq{\emon{M}{(\emon{\pvy}{(\eren{N}{\vca'}{\vca})}{\pvx})}{\pvy}}}{\dbox{\alpha;\beta}{\psi}}$.

\mycase \irref{dseqMon}:
By premisses have $\drv_1$ is $\proves{\Gemp}{\edseq{M}}{\ddiamond{\alpha;\beta}{\phi}}$, and
$\drv_2$ is $\proves{\pvx:\phi}{N}{\psi}$.
By inversion on $\drv_1$ have (1) $\proves{\Gemp}{M}{\ddiamond{\alpha}{\ddiamond{\beta}{\phi}}}$.
By \rref{lem:app-ptsub} on $\drv_2$ have $\proves{\pvz:\phi}{\esub{N}{\pvx}{\pvz}}{\psi}$, then by
\irref{mon} have $\proves{\pvy:\ddiamond{\beta}{\phi}}{\emon{\pvy}{(\esub{N}{\pvx}{\pvz})}{\pvz}}{\ddiamond{\beta}{\psi}}$, then by
\irref{mon} again have $\proves{\Gemp}{\emon{M}{\emon{\pvy}{(\esub{N}{\pvx}{\pvz})}{\pvz}}{\pvy}}{\ddiamond{\alpha}{\ddiamond{\beta}{\psi}}}$ and finally
by \irref{seqI} have $\proves{\Gemp}{\edseq{\emon{M}{(\emon{\pvy}{(\esub{N}{\pvx}{\pvz})}{\pvz})}{\pvy}}}{\ddiamond{\alpha;\beta}{\psi}}$.
 
\mycase \irref{brollMon}:
By premisses have $\drv_1$ is $\proves{\Gemp}{\ebroll{M}}{\dbox{\prepeat{\alpha}}{\phi}}$, and
$\drv_2$ is $\proves{\eren{\Gemp}{\vec{x}}{\vec{y}},\pvx:\phi}{N}{\psi}$.
Then have (2a) $\proves{\pvx:\phi}{\eren{N}{\vec{y}}{\vec{x}}}{\psi}$ by \rref{lem:app-ptren} since $\vec{y}$ are fresh.
By inversion on $\drv_1$ have (1) $\proves{\Gemp}{M}{\phi \land \dbox{\alpha}{\dbox{\prepeat{\alpha}}{\phi}}}$.
Then by \irref{dtestEL} and \irref{dtestEL} have
(1a) $\proves{\Gemp}{\eprojL{M}}{\phi}$ and
(1b) $\proves{\Gemp}{\eprojR{M}}{\dbox{\alpha}{\dbox{\prepeat{\alpha}}{\phi}}}$
By \irref{mon} on (1a) and (2a) have (L) $\proves{\Gemp}{\emon{\eprojL{M}}{\eren{N}{\vec{y}}{\vec{x}}}{\pvx}}{\psi}$.
By \irref{mon} on $\drv_2$ and weakening then have $\proves{\eren{\Gemp}{\vec{x}}{\vec{y}},\pvz:\dbox{\prepeat{\alpha}}{\phi}}{\emon{\pvz}{N}{\pvx}}{\dbox{\prepeat{\alpha}}{\psi}}$ , then by 
\irref{mon} again and (1b) have (R) $\proves{\Gemp}{\emon{\eprojR{M}}{\emon{\pvz}{N}{\pvx}}{\pvx}}{\dbox{\alpha}{\dbox{\prepeat{\alpha}}{\psi}}}$.
Finally, from (L) and (R) have by \irref{dtestI}, $\proves{\Gemp}{\edcons{\emon{\eprojL{M}}{(\eren{N}{\vec{y}}{\vec{x}})}{\pvx}}{\emon{\eprojR{M}}{\emon{\pvz}{N}{\pvx}}{\pvz}}}{\psi \land \dbox{\alpha}{\dbox{\prepeat{\alpha}}{\psi}}}$ and immediately by \irref{broll}
$\proves{\Gemp}{\ebroll{\edcons{\emon{\eprojL{M}}{(\eren{N}{\vec{y}}{\vec{x}})}{\pvx}}{\emon{\eprojR{M}}{\emon{\pvz}{N}{\pvx}}{\pvz}}}}{\dbox{\prepeat{\alpha}}{\psi}}$.

\mycase \irref{injLMon} (diamond loops):
By premisses have $\drv_1$ is $\proves{\Gemp}{\estop{M}}{\ddiamond{\prepeat{\alpha}}{\phi}}$, and
$\drv_2$ is $\proves{\eren{\Gemp}{\vec{x}}{\vec{y}},\pvx:\phi}{N}{\psi}$.
By inversion on $\drv_1$ have (1) $\proves{\Gemp}{M}{\phi}$.
By \rref{lem:app-ptren} on $\drv_2$ have $\proves{\pvx:\phi}{\eren{N}{\va'}{\va}}{\psi}$ since $\va'$ are fresh in $\phi$ and $\psi$.
Then by \rref{lem:app-ptsub} then have $\proves{\Gemp}{\esub{\eren{N}{\va'}{\va}}{\pvx}{M}}{\psi}$ so by \irref{dstop} have
$\proves{\Gemp}{\estop{\esub{\eren{N}{\va'}{\va}}{\pvx}{M}}}{\ddiamond{\prepeat{\alpha}}{\psi}}$ as desired.

\mycase  \irref{injRMon} (diamond loop):
By premisses have $\drv_1$ is $\proves{\Gemp}{\ego{M}}{\ddiamond{\prepeat{\alpha}}{\phi}}$, and
$\drv_2$ is $\proves{\eren{\Gemp}{\vec{x}}{\vec{y}},\pvx:\phi}{N}{\psi}$.
By inversion on $\drv_1$ have (1) $\proves{\Gemp}{M}{\ddiamond{\alpha}{\ddiamond{\prepeat{\alpha}}{\phi}}}$.
By \irref{mon} have $\proves{\eren{\Gemp}{\vec{x}}{\vec{y}},\pvz:\ddiamond{\prepeat{\alpha}}{\phi}}{\emon{\pvz}{N}{\pvy}}{\ddiamond{\alpha}{\psi}}$.
Note this application checks even with $\eren{\Gemp}{\vec{x}}{\vec{y}}$ because in checking $N$ we can exploit $\eren{\eren{\Gemp}{\vec{x}}{\vec{y}}}{\vec{x}}{\vec{z}} = \eren{\Gemp}{\vec{x}}{\vec{y}}$, where $\vec{z}$ is the additional vector of fresh variables used in the second application.
Then by \irref{mon} again have
$\proves{\pvy:\ddiamond{\alpha}{\ddiamond{\prepeat{\alpha}}{\phi}}}{\emon{\pvy}{\emon{\pvz}{N}{\pvx}}{\pvz}}{\ddiamond{\alpha}{\ddiamond{\prepeat{\alpha}}{\psi}}}$.
Then by \rref{lem:app-ptsub} and (1)have $\proves{\Gemp}{\esub{\emon{\pvy}{\emon{\pvz}{N}{\pvx}}{\pvz}}{\pvy}{M}}{\ddiamond{\alpha}{\ddiamond{\prepeat{\alpha}}{\psi}}}$ so by \irref{dgo} have
$\proves{\Gemp}{\ego{\esub{\emon{\pvy}{\emon{\pvz}{N}{\pvx}}{\pvz}}{\pvy}{M}}}{\ddiamond{\prepeat{\alpha}}{\psi}}$

We give the commuting conversion rules next.


\mycase \irref{dseqC}:
From premisses, $\drv_1$ is $\proves{\Gemp}{\edseq{\ecase{A}{B}{C}}}{\ddiamond{\alpha;\beta}{\phi}}$ and by inversion
(1) $\proves{\Gemp}{\ecase{A}{B}{C}}{\ddiamond{\alpha}{\ddiamond{\beta}{\phi}}}$.
(A) $\proves{\Gemp}{A}{\psi\lor\rho}$.
(B) $\proves{\pvl:\psi}{B}{\ddiamond{\alpha}{\ddiamond{\beta}{\phi}}}$.
(C) $\proves{\pvr:\rho}{C}{\ddiamond{\alpha}{\ddiamond{\beta}{\phi}}}$.
Then by \irref{seqI} have 
(B1) $\proves{\pvl:\psi}{\edseq{B}}{\ddiamond{\alpha;\beta}{\phi}}$.
(C1) $\proves{\pvr:\rho}{\edseq{C}}{\ddiamond{\alpha;\beta}{\phi}}$,
then by \irref{dchoiceE} on (A) have $\proves{\Gemp}{\ecase{A}{\edseq{B}}{\edseq{C}}}{\ddiamond{\alpha;\beta}{\phi}}$

\mycase \irref{bseqC}:
From premisses, $\drv_1$ is $\proves{\Gemp}{\ebseq{\ecase{A}{B}{C}}}{\dbox{\alpha;\beta}{\phi}}$ and by inversion
(1) $\proves{\Gemp}{\ecase{A}{B}{C}}{\dbox{\alpha}{\dbox{\beta}{\phi}}}$.
(A) $\proves{\Gemp}{A}{\psi\lor\rho}$.
(B) $\proves{\pvl:\psi}{B}{\dbox{\alpha}{\dbox{\beta}{\phi}}}$.
(C) $\proves{\pvr:\rho}{C}{\dbox{\alpha}{\dbox{\beta}{\phi}}}$.
Then by \irref{seqI} have 	
(B1) $\proves{\pvl:\psi}{\edseq{B}}{\dbox{\alpha;\beta}{\phi}}$.
(C1) $\proves{\pvr:\rho}{\edseq{C}}{\dbox{\alpha;\beta}{\phi}}$,
then by \irref{dchoiceE} on (A) have $\proves{\Gemp}{\ecase{A}{\ebseq{B}}{\ebseq{C}}}{\dbox{\alpha;\beta}{\phi}}$

\mycase \irref{bconsCL}:
From premisses, $\drv_1$ is $\proves{\Gemp}{\ebcons{\ecase{A}{B}{C}}{N}}{\dbox{\alpha\cup\beta}{\phi}}$ and by inversion
(1) $\proves{\Gemp}{\ecase{A}{B}{C}}{\dbox{\alpha}{\phi}}$.
(A) $\proves{\Gemp}{A}{\psi\lor\rho}$.
(B) $\proves{\pvl:\psi}{B}{\dbox{\alpha}{\phi}}$.
(C) $\proves{\pvr:\rho}{C}{\dbox{\alpha}{\phi}}$.
Then by \irref{bchoiceI} and weakening on $N$ have 
(B1) $\proves{\pvl:\psi}{\ebcons{B}{N}}{\dbox{\alpha\cup\beta}{\phi}}$.
(C1) $\proves{\pvr:\rho}{\ebcons{C}{N}}{\dbox{\alpha\cup\beta}{\phi}}$,
then by \irref{dchoiceE} on (A) have $\proves{\Gemp}{\ecase{A}{\ebcons{B}{N}}{\ebcons{C}{N}}}{\dbox{\alpha\cup\beta}{\phi}}$

\mycase \irref{bconsCR}:
From premisses, $\drv_1$ is $\proves{\Gemp}{\ebcons{M}{\ecase{A}{B}{C}}}{\dbox{\alpha\cup\beta}{\phi}}$ and by inversion
(1) $\proves{\Gemp}{\ecase{A}{B}{C}}{\dbox{\beta}{\phi}}$.
(A) $\proves{\Gemp}{A}{\psi\lor\rho}$.
(B) $\proves{\pvl:\psi}{B}{\dbox{\beta}{\phi}}$.
(C) $\proves{\pvr:\rho}{C}{\dbox{\beta}{\phi}}$.
Then by \irref{bchoiceI} and weakening on $M$ have 
(B1) $\proves{\pvl:\psi}{\ebcons{M}{B}}{\dbox{\alpha\cup\beta}{\phi}}$.
(C1) $\proves{\pvr:\rho}{\ebcons{M}{C}}{\dbox{\alpha\cup\beta}{\phi}}$,
then by \irref{dchoiceE} on (A) have $\proves{\Gemp}{\ecase{A}{\ebcons{M}{B}}{\ebcons{M}{C}}}{\dbox{\alpha\cup\beta}{\phi}}$

\mycase \irref{dconsCL}:
From premisses, $\drv_1$ is $\proves{\Gemp}{\edcons{\ecase{A}{B}{C}}{N}}{\ddiamond{\ptest{\psi}}{\phi}}$ and by inversion
(1) $\proves{\Gemp}{\ecase{A}{B}{C}}{\psi}$.
(A) $\proves{\Gemp}{A}{\zeta\lor\rho}$.
(B) $\proves{\pvl:\zeta}{B}{\psi}$.
(C) $\proves{\pvr:\rho}{C}{\psi}$.
Then by \irref{dtestI} and weakening on $N$ have 
(B1) $\proves{\pvl:\zeta}{\edcons{B}{N}}{\ddiamond{\ptest{\psi}}{\phi}}$.
(C1) $\proves{\pvr:\rho}{\edcons{C}{N}}{\ddiamond{\ptest{\psi}}{\phi}}$,
then by \irref{dchoiceE} on (A) have $\proves{\Gemp}{\ecase{A}{\edcons{B}{N}}{\edcons{C}{N}}}{\dbox{\ptest{\psi}}{\phi}}$

\mycase \irref{dconsCR}:
From premisses, $\drv_1$ is $\proves{\Gemp}{\edcons{M}{\ecase{A}{B}{C}}}{\ddiamond{\ptest{\psi}}{\phi}}$ and by inversion
(1) $\proves{\Gemp}{\ecase{A}{B}{C}}{\psi}$.
(A) $\proves{\Gemp}{A}{\zeta\lor\rho}$.
(B) $\proves{\pvl:\zeta}{B}{\psi}$.
(C) $\proves{\pvr:\rho}{C}{\psi}$.
Then by \irref{dtestI} and weakening on $M$ have 
(B1) $\proves{\pvl:\zeta}{\edcons{M}{B}}{\ddiamond{\ptest{\psi}}{\phi}}$.
(C1) $\proves{\pvr:\rho}{\edcons{M}{C}}{\ddiamond{\ptest{\psi}}{\phi}}$,
then by \irref{dchoiceE} on (A) have $\proves{\Gemp}{\ecase{A}{\edcons{M}{B}}{\edcons{M}{C}}}{\dbox{\ptest{\psi}}{\phi}}$

\mycase \irref{projLC}:
From premisses, $\drv_1$ is $\proves{\Gemp}{\eprojL{\ecase{A}{B}{C}}}{\psi}$ and by inversion
(1) $\proves{\Gemp}{\ecase{A}{B}{C}}{\ddiamond{\ptest{\psi}}{\phi}}$.
(A) $\proves{\Gemp}{A}{\zeta\lor\rho}$.
(B) $\proves{\pvl:\zeta}{B}{\ddiamond{\ptest{\psi}}{\phi}}$.
(C) $\proves{\pvr:\rho}{C}{\ddiamond{\ptest{\psi}}{\phi}}$.
Then by \irref{dtestEL} have 
(B1) $\proves{\pvl:\zeta}{\eprojL{B}}{\psi}$.
(C1) $\proves{\pvr:\rho}{\eprojL{C}}{\psi}$,
then by \irref{dchoiceE} on (A) have $\proves{\Gemp}{\ecase{A}{\eprojL{B}}{\eprojL{C}}}{\psi}$

\mycase \irref{projRC}:
From premisses, $\drv_1$ is $\proves{\Gemp}{\eprojR{\ecase{A}{B}{C}}}{\phi}$ and by inversion
(1) $\proves{\Gemp}{\ecase{A}{B}{C}}{\ddiamond{\ptest{\psi}}{\phi}}$.
(A) $\proves{\Gemp}{A}{\zeta\lor\rho}$.
(B) $\proves{\pvl:\zeta}{B}{\ddiamond{\ptest{\psi}}{\phi}}$.
(C) $\proves{\pvr:\rho}{C}{\ddiamond{\ptest{\psi}}{\phi}}$.
Then by \irref{dtestER} have 
(B1) $\proves{\pvl:\zeta}{\eprojR{B}}{\phi}$.
(C1) $\proves{\pvr:\rho}{\eprojR{C}}{\phi}$,
then by \irref{dchoiceE} on (A) have $\proves{\Gemp}{\ecase{A}{\eprojR{B}}{\eprojR{C}}}{\phi}$

\mycase \irref{injLC}:
From premisses, $\drv_1$ is $\proves{\Gemp}{\einjL{\ecase{A}{B}{C}}}{\ddiamond{\alpha\cup\beta}{\phi}}$ and by inversion
(1) $\proves{\Gemp}{\ecase{A}{B}{C}}{\ddiamond{\alpha}{\phi}}$.
(A) $\proves{\Gemp}{A}{\zeta\lor\rho}$.
(B) $\proves{\pvl:\zeta}{B}{\ddiamond{\alpha}{\phi}}$.
(C) $\proves{\pvr:\rho}{C}{\ddiamond{\alpha}{\phi}}$.
Then by \irref{dchoiceIL} have 
(B1) $\proves{\pvl:\zeta}{\einjL{B}}{\ddiamond{\alpha\cup\beta}{\phi}}$.
(C1) $\proves{\pvr:\rho}{\einjL{C}}{\ddiamond{\alpha\cup\beta}{\phi}}$,
then by \irref{dchoiceE} on (A) have $\proves{\Gemp}{\ecase{A}{\einjL{B}}{\einjL{C}}}{\ddiamond{\alpha\cup\beta}{\phi}}$

\mycase \irref{injRC}:
From premisses, $\drv_1$ is $\proves{\Gemp}{\einjR{\ecase{A}{B}{C}}}{\ddiamond{\alpha\cup\beta}{\phi}}$ and by inversion
(1) $\proves{\Gemp}{\ecase{A}{B}{C}}{\ddiamond{\beta}{\phi}}$.
(A) $\proves{\Gemp}{A}{\zeta\lor\rho}$.
(B) $\proves{\pvl:\zeta}{B}{\ddiamond{\beta}{\phi}}$.
(C) $\proves{\pvr:\rho}{C}{\ddiamond{\beta}{\phi}}$.
Then by \irref{dchoiceIR} have 
(B1) $\proves{\pvl:\zeta}{\einjR{B}}{\ddiamond{\alpha\cup\beta}{\phi}}$.
(C1) $\proves{\pvr:\rho}{\einjR{C}}{\ddiamond{\alpha\cup\beta}{\phi}}$,
then by \irref{dchoiceE} on (A) have $\proves{\Gemp}{\ecase{A}{\einjR{B}}{\einjR{C}}}{\ddiamond{\alpha\cup\beta}{\phi}}$

\mycase \irref{repC}:
From premisses, $\drv_1$ is $\proves{\Gemp}{\erep{\ecase{A}{B}{C}}{N}{\pvx:J}{O}}{\dbox{\prepeat{\alpha}}{\phi}}$ and by inversion
(1) $\proves{\Gemp}{\ecase{A}{B}{C}}{J}$.
(A) $\proves{\Gemp}{A}{\zeta\lor\rho}$.
(B) $\proves{\pvl:\zeta}{B}{J}$.
(C) $\proves{\pvr:\rho}{C}{J}$.
Then by \irref{bloopI} and because $N$ is being applied in the same context as before have 
(B1) $\proves{\pvl:\zeta}{\erep{B}{N}{\pvx:J}{O}}{\dbox{\prepeat{\alpha}}{\phi}}$.
(C1) $\proves{\pvr:\rho}{\erep{C}{N}{\pvx:J}{O}}{\dbox{\prepeat{\alpha}}{\phi}}$,
then by \irref{dchoiceE} on (A) have $\proves{\Gemp}{\ecase{A}{\erep{B}{N}{\pvx:J}{O}}{\erep{C}{N}{\pvx:J}{O}}}{\dbox{\prepeat{\alpha}}{\phi}}$

\mycase \irref{appCL}
From premisses, $\drv_1$ is $\proves{\Gemp}{\eapp{\ecase{A}{B}{C}}{N}}{\phi}$ and by inversion
(1) $\proves{\Gemp}{\ecase{A}{B}{C}}{\psi \limply \phi}$.
(A) $\proves{\Gemp}{A}{\zeta\lor\rho}$.
(B) $\proves{\pvl:\zeta}{B}{\psi \limply \phi}$.
(C) $\proves{\pvr:\rho}{C}{\psi \limply \phi}$.
Then by \irref{btestE} and weakening on $N$ have 
(B1) $\proves{\pvl:\zeta}{\eapp{B}{N}}{\phi}$.
(C1) $\proves{\pvr:\rho}{\eapp{C}{N}}{\phi}$,
then by \irref{dchoiceE} on (A) have $\proves{\Gemp}{\ecase{A}{\eapp{B}{N}}{\eapp{C}{N}}}{\phi}$

\mycase \irref{appCR}:
From premisses, $\drv_1$ is $\proves{\Gemp}{\eapp{M}{\ecase{A}{B}{C}}}{\phi}$ and by inversion
(1) $\proves{\Gemp}{\ecase{A}{B}{C}}{\psi}$.
(A) $\proves{\Gemp}{A}{\zeta\lor\rho}$.
(B) $\proves{\pvl:\zeta}{B}{\psi}$.
(C) $\proves{\pvr:\rho}{C}{\psi}$.
Then by \irref{btestE} and weakening on $M$ have 
(B1) $\proves{\pvl:\zeta}{\eapp{M}{B}}{\phi}$.
(C1) $\proves{\pvr:\rho}{\eapp{M}{C}}{\phi}$,
then by \irref{dchoiceE} on (A) have $\proves{\Gemp}{\ecase{A}{\eapp{M}{B}}{\eapp{M}{C}}}{\phi}$

\mycase \irref{dswapC}:
From premisses, $\drv_1$ is $\proves{\Gemp}{\edswap{\ecase{A}{B}{C}}}{\ddiamond{\pdual{\alpha}}{\phi}}$ and by inversion
(1) $\proves{\Gemp}{\ecase{A}{B}{C}}{\dbox{\alpha}{\phi}}$.
(A) $\proves{\Gemp}{A}{\zeta\lor\rho}$.
(B) $\proves{\pvl:\zeta}{B}{\dbox{{\beta}}{\phi}}$.
(C) $\proves{\pvr:\rho}{C}{\dbox{{\beta}}{\phi}}$.
Then by \irref{dualI} have 
(B1) $\proves{\pvl:\zeta}{\edswap{B}}{\ddiamond{\pdual{\alpha}}{\phi}}$.
(C1) $\proves{\pvr:\rho}{\edswap{C}}{\ddiamond{\pdual{\alpha}}{\phi}}$,
then by \irref{dchoiceE} on (A) have $\proves{\Gemp}{\ecase{A}{\edswap{B}}{\edswap{C}}}{\ddiamond{\pdual{\alpha}}{\phi}}$

\mycase \irref{bswapC}:
From premisses, $\drv_1$ is $\proves{\Gemp}{\ebswap{\ecase{A}{B}{C}}}{\dbox{\pdual{\alpha}}{\phi}}$ and by inversion
(1) $\proves{\Gemp}{\ecase{A}{B}{C}}{\ddiamond{\alpha}{\phi}}$.
(A) $\proves{\Gemp}{A}{\zeta\lor\rho}$.
(B) $\proves{\pvl:\zeta}{B}{\ddiamond{\beta}{\phi}}$.
(C) $\proves{\pvr:\rho}{C}{\ddiamond{\beta}{\phi}}$.
Then by \irref{dualI} have 
(B1) $\proves{\pvl:\zeta}{\ebswap{B}}{\dbox{\pdual{\alpha}}{\phi}}$.
(C1) $\proves{\pvr:\rho}{\ebswap{C}}{\dbox{\pdual{\alpha}}{\phi}}$,
then by \irref{dchoiceE} on (A) have $\proves{\Gemp}{\ecase{A}{\ebswap{B}}{\ebswap{C}}}{\dbox{\pdual{\alpha}}{\phi}}$

\mycase \irref{brandomC}:
From premisses, $\drv_1$ is $\proves{\Gemp}{\eapp{\ecase{A}{B}{C}}{f}}{\tsub{\phi}{x}{f}}$ and by inversion
(1) $\proves{\Gemp}{\ecase{A}{B}{C}}{\dbox{\prandom{x}}{\phi}}$.
(A) $\proves{\Gemp}{A}{\zeta\lor\rho}$.
(B) $\proves{\ell:\zeta}{B}{\dbox{\prandom{x}}{\phi}}$.
(C) $\proves{r:\rho}{C}{\dbox{\prandom{x}}{\phi}}$.
Then by \irref{brandomE} have
(B1) $\proves{\ell:\zeta}{\eapp{B}{f}}{\tsub{\phi}{x}{f}}$.
(C1) $\proves{r:\rho}{\eapp{B}{f}}{\tsub{\phi}{x}{f}}$,
then by \irref{dchoiceE} on (A) have $\proves{\Gemp}{\ecase{A}{\eapp{B}{f}}{\eapp{C}{f}}}{\tsub{\phi}{x}{f}}$

\mycase \irref{monC}:
From premisses, $\drv_1$ is $\proves{\Gemp}{\emon{\ecase{A}{B}{C}}{N}{\pvx}}{\ddiamond{\alpha}{\psi}}$ and by inversion
(1) $\proves{\Gemp}{\ecase{A}{B}{C}}{\ddiamond{\alpha}{\phi}}$.
(A) $\proves{\Gemp}{A}{\zeta\lor\rho}$.
(B) $\proves{\pvl:\zeta}{B}{\ddiamond{\alpha}{\phi}}$.
(C) $\proves{\pvr:\rho}{C}{\ddiamond{\alpha}{\phi}}$.
Then by \irref{mon} and by weakening $N$ with $\eren{\zeta}{x}{z}$ and $\eren{\rho}{x}{z}$, then
(B1) $\proves{\pvl:\zeta}{\emon{B}{N}{\pvx}}{\ddiamond{\alpha}{\psi}}$.
(C1) $\proves{\pvr:\rho}{\emon{C}{N}{\pvx}}{\ddiamond{\alpha}{\psi}}$,
then by \irref{dchoiceE} on (A) have $\proves{\Gemp}{\ecase{A}{\emon{B}{N}{\pvx}}{\emon{C}{N}{\pvx}}}{\ddiamond{\alpha}{\psi}}$.

\mycase \irref{fpC}:
From premisses, $\drv_1$ is $\proves{\Gemp}{\efp{\ecase{A}{B}{C}}{D}{E}}{\psi}$ and by inversion
(1) $\proves{\Gemp}{\ecase{A}{B}{C}}{\ddiamond{\prepeat{\alpha}}{\phi}}$.
(A) $\proves{\Gemp}{A}{\zeta\lor\rho}$.
(B) $\proves{\pvl:\zeta}{B}{\ddiamond{\prepeat{\alpha}}{\phi}}$.
(C) $\proves{\pvr:\rho}{C}{\ddiamond{\prepeat{\alpha}}{\phi}}$.
Then by \irref{dloopE} and since $D$ and $E$ are applied with the same context as before, then
(B1) $\proves{\pvl:\zeta}{\efp{B}{D}{E}}{{\psi}}$.
(C1) $\proves{\pvr:\rho}{\efp{C}{D}{E}}{{\psi}}$,
then by \irref{dchoiceE} on (A) have $\proves{\Gemp}{\ecase{A}{\efp{B}{D}{E}}{\efp{C}{D}{E}}}{\psi}$.

\mycase \irref{caseC}:
From premisses, $\drv_1$ is $\proves{\Gemp}{\ecase{\ecase{A}{B}{C}}{D}{E}}{\rho}$ and by inversion
(1) $\proves{\Gemp}{\ecase{A}{B}{C}}{\phi\lor\psi}$.
(A) $\proves{\Gemp}{A}{\zeta\lor\rho}$.
(B) $\proves{\pvl:\zeta}{B}{\phi\lor\psi}$.
(C) $\proves{\pvr:\rho}{C}{\phi\lor\psi}$.
Then by \irref{dchoiceE} and since $D$ and $E$ are applied with the same context as before (that is, we allow $\pvl$ and $\pvr$ to shadow here as the outer binding is never used), then
(B1) $\proves{\pvl:\zeta}{\ecase{B}{D}{E}}{\rho}$. 
(C1) $\proves{\pvr:\rho}{\ecase{C}{D}{E}}{\rho}$,
then by \irref{dchoiceE} on (A) have $\proves{\Gemp}{\ecase{A}{\ecase{B}{D}{E}}{\ecase{C}{D}{E}}}{\rho}$.

\mycase \irref{unpackC}:
From premisses, $\drv_1$ is $\proves{\Gemp}{\eunpack{\ecase{A}{B}{C}}{N}}{\psi}$ and by inversion
(1) $\proves{\Gemp}{\ecase{A}{B}{C}}{\ddiamond{\prandom{x}}{\phi}}$.
(A) $\proves{\Gemp}{A}{\zeta\lor\rho}$.
(B) $\proves{\pvl:\zeta}{B}{\ddiamond{\prandom{x}}{\phi}}$.
(C) $\proves{\pvr:\rho}{C}{\ddiamond{\prandom{x}}{\phi}}$.
Then by \irref{drandomE} and by weakening $N$ with $\eren{\zeta}{x}{y}$ and $\eren{\rho}{x}{y}$
(B1) $\proves{\pvl:\zeta}{\eunpack{B}{N}}{\psi}$.
(C1) $\proves{\pvr:\rho}{\eunpack{C}{N}}{\psi}$,
then by \irref{dchoiceE} on (A) have $\proves{\Gemp}{\ecase{A}{\eunpack{B}{N}}{\eunpack{C}{N}}}{\psi}$.

\mycase \irref{forC}:
From premisses, $\drv_1$ is $\proves{\Gemp}{\efor{\ecase{A}{B}{C}}{N}{O}}{\ddiamond{\prepeat{\alpha}}{(\met =  \metz \land \conv)}}$ and by inversion
(1) $\proves{\Gemp}{\ecase{A}{B}{C}}{\conv}$.
(A) $\proves{\Gemp}{A}{\zeta\lor\rho}$.
(B) $\proves{\Gemp,\ell:\zeta}{B}{\conv}$.
(C) $\proves{r:\rho}{C}{\conv}$.
Then by \irref{dloopI} and because $N$ is applied only in its original context,
(B1) $\proves{\ell:\zeta}{\efor{B}{N}{O}}{\ddiamond{\prepeat{\alpha}}{\phi}}$.
(C1) $\proves{r:\rho}{\efor{C}{N}{O}}{\ddiamond{\prepeat{\alpha}}{\phi}}$,
then by \irref{dchoiceE} on (A) have $\proves{\Gemp}{\ecase{A}{\efor{B}{N}{O}}{\efor{C}{N}{O}}}{\ddiamond{\prepeat{\alpha}}{\phi}}$.

\mycase \irref{bunrollC}:
From premisses, $\drv_1$ is $\proves{\Gemp}{\ebunroll{\ecase{A}{B}{C}}}{\phi\land\dbox{\alpha}{\dbox{\prepeat{\alpha}}{\phi}}}$ and by inversion
(1) $\proves{\Gemp}{\ecase{A}{B}{C}}{\dbox{\prepeat{\alpha}}{\phi}}$.
(A) $\proves{\Gemp}{A}{\zeta\lor\rho}$.
(B) $\proves{\ell:\zeta}{B}{\dbox{\prepeat{\alpha}}{\phi}}$.
(C) $\proves{r:\rho}{C}{\dbox{\prepeat{\alpha}}{\phi}}$.
Then by \irref{bunroll},
(B1) $\proves{\ell:\zeta}{\ebunroll{B}}{\phi\land\dbox{\alpha}{\dbox{\prepeat{\alpha}}{\phi}}}$.
(C1) $\proves{r:\rho}{\ebunroll{C}}{\phi\land\dbox{\alpha}{\dbox{\prepeat{\alpha}}{\phi}}}$,
then by \irref{dchoiceE} on (A) have $\proves{\Gemp}{\ecase{A}{\ebunroll{B}}{\ebunroll{C}}}{\phi\land\dbox{\alpha}{\dbox{\prepeat{\alpha}}{\phi}}}$.

\mycase \irref{caseC}:
From premisses, $\drv_1$ is $\proves{\Gemp}{\ercase{\ecase{A}{B}{C}}{D}{E}}{\psi}$ and by inversion
(1) $\proves{\Gemp}{\ecase{A}{B}{C}}{\ddiamond{\prepeat{\alpha}}{\phi}}$.
(A) $\proves{\Gemp}{A}{\zeta\lor\rho}$.
(B) $\proves{\Gemp,\ell:\zeta}{B}{\ddiamond{\prepeat{\alpha}}{\phi}}$.
(C) $\proves{r:\rho}{C}{\ddiamond{\prepeat{\alpha}}{\phi}}$.
Then by \irref{drcase},
(B1) $\proves{\ell:\zeta}{\ercase{B}{D}{E}}{\psi}$.
(C1) $\proves{r:\rho}{\ercase{C}{D}{E}}{\psi}$,
then by \irref{dchoiceE} on (A) have $\proves{\Gemp}{\ecase{A}{\ercase{B}{D}{E}}{\ercase{C}{D}{E}}}{\psi}$.
\end{proof}

\end{document}